\documentclass[11pt,letterpaper]{article} 
\usepackage[margin=1in]{geometry} 

\usepackage[T1]{fontenc}
\usepackage{natbib}
\usepackage{comment}
\usepackage{bm, bbm, amsmath, amssymb, amsthm, mathtools}
\usepackage{algorithm, algpseudocode}
\usepackage{enumitem}
\usepackage{float}

\usepackage{graphicx}
\graphicspath{{figures}}

\usepackage{hyperref}
\hypersetup{
	colorlinks=true,
	citecolor=black,
	linkcolor=black,
}

\newcommand{\mc}[1]{\mathcal{#1}}
\newcommand{\bvec}{\bm{\beta}}
\newcommand{\dvec}{\bm{\delta}}
\newcommand{\bvecOR}{\hat \bvec^o}
\newcommand{\Mset}{\mathcal{M}}
\newcommand{\Sset}{\mathcal{S}}

\newcommand{\MSset}{\overline \Mset}
\newcommand{\norm}[1]{\left\lVert#1\right\rVert}
\newcommand{\shortnorm}[1]{\lVert#1\rVert}
\newcommand{\Cmat}{\bm C}

\newcommand{\reals}{\mathbb R}
\newcommand{\defeq}{\vcentcolon=}
\newcommand{\avec}{\bm a}
\newcommand{\xvec}{\bm x}
\newcommand{\xmat}{\bm X}
\newcommand{\yvec}{\bm y}

\newcommand{\hatmat}[1]{\widehat {\bm #1} }
\newcommand{\hatvec}[1]{\hat {\bm #1} }
\newcommand{\empmean}{\mathbb{E}_n}
\newcommand{\empproc}{\mathbb{G}_n}

\newcommand{\zvec}{\bm z}
\newcommand{\prox}{\operatorname{Prox}}
\newcommand{\shrink}{\operatorname{Shrink}}
\newcommand{\checkloss}{\rho_\tau}

\newcommand{\eps}{\varepsilon}

\DeclareMathOperator{\E}{E}

\DeclareMathOperator{\tr}{tr}

\DeclareMathOperator*{\argmin}{arg \, min}
\DeclareMathOperator*{\argmax}{arg \, max}

\theoremstyle{plain}

\newtheorem{theorem}{Theorem}[section]

\newtheorem{corollary}[theorem]{Corollary}
\newtheorem{remark}{Remark}
\newtheorem{lemma}[theorem]{Lemma}

\allowdisplaybreaks

\begin{document}	

	\title{\bf Global Inference for High Dimensional Quantile Regression with Heterogeneity}
	
	\author{Tate Jacobson%
		\hspace{.2cm}\\
		Department of Statistics and Data Science, Oregon State University,\vspace{0.1cm} \\
		Daeyoung Ham \\
		Department of Statistics and Data Science, University of Texas at San Antonio, \vspace{0.1cm}\\
		and \\
		Ben Sherwood \\ School of Business, The University of Kansas}
	
	\maketitle
	
	\begin{abstract}
		Quantile regression has received considerable attention in high dimensional methodology due to its robustness and ability to accommodate heteroskedasticity.
		Despite this attention, current inferential methods for high dimensional quantile regression have significant limitations---only testing significance, testing at a single quantile at a time, or assuming homoskedasticity. 
		To generalize the inferences that can be drawn in high dimensional quantile regression, we introduce partial penalized Wald and rank score tests for testing linear hypotheses uniformly over a set of quantiles $\mc T$ under a heteroskedastic model.
		We propose globally adaptive local linear approximation (LLA) algorithms for folded concave penalized estimation and show that our global LLA algorithms recover ``testing oracles'' for the full and reduced models with high probability. 
		We then show that our global LLA estimators achieve the $\sqrt{s/n}$ oracle estimation rate uniformly over $\mc T$, establish Bahadur representations, and derive asymptotic distributions for the test statistics under the null and local alternatives.
		We develop a dual formulation incorporating adaptive penalty weights for the rank scores and use the strong oracle property to establish their oracle equivalence. Using empirical process techniques, we derive reference distributions based on suprema of squared Bessel processes, extending partial penalized inference beyond pointwise chi-square calibration.
	\end{abstract}
	
	
	\section{Introduction}
	In recent decades ultra-high dimensional data, in which the number of predictors $p$ grows exponentially with the number of observations $n$, have become commonplace in domains from economics to genomics.
	The growing availability of these data has spurred the development of statistical methods for estimation and inference in ultra-high dimensions.
	With extensions of linear regression having been thoroughly explored \citep{fan_strong_2014, shi2019linear}, recent efforts have shifted to developing robust methods for this setting. 
	Quantile regression \citep{koenker_regression_1978} has been a leading method in these developments due to its well-known robustness to outliers and ability to model heterogeneity across quantiles.
	Under the quantile regression model, the $\tau$th conditional quantile of the response $Y$ given the predictors $\xvec = (1, x_1, \ldots, x_p)^T$ satisfies
	$
		Q_{y}(\tau|\xvec) = \xvec^T \bvec^*(\tau)
	$
	for some $\bvec^*(\cdot): (0,1) \mapsto \reals^{p+1}$.
	Under this model, both the central tendency and the spread of the conditional distribution of $Y$ can change with $\xvec$, with changes in the spread captured by changes in $\bvec^*(\cdot)$ across quantiles.

	Much research in high dimensional quantile regression has focused on modeling $\bvec^*(\tau)$ at a single quantile $\tau$, an approach referred to as \textit{locally concerned} quantile regression \citep{zheng2015globally}.
	Several recent works \citep{wangWuLi2012, fan_strong_2014, tan2022high} have studied the estimation and selection properties of folded-concave penalized quantile regression estimators at a single $\tau$, deriving their estimation rates and showing that they recover the local oracle estimator at $\tau$ with high probability.
	For inference at a single $\tau$, \citet{yan2023confidence} have introduced confidence intervals and significance tests for individual regression coefficients based on a debiased convolution smoothed estimator, \citet{yang_communication-efficient_2025} have proposed a partial penalized Wald test for distributed data, and \citet{zhou_testing_2026} have introduced a decorrelated score test for testing a high dimensional coefficient vector.
	
	Another line of research models the coefficient process $\bvec^*(\cdot)$ over a set of quantiles $\mc T \subset (0,1)$, an approach known as \textit{globally concerned} quantile regression \citep{zheng2015globally}.	
	Within this paradigm, \citet{BC.L1penalized} studied the $\ell_1$-penalized quantile regression estimator in the sparse ultra-high dimensional setting where $s \ll n$ predictors have nonzero effects over $\mc T$;
	they showed that the $\ell_1$-penalized estimator achieves an estimation rate of $\sqrt{s \log p/n}$ uniformly over $\mc T$ and contains the true model as a submodel with high probability.
	Further developing the globally concerned approach, \citet{zheng2015globally} introduced globally adaptive $\ell_1$ penalized quantile regression estimators and proved that they recover a regularized oracle estimator with high probability and thereby achieve a near-oracle rate of $\sqrt{s \log n/n}$ uniformly over $\mc T$.
	Despite these advances, recovering the true (unregularized) oracle over $\mc T$ and achieving the oracle rate of $\sqrt{s/n}$ uniformly over $\mc T$ have remained open problems.
	
	While some research has examined the problem of making global inferences about $\bvec^*(\cdot)$ over $\mc T$ in high dimensions, current methods are limited in the range of hypotheses they can test or by restrictive model assumptions.
	\citet{belloni_valid_2019} proposed using orthogonal score functions to construct confidence bands for multiple quantile regression coefficients $\bvec_{\Mset}^*(\cdot)$, $\Mset \subset \{1, \ldots, p\}$, over $\mc T$.
	However, their procedure cannot be used to test more complicated hypotheses, such as tests of contrasts, and has not been studied in the setting where the number of parameters being tested diverges.
	\citet{bradic2017uniform} developed a debiased test for testing linear hypotheses---that is, any hypothesis of the form $\text{H}_0: \Cmat \bvec_{\mc M} (\tau) = \bm \Psi(\tau)$---over $\mc T$, but imposed the restrictive assumption that the errors $\eps_i = Y_i -  \xvec_i^T \bvec^*(\tau)$ are i.i.d. and therefore homoskedastic.

	In this paper, we introduce globally concerned estimators and hypothesis tests for the quantile regression process $\bvec^*(\cdot)$ over $\mc T \in (0,1)$ in ultra-high dimensions which improve on the selection and estimation properties of existing estimators and expand the range of hypotheses that can be tested under a heteroskedastic quantile regression model.
	Our estimators combine two key elements: a novel globally adaptive local linear approximation (LLA) algorithm for folded-concave penalized quantile regression over $\mc T$ and partial penalization.
	We use these estimators to develop partial penalized Wald and rank scores tests for testing global linear hypotheses, which take the form $\text{H}_0: \Cmat \bvec_{\mc M} (\tau) = \bm \Psi(\tau)$ for $\tau \in \mc T$.
	For the theoretical study, we introduce ``testing oracles'' for the full and reduced models under the null and prove that our globally adaptive, partial penalized LLA estimators recover the corresponding testing oracles with high probability---a result often termed the ``strong oracle property'' \citep{fan_strong_2014}.
	Leveraging this result, we establish Bahadur representations for the full and reduced model global LLA estimators and derive the asymptotic distributions of our test statistics under the null and local alternatives in the setting where the number of parameters under test and the number of nonzero parameters in the true model can grow with $n$.
	
	We highlight the following contributions of this work:
	(1) We introduce the first high dimensional tests for linear hypotheses over a set of quantiles $\mc T \subset (0,1)$ under a heteroskedastic quantile regression model (see Table \ref{tbl:checklist} for a comparison of our method with other inferential procedures for high dimensional quantile regression).
	(2) We show that our global LLA estimators recover the (unregularized) oracle estimator and achieve the $\sqrt{s/n}$ oracle estimation rate uniformly over $\mc T$, improving on the $\sqrt{s\log n/n}$ rate established in \citet{zheng2015globally}.
	(3) We formulate our novel partial penalized rank scores through the dual of the constrained, adaptively weighted LLA problem \eqref{eqn:rank scores}. Leveraging the global strong oracle property, we establish their equivalence to the oracle rank scores uniformly over $\mc T$ with probability tending to one.
	(4) In the course of studying our partial penalized tests, we generalize the theory for the classical Wald and rank score tests for quantile regression, moving beyond the location-scale model assumed in prior studies \citep{gutenbrunner1992regression, koenker_goodness_1999,  koenker_quantile_2005} and allowing the number of parameters being tested to diverge.
	(5) For a fixed number of restrictions, we use empirical process techniques to derive asymptotic distributions based on suprema of squared Bessel processes under the null hypothesis and local alternatives (Theorem \ref{thm:part pen bessel}). This extends partial penalized inference beyond pointwise chi-square calibration to a process-level reference distribution for global testing over $\mc T$.
	(6) We establish consistent variable selection for our global LLA estimators when the penalty parameter is tuned with a generalized information criterion and show that the estimation rates and Bahadur representations for our global LLA estimators still hold when the coefficient process is approximated by a c\`adl\`ag step function.
	
	\begin{table}
		\centering
		\resizebox{\textwidth}{!}{
		\begin{tabular}{lcccc}
			\hline
			& Global in $\mathcal{T}$ 
			& Heteroskedastic
			& H$_0$: $\Cmat \bvec_{\mathcal{M}}^*(\tau) = \bm \Psi(\tau)$ 
			& $|\mathcal{M}|,s \to \infty$ \\
			\hline
			\textbf{Global LLA} (Ours) &  \checkmark & \checkmark  & \checkmark  & \checkmark   \\
			\citet{bradic2017uniform} &   \checkmark &  & \checkmark &  \checkmark  \\
			\citet{belloni_valid_2019} &  \checkmark  & \checkmark &  &   \\
			\citet{yan2023confidence} &  & \checkmark  &  &   \\
			\citet{yang_communication-efficient_2025} &  & \checkmark & \checkmark &   \\
			\hline
		\end{tabular}}
		\caption{Comparison of inferential procedures for high dimensional quantile regression} \label{tbl:checklist}
	\end{table}
	
	In simulations, we find that our global LLA estimators achieve superior estimation and selection performance to $\ell_1$-penalized quantile regression \citep{BC.L1penalized} and adaptive $\ell_1$-penalized quantile regression \citep{zheng2015globally}.
	We compare our partial penalized Wald and rank score tests with the debiased tests of \citet{bradic2017uniform} and \citet{yan2023confidence} and find that they achieve higher power both theoretically and empirically in simulations.
	We demonstrate the use of our partial penalized tests in an analysis of the effects of lung disease and smoking on lung capacity.

	\subsection*{Outline}
	The rest of the paper is organized as follows.
	In Section \ref{sec:glla estimation} we introduce our partial penalized estimators and the globally concerned LLA algorithms we use to compute them.
	In Section \ref{sec:estimation theory} we introduce ``testing oracles'' for $\text{H}_0: \Cmat \bvec_{\mc M} (\tau) = \bm \Psi(\tau)$ over $\tau \in \mc T$, prove that our globally concerned LLA algorithms converge to the testing oracles with probability converging to $1$, and derive estimation rates for our global LLA estimators.
	In Sections \ref{sec:partial penalized tests} and \ref{sec:inference theory} we introduce partial penalized Wald and rank score tests for testing $\text{H}_0: \Cmat \bvec_{\mc M} (\tau) = \bm \Psi(\tau)$ over $\tau \in \mc T$ and establish their asymptotic properties.
	We discuss penalty parameter tuning and further implementation details in Section \ref{sec:implementation details}.
	We conduct simulation studies in Section \ref{sec:sim study} and apply our partial penalized tests in a study of lung capacity in Section \ref{sec:real data} .

	\subsubsection*{Notation}
	For a matrix $\bm A$ and index set $\mc B$ we let $\bm A_{(\mc B)}$ denote submatrix with row indices in $\mc B$ and  $\bm A_{\mc B}$ denote submatrix with column indices in $\mc B$.
	For a square matrix $\bm A \in \reals^{m \times m}$, we let $\lambda_{\max}(\bm A)$ and $\lambda_{\min}(\bm A)$ denote its largest and smallest eigenvalues, respectively.
	If $\bm A$ is positive semi-definite with eigen-decomposition $\bm A = \bm V \bm D \bm V^T$, we define $\bm A^{1/2} = \bm V \bm D^{1/2} \bm V^T$ and $\bm A^{-1/2} = \bm V \bm D^{-1/2} \bm V^T$.
	We use $\bm I_d$ to denote the $d\times d$ identity matrix.
	We use $\norm{\bm v}_2 = \sqrt {\sum v_i^2 } $ to denote the $\ell_2$ norm of a vector $\bm v$.
	
	For deterministic sequences $a_n,b_n$ we use $a_n = O(b_n)$ to denote that $a_n/b_n$ is bounded and $a_n = o(b_n)$ to denote that $a_n/b_n \to 0$. 
	Likewise, for random sequences $X_n, Y_n$ we use $X_n = O_p(Y_n)$ to denote that $X_n/Y_n$ is bounded in probability and $X_n = o(Y_n)$ to denote that $X_n/Y_n$ converges in probability to $0$.
	We use $C, C_1, \ldots, c, c_1, \ldots$ to denote fixed constants. These symbols may be used repeatedly and have their meanings change from line to line.	
	
	\section{Global Local Linear Approximation and Partial Penalization} \label{sec:glla estimation}
	
	\subsection{Testing the quantile process over $\mc T$}
	Let $Q_{y}(\tau|\xvec) \defeq \inf\{y: \Pr(Y \leq y|\xvec) \geq \tau \} $ denote the $\tau$th conditional quantile of $Y$ given the predictors $\xvec = (1, x_1, \ldots, x_p)^T$.
	Throughout this paper we assume
	\begin{equation}
		Q_{y}(\tau|\xvec) = \xvec^T \bvec^*(\tau) \text{ \; for \; } \tau \in \mc T \label{eqn:conditional quantile linear expr}
	\end{equation}
	where $\bvec^*(\tau) = (\beta_0^*(\tau), \ldots, \beta_p^*(\tau))^T$ is a vector of unknown coefficient functions and $\mc T \subseteq [\delta,1 - \delta]$, $0 < \delta < 1/2$, is a set of quantiles of interest. 
	Note that under \eqref{eqn:conditional quantile linear expr} the variance of the conditional distribution can vary with $\xvec$.
	The set $\mc T$ can be a countable set, such as $\{0.25, 0.5, 0.75\}$, or an uncountable set, such as the interval $[0.25, 0.75]$.
	
	Our interest lies in testing linear hypotheses, which are of the form
	\begin{equation}
		\text{H}_0: \Cmat \bvec^*_{\Mset}(\tau) = \bm \Psi (\tau) \text{\; for \;} \tau \in \mc T, \label{hyp:linear hypothesis}
	\end{equation}
	where $\Mset \subseteq \{0, \ldots, p\}$ is a set of parameters of interest with $m \defeq |\Mset|$, $\Cmat \in \reals^{r \times m}$ is a constraint matrix with full row rank, and $\bm \Psi(\tau): \mc T \to \reals^r$ is a known function.
	This is a rich class of hypotheses, covering many important cases, such as significance tests for multiple parameters for which the null is $\rm{H}_0: \bvec^*_{\Mset}(\tau) = \bm 0$ and tests of contrasts such as $\rm{H}_0:\beta_1^*(\tau) - \beta_2^*(\tau) = 0$.
	 
	We consider $\bvec_{\Mset}(\tau)$ to be the parameters of interest and $\bvec_{\Mset^c}(\tau)$ to be nuisance parameters, which we aim to estimate consistently in order to improve our inferences about $\bvec_{\Mset}(\tau)$.
	We define $\mc A_{\tau} \defeq \{ j : |\beta_j^*(\tau)| > 0 \} $ to be the true support set for $\bvec(\tau)$ at $\tau \in (0,1)$. 
	For the set of quantiles $\mc T$, we define $\mc A_{\mc T} \defeq \{ j : \sup_{\tau \in \mc T}|\beta_j^*(\tau)| > 0 \} $.
	For convenience we define $\Sset_{\mc T} = \mc A_{\mc T} \cap \Mset^c$ to be the support set among the nuisance parameters, such that $\Sset_{\mc T} \cup \Mset = \mc A_{\mc T} \cup \Mset$ and $\Sset_{\mc T}$ and $\Mset$ are disjoint, and define $s \defeq |\Sset_{\mc T}|$.
	We likewise define $\Sset_{\tau} = \mc A_{\tau} \cap \Mset^c$ for $\tau \in (0,1)$.
	We drop the $\mc T$ subscript from $\Sset = \Sset_{\mc T}$ where its meaning is clear from context.
	We define $\MSset = \Mset \cup \Sset$ for ease of presentation.
	We assume, without loss of generality, that $0 \in \Mset$ (one may set $\Cmat_{\cdot 1} = \bm 0$ if no constraint on the intercept is desired) and that the predictors are reordered so that $\bvec(\tau) = (\bvec_{\Mset}^T, \bvec_{\Sset}^T, \bvec_{\MSset^c}^T)^T(\tau)$.
	
	\subsection{Partial penalized estimation}
	Suppose we observe a sample $\{(\xvec_i, y_i)\}_{i = 1}^n$ generated from \eqref{eqn:conditional quantile linear expr}. Let $\rho_{\tau}(u) = u(\tau - I{(u < 0)})$, the check loss function \citep{koenker_regression_1978}. When the total number of predictors $p$ is smaller than $n$, the quantile regression coefficients can be estimated by minimizing the empirical loss: 
	$
		\frac{1}{n} \sum_{i=1}^{n} \rho_{\tau}( y_i - \xvec_i^T\bvec(\tau) ).
	$
		
	Our interest, however, lies in the high dimensional setting where $p$ is larger than $n$.
	In high dimensions the regular quantile regression estimator is not consistent, with the regression signal dominated by noise from uninformative predictors.
	To address these challenges in high dimensions, \citet{BC.L1penalized} proposed the $\ell_1$ penalized quantile regression estimator,
	\begin{equation}
		\hat\bvec^{\rm lasso}(\tau)=\argmin_{\bvec} \frac{1}{n}\sum_{i=1}^n \rho_{\tau}(y_i-\bm x_i^T\bvec) +  \frac{ \lambda_{\rm lasso} \sqrt{\tau (1-\tau)}}{n} \sum_{j=1}^p \hat{\sigma}_j|\beta_j| \label{eqn:BC l1 estimator}
	\end{equation}
	with tuning parameter $\lambda_{\rm lasso}>0$ and penalty weights $\hat \sigma_j^2 = \sum_{i=1}^n x_{ij}^2/n$.
	\citet{BC.L1penalized} showed that the $\ell_1$-penalized estimator is consistent when the true model is \textit{sparse}, meaning that $\sup_{\tau \in \mc T} |\mc A_{\tau}| \leq n$. 
	However, the $\ell_1$-penalized estimator is biased and sets some elements of $\bvec_{\Mset}$ exactly equal to $0$, giving $\hat\bvec^{\rm lasso}_{\Mset}(\tau)$ an intractable asymptotic distribution and hampering inference.
	
	We propose a globally adaptive, partial penalization approach to obtain estimators which are consistent, consistently select the nonzero elements of $\bvec_{\Mset^c}^*(\tau)$, and can be used to make inference about $\bvec^*_{\Mset}(\tau)$ over $\tau \in \mc T$. 
	We define the \textit{full} model estimator for testing \eqref{hyp:linear hypothesis} to be
	\begin{align}\label{eqn:FE}
		\hat\bvec_{a} (\tau) =\argmin_{\bvec } \frac{1}{n}\sum_{i=1}^n \rho_{\tau}(y_i-\bm x_i^T\bvec(\tau))+\sum_{j \in \Mset^c} p_{\lambda}(|\beta_j(\tau)|)
	\end{align}
	and the \textit{reduced} model estimator for testing \eqref{hyp:linear hypothesis} to be
	\begin{align}\label{eqn:RE}
		\hat\bvec_{0} (\tau)= \argmin_{\bvec } \frac{1}{n}\sum_{i=1}^n \rho_{\tau}(y_i-\bm x_i^T\bvec(\tau))+\sum_{j \in \Mset^c} p_{\lambda}(|\beta_j(\tau)|) {\rm \; ~s.t.~} \Cmat \bvec_{\mc M} (\tau) =\bm \Psi(\tau),
	\end{align}
	where $p_{\lambda}(\cdot)$ is a sparsity-inducing penalty.
	
	The defining property of the partial penalized estimators is that they leave the parameters being tested unpenalized.
	This ensures those parameters are included in model and avoids introducing bias in their estimation, giving our tests power when the true signals under test $\{\beta_j^*(\tau) : j \in \Mset\}$ are small.
	
	In Section \ref{sec:partial penalized tests} we introduce partial penalized Wald and rank score tests of \eqref{hyp:linear hypothesis} using these estimators. For the remainder of this section and Section \ref{sec:estimation theory}, however, we will discuss our approach to solving \eqref{eqn:FE} and \eqref{eqn:RE} over $ \mc T$ and study the estimation and variable selection properties of our estimators.
	
	\subsection{Globally concerned local linear approximations}
	To mitigate the bias in estimating the nuisance parameters $\bvec_{\Mset^c}(\tau)$, we use a folded-concave penalty (FCP) in \eqref{eqn:FE} and \eqref{eqn:RE}. 
	An FCP is defined by the following properties: 
	\begin{enumerate}[label=(\roman*)]
		\item $p_\lambda(t)$ is increasing and concave on $t\in[0,\infty)$ with $p_\lambda(0)=0$ \label{as:FCP i}
		\item $p_\lambda(t)$ is differentiable on $t\in(0,\infty)$ with $p_\lambda'(0):=p_\lambda'(0+)\geq d_1\lambda$ \label{as:FCP ii}
		\item $p_\lambda'(t)\geq d_1\lambda$ for $t\in(0,d_2\lambda]$ \label{as:FCP iii}
		\item $p_\lambda'(t)=0$ for $t\in[d\lambda,\infty)$ where $d>d_2$ \label{as:FCP iv}
	\end{enumerate}
	for some positive constants $d$, $d_1$, and $d_2$.
	Two well-known FCPs are the SCAD penalty \citep{Fan2001}, which has derivative
	$
	p_{\lambda}'(t) = \lambda I{(t \leq \lambda)} + \frac{(d\lambda - t)_+}{d - 1} I{(t > \lambda)}
	$
	where $d > 2$, and the MCP \citep{Zhang2010}, which has derivative
	$
	p_{\lambda}'(t) = \left(\lambda - \frac{t}{d}\right)_+
	$
	where $d > 1$.
	For both penalties, $d$ is specified by the user and corresponds to $d$ in the conditions above.
	One can show that the SCAD penalty with $d_1 = d_2 = 1$ and the MCP with $d_1 = 1 - d^{-1}$ and $d_2 = 1$ satisfy conditions \ref{as:FCP i} - \ref{as:FCP iv}.
	
	One complication that arises from using an FCP is that it makes the objective function potentially non-convex.
	As a consequence, \eqref{eqn:FE} and \eqref{eqn:RE} may have multiple local minima.
	To address this, we develop iterative local linear approximation (LLA) algorithms  \citep{Zou2008, fan_strong_2014} to solve \eqref{eqn:FE} and \eqref{eqn:RE}.
	
	Here we must draw a distinction between \textit{locally concerned} and \textit{globally concerned} quantile regression in high dimensions. In locally concerned quantile regression, variable selection is performed independently at each $\tau$.
	However, the set of selected variables may change when $\tau$ changes even slightly, making the locally concerned approach sensitive to the choice of $\tau$.	
	Globally concerned quantile regression avoids this problem by performing variable selection over a set of quantiles $\mc T \subset (0,1)$.
	
	Suppose we have an initial estimate $\tilde\bvec(\tau)$.
	In a locally concerned LLA algorithm for solving \eqref{eqn:FE}, one would approximate the penalty $p_{\lambda}(|\beta_j(\tau)|)$ by the linear approximation at $\tilde\bvec(\tau)$, given by $p_{\lambda}(|\beta_j(\tau)|) \approx p_{\lambda}(|\tilde \beta_j(\tau)|) + p_{\lambda}'(|\tilde \beta_j(\tau)|)|\beta_j(\tau)|$, and solve the weighted $\ell_1$-penalized problem
	$$
	\min_{\bvec} \frac{1}{n}\sum_{i=1}^n \rho_{\tau}(y_i-\bm x_i^T\bvec(\tau)) + \sum_{j \in \Mset^c} p_{\lambda}'(|\tilde \beta_j(\tau)|)|\beta_j(\tau)|
	$$
	separately at each $\tau$ to obtain an updated estimate $\hat \bvec (\tau)$.
	The process is repeated, with $\hat \bvec (\tau)$ being used to compute the LLA weights $p_{\lambda}'(|\hat \beta_j(\tau)|)$ for the next LLA iteration.
	\citet{fan_strong_2014} and \citet{tan2022high} have proposed such locally concerned LLA algorithms for canonical quantile regression and convolution smoothed quantile regression, respectively.
	
	As our interest lies in testing \eqref{hyp:linear hypothesis} over $\mc T$, we propose globally concerned LLA algorithms for \eqref{eqn:FE} and \eqref{eqn:RE}.
	We replace the local LLA weights $p_{\lambda}'(|\tilde \beta_{j}(\tau)|)$ with the global LLA weights 
	$p_{\lambda}'\left(\sup_{\tau \in \mc T} |\tilde \beta_{j}(\tau)|\right)$, solving 
	$$
		\min_{\bvec} \frac{1}{n}\sum_{i=1}^n \rho_{\tau}(y_i-\bm x_i^T\bvec(\tau)) + \sum_{j \in \Mset^c} p_{\lambda}'\left(\sup_{\tau \in \mc T} |\tilde \beta_{j}(\tau)|\right) |\beta_j(\tau)| \text{~for~} \tau \in \mc T
	$$
	in each LLA iteration.
	Our complete globally concerned LLA algorithms for \eqref{eqn:FE} and \eqref{eqn:RE} are given as Algorithms \ref{algo:LLA for full estimator} and \ref{algo:LLA for reduced estimator}, respectively.
	By using global LLA weights, we avoid the sensitivity of local weighting schemes and enable uniform selection and inference over $\mc T$.
	
	\begin{remark}
		One can substitute the following alternative global LLA weights $\hat w_j(\tau) = p_{\lambda}'(\int_{\mc T} |\tilde \beta_{j}(\tau)| d \tau)$ in Algorithms \ref{algo:LLA for full estimator} and \ref{algo:LLA for reduced estimator}. With minor adjustments, equivalent theoretical guarantees will hold, as with the analogous integrated global weights in \citet{zheng2015globally}.
	\end{remark}
	
	\begin{algorithm}[t]
		\caption{Global LLA for full model} \label{algo:LLA for full estimator}
		\begin{algorithmic}
			\State \textit{Input:} Initial estimator $\hat{\bvec}_a^{(0)}(\tau)$, penalty parameter $\lambda$, number of LLA iterations $B$
			\State Compute the adaptive weights $\hat{w}_{j}^{(0)}(\tau) = p_{\lambda}'\left(\sup_{\tau \in \mc T} |\hat \beta_{a,j}^{(0)}(\tau)|\right)$ for $j \in \Mset^c$
			\For{$b = 1, \ldots, B$}
			\State Compute $\hat \bvec_a^{(b)}(\tau)$ by solving the following weighted lasso problem for $\tau \in \mc T$:
			\begin{equation}
				\hat \bvec_a^{(b)} (\tau)
				= \argmin_{\bvec} \frac{1}{n}\sum_{i=1}^n \rho_{\tau}(y_i-\bm x_i^T\bvec(\tau)) + \sum_{j \in \Mset^c} \hat w_j^{(b-1)}(\tau)|\beta_j(\tau)| \label{algo:LLA update}
			\end{equation}
			\State Update the adaptive weights $\hat{w}_{j}^{(b)}(\tau) = p_{\lambda}'\left(\sup_{\tau \in \mc T} |\hat \beta_{a,j}^{(b)}(\tau)|\right)$ for $j \in \Mset^c$
			\EndFor
			\State \textit{Return:} $\hat \bvec_a^{(B)} (\tau)$ for $\tau \in \mc T$
		\end{algorithmic}
	\end{algorithm}
	
	\begin{algorithm}[t]
		\caption{Global LLA for reduced model} \label{algo:LLA for reduced estimator}
		\begin{algorithmic}
			\State \textit{Input:} Initial estimator $\hat{\bvec}_0^{(0)}(\tau)$, penalty parameter $\lambda$, number of LLA iterations $B$
			\State Compute the adaptive weights $\hat{w}_{j}^{(0)}(\tau) = p_{\lambda}'\left(\sup_{\tau \in \mc T} |\hat \beta_{0,j}^{(0)}(\tau)|\right)$ for $j \in \Mset^c$
			\For{$b = 1, \ldots, B$}
			\State Compute $\hat \bvec_0^{(b)}(\tau)$ by solving the following constrained weighted lasso problem for $\tau \in \mc T$:
			\begin{equation}
				\hat \bvec_0^{(b)} (\tau)
				= \argmin_{\bvec} \frac{1}{n}\sum_{i=1}^n \rho_{\tau}(y_i-\bm x_i^T\bvec(\tau)) + \sum_{j \in \Mset^c} \hat w_j^{(b-1)}(\tau)|\beta_j(\tau)| {\rm \; ~s.t. ~} \Cmat \bvec_{\mc M} (\tau) =\bm \Psi(\tau)\label{algo:LLA update reduced}
			\end{equation}
			\State Update the adaptive weights $\hat{w}_{j}^{(b)}(\tau) = p_{\lambda}'\left(\sup_{\tau \in \mc T} |\hat \beta_{0,j}^{(b)}(\tau)|\right)$ for $j \in \Mset^c$
			\EndFor
			\State \textit{Return:} $\hat \bvec_0^{(B)} (\tau)$ for $\tau \in \mc T$
		\end{algorithmic}
	\end{algorithm}
	
	\section{Estimation Rates and the Strong Oracle Property}\label{sec:estimation theory}
	As we noted above, the partial penalized estimation problems \eqref{eqn:FE} and \eqref{eqn:RE} may have multiple local solutions when $p_{\lambda}(\cdot)$ is an FCP.
	As such, we will develop theoretical guarantees specifically for the LLA solutions as in \citet{fan_strong_2014} and \citet{jacobson2025strong}.
	In this section, we introduce oracle estimators for testing \eqref{hyp:linear hypothesis}, study their asymptotic behavior, and prove that our globally adaptive LLA algorithms converge to these oracles with high probability as $n$ diverges.
	
	\subsection{Testing oracles for globally concerned quantile regression}
	A common object of study in high-dimensional regression is the \textit{oracle} estimator, which knows the true support set in advance and removes uninformative predictors prior to estimation.
	For globally concerned quantile regression the canonical oracle estimator is
	$$
		\hat \bvec^o(\tau) = \argmin_{\bvec_{\mc A_{\mc T}^c} = \bm 0} \frac{1}{n}\sum_{i=1}^n \rho_{\tau}(y_i-\bm x_i^T\bvec(\tau)).
	$$
	The canonical oracle estimator represents the best-case scenario for high-dimensional estimation: by removing the noise variables in advance, it avoids the estimation cost of performing variable selection.
	
	Because some elements $\Mset$ may not belong to $\mc A_{\mc T}$, the canonical oracle estimator may set some of the coefficients we want to test equal to be $0$, making it an unsuitable target for making inferences about $\bvec_{\Mset}^*(\tau)$. To address the shortcomings of the canonical oracle for testing \eqref{hyp:linear hypothesis}, we introduce \textit{testing oracle estimators}. Recall that $\MSset = \Mset \cup \Sset$. We define
	the \textit{full oracle estimator} to be
	\begin{align}\label{eqn:FE.oracle}
		\bvecOR_a(\tau) =\argmin_{\bvec_{\MSset^c}(\tau) = \bm 0} \frac{1}{n}\sum_{i=1}^n \rho_{\tau}(y_i-\bm x_i^T\bvec(\tau)),
	\end{align}
	and the \textit{reduced oracle estimator}, which satisfies the constraints of the null hypothesis \eqref{hyp:linear hypothesis}, to be
	\begin{align}\label{eqn:RE.oracle}
		\hat\bvec^{o}_{0}(\tau)=\argmin_{\bvec_{\MSset^c}(\tau) = \bm 0} \frac{1}{n}\sum_{i=1}^n \rho_{\tau}(y_i-\bm x_i^T\bvec(\tau)) {\rm \; ~s.t.~} \Cmat \bvec_{\Mset}(\tau) = \bm \Psi(\tau).
	\end{align}

	Unlike traditional oracle estimators, these testing oracles include the elements of ${\Mset}$ among the set of important predictors which are kept in the model.
	This enables us to make inferences about $\bvec_{\Mset}^*(\tau)$ even if some of its elements are zero in the true model.

	In the sequel, we will study the asymptotic properties of the testing oracles. We will then use these results to show that our globally adaptive LLA algorithms for the reduced and full models recover the corresponding oracle estimators with probability converging to $1$, a result referred to as the \textit{strong oracle property} \citep{fan_strong_2014, jacobson2025strong}.

	\subsection{Oracle estimation theory}
	Let $\tilde \xvec = \xvec_{\MSset}$ and $\widetilde \xmat  = \xmat_{ \MSset}$. 
	We assume the following regularity conditions for our study of the testing oracles:
	\begin{enumerate}[label=(A\arabic*), start = 1 , leftmargin=2\parindent]
		\item The data form a triangular array of random variables so that for any given $n$, the data $\mc D_n = \{(\xvec_i, y_i) : 1 \leq i \leq n\}$ is an i.i.d. random sample from \eqref{eqn:conditional quantile linear expr}. \label{as:iid array}
		\item 
		\begin{enumerate}[label = (\roman*)]
			\item The conditional density $f_{y|\xvec}(y|\xvec)$ is bounded from above uniformly over $y \in \mc Y_{\xvec}$, $\xvec \in \mc X$, and $n$ where $\mc X$ denotes the support of $\xvec$ and $\mc Y_{\xvec}$ denotes the support of $y$ conditional on $\xvec$. 
			
			\item $f_{y|\xvec}(Q_y(\tau|\xvec) )$ is bounded away from 0 uniformly over $\tau \in \overline{\mc T}$, $\xvec \in \mc X$, and $n$, where $\overline{ \mc T}$ denotes the convex hull of $\mc T$.
			
			\item the derivative of $y \mapsto f_{y|\xvec}(y|\xvec)$ is continuous and bounded in absolute value from above uniformly over $y \in \mc Y_{\xvec}$, $\xvec \in \mc X$, and $n$.
		\end{enumerate}
		\label{as:density conditions}
		\item The eigenvalues of the Gram matrix $\bm \Sigma = \E[\tilde \xvec \tilde \xvec^T]$ are bounded from above and away from zero uniformly over $n$. \label{as:Gram lower bound} 
	\end{enumerate}
	In addition, we assume the following for our study of the reduced oracle estimator (and later our hypothesis testing procedures):
	\begin{enumerate}[label=(A\arabic*), start = 4, leftmargin=2\parindent]
		\item $\Cmat \bvec_{\Mset}^*(\tau) = \bm \Psi(\tau) + \frac{\bm \eta(\tau)}{\sqrt{n}}$ for some function $\bm \eta(\cdot):[0,1] \mapsto \reals^r$, with 
		$ \sup_{\tau \in \mc T}\norm{\bm \eta(\tau)}_2 = O\left( \sqrt{ \min \{r, m-r+s\} } \right) $, $\norm{ \bm \eta (\tau_1) - \bm \eta (\tau_2) }_2 \le C' \sqrt{(m-r+s)n} \,|\tau_1 - \tau_2|$, $\forall_{\tau_1, \tau_2 \in \mc T}$ for a fixed constant $C'$, and $\lambda_{\max} \{(\Cmat \Cmat^T)^{-1} \} = O(1)$. \label{as:local alts}
	\end{enumerate}
	
	Conditions \ref{as:iid array}-\ref{as:Gram lower bound} are standard in the quantile regression literature (see, for example, \citet{BC.L1penalized, belloni_conditional_2019, zheng2015globally}) and allow the conditional distribution of $y$ to change with $\xvec$. In particular, the conditional variance of $\yvec$ can vary with $\xvec$, meaning that the model allows for heteroskedasticity. Conditions similar to \ref{as:local alts} have been used in the partial penalization literature \citep{shi2019linear,jacobson2025strong} to allow either for the null to be true or for the truth to get closer to the null as more data is collected.
	
	We define the following important quantity which will appear frequently in our results:
	\begin{equation}
		\xi_n \defeq \sup_{\xvec \in \mc X} \norm{\tilde \xvec}_2.
	\end{equation}	
	First, we establish rates for the full and reduced oracle estimators as $m$, $s$, $r$ diverge:
	\begin{theorem} \label{thm:oracle rate}
		Suppose that \ref{as:iid array} - \ref{as:Gram lower bound} hold and that $(m + s)\xi_n^2 \log^2 n = o(n)$.
		Then
		\begin{equation}
			\sup_{\tau \in \mc T} \norm{\bvecOR_{a,\MSset}(\tau) - \bvec_{\MSset}^*(\tau) }_2 = O_p\left(\sqrt{\frac{m+s}{n}} \right).
			\label{eqn:FE.oracle rate}
		\end{equation}
		If \ref{as:local alts} also holds, then
		\begin{equation}
			\sup_{\tau \in \mc T} \norm{\bvecOR_{0,\MSset}(\tau) - \bvec^*_{\MSset}(\tau) }_2 = O_p\left(\sqrt{\frac{m - r +s}{n}} \right).  \label{eqn:RE.oracle rate}
		\end{equation}
	\end{theorem} 
	The full oracle estimator rate \eqref{eqn:FE.oracle rate} matches the rate in \citet{he_parameters_2000} for quantile regression at a single $\tau$ and the rate in \citet{belloni_conditional_2019} for estimating the quantile process $\bvec^*(\tau)$ over $\tau \in \mc T$.
	The rate for the reduced oracle estimator \eqref{eqn:RE.oracle rate} is, to our knowledge, the first result of its kind for a linearly constrained quantile regression estimator, let alone with diverging dimension.

	\subsection{Global strong oracle property}
	In this section, we will show that the two-step LLA estimators possess the strong oracle property---that is, 
	$ \Pr( \sup_{\tau \in \mc T} \shortnorm{\hat \bvec^{(2)}_e(\tau) - \bvecOR_e(\tau)}_2 = 0 ) \to 1 $ for the full ($e = a$) and reduced ($e = 0$) estimators.
	To accomplish this, we will first establish conditions under which the LLA algorithms find the oracle estimators in their first iteration and return to them in all subsequent iterations.
	We will then show that these conditions hold with probability converging to $1$.	
	
	For ease of presentation, we define
	\begin{equation}
		\label{eqn:qndefine}
		\mc Q_{\tau,n}(\bvec) = \frac{1}{n} \sum_{i=1}^n \rho_\tau(y_i - \xvec_i^T\bvec).
	\end{equation} 
	Let $\psi_{\tau}(u) \defeq \tau - 1(u < 0) $ denote the quantile score function and $\mc I(\bvec) = \{i \in \{1,\ldots,n\} : y_i = \xvec_i^T\bvec\}$ denote the observations that are exactly interpolated.
	The subdifferential of $\mc Q_{\tau,n}(\bvec)$ is 
	\begin{equation}
		\label{eqn:subdiff}
		\partial \mc Q_{\tau,n}(\bvec) =  \left\{\bm \xi \in \reals^{p+1} : \xi_j =  - n^{-1}\left(\sum_{i \in (\mc I(\bvec))^c }x_{ij} \psi_\tau(y_i-\xvec_i^T\bvec) + \sum_{i \in \mc I (\bvec)} x_{ij}v_i^*\right)\right\},
	\end{equation}
	where $v_i^* \in [\tau-1,\tau]$.		
	We assume the following additional regularity conditions for our deterministic oracle recovery results:
	\begin{enumerate}[label=(A\arabic*), start = 5, leftmargin=2\parindent]
		\item There is a lower bound on the largest signal for the elements of $\bvec^*_{\Sset}(\tau)$ across the quantiles of interest. Specifically,
		\begin{equation*}
			\inf_{j \in \Sset} \sup_{\tau \in \mc T} |\beta_j^*(\tau)| > (d+1)\lambda.
		\end{equation*}\label{as:strength of signal}
		\item The oracle estimators exist and are unique. That is, for all $\tau \in \mc T$, 
		$$\bm 0 \in \partial_{\MSset} \mc Q_{\tau,n}(\hat{\bvec}_a^o(\tau))$$ 
		and there exist $\tilde{\bm w}(\tau) \in  \mc Q_{\tau,n}(\hat{\bvec}^{o}_0(\tau))$ and $\bm d(\tau)$ such that $$\tilde{\bm w}_{\MSset}(\tau) = \begin{bmatrix} \Cmat^T \bm d(\tau) \\ \bm 0 \end{bmatrix}.$$ \label{as:oracleDefined}
	\end{enumerate}
	Since $\Sset \subset \Mset^c$, the minimum signal strength condition \ref{as:strength of signal} applies only to coefficients that are not being tested. The lack of such a condition on $\bvec^*_{\Mset}(\tau)$ ensures that our tests have power when the signals $\bvec^*_j(\tau)$ are small.
	Moreover, the minimum signal condition on the supremum is weaker than pointwise-in-$\tau$ minimum signal conditions---such as (A0) in \citet{fan_strong_2014}---and allows each coefficient process to touch zero for some $\tau \in \mc T$.
	
	The following theorem establishes conditions under which the LLA algorithms find the testing oracles in their first iterations. As in Algorithms \ref{algo:LLA for full estimator} and \ref{algo:LLA for reduced estimator}, let $\hat \bvec_e^{(0)}$ denote the initial estimator for $e \in \{0,a\}$. 
	Recall that $d$, $d_1$, and $d_2$ are constants from conditions \ref{as:FCP i}-\ref{as:FCP iv}. 
	\begin{theorem}\label{thm:global33mod}
		Suppose that \ref{as:strength of signal} and \ref{as:oracleDefined} hold and define $d_0 = \min(1,d_2)$. For the reduced ($e = 0$) and full ($e = a$) models, if the events 
		\begin{eqnarray*}
			\mc E_{e1} & = & \left\{\sup_{\tau \in \mc T }||\hat \bvec^{(0)}_{e,\Mset^c}(\tau)-\bvec^*_{\Mset^c}(\tau)||_{\infty} \leq d_0\lambda \right\},\\
			\mc E_{e2} & = &  \left\{\sup_{\tau \in \mc T, \bm w(\tau) \in \partial_{{\MSset^c}} \mc Q_{\tau,n}(\hat{\bvec}_e^{o}(\tau))} ||\bm w(\tau)||_{\infty} < d_1\lambda \right\}
		\end{eqnarray*} 
		hold, then 
		\begin{equation}
			\label{eqn:LLAOracleStep1}
			\sup_{\tau \in \mc T} ||\hat{\bvec}^{(1)}_e(\tau)-\hat \bvec_e^o(\tau)||_2=0.
		\end{equation}		
	\end{theorem}

	The next theorem establishes conditions under which the LLA algorithms, having found the testing oracles, return to them in the next iteration.
	\begin{theorem}\label{thm:global33modStep2}
		Suppose that \ref{as:strength of signal} and \ref{as:oracleDefined} hold.
		For the reduced ($e = 0$) and full ($e = a$) models, if $\sup_{\tau \in \mc T}||\hat{\bvec}_e^{(b-1)}(\tau)-\hat{\bvec}_e^{o}(\tau)||_2=0$ and the events $\mc E_{e2}$
		and
		\begin{equation*}
			\mc E_{e3} = \left\{ \inf_{j \in S} \sup_{\tau \in \mc T }|\hat{\beta}^o_{e,j}(\tau)| > d\lambda \right\}
		\end{equation*}
		hold, then
		\begin{equation}
			\label{eqn:LLAOracleStep2}
			\sup_{\tau \in \mc T} ||\hat{\bvec}^{(b)}_e(\tau)-\hat \bvec_e^o(\tau)||_2=0.
		\end{equation}		
	\end{theorem}
	
	The events $\mc E_{e2}$ and $\mc E_{e3}$ are regularity conditions for the testing oracles, which we will show hold with high probability as a consequence of Theorem \ref{thm:oracle rate}.
	Whether $\mc E_{e1}$ holds depends on the choice of initial estimator.
	Note, however, that $\mc E_{e1}$ only places conditions on the the initial estimates of $\bvec^*_{\Mset^c}(\tau)$, and the initial estimator for the reduced model is not required to satisfy the constraints of \eqref{hyp:linear hypothesis}. As such, we can use the same initial estimator for both LLA algorithms.
	
	The final step to establishing the strong oracle property is to show that $\mc E_{e1}$, $\mc E_{e2}$, and $\mc E_{e3}$, $e \in \{0,a\}$, hold with probability converging to $1$.
	To do so, we make the following additional assumptions:
	\begin{enumerate}[label=(A\arabic*), start = 7, leftmargin=2\parindent]
		\item For all $j \in \{1,\ldots,p\}$, $\E[x_{ij}^2]=1$ and \begin{equation*}\lim_{n \rightarrow \infty} \Pr\left( \max_{j \in \{1,\ldots,p\}} |\hat{\sigma}_j-1| \leq 1/2\right)=1,\end{equation*}
		where $\hat{\sigma}_j=\sqrt{\sum_{i=1}^n x_{ij}^2/n}$. \label{as:xbounds}
		
		\item \label{as:x_max} $\max_{i, j \in \MSset} |x_{ij}| = o_p(n^{1/2}/((m+s)\log n)^{3/2})$ and $\max_{i, j \in \MSset^c} |x_{ij}| = o_p(n^{1/2}/((m+s)^2\log (p \vee n))^{1/2})$.
		
		\item There exists a positive constant $C$ such that $||\bvec^*_{\MSset}(\tau_1) - \bvec^*_{\MSset}(\tau_2)||_2 \leq C \sqrt{m+s}|\tau_1-\tau_2|$ for all $\tau_1,\tau_2 \in \mc T$ and $n$.  \label{as:bvec}
	
		\item 
		$
			\inf_{\avec \in \reals^{m+s}, \avec \neq \bm 0} \frac{\E[(\tilde \xvec^T\avec)^2]^{3/2}}{\E[|\tilde \xvec^T \avec|^3]}
		$ is bounded away from $0$ uniformly for all $n$ \label{as:eigens} 
		
		\item $\sqrt{\frac{(m+s) \log(p \vee n)}{n}} = o(\lambda) = o(1)$. \label{as:lambda}
		
		\item $(m+s)^2\log p = O(n^\alpha)$ for some $\alpha \in (0,1)$. \label{as:p rate}
		
		\item \label{as:new_re}
		$
			\sup_{j \in \MSset^c, \bm \delta \in \reals^{ m+s }} \frac{E[ (x_{ij} \tilde \xvec_i^T\bm \delta)^2]}{||\bm \delta||_2^2} = o\left(\frac{\log p}{s \log^2 n}\right).
		$ 
	\end{enumerate}
	
	Assumptions \ref{as:xbounds}-\ref{as:new_re} are similar to the assumptions made by \citet{zheng2015globally} to establish recovery of a regularized oracle estimator over $\mc T$ and by \citet{wangWuLi2012} to prove an oracle property at a single $\tau$, and mirror many of the assumptions of \citet{BC.L1penalized}. Assumptions \ref{as:xbounds}, \ref{as:x_max}, \ref{as:eigens} and \ref{as:new_re} provide that the covariates are well behaved. Assumption \ref{as:bvec} provides that the quantile functions are smooth with respect to $\tau$. Assumptions \ref{as:lambda} and \ref{as:p rate} provide bounds for how fast $m$, $s$, and $p$ can grow with $n$.  
	
	We use the $\ell_1$-penalized quantile regression estimator $\hat \bvec^{\rm lasso}(\tau)$, defined in \eqref{eqn:BC l1 estimator}, as the initial estimator for both global LLA algorithms.
	The following lemma, which follows directly from Theorem 2 in \citet{BC.L1penalized}, shows that $\hat \bvec^{\rm lasso}(\tau)$ satisfies $\mc E_{e1}$ with high probability.
	For clarity of presentation, we state the additional assumptions from \citet{BC.L1penalized} required for this lemma in Section S1.2 of the supplement.
	 
	\begin{lemma}\label{lemma:global33modAssumptions}
		Under \ref{as:iid array}, \ref{as:density conditions}, \ref{as:xbounds}, \ref{as:lambda}, and the additional conditions in Section S1.2
		\begin{equation}
			\label{eqn:thm33as1}
			\Pr \left(
			\sup_{\tau \in \mc T }||\hat{\bvec}^{\rm lasso}_{\Mset^c}(\tau)-\bvec^*_{\Mset^c}(\tau)||_{\infty} \leq d_0\lambda\right) \rightarrow 1.
		\end{equation}
	\end{lemma}
	
	Our next lemma establishes that the reduced ($e = 0$) and full ($e = a$) oracle estimators satisfy $\mc E_{e2}$ and $\mc E_{e3}$ with high probability. 
	\begin{lemma}\label{lemma:global33modAssumptions2}
		Suppose that \ref{as:iid array} - \ref{as:new_re} hold and that $(m + s)\xi_n^2 \log^2 n = o(n)$. 
		Then for $e \in \{0,a\}$
		\begin{equation}
			\label{eqn:thm33as2}
			\Pr \left(
			\sup_{\tau \in \mc T, \bm w(\tau) \in \partial_{{\MSset^c}} \mc Q_{\tau,n}(\hat{\bvec}_e^{o}(\tau))} ||\bm w(\tau)||_{\infty} < d_1\lambda\right) \rightarrow 1
		\end{equation}
		and 
		\begin{equation}
			\label{eqn:thm33as3}
			\Pr\left(
			\inf_{j \in \mathcal{S}} \sup_{\tau \in \mc T}|\hat{\beta}^o_{e,j}(\tau)| > d\lambda\right) \rightarrow 1. 
		\end{equation}
	\end{lemma}
	
	Together, Lemmas \ref{lemma:global33modAssumptions} and \ref{lemma:global33modAssumptions2} establish that the global LLA algorithms converge to the testing oracles in two iterations with probability converging to $1$, as is summarized in the following theorem.
	\begin{theorem}\label{thm:strong oracle}
		Suppose \ref{as:iid array} - \ref{as:new_re} and the additional conditions in Section S1.2 hold and that $(m + s)\xi_n^2 \log^2 n = o(n)$. Then for $e \in \{0,a\}$
		$$
			\Pr\left( \sup_{\tau \in \mc T} \shortnorm{\hat \bvec^{(2)}_e(\tau) - \bvecOR_e(\tau)}_2 = 0 \right) \to 1.
		$$
		That is, the two-step global LLA estimators are equal to the corresponding testing oracles uniformly over $\mc T$ with probability converging to $1$.
	\end{theorem}
	
	\begin{remark} \label{remark:faster oracle rate}
		While \citet{zheng2015globally} showed that their adaptively weighted quantile regression estimators recover a \emph{regularized} oracle estimator over a range of quantiles $\mc T$, Theorem \ref{thm:strong oracle} is the first to establish recovery of \emph{unregularized} oracle estimators over $\mc T$.
		These oracles attain different uniform estimation rates over $\mc T$---$\sqrt{s \log n/n}$ for the regularized oracle (Theorem 3.1 of \citet{zheng2015globally}) vs. $\sqrt{s/n}$ for the unregularized oracle (with $\Mset =\{0\}$ for the sake of comparison) (Theorem \ref{thm:oracle rate})---making the unregularized oracle a more attractive target for recovery.
	\end{remark}
	
	\subsection{Global LLA estimation rate}
	Having established that the two-step global LLA estimators possess the strong oracle property, we can show that they attain the same rates of convergence as the testing oracles.
	\begin{theorem} \label{thm:part pen rates}
		Suppose \ref{as:iid array} - \ref{as:new_re} and the additional conditions in Section S1.2 hold and that $(m + s)\xi_n^2 \log^2 n = o(n)$. 
		Then 
		\begin{equation}
			\sup_{\tau \in \mc T} \norm{\hat \bvec_{a}^{(2)}(\tau) - \bvec^*(\tau) }_2 = O_p\left(\sqrt{\frac{m+s}{n}} \right)\label{eqn:FE.part pen rate}
		\end{equation}
		and
		\begin{equation}
			\sup_{\tau \in \mc T} \norm{\hat \bvec_{0}^{(2)}(\tau) - \bvec^*(\tau) }_2 = O_p\left(\sqrt{\frac{m-r+s}{n}} \right).\label{eqn:RE.part pen rate}
		\end{equation}
		
	\end{theorem}
	
	\begin{remark}\label{remark:faster penalized rate}
		Supposing that $\Mset =\{ 0 \}$ for the sake of comparison to other methods, 
		Theorem \ref{thm:part pen rates} establishes that the estimation rate for the global LLA estimators is $\sqrt{s/n}$ uniformly over $\mc T$ -- faster than the $\sqrt{s \log (p \vee n) /n}$ rate for \citet{BC.L1penalized}'s $\ell_1$-penalized estimator and the $\sqrt{s \log n /n}$ rate for \citet{zheng2015globally}'s globally adaptive $\ell_1$-penalized estimator.
		To our knowledge, the global LLA estimators are the first to attain the $\sqrt{s/n}$ uniform oracle estimation rate over $\mc T$. 
	\end{remark}

	\section{Partial Penalized Wald and Rank Score Tests} \label{sec:partial penalized tests}
	In this section, we propose partial penalized Wald and rank score tests for testing \eqref{hyp:linear hypothesis} over $\mc T$ in high dimensions based the global LLA estimators.
	
	\subsection{Partial penalized Wald test}
	We use the two-step LLA estimator for the full model, which we denote by $\hat \bvec_a^{(2)}(\tau)$, to compute the partial penalized Wald process.
	Let $\hat {\Sset}_a = \{j \notin \Mset: \sup_{\tau \in \mc T}|\hat \beta_{a,j}^{(2)}(\tau)| > 0 \}$ denote the support set for the full model estimator.
	We use the selected predictors from the full model, which we denote by $\tilde{\xvec}_a = \xvec_{\Mset \cup \hat \Sset_a}$, to compute the sample Gram matrix $\hatmat \Sigma_a = \frac{1}{n} \sum_{i=1}^n \tilde{\xvec}_{a,i} \tilde{\xvec}_{a,i}^T$.
	Let
	\begin{equation}
		\bm J(\tau) \defeq \E[f_{y|\xvec} (Q_y(\tau|\xvec)) \tilde \xvec \tilde \xvec^T ] = \E[f_{y|\xvec} (\xvec^T\bvec^*(\tau))) \tilde \xvec \tilde \xvec^T ] \label{eqn:Jacobian}
	\end{equation}
	denote the Jacobian matrix.
	We use \citet{Powell1984}'s estimator of the Jacobian with the selected predictors:
	\begin{equation}	
		\hatmat J_a(\tau) =\frac{1}{2hn} \sum_{i=1}^n  I(|y_i - \xvec_i^T \hat \bvec_{a}^{(2)}(\tau)| \leq h ) \tilde{\xvec}_{a,i} \tilde{\xvec}_{a,i}^T, \label{eqn:Jacobian (part pen)}
	\end{equation}
	where $h>0$ is a bandwidth parameter.
	Lastly, we define the sandwich estimator $\hatmat K_a(\tau) = \Cmat [(\hatmat J_a (\tau))^{-1} \hatmat \Sigma_a (\hatmat J_a (\tau))^{-1} ]_{\Mset, \Mset} \Cmat^T$.
	Using these quantities, we define the \textbf{partial penalized Wald process} for testing \eqref{hyp:linear hypothesis} to be
	\begin{equation}
		W_n(\tau) = n (\Cmat \hat \bvec_{a, \Mset}^{(2)}(\tau) - \bm \Psi(\tau) )^T(\hatmat K_a(\tau))^{-1}(\Cmat \hat \bvec_{a, \Mset}^{(2)}(\tau) - \bm \Psi(\tau) )/(\tau(1 - \tau)) \text{ \; for \;  } \tau \in \mc T. \label{eqn:wald process}
	\end{equation}
	
	\subsection{Partial penalized rank score test}
	The classical quantile rank score process \citep{gutenbrunner1992regression} is limited to testing hypotheses of the form
	$
	\text{H}_0: \bvec^*_{\Mset}(\tau) = \bm 0  \text{ for } \tau \in \mc T.
	$
	We propose a generalization of the rank score process for testing all linear hypotheses \eqref{hyp:linear hypothesis}.
	
	We use the two-step LLA estimator for the reduced model, which we denote by $\hat \bvec_0^{(2)}$, to compute the partial penalized rank score process.
	Let $\hat \Sset_0 = \{j \notin \Mset: \sup_{\tau \in \mc T}|\hat \beta_{0,j}^{(2)}(\tau)| > 0 \}$ denote the support set for the reduced model estimator.
	We define the \textit{partial penalized regression rank scores} $\hat {\bm a}(\tau)$ to be the solution to the dual problem for the reduced model LLA optimization \eqref{algo:LLA update reduced} with adaptive weights $\hat w_j^{(B-1)}$ for $j \in \Mset^c$, given by
	\begin{multline}
		(\hat {\bm a}(\tau), \hat {\bm \lambda}) = \argmax_{\bm a, \bm \lambda} \{ \bm y^T\bm a - \bm \lambda \bm \Psi(\tau) : \bm a \in [0,1]^n, \, \xmat_{\Mset}^T \bm a = (1 - \tau)\xmat_{\Mset}^T \bm 1_n +  \Cmat^T\bm \lambda, \,\\ 
		|\xmat_{j}^T (\bm a - (1-\tau) \bm 1_n ) | \leq \hat w_j^{(B-1)} \; \;\forall_{j \in \Mset^c} \}. \label{eqn:rank scores}
	\end{multline}
	The rank scores can be expressed in terms of the reduced model estimator $\hat \bvec_0$, as
	\begin{equation*}
		\hat a_i (\tau) = 
		\begin{cases}
			1 \text{\; if \;} y_i > \xvec_{i}^T \hat{\bvec}_{0}^{(2)}\\
			0 \text{\; if \;} y_i < \xvec_{i}^T \hat{\bvec}_{0}^{(2)}
		\end{cases}
	\end{equation*}
	with $\hat a_i (\tau) $ for $i \in \mathcal{I}(\hat \bvec_0^{(2)}) = \{i: y_i = \xvec_{i}^T \hat{\bvec}_{0}^{(2)}\}$ satisfying
	the constraints in \eqref{eqn:rank scores}.
	
	Define $\tilde \xvec_0$, $\hatmat \Sigma_0$, $\hatmat J_0(\tau)$, and $\hatmat K_0(\tau)$ based on $\hat \bvec_0^{(2)}$ equivalently to their full model counterparts.
	Using these terms, we define the \textbf{partial penalized rank score process} for testing \eqref{hyp:linear hypothesis}  to be
	\begin{equation}
		R_n(\tau) = \norm{ \frac{1}{\sqrt{n}} [\widehat{\bm K}_0(\tau)]^{-1/2} 
			\bm C [\widehat{\bm J}_0^{-1}(\tau)]_{(\Mset)}\sum_{i \in [n]} \tilde \xvec_i ( \hat a_i(\tau) - (1 - \tau) ) }_2^2/(\tau(1 - \tau))  \label{eqn:rs process}
	\end{equation}
	
	\subsection{Testing procedure}
	Our Wald and rank score statistics for testing \eqref{hyp:linear hypothesis} over $\mc T$ are given by $ \sup_{\tau \in \mc T} W_n(\tau)$ and $\sup_{\tau \in \mc T} R_n(\tau)$, respectively. For testing at a single $\tau$, the test statistics simplify to $W_n(\tau)$ and $R_n(\tau)$.
	In Section \ref{sec:inference theory}, we will show that the partial penalized Wald and rank score processes converge in distribution to a squared Bessel process of order $r$, which we denote by $Q^2_{r}(\tau)$, under \eqref{hyp:linear hypothesis}.
	As such, we reject the null over $\mc T$ when the test statistics exceed the upper-$\alpha$ critical value for the Bessel process $\sup_{\tau \in \mc T} Q^2_{r,\alpha}(\tau)$.
	For testing at a single $\tau$, $W_n(\tau)$ and $R_n(\tau)$ converge to chi-square random variables with $r$ degrees of freedom, denoted by $\chi^2_r$, under \eqref{hyp:linear hypothesis}.
	As such, we reject the null at significance level $\alpha$ when the test statistics exceed the upper-$\alpha$ critical value $\chi^2_{r,\alpha}$.

	\section{Asymptotic Theory for Inference}\label{sec:inference theory}
	Having proven that our two-step global LLA estimators possess the strong oracle property, we can leverage the asymptotic properties of the full and reduced oracle estimators to derive limiting distributions for the partial penalized Wald and rank score processes evaluated at the global LLA solutions.
	First, we introduce the oracle Wald and rank score tests and study their asymptotic behavior.
	
	\subsection{Oracle Wald and rank score tests}
	We define oracle Wald and rank score tests based on the full and reduced oracle estimators.
	Let $\hatmat \Sigma^o = \frac{1}{n} \sum_{i = 1}^n \tilde{\xvec}_i \tilde{\xvec}_i^T$ and 
	$\hatmat{J}_a^o(\tau) =\frac{1}{2hn} \sum_{i=1}^n  I(|y_i - \xvec_i^T \hat \bvec_{a}^{o}(\tau)| \leq h ) \tilde{\xvec}_{i} \tilde{\xvec}_{i}^T $
	denote the full oracle estimators of the Gram and Jacobian matrices.
	Let $\widetilde{\Cmat} = [\Cmat \text{\;} \bm 0] \in \reals^{r \times m+s}$ and
	$\hatmat K_a^o(\tau) = \widetilde{\Cmat} (\hatmat J_a^o (\tau))^{-1} \hatmat \Sigma (\hatmat J_a^o (\tau))^{-1} \widetilde\Cmat^T$.
	Using these terms, the \textit{oracle Wald process} for testing \eqref{hyp:linear hypothesis} is
	\begin{equation}
		W_n^o(\tau) = n (\Cmat \hat \bvec^o_{a, \Mset }(\tau) - \bm \Psi(\tau) )^T(\hatmat K_a^o(\tau))^{-1}(\Cmat \hat \bvec^o_{a, \Mset }(\tau) - \bm \Psi(\tau) )/(\tau(1 - \tau)). \label{eqn:oracle wald process}
	\end{equation}
	
	We define the oracle regression rank scores $\hat {\bm a}^o(\tau)$ as the solution to the dual problem for the reduced oracle optimization problem \eqref{eqn:RE.oracle}, given by
	\begin{equation}
		(\hat {\bm a}^o(\tau), \hat {\bm \lambda}) = \argmax_{\bm a, \bm \lambda} \{ \bm y^T\bm a - \bm \lambda \bm \Psi(\tau) : \widetilde \xmat^T \bm a = (1 - \tau)\widetilde\xmat^T \bm 1_n + \widetilde \Cmat^T  \bm \lambda , \, \bm a \in [0,1]^n  \}.  \label{eqn:oracle rank scores}
	\end{equation}
	The oracle rank scores can be expressed in terms of the reduced oracle estimator, as
	\begin{equation*}
		\hat a_i^o (\tau) = 
		\begin{cases}
			1 \text{\; if \;} y_i >  \xvec_{i}^T \hat{\bvec}_{0}^o\\
			0 \text{\; if \;} y_i < \xvec_{i}^T \hat{\bvec}_{0}^o
		\end{cases}
	\end{equation*}
	with $\hat a_i^o (\tau) $ for $i \in \mathcal{I}(\hat \bvec^o_0) = \{i: y_i = \xvec_{i}^T \hat{\bvec}_{0}^o\}$ satisfying
	$
	\sum_{i \in \mathcal{I}} \tilde \xvec_{i} \hat a_i^o
	+ \sum_{i=1}^n \tilde \xvec_{i} I(y_i > \xvec_{i}^T \hat{\bvec}_{0}^o)
	= (1 - \tau)  \sum_{i=1}^n \tilde \xvec_{i} + \widetilde \Cmat^T  \bm \lambda  .
	$
	
	Define $\hatmat J_0^o(\tau)$ and $\hatmat K_0^o(\tau)$ based on $\hat \bvec_0^o$ equivalently to their full model counterparts.
	Using these terms, we define the \textit{oracle rank score process} for testing \eqref{hyp:linear hypothesis} to be
	\begin{equation}
		R_n^o(\tau) = \norm{ \frac{1}{\sqrt{n}} [\widehat{\bm K}_0^o(\tau)]^{-1/2} 
			\widetilde{\bm C} [\widehat{\bm J}_0^o(\tau)]^{-1}\sum_{i \in [n]} \tilde \xvec_i ( \hat a_i^o(\tau) - (1 - \tau)) }_2^2/(\tau(1 - \tau)). \label{eqn:oracle rs process}
	\end{equation}
	
	\subsection{Asymptotic results for oracle tests}	
	Our first theorem establishes Bahadur representations for the testing oracles and the oracle rank scores.
	Recall that $\psi_{\tau}(u) = \tau - 1(u < 0) $ and define $ \bm H(\tau) = \widetilde{\bm C} [\bm J(\tau)]^{-1} \widetilde{\bm C}^T $ and
	$ \bm P(\tau)  = [\bm J(\tau)]^{-1/2}\widetilde{\bm C}^T [\bm H(\tau)]^{-1}\widetilde{\bm C}[\bm J(\tau)]^{-1/2}$.
	\begin{theorem} \label{thm:oracle bahadur}
		Suppose that \ref{as:iid array} - \ref{as:Gram lower bound} hold and that $(m + s)^3 \xi_n^2 = o(n^{1 -\epsilon})$ for some constant $\epsilon > 0$.
		Then the full oracle estimator satisfies
		\begin{equation}
			\sqrt{n}(\bvecOR_{a, \MSset}(\tau) - \bvec_{\MSset}^*(\tau)) =  \frac{1}{\sqrt n} [\bm J(\tau)]^{-1} \sum_{i \in [n]} \tilde \xvec_i \psi_{\tau}(y_i - \xvec_i^T \bvec^*(\tau)) + \bm r_{1n}(\tau) \label{eqn:FE.oracle bahadur}
		\end{equation}
		for all $\tau \in \mc T$,
		where 
		$
		\sup_{\tau \in \mc T} \norm{\bm r_{1n}(\tau)}_2 = O_p\left(n^{-1/4}(m + s)^{3/4}\xi_n^{1/2} \log^{1/2}n\right) = o(n^{- \epsilon'})
		$
		for some $\epsilon' > 0$.
		If \ref{as:local alts} also holds, then the reduced oracle estimator satisfies
		\begin{align}
			\sqrt{n}(\bvecOR_{0, \MSset}(\tau) - \bvec_{\MSset}^*(\tau)) 
			& =  [\bm J(\tau)]^{-1/2}(\bm I - \bm P(\tau)) [\bm J(\tau)]^{-1/2} \frac{1}{\sqrt n} \sum_{i \in [n]}  \tilde \xvec_i \psi_{\tau}(y_i - \xvec_i^T \bvec^*(\tau)) \notag \\
			& \text{\; \;} - [\bm J(\tau)]^{-1} \widetilde{\bm C}^T \bm H^{-1}(\tau) \bm \eta(\tau) + \bm r_{2n}(\tau)  \label{eqn:RE.oracle bahadur}
		\end{align}
		for all $\tau \in \mc T$
		and the oracle rank scores satisfy
		\begin{align}
			\frac{1}{\sqrt{n}} \sum_{i \in [n]} \tilde \xvec_i ( \hat a_i^o(\tau) - (1 - \tau) ) 
			& =  [\bm J(\tau)]^{1/2} \bm P(\tau) [\bm J(\tau)]^{-1/2}\frac{1}{\sqrt n} \sum_{i \in [n]} \tilde \xvec_i \psi_{\tau}(y_i - \xvec_i^T \bvec^*(\tau)) \notag \\ 
			& \text{\; \;} + \widetilde{\bm C}^T \bm H^{-1}(\tau) \bm \eta(\tau) + \bm r_{3n}(\tau) \label{eqn:oracle rs rep}
		\end{align}
		for all $\tau \in \mc T$, where $ \bm r_{2n}(\tau) $ and $ \bm r_{3n}(\tau) $ attain the same rate as $ \bm r_{1n}(\tau) $.
	\end{theorem}
	
	\begin{remark}
		As noted above, we have proposed more general versions of the reduced model estimator \eqref{eqn:RE.oracle} and rank scores \eqref{eqn:oracle rank scores} than in prior studies of quantile regression (see, for example, \citet{koenker_quantile_2005}). As such, the Bahadur representations for these quantities in Theorem \ref{thm:oracle bahadur} are entirely novel.  
	\end{remark}
	
	Together the oracle convergence rates in Theorem \ref{thm:oracle rate} and the Bahadur representations in Theorem \ref{thm:oracle bahadur} enable us to study the asymptotic properties of the oracle Wald and rank score processes.
	Let $\bm B_r(\tau)$ denote a standard $r$-dimensional Brownian bridge.
	Then
	$$Q^2_{r,\bm \mu(\tau)}(\tau) \defeq \norm{ \bm B_r(\tau) + \bm \nu(\tau)}_2^2/(\tau(1 - \tau))$$ 
	is a squared non-central Bessel process of order $r$ with $\bm \mu(\tau) = \bm \nu (\tau)^T \bm \nu(\tau)/( \tau(1 - \tau) )$.
	Our final result for the testing oracles establishes the limiting distribution of the oracle Wald and rank score test statistics under the local alternatives in \ref{as:local alts}.
	
	\begin{theorem}\label{thm:oracle test stat convergence}
		Suppose that \ref{as:iid array} - \ref{as:local alts} hold, $r$ is fixed, and $(m + s)^3 \xi_n^2 = o(n^{1 - \delta_1})$ and $\xi_n^6 = o(n^{1 - \delta_2})$ for some $\delta_1, \delta_2 >0$. 
		If the bandwidth parameter $h$ for $\hatmat J_e^o(\tau)$, $e \in \{0,a\}$, satisfies
		$(m+s)^2 \xi_n^2 \log n = o(n^{1 - \epsilon} h)$ and $h = o((m+s)^{-1/2}n^{-\epsilon/2})$ and $m+s = o(n^{\epsilon})$ for some $\epsilon > 0$, then
		\begin{eqnarray}
			& & \sup_{\tau \in \mc T} W_n^o(\tau) \overset{d}{\to} \sup_{\tau \in \mc T} Q^2_{r,\bm \mu(\tau)}(\tau) \label{eqn:thm:Wald to Bessel}\\
			& & \sup_{\tau \in \mc T} R_n^o(\tau) \overset{d}{\to} \sup_{\tau \in \mc T} Q^2_{r,\bm \mu(\tau)}(\tau) \label{eqn:thm:rs to Bessel}
		\end{eqnarray}
		as $n \to \infty$, where $\bm \mu(\tau) = \bm \eta (\tau)^T [\bm K(\tau)]^{-1} \bm \eta(\tau)/(\tau(1 - \tau))$.
	\end{theorem}
	
	\subsection{Asymptotic results for partial penalized tests}
	By combining our strong oracle results for the partial penalized LLA estimators with our asymptotic results for the oracle Wald and rank score tests, we establish corresponding guarantees for the partial penalized tests.
	First, we establish Bahadur representations for the global LLA estimators.
	\begin{theorem}\label{thm:part pen bahadur}
		Suppose that the conditions of Theorems \ref{thm:strong oracle} and \ref{thm:oracle bahadur} hold.
		Then
		\begin{equation}
			\sqrt{n}(\hat \bvec_{a, \MSset}^{(2)}(\tau) - \bvec_{\MSset}^*(\tau)) =  \frac{1}{\sqrt n} [\bm J(\tau)]^{-1} \sum_{i \in [n]} \tilde \xvec_i \psi_{\tau}(y_i - \xvec_i^T \bvec^*(\tau)) + \bm r_n(\tau) \label{eqn:FE.part pen bahadur}
		\end{equation}
		and
		\begin{align}
			\sqrt{n}(\hat \bvec_{0, \MSset}^{(2)}(\tau) - \bvec_{\MSset}^*(\tau)) 
			& =  [\bm J(\tau)]^{-1/2}(\bm I - \bm P(\tau)) [\bm J(\tau)]^{-1/2} \frac{1}{\sqrt n} \sum_{i \in [n]}  \tilde \xvec_i \psi_{\tau}(y_i - \xvec_i^T \bvec^*(\tau)) \notag \\
			& \text{\; \;} - [\bm J(\tau)]^{-1} \widetilde{ \bm C }^T \bm H^{-1}(\tau) \bm \eta(\tau) + \tilde{\bm r}_{n}(\tau)  \label{eqn:RE.part pen bahadur}
		\end{align}
		where 
		$
		\sup_{\tau \in \mc T} \norm{\tilde{ \bm r }_n(\tau)}_2 \asymp_p \sup_{\tau \in \mc T} \norm{\bm r_n(\tau)}_2 = O_p\left(\frac{(m + s)^{3/4}\xi_n^{1/2} \log^{1/2}n}{n^{1/4}}\right) = o(n^{- \epsilon'})
		$
		for some $\epsilon' > 0$.
	\end{theorem}
	\noindent Next, we derive the limiting distributions of the partial penalized test statistics.
	\begin{theorem}\label{thm:part pen bessel}
		Suppose that the conditions of Theorems \ref{thm:strong oracle} and \ref{thm:oracle test stat convergence} hold.
		Let $W_n^{(2)}$ and $R_n^{(2)}$ denote the partial penalized Wald and rank score processes evaluated at the two-step LLA solutions.
		Then
		\begin{eqnarray}
			& & \sup_{\tau \in \mc T} W_n^{(2)}(\tau) \overset{d}{\to} \sup_{\tau \in \mc T} Q^2_{r,\bm \mu(\tau)}(\tau), \label{eqn:thm:pp Wald to Bessel}\\
			& & \sup_{\tau \in \mc T} R_n^{(2)}(\tau) \overset{d}{\to} \sup_{\tau \in \mc T} Q^2_{r,\bm \mu(\tau)}(\tau) \label{eqn:thm:pp RS to Bessel}
		\end{eqnarray}
		as $n \to \infty$, where $\bm \mu(\tau) = \bm \eta (\tau)^T [\bm K(\tau)]^{-1} \bm \eta(\tau)/(\tau(1 - \tau))$.
	\end{theorem}
	
	\begin{remark}\label{rem:Power_comparison}
		As a corollary to this theorem, one can compare the power of our tests with that of the global-over-$\mc T$ test of \citet{bradic2017uniform} and pointwise-in-$\tau$ test of \citet{yan2023confidence}.
		Under the homoskedastic common-error model, Lemma S.1 of \citet{shi2019linear} gives our single-coefficient test a pointwise-in-$\tau$ asymptotic local power no smaller than \citet{bradic2017uniform} and \citet{yan2023confidence}.
		For $\beta_j^*(\tau)=n^{-1/2}h$, with $h\ne0$, Anderson's inequality gives an analogous ordering against \citet{bradic2017uniform} for global-over-$\mc T$ testing. See Section S2.5 of the supplement for details.
	\end{remark}

	\section{Model Fitting and Tuning}\label{sec:implementation details}
	
	\subsection{ADMM for LLA updates}\label{sec:admm}
	In this subsection, we describe how we compute the reduced model LLA updates \eqref{algo:LLA update reduced} in Algorithm \ref{algo:LLA for reduced estimator} (the updates in Algorithm \ref{algo:LLA for full estimator} are a special case with $\Cmat = \bm 0$ and $\bm \Psi \equiv \bm 0$ and are therefore omitted). 
	
	Consider the $\ell_1$-penalized quantile regression problem with linear constraints
	$$
	\min_{\bvec\in \reals^{p+1}} \frac{1}{n}\sum_{i=1}^{n}\checkloss\big(y_i - \xvec_i^\top \bvec\big)
	+ \lambda \|\bm w \odot \bvec \|_{1} \text{ s.t. } \overline \Cmat \bvec = \bm \Psi
	$$
	where $\bm w = (w_0, \ldots, w_p)^T$ are weights with $\bm w_{\Mset} = 0$, ``$\odot$'' denotes the elementwise product, and $\overline \Cmat = [\Cmat \; \bm 0]$.
	Following \citet{gu_admm_2018} we introduce the dummy variables $z_i = y_i - \xvec_i^T \bvec$, making the problem
	\begin{equation}	
		\min_{\bvec\in \reals^{p+1}} \frac{1}{n}\sum_{i=1}^{n}\checkloss\big(z_i \big)
	+ \lambda \|\bm w \odot \bvec \|_{1} \text{ s.t. } \zvec = \yvec - \xmat \bvec, \overline\Cmat \bvec = \bm \Psi .\label{eqn:l1 with dummy z}
	\end{equation}
	Letting $\sigma > 0$ be a fixed penalty parameter, the augmented Lagrangian function for this problem is
	\begin{align*}
		\mc L_{\sigma}(\bvec, \zvec, \bm \theta) 
		& = \frac{1}{n}\sum_{i=1}^{n}\checkloss\big(z_i \big)
		+ \lambda \|\bm w \odot \bvec \|_{1} - \bm \theta_1^T\left( \zvec - \yvec + \xmat \bvec \right)
		+ \frac{\sigma}{2} \norm{ \zvec - \yvec + \xmat \bvec }_2^2 \\
		& \text{\; \; } - \bm \theta_2^T\left( \overline \Cmat \bvec  - \bm \Psi\right)
		+ \frac{\sigma}{2} \norm{\overline \Cmat \bvec - \bm \Psi }_2^2
	\end{align*}
	where $\bm \theta_1 \in \reals^n, \bm \theta_2 \in \reals^r$ are the Lagrange multipliers.
	We  solve \eqref{eqn:l1 with dummy z} using the alternating direction method of multipliers (ADMM) \citep{Boyd2010}.
%
	We perform the $\bvec$ updates using coordinate descent and the $\bm z$ updates elementwise using the proximal mapping of $\checkloss$ as in Algorithm 2 of \citet{gu_admm_2018}.
	We declare convergence when the primal and dual residuals, given by $\bm r_{\rm p}^{k+1} = ( (\xmat \bvec^{k+1} + \bm z^{k+1} - \bm y)^T, (\overline \Cmat \bvec^{k+1} - \bm \Psi)^T)^T$ and $\bm r_{\rm d}^{k+1} = \sigma\xmat^T(\bm z^{k+1} - \bm z^{k}) $ respectively, are sufficiently small.
	By default, we set the tolerances for convergence to $\eps_{\rm abs} = \eps_{\rm rel} = 10^{-3}$.
	Denoting the shrinkage operator by
	$\shrink(x, \lambda) = \text{sgn} (x) \max\{ |x| - \lambda, 0 \}$
	and the proximal mapping of $\checkloss$ by
	$\prox_{\checkloss}(x, \alpha) = x - \max\{ \frac{\tau - 1}{\alpha}, \min \left( x, \frac{\tau}{\alpha} \right) \} $, we present our complete ADMM algorithm in Algorithm \ref{algo:ADMM reduced}.
	
	\begin{algorithm}
		\caption{ADMM for reduced model LLA updates}
		\label{algo:ADMM reduced}
		\begin{algorithmic}
			\State \textit{Input:} $\lambda$, $\bm w$, $\bvec^0$, $\zvec^0$, $\bm \theta_1^0$, $\bm \theta_2^0$, $\epsilon_{\rm rel}$, $\epsilon_{\rm abs}$
			\State $k = 0$
			\Repeat:
			\State $m = 0$
			\Repeat
			\For{$j = 0,1,2,\ldots,p$}
			\State
			{\tiny
				$v_j = \xmat_j^T \left[\bm \theta_1^k  + \sigma \left(\yvec - \sum_{t \neq j }\xmat_{t}\beta_t^{k, m + 1(t < j)} - \zvec^k\right) \right] 
				+ \Cmat_j^T \left[ \bm \theta_2^k + \sigma\left( \bm \Psi - \sum_{t \neq j }\overline \Cmat_{t}\beta_t^{k, m + 1(t < j)} \right) \right] $}
			\State $s_j  = \sigma \left(\norm{\xmat_j}_2^2 + \norm{\overline \Cmat_j}_2^2 \right)$
			\State $ \beta_j^{k, m+1} = \shrink \left( v_j , \lambda w_j\right)/s_j$
			\EndFor
			\State $m = m+1$
			\Until $\max_j (\beta_j^{k, m+1} - \beta_j^{k, m})^2  < \epsilon $
			\State $\bvec^{k+1} = \bvec^{k, m+1}$
			\State
			$
			z_i^{k+1} = \left( \prox_{\checkloss}( y_i - \xvec_i^T \bvec^{k+1} + \theta_{1i}^k/\sigma, n \sigma )  \right)_{i \in [n]}
			$
			\State
			$
			\bm \theta_1^{k+1} =  \bm \theta_1^k - \sigma(\xmat \bvec^{k+1} + \zvec^{k+1} - \yvec ).
			$
			\State
			$
			\bm \theta_2^{k+1}  =  \bm \theta_2^k - \sigma(\overline \Cmat \bvec^{k+1} - \bm \Psi ).
			$
			\State $k = k+1$
			\Until $ \shortnorm{\bm r_p^{k+1}}_2 \leq \eps_{\rm abs} \sqrt{n+r}  + \eps_{\rm rel} \max\{  \shortnorm{ ((\xmat \bvec^{k+1})^T, (\overline \Cmat \bvec^{k+1})^T)^T }_2, \shortnorm{ \bm z^{k+1} }_2, \shortnorm{(\yvec^T, \bm \Psi^T)^T}_2 \}$, 
			$ \shortnorm{\bm r_d^{k+1}}_2 \leq \eps_{\rm abs}\sqrt{p+1}  + \eps_{\rm rel} \shortnorm{ ( \bm \theta_1^{k+1} \xmat ,  \bm \theta_1^{k+2} \overline \Cmat )^T }_2$ 
		\end{algorithmic}
	\end{algorithm}
	
	We accelerate computation of the solution path over $\lambda_1 > \cdots > \lambda_L$ using warm starts---initializing Algorithm \ref{algo:ADMM reduced} for $\lambda_{\ell}$ with the solution for $\lambda_{\ell-1}$---and select $\lambda_1$ such that all the coefficients $\bvec_{\Mset^c}$ are shrunk to $0$ based on the KKT conditions, setting 
	$$
	\lambda_{1}
	= \max_{j \in \Mset^c} w_j^{-1}
	\left\{
	\left|
	\frac{2\tau - 1}{2n} \sum_{i=1}^n x_{ij}
	\;+\;
	\frac{1}{2n} \sum_{i \notin \mc I(\tilde \bvec_0)} \operatorname{sgn}(y_i - \xvec_i^T\tilde \bvec_0)\, x_{ij}
	\right|
	\;+\;
	\frac{1}{2n} \sum_{i \in \mc I(\tilde \bvec_0)} |x_{ij}|
	\right\}
	$$
	where $\tilde \bvec_0 = \argmin_{\bvec \in \reals^{p+1}} \frac{1}{n} \sum_{i = 1}^{n} \checkloss(y_i - \xvec_i^T \bvec)$ s.t. $\overline \Cmat \bvec = \bm \Psi, \bvec_{\Mset^c} = \bm 0$.
	We further speed up the $\bvec$ updates in Algorithm \ref{algo:ADMM reduced} with active set cycling, though we have omitted the details here in the interest of space.
	
	\subsection{Tuning $\lambda$}
	For this subsection, let $\mc T=[\tau_L,\tau_U]\subsetneq(0,1)$ with $\tau_L<\tau_U$.
	Following Section 3.3 of \citet{zheng2015globally}, let
	$\mathcal P_{e,n}$ be a nonempty finite set of candidate tuning parameters,
	$e\in\{0,a\}$, and let $K_{\mathrm G,n}\ge s\vee1$ be the maximum
	number of selected nuisance coordinates.
	Let
	\begin{align}
		\widehat{\Sset}_e(\lambda)
		&=\left\{j\in\Mset^c\setminus\{0\}:
		\sup_{\tau\in\mc T}|\widehat\beta^{(2)}_{e,j}(\tau;\lambda)|>0\right\}.
		\label{eqn:gic_selected_support}
	\end{align}
	Using the loss in \eqref{eqn:qndefine}, we define the integrated generalized information criterion (GIC) by
	\begin{align}
		\operatorname{GIC}_{e;\mc T}(\lambda)
		&=\int_{\mc T}\log\mc Q_{\tau,n}
		\{\widehat{\bvec}_e^{(2)}(\tau;\lambda)\}\,d\tau
		+|\widehat{\Sset}_e(\lambda)|\phi_n,
		\qquad \phi_n>0,
		\label{eqn:gic_lla}
	\end{align}
	and select
	\begin{equation}\label{eqn:gic_lambda_selector}
		\widehat\lambda_{e,n}\in
		\argmin_{\lambda\in\mathcal P_{e,n}:\,
			|\widehat{\Sset}_e(\lambda)|\le K_{\mathrm G,n}}
		\operatorname{GIC}_{e;\mc T}(\lambda).
	\end{equation}
	Our numerical implementation uses
	$
	\phi_n=\frac{\log(\log n)\log p}{n}.
	$
	
	Our next theorem establishes the selection consistency of the integrated GIC criterion \eqref{eqn:gic_lla}. For clarity of presentation, we state the additional conditions for this theorem in Section S2.4 of the supplement.
	\begin{theorem}\label{thm:GIC.consistency}
		Suppose \ref{as:iid array}, \ref{as:density conditions}, and
		the additional conditions in Section S2.4 hold. Then for $e=a$,
		\begin{equation}\label{eqn:gic_selection_consistency}
			P\{\widehat{\Sset}_e(\widehat\lambda_{e,n})\ne\Sset\}\longrightarrow0.
		\end{equation}
		For $e=0$, \eqref{eqn:gic_selection_consistency} holds under
		\eqref{hyp:linear hypothesis} or
		\ref{as:local alts} if
		\begin{equation}\label{eqn:gic_local_alt_rate}
			\sup_{\tau\in\mc T}\|\bm\eta(\tau)\|_2^2=o(n\phi_n).
		\end{equation}
	\end{theorem}
	
	\subsection{Estimating the process via a c\`adl\`ag grid approximation}\label{sec:grid_interpolation}
	
	Computing the complete two-step LLA coefficient path over
	$\mc T$ can be costly when $p$ is large.  Motivated by the
	deterministic-grid construction in Section 2.3 and Corollary B.1 of
	\citet{zheng2015globally}, we compute the reduced two-step LLA
	estimator only at finitely many quantile indices and extend the
	resulting knot estimates over $\mc T$.  
	Since the reduced estimator
	must satisfy a possibly nonconstant linear restriction, the ordinary
	step extension is corrected so that the restriction holds at every
	target quantile.
	
	For this subsection, let $\mc T=[\tau_L,\tau_U]\subsetneq(0,1)$ and let
	\[
	\mathcal G_n
	=
	\{\tau_{0,n},\tau_{1,n},\ldots,\tau_{N_{\mathcal G,n},n}\},
	\qquad
	\tau_{0,n}=\tau_L
	<
	\cdots
	<
	\tau_{N_{\mathcal G,n},n}=\tau_U,
	\]
	be a deterministic grid with mesh
	\[
	\Delta_n
	=
	\max_{1\leq k\leq N_{\mathcal G,n}}
	(\tau_{k,n}-\tau_{k-1,n}).
	\]
	We execute Algorithm \ref{algo:LLA for reduced estimator} only for
	$t\in\mathcal G_n$, replacing every supremum over $\mc T$ in its
	adaptive weights by the corresponding maximum over $\mathcal G_n$.
	Specifically,
	\[
	\widehat{\bvec}_{0,\mathcal G_n}^{(0)}(t)
	=
	\widehat{\bvec}_{0}^{(0)}(t),
	\qquad
	t\in\mathcal G_n,
	\]
	and the weights used in the $b^{\rm th}$ grid update are
	\[
	p_\lambda'\left\{
	\max_{u\in\mathcal G_n}
	\left|
	\widehat\beta_{0,\mathcal G_n,j}^{(b-1)}(u)
	\right|
	\right\},
	\qquad
	j\in \Mset^c.
	\]
	The corresponding constrained weighted quantile-regression problem in \eqref{algo:LLA update reduced} is then solved at each $t\in\mathcal G_n$, producing the knot estimate $\widehat{\bvec}_{0,\mathcal G_n}^{(b)}(t)$.
	
	Define $\pi_n(\tau)=\tau_{k-1,n}$ for
	$\tau\in[\tau_{k-1,n},\tau_{k,n})$, and set
	$\pi_n(\tau_U)=\tau_U$.  We extend the two-step knot estimates by
	\begin{align}
		\widehat{\bvec}_{0,\mathcal G_n,\Mset}^{(2)}(\tau)
		&=
		\widehat{\bvec}_{0,\mathcal G_n,\Mset}^{(2)}
		\{\pi_n(\tau)\}
		+
		\Cmat^T(\Cmat\Cmat^T)^{-1}
		\left[
		\bm\Psi(\tau)
		-
		\bm\Psi\{\pi_n(\tau)\}
		\right],
		\label{eqn:grid_interpolated_estimator}\\
		\widehat{\bvec}_{0,\mathcal G_n,\Mset^c}^{(2)}(\tau)
		&=
		\widehat{\bvec}_{0,\mathcal G_n,\Mset^c}^{(2)}
		\{\pi_n(\tau)\}.
		\nonumber
	\end{align}
	At every $t\in\mathcal G_n$, $\pi_n(t)=t$, so
	\eqref{eqn:grid_interpolated_estimator} agrees with the knot estimate.
	Moreover,
	\[
	\Cmat
	\widehat{\bvec}_{0,\mathcal G_n,\Mset}^{(2)}(\tau)
	=
	\bm\Psi(\tau),
	\qquad
	\tau\in\mc T.
	\]
	Under the conditions below,
	\eqref{eqn:grid_interpolated_estimator} is c\`adl\`ag.  When
	$\bm\Psi(\cdot)$ is constant, its correction term vanishes and the
	estimator reduces to the ordinary right-continuous grid step process.
	
	\begin{theorem}\label{thm:grid_interpolation}
		Suppose
		\ref{as:iid array}--\ref{as:new_re} and
		the additional conditions in Section S1.2 hold, together with
		$
		(m+s)\xi_n^2\log^2n=o(n).
		$
		Suppose also that $m-r+s\geq1$ and
		$
		\sqrt n\,\Delta_n=O(1).
		$
		Then the interpolated two-step LLA estimator recovers the reduced
		oracle at every knot:
		\[
		\Pr\left\{
		\max_{t\in\mathcal G_n}
		\left\|
		\widehat{\bvec}_{0,\mathcal G_n}^{(2)}(t)
		-
		\widehat{\bvec}_0^o(t)
		\right\|_2
		=0
		\right\}
		\longrightarrow1.
		\]
		Furthermore, for a finite constant $C_\star$ independent of $n$,
		\begin{equation}\label{eqn:grid_interpolation_bound}
			\sup_{\tau\in\mc T}
			\left\|
			\widehat{\bvec}_{0,\mathcal G_n}^{(2)}(\tau)
			-
			\bvec^*(\tau)
			\right\|_2
			\leq
			\max_{t\in\mathcal G_n}
			\left\|
			\widehat{\bvec}_{0,\mathcal G_n}^{(2)}(t)
			-
			\bvec^*(t)
			\right\|_2
			+
			C_\star\sqrt{m-r+s}\,\Delta_n.
		\end{equation}
		Consequently,
		\[
		\sup_{\tau\in\mc T}
		\left\|
		\widehat{\bvec}_{0,\mathcal G_n}^{(2)}(\tau)
		-
		\bvec^*(\tau)
		\right\|_2
		=
		O_p\left(
		\sqrt{\frac{m-r+s}{n}}
		\right).
		\]
	\end{theorem}
	\noindent We also establish the following result, which shows that the interpolated estimator admits the same uniform Bahadur representation as in Theorem \ref{thm:part pen bahadur}.
	
	\begin{corollary}\label{cor:grid_bahadur}
		Suppose \ref{as:iid array}--\ref{as:new_re} and
		the additional conditions in Section S1.2 hold, and that
		$
		(m+s)^3\xi_n^2=o(n^{1-\epsilon})
		$
		for some constant $\epsilon>0$.  If $m-r+s\geq1$ and
		$
		\sqrt{n(m+s)}\,\Delta_n=o(1),
		$
		then
		\begin{equation}\label{eqn:grid_exact_equivalence}
			\sup_{\tau\in\mc T}
			\sqrt n
			\left\|
			\widehat{\bvec}_{0,\mathcal G_n,\MSset}^{(2)}(\tau)
			-
			\widehat{\bvec}_{0,\MSset}^{(2)}(\tau)
			\right\|_2
			=o_p(1).
		\end{equation}
		Moreover, for every $\tau\in\mc T$,
		\begin{align}
			\sqrt n
			\bigg\{
			\widehat{\bvec}_{0,\mathcal G_n,\MSset}^{(2)}(\tau)
			-
			& \bvec_{\MSset}^*(\tau)
			\bigg\} \notag\\
			\notag
			&=
			[\bm J(\tau)]^{-1/2}
			\{\bm I-\bm P(\tau)\}
			[\bm J(\tau)]^{-1/2}
			\frac1{\sqrt n}
			\sum_{i\in[n]}
			\tilde\xvec_i
			\psi_\tau(y_i-\xvec_i^T\bvec^*(\tau))
			\notag\\
			&\quad-
			[\bm J(\tau)]^{-1}
			\widetilde{\bm C}^{T}
			\bm H^{-1}(\tau)
			\bm\eta(\tau)
			+
			\widetilde{\bm r}_{\mathcal G,n}(\tau),
			\label{eqn:RE.grid bahadur}
		\end{align}
		where $\sup_{\tau\in\mc T}
		\left\|
		\widetilde{\bm r}_{\mathcal G,n}(\tau)
		\right\|_2
		=o_p(1)$.
	\end{corollary}

	\section{Simulation Studies}\label{sec:sim study}
	
	\subsection{Estimation comparison}\label{sec: estimation comparison}
	To complement our theoretical results in Section \ref{sec:estimation theory}, we conduct a simulation study comparing the estimation and selection performance of our globally adaptive LLA estimator (denoted by ``LLA'' below) with the $\ell_1$ penalized quantile regression estimator \citep{BC.L1penalized} (``L1''), \citet{zheng2015globally}'s globally adaptive $\ell_1$ estimator (``Zheng''), and the oracle estimator (``Oracle'').
	We include LLA estimators for both the SCAD penalty \citep{Fan2001} and MCP \citep{Zhang2010}.
	In addition, we include locally adaptive versions of our LLA estimators and Zheng in the comparison, denoting the global and local versions by ``G-'' and ``L-'' respectively. 
	We initialize the LLA algorithms and Zheng with L1, selecting the penalty parameter $\lambda_{\rm lasso}$ using the procedure described in \citet{BC.L1penalized} with the $\alpha$ parameter tuned via 10-fold CV.
	We tune all estimators for the final comparison, including L1, using the integrated GIC \eqref{eqn:gic_lla}.
	
	We generate the response from a location-scale model 
	$$y_i = \beta_0^* + \xvec_i^T \bvec^* + ((x_{i1} + 4)/5)\eps_i \text{\; for \;} i= 1, \ldots, n,$$
	with $\beta_0^* = 1 $, $\bvec^* = (2, 0, -2, 0, 1, 0, -1, 0, \ldots)^T$, where $\xvec_i \overset{iid}{\sim} N(\bm 0, \bm \Sigma_{\xvec})$ with $ [\bm \Sigma_{\xvec}]_{ij} = 0.5^{|i-j|}$ for $i,j \in [p]$ and $\eps_i \overset{iid}{\sim} t_2$.
	In each replication we generate $n = 200$ observations with either $p = 100$ or $p = 400$ predictors.
	We fit the models at $80$ equally-spaced quantiles in $\mc T = [0.25, 0.75]$.
	
	We compute the mean $\ell_1$ estimation error, the mean $\ell_2$ estimation error, the mean number of false positive variable selections (FP), and the mean number of false negative variable selections (FN) across $\mc T$.
	The reported values are the means (and standard errors) of these metrics across $1000$ replications.
	The lowest two values for each metric in each setting (excepting the oracle estimator, which we include as a benchmark) are in \textbf{bold} .

	Table \ref{tbl:het est comp} reports the results from the estimation comparison.
	We see that the globally adaptive LLA estimators achieve the lowest mean $\ell_1$ and $\ell_2$ estimation errors and lowest mean number of false negative variable selections for both $p = 100$ and $p = 400$.
	The mean number of false positive variable selections is low for all methods, and the globally adaptive LLA estimators are at or near the lowest FP in both cases.
	Within each of the adaptive weighting schemes (LLA (SCAD), LLA (MCP), and Zheng), the globally adaptive weights consistently outperform the locally adaptive weights in terms of the $\ell_1$ and $\ell_2$ estimation errors and FN. 
	
	\begin{table}
		\centering
		\caption{Heteroskedastic model estimation comparison} \label{tbl:het est comp}
		\begin{tabular}[t]{l|lllll}
			& $p$ & $\ell_1$ error & $\ell_2$ error & FP & FN\\
			\hline
			G-LLA(SCAD) & 100 & \textbf{0.369 (0.004)} & \textbf{0.045 (0.001)} & \textbf{0.004 (0.001)} & \textbf{0 (0)}\\
			G-LLA(MCP) & 100 & \textbf{0.362 (0.003)} & \textbf{0.042 (0.001)} & \textbf{0.003 (0.002)} & \textbf{0 (0)}\\
			L-LLA(SCAD) & 100 & 0.559 (0.01) & 0.185 (0.007) & 0.038 (0.004) & 0.031 (0.004)\\
			L-LLA(MCP) & 100 & 0.465 (0.007) & 0.116 (0.004) & 0.017 (0.003) & 0.018 (0.002)\\
			G-Zheng & 100 & 0.562 (0.016) & 0.175 (0.012) & 0.05 (0.007) & 0.087 (0.009)\\
			L-Zheng & 100 & 1.104 (0.022) & 0.643 (0.019) & 0.016 (0.002) & 0.45 (0.015)\\
			L1 & 100 & 3.568 (0.022) & 3.033 (0.037) & 0.088 (0.005) & 0.504 (0.018)\\
			Oracle & 100 & 0.358 (0.003) & 0.041 (0.001) & 0 (0) & 0 (0)\\
			\hline
			G-LLA(SCAD) & 400 & \textbf{0.516 (0.012)} & \textbf{0.132 (0.01)} & 0.024 (0.003) & \textbf{0.04 (0.009)}\\
			G-LLA(MCP) & 400 & \textbf{0.417 (0.008)} & \textbf{0.072 (0.005)} & \textbf{0.015 (0.003)} & \textbf{0.011 (0.005)}\\
			L-LLA(SCAD) & 400 & 1.135 (0.022) & 0.686 (0.021) & 0.055 (0.004) & 0.294 (0.019)\\
			L-LLA(MCP) & 400 & 0.87 (0.017) & 0.444 (0.015) & 0.034 (0.003) & 0.171 (0.013)\\
			G-Zheng & 400 & 1.384 (0.026) & 0.859 (0.022) & 0.033 (0.006) & 0.632 (0.019)\\
			L-Zheng & 400 & 2.199 (0.021) & 1.676 (0.022) & \textbf{0.004 (0.001)} & 1.309 (0.016)\\
			L1 & 400 & 4.731 (0.019) & 5.161 (0.043) & 0.024 (0.003) & 1.359 (0.017)\\
			Oracle & 400 & 0.357 (0.003) & 0.041 (0.001) & 0 (0) & 0 (0)\\
			\hline
		\end{tabular}
	\end{table}
	
	\subsection{Test comparison}\label{sec:additional simulation}
	For our next set of simulations, we compare our partial penalized Wald and rank score tests evaluated at the two-step global LLA estimators (denoted by ``Q/Wald'' and ``Q/RS'' below) with their oracle counterparts (``O/Wald'' and ``O/RS''), the debiased test of \citet{bradic2017uniform} (``BK''), and the debiased test of \citet{yan2023confidence} (``Yan'') at global-over-$\mc T$ and pointwise-at-$\tau$ hypothesis testing. Note that \citet{yan2023confidence}'s test only applies pointwise-at-$\tau$ and therefore is excluded from the global test comparison.
	
	Throughout the simulations in this section, we generate the response from the following location-scale model:
	\begin{equation}	
		y_i= \beta_0^* +\xvec_i^T\bvec_h^*+
		\left(1+\frac{x_{i1}}4\right)\eps_i,
		\qquad \eps_i\overset{\mathrm{iid}}{\sim}t_5,
		\label{eqn:testing sim model}
	\end{equation}
	where $\beta_0^* = 1$ and
	$\bvec_h^*=(2,h,-2,0,1+h/2,0,-1+h/2,0,\ldots)^T$.
	In both the global and pointwise inference comparisons, we set $n=300$ and $p\in\{150,450\}$.
	Let $U_i$ be Uniform$(0,1)$ and let $\eta_{ij}$ be independent standard
	normal variables.  We generate
	$
	x_{i1}=\sqrt{3}(2U_i-1)$ and $x_{ij}=\frac12x_{i,j-1}+\frac{\sqrt3}{2}\eta_{ij}$ for $j\geq2$.
	Every coordinate has mean zero and variance one, with
	$\operatorname{Cov}(x_{ij},x_{ik})=2^{-|j-k|}$.  
	Because $x_{i1}\in[-\sqrt3,\sqrt3]$, the conditional scale lies in
	$[1-\sqrt3/4,1+\sqrt3/4]$ and is positive. 
	Under this model, the conditional-quantile intercept is
	$1+F_{t_5}^{-1}(\tau)$, the coefficient of $x_{i1}$ is $2+F_{t_5}^{-1}(\tau)/4$, and all other
	slopes equal their location coefficients.

	We consider tests of the following hypotheses:
	\begin{align*}
	  \text{H}_0^{(1)}:\quad & \beta_2^*(\tau)=0,\\
	  \text{H}_0^{(2)}:\quad &
	  \bigl(\beta_2^*(\tau),\beta_4^*(\tau),\beta_6^*(\tau)\bigr)^T=\bm 0,\\
	  \text{H}_0^{(3)}:\quad & \beta_5^*(\tau)+\beta_7^*(\tau)=0.
	\end{align*}
	These cases demonstrate the range of hypotheses our procedures can test, including a simultaneous test of multiple coefficients in H$_0^{(2)}$ and a test of contrasts in H$_0^{(3)}$. 
	The departures from these null hypotheses under \eqref{eqn:testing sim model} are $h$, $(h,0,0)^T$, and $h$, respectively.
	We vary $h\in\{0,0.1,0.2,0.3,0.4\}$ to include simulations under the null hypotheses and increasingly different alternatives.
	 
	\subsubsection{Global inference comparison}
	For global inference, we estimate the coefficient process at 81 equally-spaced knots in $\mc T = [0.25, 0.75]$ and use the interpolation described in Section \ref{sec:grid_interpolation}. 
	The QLLA estimators use the globally concerned LLA weights for the SCAD penalty, and $\lambda$ is selected by the integrated GIC \eqref{eqn:gic_lla}.  
	For the QLLA and oracle procedures, we estimate the Jacobian using \citet{Powell1984}'s estimator with the Hall--Sheather bandwidth \citep[Section~3.4]{koenker_quantile_2005}.
	
	Figures \ref{fig:add_hetero_p150_band_rates} and \ref{fig:add_hetero_p450_band_rates} show the rejection rates from the global inference comparison over 500 simulation replications when $p = 150$ and $p = 450$, respectively, with all tests being performed at level $\alpha = 0.05$.
	For both $p = 150$ and $p = 450$, we see that the rejection rates for the QLLA Wald and rank score tests track their oracle counterparts closely.
	The QLLA and oracle tests achieve rejection rates near their nominal size when $h = 0$ for all three tests while BK is notably anti-conservative for H$_{0}^{(3)}$.
	As $h$ increases and the truth moves farther from the null hypotheses, we see that the QLLA tests are consistently more powerful than BK and that this gap grows with $h$, as the theoretical power analysis in Section S2.5 of the supplement would suggest.
	Among the two QLLA tests, the rank score test outperforms the Wald test, nearly matching its oracle counterpart's performance and achieving greater power for all values of $h$.

	\begin{figure}[t]
		\centering
		\includegraphics[width=\textwidth]{"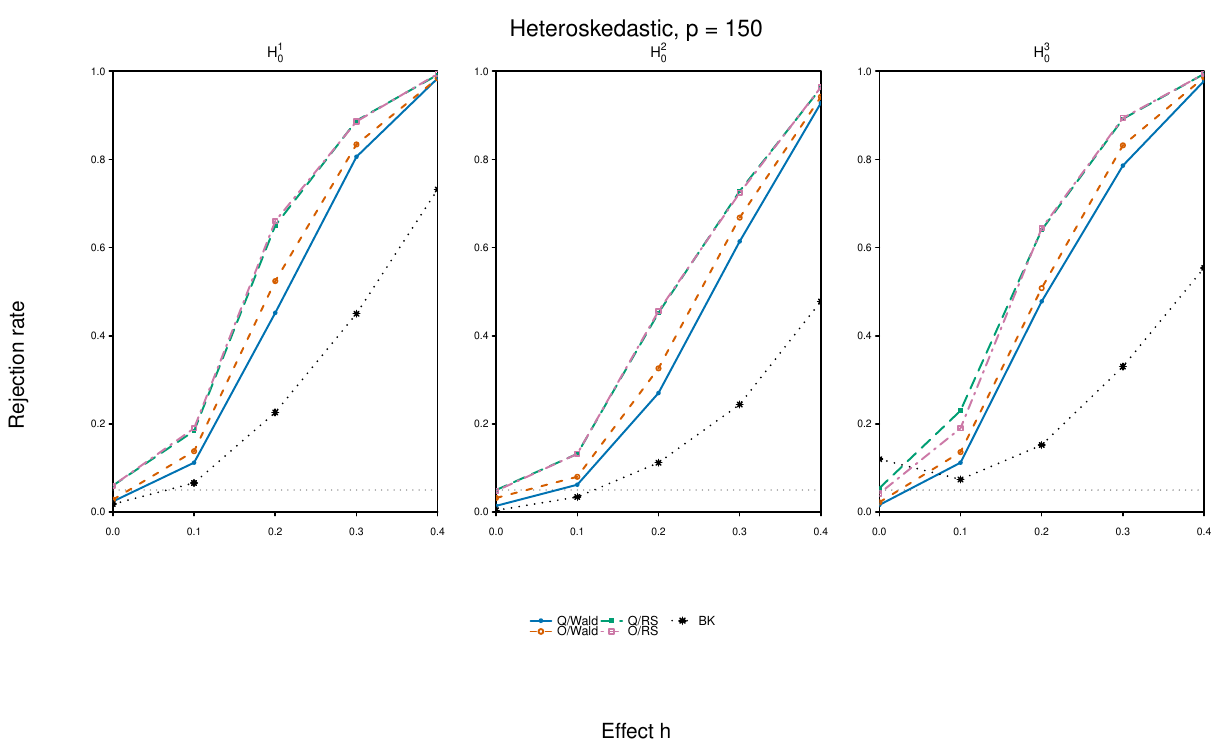"}
		\caption{Global rejection rates at level $0.05$ in the heteroskedastic $p=150$ design.}
		\label{fig:add_hetero_p150_band_rates}
	\end{figure}
	
	\begin{figure}[t]
		\centering
		\includegraphics[width=\textwidth]{"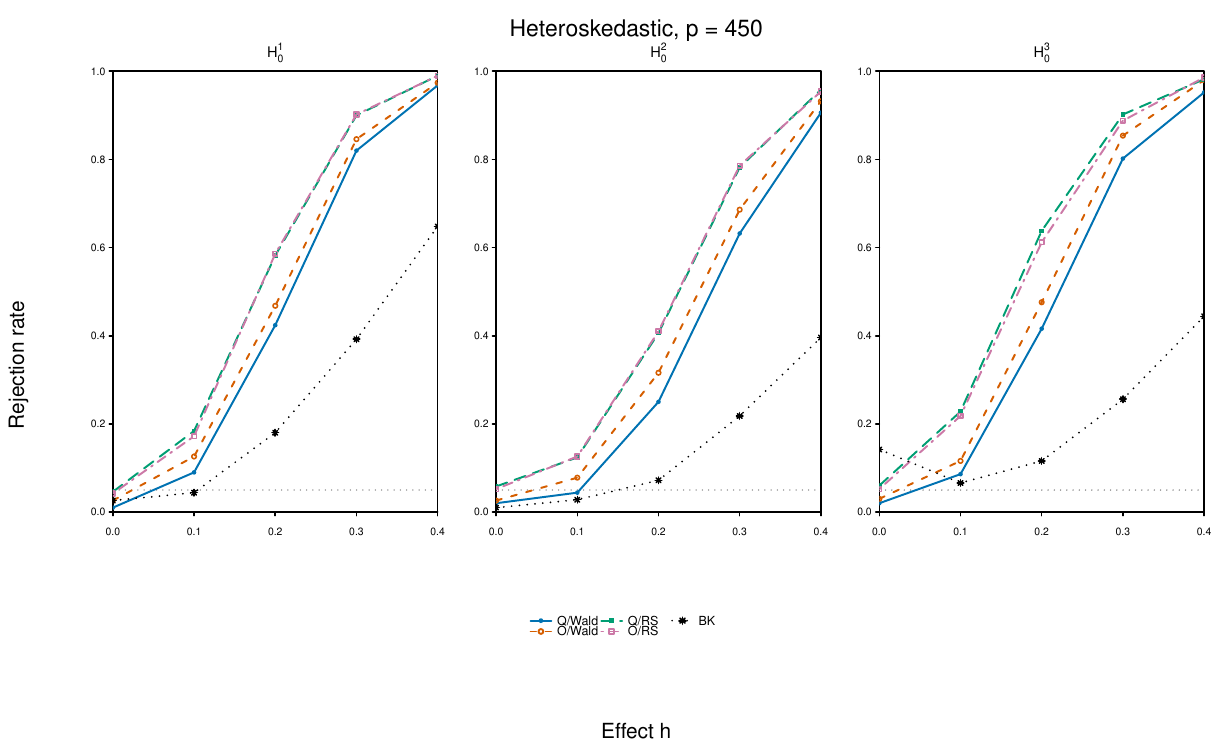"}
		\caption{Global rejection rates at level $0.05$ in the heteroskedastic $p=450$ design.}
		\label{fig:add_hetero_p450_band_rates}
	\end{figure}
	
	\subsubsection{Pointwise inference comparison}
	For pointwise inference, we examine the quantiles $\tau\in\{0.275,0.5,0.725\}$, which are knots of the grid from the global analysis. 
	In the results below, ``Q-G'' denotes the globally weighted QLLA estimator evaluated at a given quantile without refitting and ``Q-L'' denotes a separate QLLA fit at that quantile using the local weights	$p_\lambda'(|\widetilde\beta_j(\tau)|)$ and the pointwise analogue of \eqref{eqn:gic_lla}, obtained by replacing the integral and global selected support by the loss and support at that quantile.  The oracle and BK estimators are likewise evaluated from their global fits.  
	We also include Yan in the comparison for H$_0^{(1)}$ and H$_0^{(3)}$, though not H$_0^{(2)}$ as it structurally inapplicable there.  
	
	Figures \ref{fig:add_hetero_p150_point_rates} and \ref{fig:add_hetero_p450_point_rates} show the rejection rates from the pointwise inference comparison when $p = 150$ and $p = 450$, respectively; again the results are from 500 replications, and all tests are performed at level $\alpha = 0.05$.
	As in the global comparison, the QLLA tests track their oracle counterparts closely and achieve rejection rates near their nominal size under the null hypotheses.
	Notably, the tests using global QLLA weights achieve rejection rates near those of the tests using local weights. Given these results and the superior performance of the global weights in the estimation comparison in Section \ref{sec: estimation comparison}, our recommendation is to use the global weights in practice, even for pointwise analyses.
	Comparing the power with Yan and BK, we see that Yan achieves rejection rates closer to the QLLA rates than BK does, particularly with $\tau = 0.5$ and $p = 150$; for $\tau \in \{0.275,0.725\}$ and $p = 450$, however, the QLLA tests clearly attain higher rejection rates than Yan.
	
	\begin{figure}
	  \centering
	  \includegraphics[width=\textwidth]{"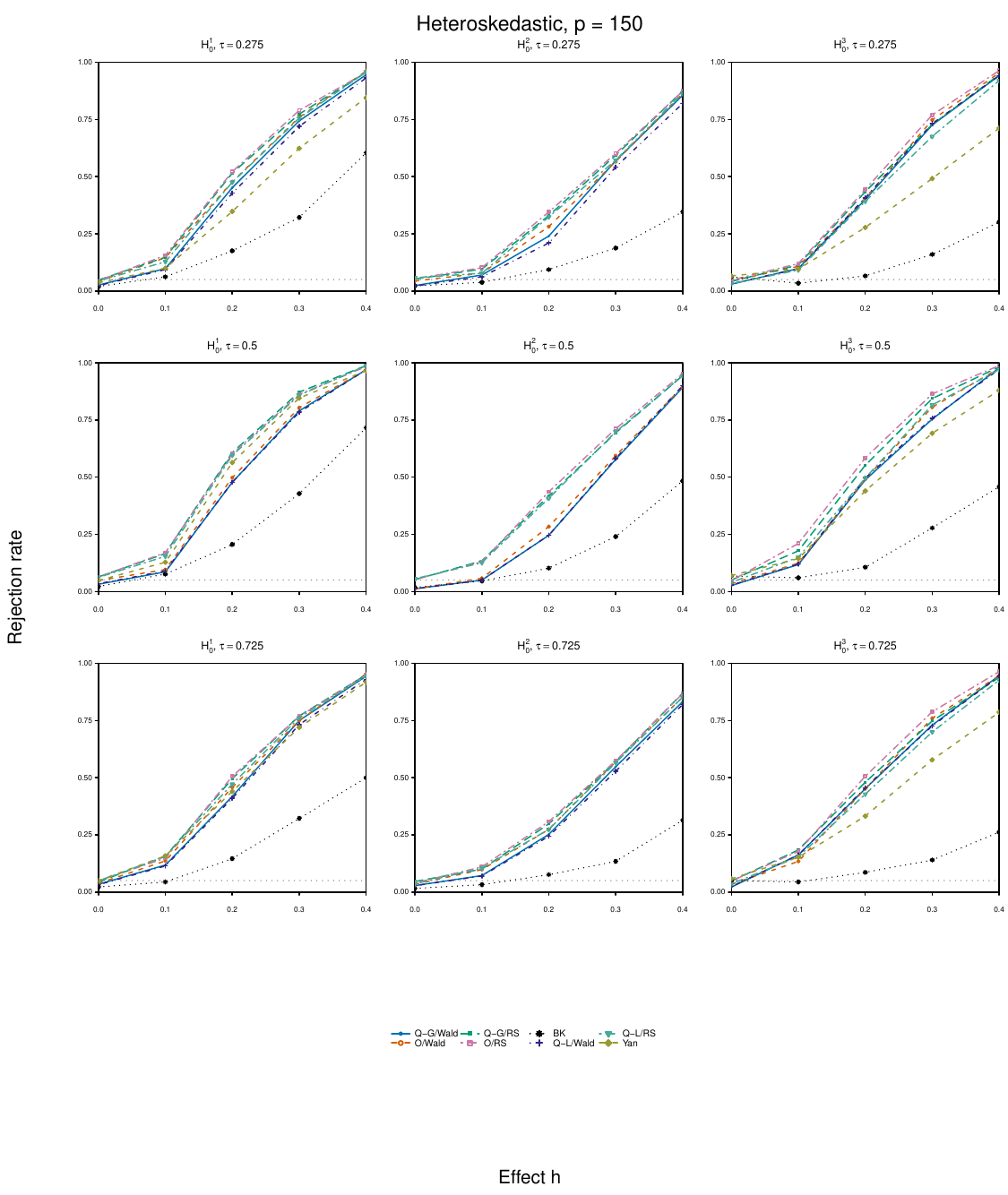"}
	  \caption{Pointwise rejection rates at level $0.05$ in the
	  heteroskedastic $p=150$ design.}
	  \label{fig:add_hetero_p150_point_rates}
	\end{figure}
	
	\begin{figure}
	  \centering
	  \includegraphics[width=\textwidth]{"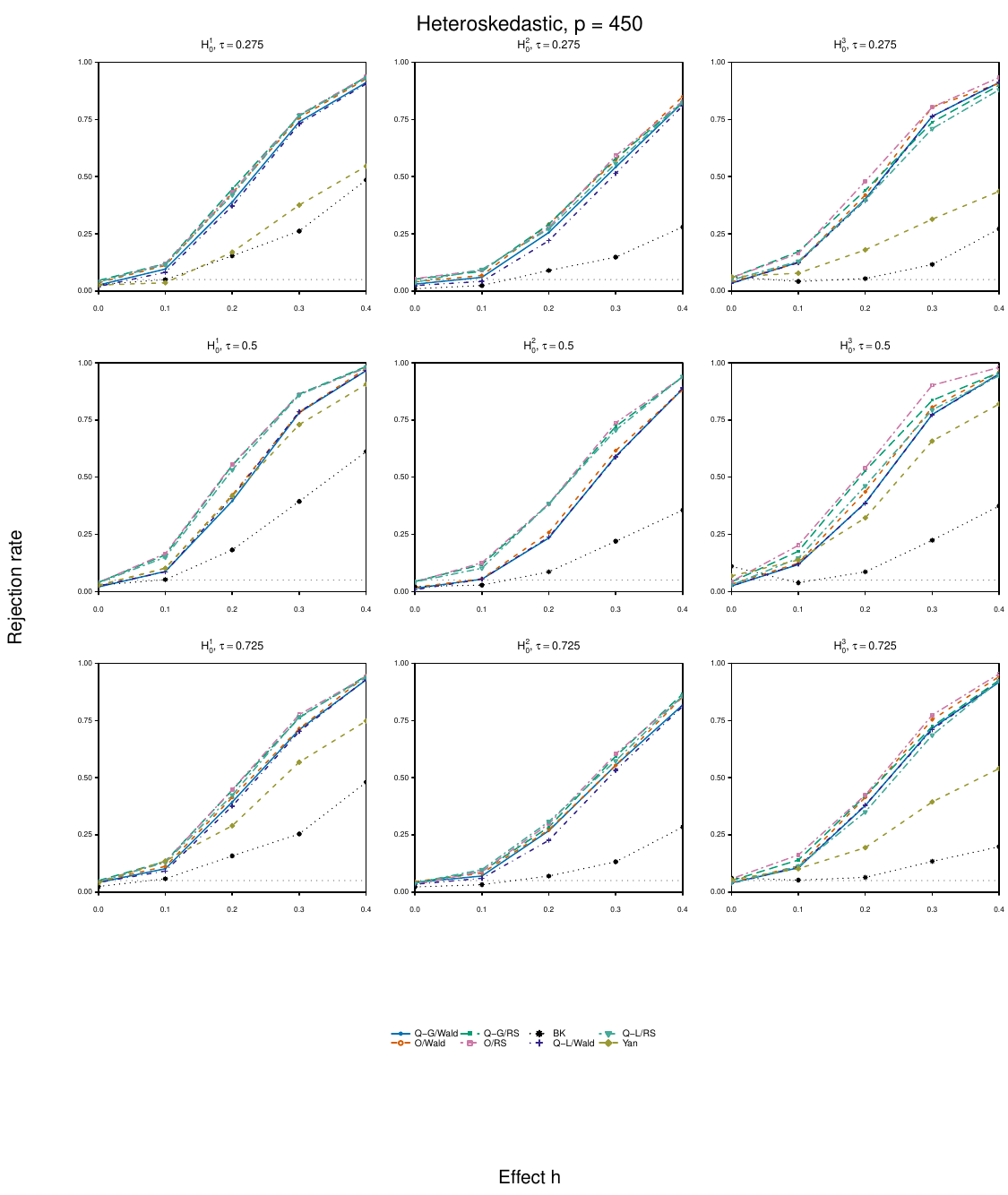"}
	  \caption{Pointwise rejection rates at level $0.05$ in the
	  heteroskedastic $p=450$ design.}
	  \label{fig:add_hetero_p450_point_rates}
	\end{figure}
	
	\subsubsection{Additional simulations}
	In addition to the power comparison results presented here, Section S3.1 of the supplement contains QQ-plots of the $p$-values under the null hypotheses for each of the tests performed in this section.
	Additional estimation and inference comparisons under a homoskedastic error model are presented in Sections S3.2 and S3.3, respectively.
	
	\section{Analysis of Lung Capacity Data}\label{sec:real data}
	To demonstrate the use of the proposed partial penalized tests, we analyze the lung capacity data set from \citet{pmid27609773}, regressing participants' FEV1\%---a measure of actual lung strength relative to predicted lung strength given one's demographics---on their lung disease status, smoking status, age, sex, and expression levels for 15,261 genes. The data come from 148 participants, including 72 with chronic obstructive lung disease, 59 with interstitial lung disease, and 17 with no lung disease and 115 smokers and 33 non-smokers.
	
	Let $\beta_{\rm CO}$, $\beta_{\rm I}$, and $\beta_{\rm SM}$ denote the coefficients for the indicators for  chronic obstructive lung disease, interstitial lung disease, and smoking, respectively.
	We test the following hypotheses:
	\begin{itemize}
		\item ``Disease vs. Smoking'' - ${\rm H}_0: (\beta_{\rm CO}^*, \beta_{\rm I}^* )^T = (\beta_{\rm SM}^*, \beta_{\rm SM}^*)^T $
		\item ``Comparing Diseases'' - ${\rm H}_0: \beta_{\rm CO}^* =  \beta_{\rm I}^* $
	\end{itemize}
	The first test compares the effects of the two lung diseases on FEV1\%  with the effect of smoking while the second compares the disease effects with each other.
	
	We examine three different intervals of quantiles $\mc T \in \{ [.1,.3], [.4,.6], [.5,.7]\}$, capturing the lower, middle, and upper quantiles of the conditional distribution of FEV1\%.
	The coefficients for age, sex, and the gene expression levels are nuisance parameters and comprise the elements of $\bvec_{\Mset^c}$. We standardized all continuous predictors to have mean 0 and standard deviation 1 prior to our analysis.
	As in our simulation studies, we use \citet{Powell1984}'s estimator of the Jacobian \eqref{eqn:Jacobian (part pen)}, with the Hall-Sheather bandwidth \citep[Section~3.4]{koenker_quantile_2005}.
	
	Table \ref{tab:empiricalResults} presents the test statistics and $p$-values for the partial penalized rank score and Wald tests of these hypotheses. Testing at the $0.05$ significance level, the rank score and Wald test results agree in every case but one---namely, testing the ``disease vs. smoking'' hypothesis at $\mc T = [0.1,0.3]$. In every other case, both tests reject the null hypothesis, leading us to conclude that the effects of the two diseases on FEV1\% differ from that of smoking over $\mc T = [0.4,0.6]$ and $\mc T = [0.7,0.9]$ and that the effects of the two diseases differ from each other over all three intervals.
	
	\begin{table}[ht]
		\centering
		\caption{Rank score and Wald test results regarding FEV1\%}
		\label{tab:empiricalResults}
		\begin{tabular}{llrrrr}
			\hline
			& & \multicolumn{2}{c}{Rank score test} & \multicolumn{2}{c}{Wald test} \\
			\cline{3-4} \cline{5-6}
			Test & $\mathcal{T}$ & Statistic & $p$-value & Statistic & $p$-value \\
			\hline
			Disease vs.\ Smoking
			& $[.1,.3]$    & 8.02  & 0.08 & 10.01 & 0.04 \\
			& $[.4,.6]$ & 9.99  & 0.03 & 16.95 & $<0.01$ \\
			& $[.7,.9]$   & 11.45 & 0.02 & 11.06 & 0.02 \\
			\hline
			Comparing Diseases
			& $[.1,.3]$    & 8.65  & 0.02 & 9.83 & 0.01 \\
			& $[.4,.6]$ & 12.38 & <0.01 & 10.96 & <0.01 \\
			& $[.7,.9]$   & 10.29 & <0.01 & 6.82 & 0.04 \\
			\hline
		\end{tabular}
	\end{table}
	
	\bibliographystyle{apalike}
	\bibliography{lla}
	
	\pagebreak
	
	\appendix
	
	\setcounter{section}{0}
	\setcounter{equation}{0}
	
	\def\theequation{S\arabic{equation}}
	\def\thesection{S\arabic{section}}
	\def\thefigure{S\arabic{figure}}
	\def\thetable{S\arabic{table}}
	
	\newcommand\myspacing{1}
	
	\def\spacingset#1{\renewcommand{\baselinestretch}%
		{#1}\small\normalsize} 
	

\section{Technical conditions}

\subsection{Equivalent assumptions under a location-scale model} \label{app:equivalent location-scale assumptions}
Studies of inference for the quantile regression process often assume the data are generated from a \emph{location-scale model}:
\begin{equation}
		y_i = \xvec_i^T\bvec^* + \xvec_i^T \bm \gamma^* \varepsilon_i \label{as:location-scale model}
	\end{equation}
where $\varepsilon_1, \ldots, \varepsilon_n$ are i.i.d. random errors with distribution function $F_{\eps}(\cdot)$ and density function $f_{\eps}(\cdot)$, and $\bvec^*, \bm \gamma^* \in \reals^{p+1}$ are vectors of unknown parameters.
To situate our assumptions within the quantile inference literature, we will briefly show that Conditions \ref{as:iid array} and \ref{as:density conditions} are satisfied if the data are generated from a location-scale model and some standard conditions hold.

We can rewrite \eqref{as:location-scale model} as
\begin{equation}
		Q_{y}(\tau|\xvec)
		= \xvec^T[ \bvec^* + \bm \gamma^* F_{\eps}^{-1}(\tau) ]
		= \xvec^T \bvec^*(\tau)
	\end{equation}
where $ \bvec^*(\tau) = \bvec^* + \bm \gamma^* F_{\eps}^{-1}(\tau) $, making the location-scale model a special case of \eqref{eqn:conditional quantile linear expr}. \ref{as:iid array} is clearly satisfied since the data are i.i.d.
Under \eqref{as:location-scale model} the following regularity conditions guarantee that \ref{as:density conditions} holds:
\begin{enumerate}[label=(LS\arabic*), start = 1 , leftmargin=2\parindent]
		\item $f_{\eps}$ is bounded from above uniformly over the domain of $\varepsilon$ and $n$. 
		
		\item $f_{\eps}(F_{\eps}^{-1}(\tau))$ is bounded away from $0$ uniformly over $\bar{\mc T}$ and $n$.
		
		\item $f(\varepsilon)$ is continuously differentiable and its derivative is bounded in absolute value from above uniformly over the domain of $\varepsilon$ and $n$.
		
		\item $\xvec^T\bm \gamma$ is bounded uniformly from above and away from $0$ over $\xvec \in \mc X$ and $n$.
	\end{enumerate}
These conditions are standard in studies of inferential procedures for the location-scale model (see, for example, \citet{koenker_goodness_1999}).
To establish the equivalence of these conditions to \ref{as:density conditions}, we note that under \eqref{as:location-scale model}, $f_{y|\xvec}(y) = f_{\eps}( (y - \xvec^T\bvec^*)/\xvec^T\bm\gamma^* ) \cdot |\xvec^T \bm\gamma^*|^{-1} $,
$ 
f_{y|\xvec}(Q_y(\tau|\xvec) ) 
= f_{\eps}(F_{\eps}^{-1}(\tau)) \cdot |\xvec^T \bm\gamma^*|^{-1}
$, and 
$\frac{d}{dy} f_{y|\xvec}(y) = f_{\epsilon}( (y - \xvec^T\bvec^*)/\xvec^T\bm\gamma^* ) \cdot |\xvec^T \bm\gamma^*|^{-2}$.

\subsection{Conditions for $\ell_1$-penalized estimator convergence} \label{app:BC l1 conditions}

The following additional conditions from \citet{BC.L1penalized} are required to prove Lemma \ref{lemma:global33modAssumptions}.
The first condition is similar to Assumption \ref{as:eigens}. Before stating this condition, let $c_0$ be a positive constant and define
\begin{equation*}
	D_\tau \defeq \left\{\bm \delta \in \reals^{p+1}:||\bm \delta_{\mc A^c_{\tau}}||_1 \leq c_0||\bm \delta_{\Sset_{\tau}}||_1, ||\bm \delta_{\Sset_{\tau}}^c||_0 \leq n \right\}.
\end{equation*} 
In addition, let
$\bar{\mc A}_{\tau}(\bm \delta, k)$ as the $k$ largest in absolute value elements of $\bm \delta$ in $\mc A^c_{\tau}$. 
\begin{enumerate}[label=(Z\arabic*), start = 1 , leftmargin=2\parindent]
	\item \label{as:l1re} For some positive $k$ and $c_0 \geq 9$, $\inf_{\tau \in \mc T}\inf_{\bm \delta \in D_\tau, ||\bm \delta||_2 \neq 0} \frac{\bm \delta^T E[\xvec_i\xvec_i^T]\bm \delta}{||\bm \delta_{\mc A_{\tau} \cup \bar{\mc A}_{\tau}(\bm \delta, k)||}}$ has a positive uniform lower bound. In addition, $\inf_{\tau \in \mc T}\inf_{\bm \delta \in D_{\tau}, ||\bm \delta||_2 \neq 0} \frac{E[|\xvec_i^T\bm \delta|^2]^{3/2}}{E|\xvec_i^T\bm \delta|^3}$ has a positive uniform lower bound. 
	\item \label{as:lambda lasso conditions}
	$\lambda_{\rm lasso} = O(\sqrt{n \log p})$ and $\lambda_{\rm lasso} \geq \lambda_{\rm lasso,0}$ where $\lambda_{\rm lasso,0} = \frac{c_0+3}{c_0-3} \Lambda(1-\alpha|X)$, where $1/p \leq \alpha$, and $\Lambda(1-\alpha|X)$ is the $1-\alpha$ quantile of $$\Lambda = n \sup_{\tau \in \mc T}\max_{j \in \{1,\ldots,p\}}\left|\frac{1}{n} \sum_{i=1}^n x_{ij}(\tau-I(u_i \leq \tau)) \hat{\sigma}_j^{-1}[\tau(1-\tau)]^{-1/2} \right|.$$
	\item \label{as:s_size} $s \leq n/\log(n \vee p)$. 
\end{enumerate}

\section{Technical proofs}

\subsection{Strong oracle proofs}

In this section we provide proofs for Theorems \ref{thm:global33mod} and \ref{thm:global33modStep2} and Lemmas \ref{lemma:global33modAssumptions} and \ref{lemma:global33modAssumptions2}. A proof of Theorem \ref{thm:oracle rate} is provided in the next section. Only the proofs for the case of $e=0$, the restricted case, are provided. This is because the proof for $e=a$ is a simpler version that does not have to deal with the case of the estimator being potentially biased by the restriction that $\Cmat \bvec_{\mc M}(\tau)=\bm \Psi (\tau)$. 

\begin{proof}[Proof of Theorem \ref{thm:global33mod}]
	
	Under \ref{as:strength of signal} and $\sup_{\tau \in \mc T }||\hat{\bvec}^{(0)}_{0,\Mset^c}(\tau)-\bvec^*_{\Mset^c}(\tau)||_{\infty} \leq d_0\lambda$, we have
	\begin{equation*}
		\inf_{j \in \Sset}\sup_{\tau \in \mc T}|\hat{\beta}_{0,j}^{(0)}(\tau)| \geq \inf_{j \in \Sset}\sup_{\tau \in \mc T} |\beta_j^*(\tau)| - \sup_{\tau \in \mc T }||\hat{\bvec}^{(0)}_{0,\Mset^c}(\tau)-\bvec^*_{\Mset^c}(\tau)||_{\infty} > (d+1)\lambda - d_0\lambda \geq 
		d\lambda. 
	\end{equation*}
	Therefore by the definition of $p_\lambda(\cdot)$, $p'_\lambda(\sup_{\tau \in \mc T}|\hat{\beta}_{0,j}^{(0)}(\tau)|)=0$ for all $j \in \Sset$ and $\tau \in \mc T$. Thus under the assumed conditions, $\hat{\bvec}^{(1)}_0(\tau)$ is equivalent to 
	\begin{equation}\label{eqn:reMod}
		\hat\bvec_{0}^{(1)} (\tau)= \argmin_{\bvec \in \reals^{p+1}} \frac{1}{n}\sum_{i=1}^n \rho_{\tau}[y_i-\bm x_i^T\bvec(\tau)]+\sum_{j \in \MSset^c} p'_{\lambda}(|\beta^{(0)}_j(\tau)|)|\beta_j| {\rm ~s.t.~} \Cmat \bvec_{\mc M} (\tau) =\bm \Psi (\tau) \text{~for~} \tau \in \mc T.
	\end{equation}
	
	By definition of a subdifferential and the definition of $\tilde{\bm w}(\tau)$ in Assumption \ref{as:oracleDefined}, for any $\bvec$ where $\bm C\bvec_{\mc M}=\bm \Psi(\tau)$,
	\begin{align}
		\mc Q_{\tau,n}(\bvec)- \mc Q_{\tau,n}[\hat{\bvec}^o_0(\tau)]  &\geq \tilde{\bm w}_{\MSset^c}^T(\tau)[\bvec_{\MSset^c} - \hat{\bvec}^o_{0,\MSset^c}(\tau)] + \bm d(\tau)^T \Cmat[\bvec_{(\MSset)} - \hat{\bvec}^o_{0,(\MSset)}(\tau)] \nonumber \\
		&= \tilde{\bm w}_{\MSset^c}^T(\tau)\bvec_{\MSset^c}. \label{objDiff}
	\end{align} 
	
	
	Under Assumption \ref{as:strength of signal} and $\sup_{\tau \in \mc T }||\hat{\bvec}^{(0)}_{0,\Mset^c}(\tau)-\bvec^*_{\Mset^c}(\tau)||_{\infty} \leq d_0\lambda$, for $j \in \MSset^c$,
	\begin{equation*}
		\sup_{\tau \in \mc T}|\hat{\beta}_{0,j}^0(\tau)| \leq \sup_{\tau \in \mc T}||\hat{\bvec}^{(0)}_{0,\Mset^c}||_{\infty} = \sup_{\tau \in \mc T}||\hat{\bvec}^{(0)}_{0,\Mset^c}-\bvec^*_{\Mset^c}(\tau)||_{\infty}\leq d_0\lambda \leq d_2\lambda. 
	\end{equation*}  
	Therefore $p'_\lambda[\hat{\beta}_j^{(0)}] \geq d_1\lambda$ for $j \in \MSset^c$. Combining this with $\sup_{\tau \in \mc T}||\tilde{\bm w}_{\MSset^c}(\tau)||_{\infty} < d_1\lambda$, and \eqref{objDiff} we have that for every $\tau \in \mc T$
	\begin{align*}
		\mc Q_{\tau,n}(\beta) + & \sum_{j \in \MSset^c} p'_{\lambda}(|\hat{\beta}^{(0)}_j(\tau)|)|\beta_j| - \left( \mc Q_{\tau,n}[\hat{\beta}^o(\tau)] +
		\sum_{j \in \MSset^c} p'_{\lambda}(|\hat{\beta}^{o}_j(\tau)|)|\hat{\beta}_j^o|\right) \\
		&= \mc Q_{\tau,n}(\beta) - \mc Q_{\tau,n}[\hat{\beta}^o(\tau)] + \sum_{j \in \MSset^c} p'_{\lambda}(|\hat{\beta}^{(0)}_j(\tau)|)|\beta_j| \\
		&\geq \tilde{\bm w}(\tau)_{\MSset^c}^T\bvec_{\MSset} + \sum_{j \in \MSset^c} p'_{\lambda}(|\hat{\beta}^{(0)}_j(\tau)|)|\beta_j|\\
		&= \sum_{j \in \MSset^c} p'_{\lambda}(|\hat{\beta}^{(0)}_j(\tau)|)|\beta_j|+\tilde{w}_j(\tau)\beta_j   \\
		&= \sum_{j \in \MSset^c} \left[p'_{\lambda}(|\hat{\beta}^{(0)}_j(\tau)|)+\tilde{w}_j(\tau)\rm sgn(\beta_j)\right]|\beta_j| \\
		&\geq \sum_{j \in \MSset^c} \left[d_1\lambda -|\tilde{w}_j(\tau)\rm sgn(\beta_j) | \right]|\beta_j| \geq 0.
	\end{align*}
	Equality will hold for the final lower bound only if $\beta_j=0$ for all $j \in \MSset^c$. Thus for all $\tau \in \mc T$, there does not exist a $\beta$ that will provide a lower value of the penalized objective function than $\hat{\bvec}^o_0(\tau)$ and thus for all $\tau \in \mc T$ we have $\hat{\bvec}^{(1)}_0(\tau)=\hat{\bvec}^o_0(\tau)$. 
\end{proof}

\begin{proof}[Proof of Theorem \ref{thm:global33modStep2}]
	Recall that we assume that the oracle estimator has been found in the previous iteration, that is $\hat{\bvec}^{(b-1)}_0(\tau)=\hat{\bvec}^{o}_0(\tau)$. Thus the proof can be completed by modifying the proof of Theorem \ref{thm:global33mod} with the oracle estimator as the initial estimator. 
	Under the event $\inf_{j \in \Sset} \sup_{\tau \in \mc T}|\hat{\beta}^o_{0,j}(\tau)| > d\lambda$, it follows that $p_\lambda'(\sup_{\tau \in \mc T}|\hat{\beta}^o_{0,j}(\tau)|)=0$ for all $j \in \Sset$. 
	In addition, since $\sup_{\tau \in \mc T, j \in \MSset^c} |\hat{\beta}^o_{0,j}(\tau)|=0$ it follows that $p'_{\lambda}(\sup_{\tau \in \mc T}|\hat{\beta}^o_{0,j}(\tau)|) = p'_\lambda(0) \geq d_1\lambda$ for all $j \in \MSset^c$. 
	Combining these two results with the event $\sup_{\tau \in \mc T}||\tilde{\bm w}_{\MSset^c}(\tau)||_{\infty} < d_1\lambda$, the proof follows the same structure as the proof of Theorem \ref{thm:global33mod}. 
\end{proof}

\begin{proof}[Proof of Lemma \ref{lemma:global33modAssumptions}]
	Theorem 2 of \citet{BC.L1penalized} provides that
	$$ \sup_{\tau \in \mc T} \norm{\hat{\bvec}^{\rm lasso}(\tau)-\bvec^*(\tau)}_2=O_p\left(\sqrt{\frac{(s+m)\log (p \vee n)}{n}}\right)$$ 
	under assumptions \ref{as:iid array}, \ref{as:density conditions}, \ref{as:xbounds}, and \ref{as:l1re} - \ref{as:s_size}.
	Therefore under \ref{as:lambda} 
	$$\sup_{\tau \in \mc T}\norm{ \hat{\bvec}^{\rm lasso}_{\Mset^c}(\tau)-\bvec^*_{\Mset^c}(\tau)}_{\infty} \leq \sup_{\tau \in \mc T} \norm{\hat{\bvec}^{\rm lasso}(\tau)-\bvec^*(\tau)}_2=o_p(\lambda),$$ 
	and thus $P(\sup_{\tau \in \mc T}\shortnorm{ \hat{\bvec}^{\rm lasso}_{\Mset^c}(\tau)-\bvec^*_{\Mset^c}(\tau)}_{\infty} \leq d_0\lambda) \rightarrow 1$. 
\end{proof}

\begin{proof}[Proof of Lemma \ref{lemma:global33modAssumptions2}]
	We begin by proving \eqref{eqn:thm33as2}.
	It suffices to show that for all $v_i \in [\tau-1,\tau]$,  
	\begin{equation}
		\label{eqn:as2Target}
		\Pr\left\{ \sup_{\tau \in \mc T, j \in \MSset^c} \left| n^{-1}\left[\sum_{i \in \mc I(\hat{\bvec}_0^{o}(\tau))^c}x_{ij} \psi_\tau(y_i-\xvec_i^T\hat{\bvec}^{o}(\tau)) + \sum_{i \in \mc I(\hat{\bvec}_0^{o}(\tau))} x_{ij}v_i(\tau)\right] \right| < d_1\lambda \right\} \rightarrow 1. 
	\end{equation}
	Note that $|\mc I[\hat{\bvec}_0^{o}(\tau)]|=m-r+s$ (to see this, one can express \eqref{eqn:RE.oracle} as a linear program and count the number of non-basic variables). Using assumptions \ref{as:x_max} and \ref{as:lambda} it follows that
	\begin{equation}
		\label{eqn:pfSum}
		\sup_{\tau \in \mc T,j \in \MSset^c} \left|n^{-1} \sum_{i \in \mc I[\hat{\bvec}_0^{o}(\tau)]} x_{ij}\tilde{v}_i(\tau)\right| \leq (m-r+s)n^{-1}\max_{i,j \in \MSset^c} |x_{ij}| = o_p(((m+s)n^{-1} \log (n \vee p))^{1/2}), 
	\end{equation}
	which is $o_p(\lambda)$.
	Therefore
	\begin{equation}
		\label{eqn:as2SimpleSum}
		\sup_{\tau \in \mc T,j \in \MSset^c} \left| n^{-1}\sum_{i \in \mc I(\hat{\bvec}_0^{o}(\tau))^c }x_{ij} \psi_\tau(y_i-\xvec_i^T\hat{\bvec}^{o}(\tau))\right| \leq \sup_{\tau \in \mc T,j \in \MSset^c} \left| n^{-1}\sum_{i =1}^nx_{ij} \psi_\tau(y_i-\xvec_i^T\hat{\bvec}^{o}(\tau))\right| + o_p(\lambda).
	\end{equation}
	
	Define the set $\mathcal{R}_{\MSset}(B) \equiv \left\{\bm \delta \in \reals^{p+1} : \bm \delta_{\MSset^c} = \bm 0, ||\bm \delta_{\MSset}||_2 \leq \sqrt{(m+s)\log^2 n/n}B\right\}$ for some $B>0$. By Theorem \ref{thm:oracle rate}, \eqref{eqn:pfSum}, \eqref{eqn:as2SimpleSum} and assumption \ref{as:lambda}, one can demonstrate that \eqref{eqn:as2Target} holds by proving that for some $B>0$
	\begin{equation}
		\label{eqn:as2Target2}
		\Pr\left( \sup_{\tau \in \mc T, j \in \MSset^c, \bm \delta \in \mathcal{R}_{\MSset}(B) } \left| \sum_{i=1}^nx_{ij} \psi_\tau[y_i-\xvec_i^T(\bvec^*(\tau)+\bm \delta)] \right| < \sqrt{n (s+m) \log(p \vee n)} d_1/2 \right) \rightarrow 1.
	\end{equation}
	Consider the design events $\Omega_1 \defeq \{\max_{j \in \{1,\ldots,p\}} |\hat{\sigma}_j-1| \leq 1/2\}$ and $\Omega_2 \defeq \{ \max_{i, j \in \MSset^c} |x_{ij}| < C n^{1/2}((s+m)^2\log (p \vee n))^{-1/2}) \}$ for some constant $C$, which hold with probability converging to $1$ by assumptions \ref{as:xbounds} and \ref{as:x_max}. 
	Define $\bm e_j$ as the standard basis vector for $j$ and $M_n(\tau,\bm \delta) \defeq n^{-1/2}\sum_{i=1}^n \bm x_i \psi_\tau( y_i - \bm x_i^T (\bvec^*(\tau) + \bm \delta))$. Notice, $n^{1/2} \bm e_j^T M_n(\tau,\bm \delta) = \sum_{i=1}^nx_{ij} \psi_\tau[y_i-\xvec_i^T(\bvec^*(\tau)+\bm \delta)]$ and
	\begin{eqnarray*}
		n^{1/2} \bm e_j^T M_n(\tau,\dvec) &=& n^{1/2} \bm e_j^T \left\{ M_n(\tau,\dvec)-E[M_n(\tau,\bm \delta)] - M_n(\tau, \bm 0) + E[M_n(\tau, \bm 0)] \right\} \\
		&& + n^{1/2} \bm e_j^T \left\{E[M_n(\tau,\bm \delta)] - E[M_n(\tau,\bm 0)] \right\} \\
		&& + n^{1/2} \bm e_j^T M_n(\tau,\bm 0) \\
		&\defeq& \Pi_1(\tau,\bm \delta) + \Pi_2(\tau,\bm \delta) + \Pi_3(\tau). 
	\end{eqnarray*}
	
	Lemma \ref{lemma:Pi1 prob bound} provides that for any constant $C_1 > 1$
	\begin{align*}
		\Pr\bigg(\sup_{\tau \in \mc T, \bm \delta \in \mathcal{R}_{\MSset}(B)} |\Pi_1(\tau,\bm \delta)| > & 4 C_1 \sqrt{n (m+s) \log (p \vee n)} | \Omega_2 \bigg) \\
		& \leq 4 \exp\left\{- \frac{1}{2} C_1 (m+s)^2 \log (p \vee n) + 4 (m+s) \log n \right\}.
	\end{align*}
	Moreover, it is shown in the proof of Theorem 3.3 from \citet{zheng2015globally} that 
	$$
		\sup_{\tau \in \mc T, \bm \delta \in \mathcal{R}_{\MSset}(B)} |\Pi_2(\tau, \bm \delta)| \leq \sqrt{n \log p} \leq \sqrt{ n (m+s) \log(p \vee n) }
	$$	
	and
	\begin{multline*}
	\Pr\left(\sup_{\tau \in \mc T, \bm \delta \in \mathcal{R}_{\MSset}(B)} | \Pi_3(\tau) | > 20\sqrt{n (m+s) \log (p \vee n)} | \Omega_1 \right) \\
	\leq \Pr\left(\sup_{\tau \in \mc T, \bm \delta \in \mathcal{R}_{\MSset}(B)} | \Pi_3(\tau) | > 20\sqrt{n \log p} \right) \leq 12p^{-2} .
	\end{multline*}
	Therefore by applying Bonferroni's inequality, we find that
	\begin{align*}
		\Pr\bigg( & \sup_{\tau \in \mc T, j \in \MSset^c, \bm \delta \in \mathcal{R}_{\MSset}(B) } |n^{1/2} \bm e_j^T M_n(\tau,\dvec)|  > 35\sqrt{n (m+s) \log (p \vee n)} | \Omega_1 \cap \Omega_2  \bigg)\\
		& \leq (p - m - s) \; \frac{12}{p^2} + (p-m-s) \; 4 \exp\left\{-7 (m+s)^2 \log (p \vee n) + 4 (m+s) \log n \right\}\\
		& \leq \frac{12}{p} + 4 \exp\left\{- 7 (m+s)^2 \log (p \vee n) + 4 (m+s) \log n +  \log p \right\}\\
		& \leq \frac{12}{p} + 4\exp\left\{[- 5 (m+s) + 4] (m+s) \log n + [- 2 (m+s)^2  + 1] \log p \right\}\\
		& \leq \frac{12}{p} + 4 \exp\left\{ - (m+s)\log n - \log p \right\} = \frac{4}{n^{m+s}} + \frac{16}{p}.
	\end{align*} 
	Since $\Pr(\Omega_1^c) \to 0$ and $\Pr(\Omega_2^c) \to 0$ by \ref{as:xbounds} and \ref{as:x_max}, this is sufficient to prove \eqref{eqn:as2Target2}.
	
	Next, we prove \eqref{eqn:thm33as3}.
	When \ref{as:iid array} - \ref{as:local alts} hold and $(m + s)\xi_n^2 \log^2 n = o(n)$ the conditions of Theorem \ref{thm:oracle rate} are satisfied, and $\sup_{\tau \in \mc T} \norm{\bvecOR_{0}(\tau) - \bvec^*(\tau) }_2 = O_p\left(\sqrt{\frac{m+s}{n}} \right)$. Combining this with assumption \ref{as:lambda} provides that $\sup_{\tau \in \mc T} \norm{\bvecOR_{0}(\tau) - \bvec^*(\tau) }_2=o(\lambda)$. Then 
	\begin{eqnarray*}
		\inf_{j \in \Sset}\sup_{\tau \in \mc T}|\hat{\beta}_{0,j}^{o}(\tau)| &\geq& \inf_{j \in \Sset}\sup_{\tau \in \mc T} |\beta_j^*(\tau)| - \sup_{\tau \in \mc T }||\hat{\beta}^{o}_{0,\Mset^c}(\tau)-\beta^*_{\Mset^c}(\tau)||_{\infty} \\
		&\geq& \inf_{j \in \Sset}\sup_{\tau \in \mc T} |\beta_j^*(\tau)| - \sup_{\tau \in \mc T }||\hat{\beta}^{o}_{0,\Mset^c}(\tau)-\beta^*_{\Mset^c}(\tau)||_{2} \\
		&\geq& (d+1)\lambda - \lambda \geq d\lambda,
	\end{eqnarray*}
	where the last inequality relies on assumption \ref{as:strength of signal}. 
\end{proof}

\begin{proof}[Proof of Theorem \ref{thm:strong oracle}]
	Lemmas \ref{lemma:global33modAssumptions} and \ref{lemma:global33modAssumptions2} establish that for $e \in \{0,a\}$ the events $\mc E_{e1}$, $\mc E_{e2}$, and $\mc E_{e3}$ from Theorems \ref{thm:global33mod} and \ref{thm:global33modStep2} hold with probability converging to $1$, completing the proof.
\end{proof}

\subsubsection{Supporting lemma}

\begin{lemma}\label{lemma:Pi1 prob bound}
	Under the same conditions as Lemma \ref{lemma:global33modAssumptions2}, we have that for any constant $C_1 > 1$ the event
	$$ \sup_{\tau \in \mc T, \bm \delta \in \mathcal{R}_{\MSset}(B)} |\Pi_1(\tau,\bm \delta)| \leq 4 C_1 \sqrt{n (m+s) \log (p \vee n)} $$
	holds with probability at least
	$1 - 4 \exp\left\{- \frac{1}{2} C_1 (m+s)^2 \log (p \vee n) + 4 (m+s) \log n \right\} - \Pr(\Omega_2^c)$.
\end{lemma}

\begin{proof}[Proof of Lemma \ref{lemma:Pi1 prob bound}]
	This lemma is equivalent to Lemma 7.11 of \citet{zheng2015globally} with $\log p$ replaced by $(s+m)\log (p \vee n)$.
	As such, their proof suffices with minor adaptations which follow from our assumptions that $(m+s)^2\log p = O(n^\alpha)$ for some $\alpha \in (0,1)$ and $\max_{i, j \in \MSset^c} |x_{ij}| = o_p(n^{1/2}((s+m)^2\log (p \vee n))^{-1/2})$ 
\end{proof}

\subsection{Proofs of oracle estimator results}

\subsubsection{Oracle estimator rates and Bahadur representations}
In this section, we will prove Theorems \ref{thm:oracle rate} and \ref{thm:oracle bahadur}.
First, we define a few key quantities.
For random vectors $(\bm V_i)_{i \in [n]}$, we denote the empirical mean by $\empmean( \bm V_i ) = \frac{1}{n} \sum_{i=1}^{n} \bm V_i$ and the empirical process by $\empproc(\bm V_i) = \frac{1}{\sqrt{n}} \sum_{i = 1}^n \{\bm V_i - \E[\bm V_i]\}$.
We use $\lesssim$ and $\ll$ (and their in-probability counterparts) interchangeable with $O(\cdot)$ and $o(\cdot)$, respectively, for ease of presentation.

We further define the pivotal quantity: $U_i \defeq \inf\{ u: y_i < \xvec_i^T \bvec^*(u) \}$ for $i \in [n]$.
The following lemmas establish key properties of $U_i$ which we will rely on in our proofs.
\begin{lemma}\label{lemma:Ui uniform}
	$U_i \sim \text{Uniform}(0,1)$ conditional on $\xvec_i$ for $i \in [n]$. 
\end{lemma}

\begin{proof}[Proof of Lemma \ref{lemma:Ui uniform}]
	First, note $F_{y_i|\xvec_i}$ is continuous because it has a density.
	Since $Q_{y_i}(\tau|\xvec_i) = \xvec_i^T\bvec^*(\tau)$, we see that
	$U_i = \inf\{ u: Y_i < Q_{y_i}(u|\xvec_i) \}$. 
	By Proposition 2.3(4) of \citet{embrechts_note_2013}, we have that $F_{y_i|\xvec_i}(Q_{y_i}(u|\xvec_i)|\xvec_i) = u$.
	Therefore
	$U_i = \inf\{ u:  F_{y_i|\xvec_i}(Y_i|\xvec_i) < u \} = F_{y_i|\xvec_i}(Y_i|\xvec_i)$ where the last equality follows from the intermediate value theorem.
	Lastly because the CDF is continuous, $ F_{y_i|\xvec_i}(Y_i|\xvec_i) \sim \text{Uniform}(0,1)$ by Proposition 3.1(1) of \citet{embrechts_note_2013}.
\end{proof}

\begin{lemma} \label{lemma:Ui equivalence}
	$U_i < \tau$ if and only if $y_i < \xvec_i^T\bvec^*(\tau)$ a.s.
\end{lemma}
\begin{proof}[Proof of Lemma \ref{lemma:Ui equivalence}]
	Define $\mc S_i = \{u: y_i < \xvec_i^T\bvec^*(u) \} $ such that $U_i = \inf \mc S_i$.
	We'll start by proving $U_i < \tau \implies y_i < \xvec_i^T\bvec^*(\tau)$.
	Suppose $U_i < \tau$. Then there exists $u_0 \in \mc S_i$ such that $U_i < u_0 < \tau$ (if such a $u_0$ did not exist, then $\tau$ would be a greater lower bound of $\mc S_i$ than $U_i$, a contradiction). Because $Q_{y_i}(u|\xvec_i) = \xvec_i^T\bvec^*(u)$ is increasing in $u$, we see that $y_i < \xvec_i^T\bvec^*(u_0) \leq \xvec_i^T\bvec(\tau)$.
	
	Next we'll prove $y_i < \xvec_i^T\bvec^*(\tau) \implies U_i < \tau$ a.s.
	Suppose $y_i < \xvec_i^T\bvec^*(\tau)$. Then $\tau \in \mc S_i$ by definition and $U_i \leq \tau$ by the properties of the infimum.
	We know from Lemma \ref{lemma:Ui uniform} that $\Pr(U_i = \tau) = \E[\E[1(U_i = \tau)|\xvec_i] ] = \E[\Pr(U_i = \tau|\xvec_i)] = 0$, since $U_i \sim \text{Uniform}(0,1)$ conditional on $\xvec_i$. Therefore $\Pr( y_i < \xvec_i^T\bvec^*(\tau) \land U_i = \tau ) = 0$ and $y_i < \xvec_i^T\bvec^*(\tau) \implies U_i < \tau$ a.s.
\end{proof}

We will also use the following lemma in our proof of Theorem \ref{thm:oracle rate}.
\begin{lemma} \label{lemma:product bracketing number} 
	Suppose that $\mc F$ and $\mc G$ are two function classes which are uniformly bounded, with $\norm{f}_{\infty} \leq M_F$ and $\norm{g}_{\infty} \leq M_G$ for all $f \in \mc F$ and $g \in \mc F$.
	Then the bracketing number of the product function class
	$$
	\mc F\mc G = \{ fg: f \in \mc F, g \in \mc G  \}
	$$
	satisfies
	$$
	N_{[]}(\varepsilon, \mc {FG}, L_2(P)) =  N_{[]}\left(\frac{\varepsilon}{2 M_G}, \mc {F}, L_2(P)\right)  N_{[]}\left(\frac{\varepsilon}{2 M_F}, \mc {G}, L_2(P)\right)
	$$
	for $0 < \varepsilon < 1$.
\end{lemma}

\begin{proof}[Proof of Lemma \ref{lemma:product bracketing number}]
	Let $f, f' \in \mc F$ and $g, g' \in \mc G$.
	We see that
	\begin{align*}
		|fg - f'g'|
		= | fg - fg' + fg' - f'g' |
		\leq | fg - fg'| + | fg' - f'g'|
		\leq M_F | g - g'| + M_G | f - f'|.
	\end{align*}
	Suppose, then, that $[f_l, f_u]$ is a bracket for $\mc F$ with $\norm{f_u - f_l}_{P,2} \leq \varepsilon/(2 M_G)$ and that $[g_l, g_u]$ is a bracket for $\mc G$ with $\norm{g_u - g_l}_{P,2} \leq \varepsilon/(2 M_F)$. 
	Define $L = \min\{ f_lg_l, f_lg_u, f_ug_l, f_ug_u \}$ and $U = \max\{ f_lg_l, f_lg_u, f_ug_l, f_ug_u \}$ pointwise.
	Clearly for any $f\in[f_l, f_u], g\in[g_l, g_u]$, $L < fg < U$.
	The inequality above implies that 
	\begin{align*}
		U - L 
		& \leq \sup_{f,f'\in [f_l, f_u], g,g'\in [g_l, g_u]}|fg - f'g'| \\
		& \leq \sup_{g,g'\in [g_l, g_u]} M_F | g - g'| + \sup_{f,f'\in [f_l, f_u]} M_G | f - f'| \\
		& = M_F | g_u - g_l| + M_G | f_u - f_l|. 
	\end{align*}		
	Hence $\norm{U - L}_{P,2} \leq \varepsilon $.
	Therefore we can construct a $\varepsilon$-bracketing cover for $\mc {FG}$ from the bracketing covers $\{[f_{l,i}, f_{u,i}]: 1 \leq i \leq N_F\}$ and $\{[g_{l,j}, g_{u,j}]: 1 \leq i \leq N_G\}$ for $\mc F$ and $\mc G$, respectively, with $N_F N_G$ brackets $[L_{i,j}, U_{i,j}]$ defined as above.
\end{proof}


\begin{proof}[Proof of Theorem \ref{thm:oracle rate}]
	We only present the proof of \eqref{eqn:RE.oracle rate} as the proof of \eqref{eqn:FE.oracle rate} is a simpler version that does not have to deal with the potential bias introduced by the restriction that $\Cmat \bvec_{\MSset}(\tau)=\bm \Psi (\tau)$.
	
	Define $\tilde\bvec^*(\tau)$ by
	\begin{equation*}
		\begin{cases}
			\tilde \bvec^*_{\Mset}(\tau) & =  \bvec^*_{\Mset}(\tau) - \Cmat^T ( \Cmat \Cmat^T )^{-1} \bm \eta(\tau)/\sqrt{n} \\
			\tilde \bvec^*_{\Mset^c}(\tau) & = \bvec^*_{\Mset^c}(\tau).
		\end{cases}
	\end{equation*}
	Under \ref{as:local alts}, we have
	\begin{align*}
		\Cmat \tilde \bvec^*_{\Mset}(\tau) - \bm \Psi(\tau) 
		& = \Cmat\left(\bvec^*_{\Mset}(\tau) - \Cmat^T  (\Cmat \Cmat^T )^{-1} \bm \eta(\tau)/\sqrt{n}  \right) - \bm \Psi(\tau) \\
		& = \Cmat \bvec^*_{\Mset}(\tau)  - \Cmat\Cmat^T ( \Cmat \Cmat^T )^{-1} \bm \eta(\tau)/\sqrt{n} - \bm \Psi(\tau)\\  
		& = \Cmat \bvec^*_{\Mset}(\tau)  - \bm \eta(\tau)/\sqrt{n} - \bm \Psi(\tau)
		= \bm 0.
	\end{align*}
	In addition, \ref{as:local alts} implies
	\begin{align*}
		\norm{ \tilde \bvec^*_{\Mset}(\tau) - \bvec^*_{\Mset}(\tau) }_2^2 
		& = \norm{  \Cmat^T  (\Cmat \Cmat^T )^{-1} \bm \eta(\tau)/\sqrt{n} }_2^2 \\
		& = \frac{1}{n} \bm \eta(\tau)^T  (\Cmat \Cmat^T )^{-1}\Cmat \Cmat^T  (\Cmat \Cmat^T )^{-1} \bm \eta(\tau)^T\\
		& \leq \frac{1}{n} \norm{\bm \eta(\tau)}_2^2 \lambda_{\max}\left\{ (\Cmat \Cmat^T )^{-1} \right\}\\
		& \lesssim \frac{m - r+s}{n}
	\end{align*}
	uniformly for $\tau \in \mc T$.
	Since $\tilde \bvec^*_{\Mset^c}(\tau) = \bvec^*_{\Mset^c}(\tau)$, this further implies
	\begin{equation}	
		\sup_{\tau \in \mc T} \norm{ \tilde \bvec^*(\tau) - \bvec^*(\tau) }_2 \lesssim \sqrt{\frac{m - r+s}{n}}. \label{eqn:beta tilde rate}
	\end{equation}
	
	For any $\bvec(\tau)$ satisfying $\Cmat \bvec_{\Mset}(\tau) = \bm \Psi(\tau)$, clearly
	$
	\Cmat(\bvec_{\Mset}(\tau) - \tilde \bvec^*_{\Mset}(\tau) ) = \bm \Psi(\tau) - \bm \Psi(\tau) = \bm 0.
	$ 
	Therefore $\bvec_{\Mset}(\tau) - \tilde \bvec^*_{\Mset}(\tau)$ belongs to the null space of $\Cmat$.
	Let $\bm Z \in \reals^{m \times (m - r)}$ be an orthonormal basis for the null space of $\Cmat$, such that $\bm Z^T\bm Z = \bm I_{m-r}$. 
	Then for any $\bvec(\tau)$ satisfying $\Cmat \bvec_{\Mset}(\tau) = \bm \Psi(\tau)$, there exists a unique $\bm v(\tau)$ such that $\bvec_{\Mset}(\tau) - \tilde \bvec^*_{\Mset}(\tau) = \bm Z \bm v(\tau)$.
	Define 
	$\widetilde{ \bm Z }
	=\begin{bsmallmatrix}
		\bm Z & \bm 0\\
		\bm 0 & \bm I_{s}
	\end{bsmallmatrix}
	\in \reals^{m+s \times (m-r+s)}
	$
	and consider the affine transformation 
	$$\dvec 
	\mapsto \tilde \bvec_{\MSset}^*(\tau) + \widetilde{\bm Z} \dvec
	=  
	\begin{bmatrix}
		\tilde \bvec_{\Mset}^*(\tau) + \bm Z \dvec_1\\
		\tilde \bvec_{\Sset}^*(\tau) + \dvec_2
	\end{bmatrix},
	$$
	where $\dvec = (\dvec_1^T, \dvec_2^T)^T \in \reals^{m-r+s}$.
	We see that for any $\dvec_1 \in \reals^{m-r}$,
	$ \norm{\bm Z \dvec_1}_2^2 = \dvec_1^T \bm Z^T \bm Z \dvec_1 =  \dvec_1^T \dvec_1 = \norm{\dvec_1}_2^2 $.
	Therefore
	$
	\shortnorm{ \widetilde{\bm Z} \dvec }_2^2 = \norm{\bm Z \dvec_1}_2^2 + \norm{\dvec_2}_2^2 = \norm{\dvec_1}_2^2 + \norm{\dvec_2}_2^2 = \norm{\dvec}_2^2
	$ for any $\dvec \in \reals^{m-r+s}$.
	
	We define
	$$
	\bar \ell(\dvec) 
	\defeq \mathbb{E}_n[ \rho_{\tau}( y_i - \tilde{\xvec}_{i}^T(\tilde \bvec_{\MSset}^*(\tau) + \widetilde{ \bm Z }\dvec) )  ) ].
	$$ 
	and let $\hat \dvec(\tau) = \argmin_{\dvec} \bar \ell(\dvec) $ for $\tau \in \mc T$. 
	By definition, we have that
	$\bvecOR_{0,\MSset}(\tau) = \tilde \bvec_{\MSset}^*(\tau) + \widetilde{ \bm Z } \hat \dvec (\tau)$.
	We further define
	$$
	\bar \psi_i(\dvec, \tau) = \widetilde{\bm Z}^T\tilde \xvec_{i}[1(y_i < \tilde\xvec_{i}^T( \tilde \bvec_{\MSset}^*(\tau) + \widetilde{\bm Z} \dvec ) ) - \tau].
	$$
	Let $\mc B_q(d) = \{\bm v \in \reals^q: \norm{\bm v}_2 = d  \}$ denote the ball of radius $d$ in $\reals^q$.
	Define $d_n = B \sqrt{\frac{m-r+s}{n}}$ for some $B > 0$ and consider the event
	\begin{equation}
		\mc E_1 \defeq \left\{  \inf_{\tau \in \mc T} \inf_{\bm \alpha \in \mc B_{m-r+s}(1)} \bm \alpha^T \mathbb{E} [\bar \psi_i(\bm \alpha d_n, \tau) ] > 0 \right \}. \label{eqn:bar l subgrad event}
	\end{equation}	
	Because $\bar \ell(\dvec)$ is strictly convex and $\mathbb{E}_n[\bar \psi_i(\dvec, \tau)]$ is a subgradient of $\bar \ell (\dvec)$, if $\mc E_1$ holds, then the global minimizer of $\bar \ell(\delta)$ is in $\mc B_{m-r+s}(d_n)$ for all $\tau \in \mc T$---that is, $\sup_{\tau \in \mc T} \norm{ \bvecOR_{0,\MSset}(\tau) - \tilde \bvec^*_{\MSset}(\tau) }_2 \lesssim \sqrt{\frac{m-r+s}{n}}.$
	Combining this with \eqref{eqn:beta tilde rate} yields
	$$
	\sup_{\tau \in \mc T} \norm{ \bvecOR_{0}(\tau) - \bvec^*(\tau) }_2
	\leq \sup_{\tau \in \mc T}\norm{ \bvecOR_{0}(\tau) - \tilde \bvec^*(\tau) }_2 + \sup_{\tau \in \mc T}\norm{ \tilde \bvec^*(\tau) - \bvec^*(\tau) }_2 \lesssim \sqrt{\frac{m-r+s}{n}}.
	$$		
	As such, showing that $\mc E_1$ holds with probability converging to $1$ as $n \to \infty$ for a sufficiently large $B > 0$ is sufficient to prove \eqref{eqn:RE.oracle rate}.
	
	Define the sets
	\begin{equation}
		\bar{\mc R}(d_n) \defeq \{(\tau,\dvec) \in \mc T \times \reals^{m-r+s} : \norm{\dvec }_2 \leq d_n \}, \label{eqn:bar Rn set}
	\end{equation}
	which are indexed by $d_n$, and define
	\begin{eqnarray*}
		\tilde \epsilon_1 & \defeq & \sup_{(\tau, \dvec) \in \bar {\mc R}(d_n) }\norm{ \empproc (\bar \psi_i(\dvec, \tau)) -  \empproc (\bar \psi_i(\bm 0, \tau))}_2\\
		\tilde \epsilon_2 & \defeq & \sup_{(\tau, \dvec) \in \bar{ \mc R}(d_n) } n^{1/2} \norm{\E[\bar \psi_i(\dvec, \tau)] - \E[\bar \psi_i(\bm 0, \tau)] - \widetilde{\bm Z}^T \bm J(\tau)\widetilde{\bm Z}\dvec }_2.
	\end{eqnarray*}
	We see that for all $\tau \in \mc T$ and $\bm \alpha \in \mc B_{m-r+s}(1)$
	\begin{align}\label{eqn:bar psi lower bound}
		\sqrt n \bm \alpha^T \empmean [\bar \psi_i(\bm \alpha d_n, \tau) ]
		& = \bm \alpha^T \empproc( \bar \psi_i(\bm \alpha d_n, \tau) ) - \bm \alpha^T \empproc( \bar \psi_i(\bm 0, \tau) ) + \bm \alpha^T \empproc( \bar \psi_i(\bm 0, \tau) ) \notag \\
		& \text{\; \;} + \sqrt n \bm \alpha^T \E [\bar \psi_i(\bm \alpha d_n, \tau) ] - \sqrt n \bm \alpha^T \E [\bar \psi_i(\bm 0, \tau) ] -  \sqrt n \bm \alpha^T \widetilde{\bm Z}^T \bm J(\tau)\widetilde{\bm Z}\bm \alpha d_n \notag \\
		& \text{\; \;} +  \sqrt n \bm \alpha^T \E [\bar \psi_i(\bm 0, \tau) ] +  \sqrt n \bm \alpha^T \widetilde{\bm Z}^T\bm J(\tau)\widetilde{\bm Z} \bm \alpha d_n \notag \\
		& \geq \bm \alpha^T \empproc( \bar \psi_i(\bm 0, \tau) )  + \sqrt n \bm \alpha^T \E [\bar \psi_i(\bm 0, \tau) ] \notag\\
		& \text{\; \;}  + \sqrt n \bm \alpha^T\widetilde{\bm Z}^T\bm J(\tau)\widetilde{\bm Z} \bm \alpha d_n - \tilde \epsilon_1 - \tilde \epsilon_2.
	\end{align}
	Conditions \ref{as:density conditions} and \ref{as:Gram lower bound} provide that $\lambda_{\min}\{ \bm J(\tau) \}$ is bounded away from $0$ uniformly over $n$ and $\tau \in \overline{ \mc T}$.
	Moreover, we have that $\lVert \widetilde{\bm Z} \bm \alpha \rVert_2 = \norm{\bm \alpha}_2 = 1$.
	As such, there exists $c > 0$ such that 
	$\bm \alpha^T\widetilde{\bm Z}^T\bm J(\tau)\widetilde{\bm Z} \bm \alpha \geq \min_{\bm v : \norm{\bm v}_2 = 1} \bm v^T \bm J(\tau) \bm v = \lambda_{\min}\{ \bm J(\tau) \}  > c$ for all $n, \tau \in \overline{ \mc T}$.
	Therefore if
	$$
	\mc E_2 \defeq \left\{ c B \sqrt{m-r+s}  - |\bm \alpha^T \empproc( \bar \psi_i(\bm 0, \tau) ) | - \sqrt n |\bm \alpha^T \E [\bar \psi_i(\bm 0, \tau) ]|   - \tilde \epsilon_1 - \tilde \epsilon_2 > 0\right \}
	$$	
	holds, then so does $\mc E_1$.
	Suppose that
	\begin{eqnarray}
		\sqrt n | \bm \alpha^T \E [\bar \psi_i(\bm 0, \tau) ]| & \lesssim & \sqrt{m-r+s} \label{eqn:E2a}  \\
		|\bm \alpha^T \empproc( \bar \psi_i(\bm 0, \tau) )| & \lesssim_p & \sqrt{m-r+s} \label{eqn:E2b} \\
		\tilde \epsilon_1 & = & o_p(\sqrt{m-r+s}) \label{eqn:E2c}\\
		\tilde \epsilon_2 & = & o(\sqrt{m-r+s}) \label{eqn:E2d}
	\end{eqnarray}
	where the bounds in \eqref{eqn:E2a} and \eqref{eqn:E2b} do not depend on $B$.
	Then by selecting a sufficiently large $B$ we can make the probability of $\mc E_2$ arbitrarily close to $1$, completing the proof.
	
	All that remains is to prove \eqref{eqn:E2a} - \eqref{eqn:E2d}. 
	We begin with \eqref{eqn:E2a}. By the Tower law, we have
	\begin{align*}
		\bm \alpha^T \E [\bar \psi_i(\bm 0, \tau)] 
		& = \bm \alpha^T \widetilde{\bm Z}^T \E\left[ \E \left[ \tilde \xvec \left( 1(y < \xvec^T \tilde \bvec^*(\tau)) - \tau\right) \middle| \xvec \right] \right]\\
		& = \bm \alpha^T \widetilde{\bm Z}^T \E\left[\tilde \xvec \left( F_{y|\xvec}(\xvec^T \tilde \bvec^*(\tau)) - \tau\right) \right]\\
		& = \bm \alpha^T \widetilde{\bm Z}^T \E\left[ \tilde \xvec \left( F_{y|\xvec}(\xvec^T \tilde \bvec^*(\tau)) - F_{y|\xvec}(\xvec^T \bvec^*(\tau))\right) \right]
	\end{align*}
	where the final equality uses the fact that $F_{y|\xvec}(\xvec^T \bvec^*(\tau)) = \tau$ as shown in the proof of Lemma \ref{lemma:Ui uniform}.
	By the mean value theorem, for each $\tau$ there exists $\bar \bvec (\tau)$ on the line segment between $\tilde \bvec^*(\tau)$ and $\bvec^*(\tau)$ such that
	\begin{align*}
		F_{y|\xvec}(\xvec^T \tilde \bvec^*(\tau)) - F_{y|\xvec}(\xvec^T \bvec^*(\tau))  
		& =  f_{y|\xvec}(\xvec^T \bar \bvec(\tau)) \xvec^T( \tilde \bvec^*(\tau) - \bvec^*(\tau) )  \\
		& =  f_{y|\xvec}(\xvec^T \bar \bvec(\tau)) \tilde \xvec^T( \tilde \bvec^*_{\MSset}(\tau) - \bvec^*_{\MSset}(\tau) ).
	\end{align*}
	Plugging this into the previous expression yields
	\begin{equation*}
		\sqrt{n}|\bm \alpha^T \E [\bar \psi_i(\bm 0, \tau)]| 
		\leq \sqrt{n} \shortnorm{\widetilde{\bm Z} \bm \alpha}_2 |f_{y|\xvec}(\xvec^T \bar \bvec(\tau))| \norm{\E [\tilde \xvec \tilde \xvec^T] }_2 \shortnorm{\tilde \bvec^*_{\MSset}(\tau) - \bvec^*_{\MSset}(\tau)}_2
		\lesssim \sqrt{m-r+s}
	\end{equation*}
	uniformly for $\tau \in \mc T$ by \ref{as:density conditions}(i) and \ref{as:Gram lower bound}.
	
	Moving on to \eqref{eqn:E2b}, we see that 
	\begin{align*}
		\bm \alpha^T \empproc( \bar \psi_i(\bm 0, \tau) )
		& = \bm \alpha^T  \bigg[ \empproc\left(   \widetilde{\bm Z}^T \tilde \xvec ( 1(y < \xvec^T \bvec^*(\tau)) - \tau ) \right)\\ 
		& \text{ \; }+ \empproc\left(  \widetilde{\bm Z}^T \tilde \xvec ( 1(y < \xvec^T \tilde\bvec^*(\tau)) - \tau ) \right) - \empproc\left(  \widetilde{\bm Z}^T \tilde \xvec ( 1(y < \xvec^T \bvec^*(\tau)) - \tau ) \right) \bigg]\\
		& \leq \sup_{\tau \in \mc T}  \norm{\empproc\left(  \widetilde{\bm Z}^T  \tilde \xvec ( 1(y < \xvec^T \bvec^*(\tau)) - \tau ) \right)}_2\\
		& \text{ \; }+ \sup_{\tau \in \mc T}  \norm{\empproc\left(  \widetilde{\bm Z}^T \tilde \xvec ( 1(y < \xvec^T \tilde\bvec^*(\tau)) - \tau ) \right) - \empproc\left( \widetilde{\bm Z}^T \tilde \xvec ( 1(y < \xvec^T \bvec^*(\tau)) - \tau ) \right) }_2\\
		& \defeq \tilde \epsilon_{01} + \tilde \epsilon_{02},
	\end{align*}
	uniformly for all $\tau \in \mc T$,	where the inequality follows from the Cauchy Schwarz inequality.
	Lemmas \ref{lemma:Ui uniform} and \ref{lemma:Ui equivalence} provide that 
	\begin{align*}
		\tilde \epsilon_{01} 
		& = \sup_{\tau \in \mc T}  \norm{\empproc\left(  \widetilde{\bm Z}^T  \tilde \xvec ( 1(y < \xvec^T \bvec^*(\tau)) - \tau ) \right)}_2\\
		& = \sup_{\tau \in \mc T}  \norm{\empproc\left(  \widetilde{\bm Z}^T  \tilde \xvec ( 1(U_i \leq \tau) - \tau ) \right)}_2\\
		& = \sup_{\tau \in \mc T} \norm {\frac{1}{\sqrt n} \sum_{i = 1}^{n} \widetilde{\bm Z}^T  \tilde \xvec  (1(U_i \leq \tau) - \tau) }_2
	\end{align*}
	where $U_1, \ldots, U_n$ are i.i.d. Uniform$(0,1)$ conditional on $(\xvec_i)_{i \in [n]}$.
	From here, Lemma 31 of \citet{belloni_conditional_2019} provides that
	$\tilde \epsilon_{01} = \sup_{\tau \in \mc T} \norm {\frac{1}{\sqrt n} \sum_{i = 1}^{n} \widetilde{\bm Z}^T  \tilde \xvec ( 1(U_i \leq \tau) - \tau) }_2 \lesssim_p \sqrt{m-r+s}$ since $\widetilde{\bm Z}^T  \tilde \xvec$ has dimension ${m-r+s}$.
	
	To bound $\tilde \epsilon_{02}$, we define the following class of functions
	\begin{equation*}
		\mc F \defeq \bigg\{ (\tilde \xvec, y) \mapsto (\bm \alpha^T \widetilde{\bm Z}^T \tilde \xvec)\left[ 1(y < \tilde \xvec^T\tilde \bvec^*_{\MSset}(\tau) ) -  1(y < \tilde \xvec^T\bvec_{\MSset}^*(\tau) ) \right]: 
		\tau \in \mc T, \bm \alpha \in \mc B_{m-r+s}(1) \bigg\}.
	\end{equation*}
	Clearly $\tilde \epsilon_{02} \leq \sup_{f \in \mc F}| \empproc (f)| $. Our strategy will be to bound the latter using the bracketing entropy of $\mc F$.
	To this end, we further define 
	\begin{eqnarray*}
		\mc F_1 & \defeq &  \bigg\{ \tilde \xvec \mapsto (\bm \alpha^T \widetilde{\bm Z}^T \tilde \xvec): \bm \alpha \in \mc B_{m-r+s}(1) \bigg\} \\
		\mc F_2 & \defeq & \bigg\{ (\tilde \xvec, y) \mapsto  1(y < \tilde \xvec^T\tilde \bvec^*_{\MSset}(\tau) ) -  1(y < \tilde \xvec^T\bvec_{\MSset}^*(\tau) ): 
		\tau \in \mc T \bigg\},
	\end{eqnarray*}
	noting that for every $f \in \mc F$ there exist $f_1 \in \mc F_1$ and $f_2 \in \mc F_2$ such that $f = f_1 f_2$.
	
	To bound the bracketing entropy of $\mc F_1$, we note that for any $\bm \alpha_1,\bm \alpha_2 \in \mc B_{m-r+s}(1)$, we have 
	$ |\bm \alpha_1^T \widetilde{\bm Z}^T \tilde \xvec - \bm \alpha_2^T \widetilde{\bm Z}^T \tilde \xvec |
	\leq \shortnorm{ (\bm \alpha_1 - \bm \alpha_2)^T\widetilde{\bm Z}^T  }_2 \norm{ \tilde \xvec  }_2 \leq \shortnorm{ \bm \alpha_1 - \bm \alpha_2  }_2 \norm{ \tilde \xvec  }_2 \leq \shortnorm{ \bm \alpha_1 - \bm \alpha_2  }_2 F_1(\xvec)$, where $F_1(\tilde \xvec) \defeq \norm{ \tilde \xvec }_2 \leq \xi_n$ is the envelope function for $\mc F_1$.
	By Theorem 2.7.11 of \citet{van_der_vaart_weak_1996} we have that for some universal constant $C_1$
	$$
	N_{[]}(\varepsilon\norm{ F_1}_{P,2}, \mc F_1, L_2(P) ) \leq \left(\frac{C_1}{\varepsilon}\right)^{m-r+s} \text{, \; } 
	$$
	and therefore
	$$
	\log N_{[]}(\varepsilon\norm{ F_1}_{P,2}, \mc F_1, L_2(P) ) \leq (m-r+s) \log \left(\frac{C_1}{\varepsilon}\right) \text{, \; } 
	$$ 
	for	$0 < \varepsilon < 1$, where $P$ denotes the true probability measure.
	
	Next, we will bound the bracketing entropy of $\mc F_2$. We define $\tilde h(\tau) = \tilde \xvec^T\tilde \bvec^*_{\MSset}(\tau)$, $h(\tau) = \tilde \xvec^T\bvec^*_{\MSset}(\tau)$, and $f_{\tau} =  1\{y <\tilde h(\tau) \} -   1\{y < h(\tau) \}$ for $\tau \in \mc T$ for ease of presentation.
	Let $0 < \varepsilon < 1$ and let $\mc T_{\varepsilon} \subset \mc T$ be a finite net such that $\max_{\tau \in \mc T} \min_{\tau_0 \in \mc T_{\varepsilon}} |\tau-\tau_0| \leq  \frac{C_2\varepsilon^2}{\xi_n \sqrt{m+s}}$ for some constant $C_2$. 
	Since $\mc T$ is contained in a fixed compact interval, the next may be chosen so that extension $$|\mc T_{\varepsilon}| \leq 1 + \frac{\xi_n \sqrt{m+s}}{c_2 \varepsilon^2}$$ for some constant $c_2$. 
	Let $\tau \in \mc T$ and choose the closest $\tau_0 \in \mc T_{\varepsilon}$.
	We see that
	\begin{align*}
		\big| f_{\tau} - f_{\tau_0} \big|
		& = 
		\big|  1\{y <\tilde h(\tau) \} -   1\{y < h(\tau) \} - [1\{ y < \tilde h(\tau_0) \} -  1\{y < h(\tau_0) \}] \big|\\
		& \leq	\big|  1\{y <\tilde h(\tau)\}  - 1\{y < \tilde h(\tau_0)\} \big| + \big| 1\{y < h(\tau)\}  -  1\{y < h(\tau_0) \} \big|\\
		& \defeq I_1 + I_2
	\end{align*}
	If we can show that $I_1 + I_2 \leq H(\tilde \xvec, y, \tau_0)$ for some $H(\tilde \xvec, y, \tau_0)$ that does not depend on $\tau$, then the brackets $[ f_{\tau_0} - H(\tilde \xvec, y, \tau_0), f_{\tau_0} + H(\tilde \xvec, y, \tau_0) ]$, $\tau \in \mc T_{\varepsilon}$ cover $\mc F_2$.
	Starting with $I_1$, we find
	\begin{align*}
		I_1
		& \leq 1 \{ |y - \tilde h (\tau_0)| < |\tilde h(\tau) - \tilde h(\tau_0)| \}\\
		& = 1 \{ |y - \tilde h (\tau_0)| < |\tilde \xvec^T\tilde \bvec^*_{\MSset}(\tau) - \tilde \xvec^T\tilde \bvec^*_{\MSset}(\tau_0)  | \}\\
		& \leq 1 \{ |y - \tilde h (\tau_0)| < \shortnorm{\tilde \xvec}_2 \shortnorm {\tilde \bvec^*_{\MSset}(\tau) - \tilde \bvec^*_{\MSset}(\tau_0)}_2  \}\\
		& \leq 1 \{ |y - \tilde h (\tau_0)| < \xi_n \shortnorm {\tilde \bvec^*_{\MSset}(\tau) - \tilde \bvec^*_{\MSset}(\tau_0)}_2  \}.
	\end{align*}
	By the definition of $\tilde \bvec^*_{\MSset}(\tau)$, we have  
	\begin{align*}
		&\norm {\tilde \bvec^*_{\MSset}(\tau) - \tilde \bvec^*_{\MSset}(\tau_0)}_2\\
		& \hspace*{0.6cm} = \norm {\left(\bvec^*_{\MSset}(\tau) - \begin{bmatrix}
				\Cmat^T(\Cmat \Cmat^T)^{-1}\\
				\bm 0
			\end{bmatrix}\bm \eta(\tau)/\sqrt{n} \right)  - \left( \bvec^*_{\MSset}(\tau_0) - \begin{bmatrix}
				\Cmat^T(\Cmat \Cmat^T)^{-1}\\
				\bm 0
			\end{bmatrix}\bm \eta(\tau_0)/\sqrt{n} \right)}_2\\
		& \hspace*{0.6cm} \leq \norm { \bvec^*_{\MSset}(\tau) - \bvec^*_{\MSset}(\tau_0)}_2
		+ \norm{ \begin{bmatrix}
				\Cmat^T(\Cmat \Cmat^T)^{-1}\\
				\bm 0
			\end{bmatrix} (\bm \eta(\tau) - \bm \eta(\tau_0))/\sqrt{n} }_2\\
		& \hspace*{0.6cm} \leq \norm { \bvec^*_{\MSset}(\tau) - \bvec^*_{\MSset}(\tau_0)}_2
		+ \frac{1}{\sqrt n} \norm{\Cmat^T(\Cmat \Cmat^T)^{-1}}_2 \norm{\bm \eta(\tau) - \bm \eta(\tau_0)}_2\\
		& \hspace*{0.6cm} \leq A_1\sqrt{m+s}\,|\tau-\tau_0| + A_2A_3\sqrt{m-r+s}\,|\tau-\tau_0|\\
		& \hspace*{0.6cm} \leq (A_1 + A_2 A_3)\sqrt{m+s} \,|\tau - \tau_0|\\
		& \hspace*{0.6cm} \leq (A_1 + A_2A_3)C_2 \varepsilon^2/\xi_n
	\end{align*}
	where the penultimate inequality follows from \ref{as:local alts} and \ref{as:bvec} and $A_1$, $A_2$, and $A_3$ are universal constants for all $n$ and $\tau$.
	Therefore
	$$
	I_1  \leq 1 \{ |y - \tilde h (\tau_0)| < C_3 \varepsilon^2  \}
	$$
	where $C_3 = (A_1 + A_2 A_3)C_2$.
	Likewise, we find that
	\begin{align*}
		I_2 & \leq 1 \{ |y -  h (\tau_0)| < | h(\tau) -  h(\tau_0)| \}\\
		& = 1 \{ |y -  h (\tau_0)| < |\tilde \xvec^T \bvec^*_{\MSset}(\tau) -  \xvec^T \bvec^*_{\MSset}(\tau_0)  | \}\\
		& \leq 1 \{ |y -  h (\tau_0)| < \shortnorm{\tilde \xvec}_2 \shortnorm { \bvec^*_{\MSset}(\tau) -  \bvec^*_{\MSset}(\tau_0)}_2  \}\\
		& \leq 1 \{ |y -  h (\tau_0)| < \xi_n \shortnorm { \bvec^*_{\MSset}(\tau) -  \bvec^*_{\MSset}(\tau_0)}_2  \}\\
		& \leq 1 \{ |y -  h (\tau_0)| < \xi_n A_1\sqrt{m+s}\,|\tau - \tau_0| \}\\
		& \leq 1 \{ |y -  h (\tau_0)| <  C_3\varepsilon^2 \}.
	\end{align*}
		
	We have shown that $H(\tilde \xvec, y, \tau_0) =  1 \{ |y - \tilde h (\tau_0)| < C_3 \varepsilon^2  \} +  1 \{ |y -  h (\tau_0)| < C_3 \varepsilon^2  \}$ satisfies $|f_{\tau} - f_{\tau_0}|\leq H(\tilde \xvec, y, \tau_0) $, so the brackets $[ f_{\tau_0} - H(\tilde \xvec, y, \tau_0), f_{\tau_0} + H(\tilde \xvec, y, \tau_0) ]$, $\tau \in \mc T_{\varepsilon}$ cover $\mc F_2$.
	Our final step is to determine the bracket width. We find
	\begin{align*}
		\norm{ H(\tilde \xvec, y, \tau_0) }_{P,2}
		& \leq (\E_P[ ( 1 \{ |y - \tilde h (\tau_0)| < C_3 \varepsilon^2  \})^2])^{1/2} + (\E_P \left[(1 \{ |y -  h (\tau_0)| < C_3 \varepsilon^2  \})^2 \right])^{1/2} \\
		& = (\E_P[  1 \{ |y - \tilde h (\tau_0)| < C_3 \varepsilon^2  \}])^{1/2} + (\E_P \left[1 \{ |y -  h (\tau_0)| < C_3 \varepsilon^2  \}\right])^{1/2} .
	\end{align*}
	The Tower law provides
	$
	\E_P\left[  1 \{ |y - \tilde h (\tau_0)| < C_3 \varepsilon^2  \}\right]
	= \E_P\left[  \Pr \left(|y - \tilde h (\tau_0)| < C_3 \varepsilon^2 \middle |\tilde \xvec \right)\right].
	$
	Applying the mean value theorem and \ref{as:density conditions}(i), we see that 
	\begin{align*}
		\Pr \left(|y - \tilde h (\tau_0)| < C_3 \varepsilon^2 \middle|\tilde \xvec \right)
		& = \Pr \left(\tilde h (\tau_0) - C_3 \varepsilon^2 < y <  \tilde h (\tau_0) + C_3 \varepsilon^2 \middle |  \xvec \right)\\
		& = F_{y|\xvec}(  \tilde h (\tau_0) + C_3 \varepsilon^2) -  F_{y|\xvec}( \tilde h (\tau_0) - C_3 \varepsilon^2)\\
		& = f_{y|\xvec}( \bar h) 2 C_3 \varepsilon^2
		\lesssim C_3 \varepsilon^2 
	\end{align*}
	where $\bar h \in (\tilde h (\tau_0) - C_3 \varepsilon^2,\tilde h (\tau_0) + C_3 \varepsilon^2)$.
	Thus $\E_P\left[  1 \{ |y - \tilde h (\tau_0)| < C_3 \varepsilon^2  \}\right] \lesssim C_3 \varepsilon^2$.
	One can similarly use the Tower law, the mean value theorem, and \ref{as:density conditions}(i) to show $\E_P\left[  1 \{ |y - h (\tau_0)| < C_3 \varepsilon^2  \}\right] \lesssim C_3 \varepsilon^2$.
	Thus the bracket width is 
	$$\norm{ f_{\tau_0} + H(\tilde \xvec, y, \tau_0) - (f_{\tau_0} - H(\tilde \xvec, y, \tau_0)) }_{P,2} = 2 \norm{ H(\tilde \xvec, y, \tau_0) }_{P,2} \leq C_4 \varepsilon$$ 
	for some constant $C_4$.
	Thus we have
	$$
	N_{[]}(\varepsilon, \mc F_2, L_2(P) ) \leq 1 + \frac{\xi_n \sqrt{m+s}}{c_3 \varepsilon^2} ,
	$$
	where $c_3 = c_2/C_4^2$, and therefore
	$$
	\log N_{[]}(\varepsilon, \mc F_2, L_2(P) ) \lesssim \log \left(1 + \frac{\xi_n \sqrt{m+s}}{c_3 \varepsilon^2} \right) 
	$$
	for $0 < \varepsilon < 1$.

	By Lemma \ref{lemma:product bracketing number}, since $ |\bm \alpha^T \widetilde{ \bm Z }^T\tilde \xvec| \leq \shortnorm{\bm \alpha^T \widetilde{ \bm Z }^T}_2 \shortnorm{\tilde \xvec}_2  \leq \xi_n$ for all $\alpha \in \mc B_{m-r+s}(1), \xvec \in \mc X$ and $ |1\{y <\tilde h(\tau) \} -   1\{y < h(\tau) \}| \leq 1$ for all $\tau \in \mc T$, we have
	\begin{align*}
		\log N_{[]}(\varepsilon \xi_n, \mc F, L_2(P)) 
		& \leq \log N_{[]}\left(\frac{\varepsilon \xi_n}{2}, \mc F_1, L_2(P) \right) + \log N_{[]}\left(\frac{\varepsilon}{2}, \mc F_2, L_2(P) \right)\\
		& \lesssim  (m-r+s) \log \left(\frac{C_1}{\varepsilon}\right) + \log \left( 1 + \frac{\xi_n \sqrt{m+s}}{c_3 \varepsilon^2} \right) \\
		& \lesssim (m-r+s)   \log \left( \frac{D(1 + \xi_n \sqrt{m+s})}{ \varepsilon} \right)
	\end{align*}
	for $0 < \eps < 1$ and $m-r+s\geq1$, where $D$ is a fixed constant for all $n$.
	
	Let $\delta = \xi_n\sqrt{\frac{m-r+s}{n}}$.
	By definition, for any $f \in \mc F$ there exist $\bm \alpha \in \mc B_{m-r+s}(1)$ and $\tau \in \mc T$ such that
	\begin{align*}
		|f(\tilde \xvec, y)|
		& = |(\bm \alpha^T \widetilde{\bm Z}^T \tilde \xvec)| \left|\left[ 1(y < \tilde \xvec^T\tilde\bvec_{\MSset}^*(\tau) ) -  1(y < \xvec^T\bvec_{\MSset}^*(\tau) ) \right] \right|\\
		& \leq |(\bm \alpha^T \widetilde{\bm Z}^T \tilde \xvec)| \left| 1\{|y - \tilde \xvec^T\bvec^*_{\MSset}(\tau)| < |\tilde \xvec^T (\tilde \bvec^*_{\MSset}(\tau)  -  \bvec^*_{\MSset}(\tau) )|\} \right |\\
		& \leq |(\bm \alpha^T \widetilde{\bm Z}^T \tilde \xvec)| \left| 1\{|y - \xvec^T\bvec^*(\tau)| < C \delta\} \right |
	\end{align*}
	for some constant $C$, where the last inequality follows since
	$
	|\tilde \xvec^T (\tilde \bvec^*_{\MSset}(\tau) -  \bvec^*_{\MSset}(\tau) )| \leq \shortnorm{\tilde \xvec}_2 \shortnorm{\tilde \bvec^*_{\MSset}(\tau) -  \bvec^*_{\MSset}(\tau)}_2 \lesssim \delta
	$ by \ref{as:Gram lower bound} and \eqref{eqn:beta tilde rate}.
	Leveraging this result and the Tower law, we find
	\begin{align*}
		\sup_{f \in \mc F} \E [f^2]
		& \leq \sup_{\bm \alpha \in \mc B_{m-r+s}(1), \tau \in \mc T} \E\left[ \E\left[ (\bm \alpha^T\widetilde{\bm Z}^T \tilde \xvec)^2  1\{|y -  \xvec^T\bvec^*(\tau)| < C\delta\}  \right] \middle | \xvec \right]\\
		& = \sup_{\bm \alpha \in \mc B_{m-r+s}(1), \tau \in \mc T} \E\left[ (\bm \alpha^T \widetilde{\bm Z}^T\tilde \xvec)^2  \Pr \left(|y - \xvec^T\bvec^*(\tau)| < C\delta \middle | \xvec \right)  \right].
	\end{align*}
	Applying the mean value theorem and \ref{as:density conditions}(i), we see that 
	\begin{align*}
		\Pr \left(|y - \xvec^T\bvec^*(\tau)| < C\delta \middle |  \xvec \right)
		& = \Pr \left(\xvec^T\bvec^*(\tau) - C\delta < y <  \xvec^T\bvec^*(\tau) + C\delta\middle |  \xvec \right)\\
		& = F_{y|\xvec}( \xvec^T\bvec^*(\tau) + C\delta) -  F_{y|\xvec}(\xvec^T\bvec^*(\tau) - C\delta)\\
		& \lesssim \delta
	\end{align*}
	uniformly for $\tau \in \overline{\mc T}$ and $\xvec \in \mc X$.
	Combining this with the fact that $ \sup_{\bm \alpha \in \mc B_{m-r+s}(1)} \E\left[ (\bm \alpha^T \widetilde{\bm Z}^T \tilde \xvec)^2 \right] =  \bm \alpha^T \widetilde{\bm Z}^T \E\left[  \tilde \xvec \tilde \xvec^T \right] \widetilde{\bm Z}\bm \alpha \lesssim 1$ uniformly for all $n$ by \ref{as:Gram lower bound} yields
	\begin{equation}
		\sup_{f \in \mc F} \E [f^2] \lesssim \delta. \label{eqn:Ef2 bound}
	\end{equation}
	
	Replacing $\varepsilon$ by $\varepsilon/\xi_n$ in the entropy
	bound above gives
	\[
	\log N_{[]}(\varepsilon,\mc F,L_2(P))
	\lesssim
	(m-r+s)
	\log\left\{
	\frac{
		D\xi_n(1+\xi_n\sqrt{m+s})
	}{\varepsilon}
	\right\}.
	\]
	Consequently,
	\begin{align*}
		J_{[]}
		\left(
		C\delta^{1/2},
		\mc F,L_2(P)
		\right)
		&\lesssim
		\int_0^{C\delta^{1/2}}
		\sqrt{
			(m-r+s)
			\log\left\{
			\frac{
				D\xi_n(1+\xi_n\sqrt{m+s})
			}{\varepsilon}
			\right\}
		}
		\,d\varepsilon\\
		&\lesssim
		\delta^{1/2}
		\sqrt{
			(m-r+s)
			\log\left\{
			\frac{
				D\xi_n(1+\xi_n\sqrt{m+s})
			}{\delta^{1/2}}
			\right\}
		}.
	\end{align*}
	Condition \ref{as:Gram lower bound} implies
	\[
	\xi_n^2
	\geq
	\E\norm{\tilde\xvec}_2^2
	=
	\operatorname{tr}(\bm\Sigma)
	\gtrsim m+s.
	\]
	Together with
	$(m+s)\xi_n^2\log^2n=o(n)$ and $m-r+s\geq1$,
	this shows that the logarithm in the preceding display is
	$O(\log n)$ and therefore
	$J_{[]}
	\left(
	C\delta^{1/2},
	\mc F,L_2(P)
	\right) \lesssim \delta^{1/2} \sqrt{(m-r+s)\log n}$.
	
	Since \eqref{eqn:Ef2 bound} holds, the bounded bracketing maximal inequality in Lemma 3.4.2 of \citet{van_der_vaart_weak_1996} provides that
	\begin{align*}
		\E\left[
		\sup_{f\in\mc F}
		|\empproc(f)|
		\right]
		&\lesssim
		\delta^{1/2}
		\sqrt{(m-r+s)\log n}
		\left\{
		1+
		\frac{\xi_n}{
			\delta\sqrt n
		}
		\delta^{1/2}
		\sqrt{(m-r+s)\log n}
		\right\}\\
		&\lesssim
		\frac{
			\xi_n^{1/2}
			(m-r+s)^{3/4}
		}{n^{1/4}}
		\log^{1/2}n
		+
		\frac{
			\xi_n(m-r+s)
		}{\sqrt n}
		\log n\\
		&\ll
		\sqrt{m-r+s},
	\end{align*}
	under our assumption that $ (m+s)\xi_n^2 \log^2 n  \ll n $. Markov's inequality then gives
	\[
	\tilde\epsilon_{02}
	\leq
	\sup_{f\in\mc F}|\empproc(f)|
	=
	o_p\left(\sqrt{m-r+s}\right),
	\]
	which completes the proof of \eqref{eqn:E2b}.
	
	Next we move to \eqref{eqn:E2c}.
	To bound $\tilde \epsilon_{1}$, we define the class of functions
	\begin{multline*}
		\mc G \defeq \bigg\{ (\tilde \xvec, y) \mapsto (\bm \alpha^T \widetilde{\bm Z}^T\tilde \xvec)\left[ 1(y < \tilde \xvec^T(\tilde \bvec^*_{\MSset}(\tau) + \widetilde{ \bm Z } \dvec ) -  1(y < \tilde \xvec^T\tilde \bvec^*_{\MSset}(\tau)) \right] : \\
		\tau \in \mc T, \bm \alpha \in \mc B_{m-r+s}(1), \norm{\dvec}_2 \leq d_n \bigg\}
	\end{multline*}
	with $d_n = B\sqrt{\frac{m-r+s}{n}}$ and let $G = \norm{\tilde \xvec}_2$ be its envelope.
	Clearly $\tilde \epsilon_1 = \sup_{f \in \mc G}| \empproc (f) | $.
	By Lemma 25 of \citet{belloni_conditional_2019}, the uniform entropy numbers of $\mc G$ satisfy
	$$
	\sup_Q \log N(\epsilon \norm{G }_{Q,2}, \mc G, L_2(Q) ) \lesssim (m-r+s) \log\left(\frac{1}{ \epsilon}\right)
	$$
	uniformly over $0 < \epsilon < 1$, where supremum is taken over all finitely discrete probability measures $Q$.
	Applying the third inequality in Lemma 21 of \citet{belloni_conditional_2019}, this implies
	$$
	\sup_{f \in \mc G} |\empproc (f)| \lesssim_p \sqrt{m-r+s} \bigg(\sup_{f \in \mc G} \E [f^2]  \bigg)^{1/2} \log^{1/2}n + \frac{m-r+s}{\sqrt n} \xi_n \log n.
	$$
	Using the same argument we used to establish \eqref{eqn:Ef2 bound}, we can show
	$
	\sup_{f \in \mc G} \E [f^2] \lesssim \xi_n d_n
	$.
	Therefore we have
	$
	\tilde \epsilon_{1} = \sup_{f \in \mc G}|\empproc (f)| \ll_p \sqrt{ m-r+s }
	$
	under our assumption that $(m-r+s)\xi_n^2 \log^2 n \ll n$. 
	
	Lastly, we seek to prove \eqref{eqn:E2d}.
	Let
	$
	\bar{\bm J}(\tau) \defeq \E\left[ f(\xvec^T \tilde \bvec^*(\tau) )\tilde \xvec \tilde \xvec^T \right].
	$
	We see that	
	$
	\norm{\bar{\bm J}(\tau) - \bm J(\tau) }_2 = \sup_{\bm v \in \mc B_{m+s}(1) } \left|\bm v^T\left(\bar{\bm J}(\tau) - \bm J(\tau) \right)\bm v \right|
	$. 
	Applying the mean value theorem, we find that for all $\bm v \in \mc B_{m+s}(1)$
	\begin{align*}
		\left|\bm v^T\left(\bar{\bm J}(\tau) - \bm J(\tau) \right)\bm v  \right|
		& \leq \E \left[  \left| f(\xvec^T \bvec^*(\tau)) - f(\xvec^T \tilde \bvec^*(\tau)) \right| (\tilde \xvec^T \bm v)^2 \right]\\
		& \leq \E \left[ C \, \left|\xvec^T(\bvec^*(\tau) - \tilde \bvec^*(\tau)) \right| \, (\tilde \xvec^T \bm v)^2 \right]\\
		& = \E \left[ C \, \left|\xvec_{\Mset}^T(\bvec_{\Mset}^*(\tau) - \tilde \bvec_{\Mset}^*(\tau)) \right| \, (\tilde \xvec^T \bm v)^2 \right]\\
		& \lesssim \E \left[ (\tilde \xvec^T \bm v)^2 \right] \xi_n \norm{ \bvec_{\Mset}^*(\tau) - \tilde \bvec_{\Mset}^*(\tau) }_2 \\
		& \lesssim \lambda_{\max}\{ \bm \Sigma \} \, \xi_n \sqrt{\frac{m-r+s}{n}}\\ 
		& \lesssim  \xi_n \sqrt{\frac{m-r+s}{n}},
	\end{align*}
	where the second inequality holds for some $C>0$ uniformly over $\xvec \in \mc X$, $\tau \in \overline{\mc T}$, and $n$ per \ref{as:density conditions} (iii) and the final inequality follows from \ref{as:Gram lower bound}. 
	
	Using this bound, we find
	\begin{align*}
		\sqrt{n} \left| \bm \alpha^T\widetilde{\bm Z}^T\left( \bm J(\tau) - \bar{\bm J}(\tau) \right) \widetilde{\bm Z}\dvec  \right| 
		& \leq \sqrt{n} \shortnorm{\widetilde{\bm Z} \bm \alpha }_2 \norm{\bm J(\tau) - \bar{\bm J}(\tau)}_2 \shortnorm{\widetilde{\bm Z} \dvec}_2
		\lesssim \xi_n d_n \sqrt{m - r+s}
	\end{align*}
	uniformly over $\bm \alpha \in \mc B_{m-r+s}(1)$ and  $(\tau, \dvec) \in \bar{ \mc R}_n$.
	By extension, we have
	\begin{equation}\label{eqn:tilde ep2 bound}
		\tilde \epsilon_2 \lesssim \sup_{\substack{\bm \alpha \in \mc B_{m-r+s}(1)\\ (\tau, \dvec) \in \bar{ \mc R}_n }} n^{1/2} \norm{\bm \alpha^T( \E[\bar \psi_i(\dvec, \tau)] - \E[\bar \psi_i(\bm 0, \tau)]) - \bm \alpha^T \widetilde{\bm Z}^T \bar{ \bm J}(\tau)\widetilde{ \bm Z } \dvec }_2 + \xi_n d_n \sqrt{m-r+s}.
	\end{equation}
	We see that
	$
	\E[\bar \psi_i(\dvec, \tau)|\xvec_i] = \widetilde{\bm Z}^T\tilde \xvec \left[\, F_{y_i|\xvec_i} (\tilde \xvec_{i}^T(\tilde \bvec^*_{\MSset}(\tau) + \widetilde{\bm Z}\dvec )) - \tau\right] 
	$.
	Therefore by the Tower law and the mean value theorem, we find
	\begin{align*}
		\bm \alpha^T\left(\E[\bar \psi_i(\dvec, \tau)] - \E[\bar \psi_i(\bm 0, \tau)] \right)
		& = \E \left[ (\bm \alpha^T\widetilde{\bm Z}^T\tilde \xvec) \left\{ F_{y|\xvec} (\tilde \xvec^T(\tilde \bvec^*_{\MSset}(\tau) + \widetilde{\bm Z}\dvec ))  -  F_{y|\xvec} (\tilde \xvec^T\tilde \bvec^*_{\MSset}(\tau))\right\}  \right]\\
		& = \E \left[ (\bm \alpha^T\widetilde{\bm Z}^T\tilde \xvec) f_{y|\xvec} (\tilde \xvec^T(\tilde \bvec^*_{\MSset}(\tau) + \widetilde{\bm Z} \bar \dvec )) \,  \tilde \xvec^T\widetilde{\bm Z}\dvec \right]
	\end{align*}
	where $\bar \dvec$ is on the line segment between $\bm 0$ and $\dvec$.
	Applying the mean value theorem, we find that
	$$
	\left|f_{y|\xvec} (\tilde \xvec^T(\tilde \bvec^*_{\MSset}(\tau) + \widetilde{\bm Z} \bar \dvec )) - f_{y|\xvec} (\tilde \xvec^T\tilde \bvec^*_{\MSset}(\tau))\right| 
	\leq C \, | \tilde \xvec^T\widetilde{\bm Z}\bar\dvec| 
	\leq \norm{\tilde \xvec}_2 \shortnorm{\widetilde{\bm Z}\dvec}_2 
	\lesssim \xi_n d_n
	$$
	where the first inequality holds for some $C>0$ uniformly over $\xvec \in \mc X$, $(\tau, \dvec) \in \bar{ \mc R}_n$, and $n$ by \ref{as:density conditions} (iii).
	Thus we see that
	\begin{align*}
		n^{1/2} & \norm{\bm \alpha^T( \E[\bar \psi_i(\dvec, \tau)] - \E[\bar \psi_i(\bm 0, \tau)]) - \bm \alpha^T \widetilde{\bm Z}^T \bar{ \bm J}(\tau)\widetilde{ \bm Z } \dvec }_2\\
		& = n^{1/2} \norm{ \E \left[ (\bm \alpha^T\widetilde{\bm Z}^T\tilde \xvec ) f_{y|\xvec} (\tilde \xvec^T(\tilde \bvec^*_{\MSset}(\tau) + \widetilde{\bm Z}\bar \dvec )) - f_{y|\xvec} (\tilde \xvec^T\tilde \bvec^*_{\MSset}(\tau))   \right] \tilde \xvec^T \widetilde{\bm Z} \dvec }_2\\
		& \lesssim  \sqrt{n} \left|  \bm \alpha^T \widetilde{\bm Z}^T \E[\tilde \xvec \tilde \xvec^T] \widetilde{\bm Z} \dvec  \right| \xi_n d_n\\
		& \lesssim \sqrt{n} \xi_n d_n^2
	\end{align*}
	uniformly over $\bm \alpha \in \mc B_{m-r+s}(1)$ and $(\tau, \dvec) \in \bar{ \mc R}_n$ by \ref{as:Gram lower bound}.
	Combining this with \eqref{eqn:tilde ep2 bound} yields
	$
	\tilde \epsilon_2 \lesssim  \sqrt{n} \xi_n d_n^2 + \xi_n d_n \sqrt{m+s-r} \lesssim \sqrt{m+s-r} \xi_n d_n  = o(\sqrt{m+s-r})
	$
	since $\xi_n d_n = o(1)$ by our assumptions. 
	
	Having shown that \eqref{eqn:E2a} - \eqref{eqn:E2d} hold, the proof is complete.
\end{proof}

\begin{proof}[Proof of Theorem \ref{thm:oracle bahadur}]
	We begin by proving \eqref{eqn:RE.oracle bahadur}.
	We see that
	\begin{align}
		\frac{1}{\sqrt{n}} \sum_{i \in [n]} \tilde \xvec_i ( \hat a_i^o(\tau) - (1 - \tau) ) 
		& = \frac{1}{\sqrt{n}} \sum_{i \in [n]} \tilde \xvec_i ( \hat a_i^o(\tau) 1( y_i \neq \xvec^T \hat \bvec^o_0 (\tau) ) - (1 - \tau) ) \notag \\
		& \hspace*{0.5cm} +  \frac{1}{\sqrt{n}} \sum_{i \in [n]} \tilde \xvec_i ( \hat a_i^o(\tau) 1( y_i = \xvec^T \hat \bvec^o_0 (\tau) )  ) \notag \\
		& \defeq S_1(\tau) + S_2(\tau) \label{eqn:rs1}.
	\end{align}
	Recall that $\mc I(\hat \bvec^o_0 (\tau)) = \{i \in [n]: y_i = \xvec_i^T \hat \bvec^o_0 (\tau) \}$. We know that $|\mc I| \leq m-r+s$ for every $\tau \in \mc T$ (one need simply express \eqref{eqn:RE.oracle} as a linear programming problem and count the number of non-basic variables) and that $|\hat a_i^o(\tau)| \leq 1$.
	Therefore
	\begin{equation}
		\sup_{\tau \in \mc T} \norm{S_2 (\tau)}_2 \leq n^{-1/2} (m-r+s) \max_i \norm{\tilde \xvec_i}_2 \leq \frac{(m-r+s) \xi_n}{\sqrt{n}} \label{eqn:S2 bound}.
	\end{equation}
	
	Define $\psi_i(\bvec, \tau) = \tilde{\xvec}_i (1(y_i < \xvec_i^T \bvec) - \tau)$.	
	Since $\hat a_i^o(\tau) 1( y_i \neq \xvec^T \hat \bvec^o_0 (\tau) ) =  1( y_i > \xvec^T \hat \bvec^o_0 (\tau) )$, $S_1(\tau)$ simplifies to 
	\begin{align}
		S_1(\tau) 
		& = \frac{1}{\sqrt{n}} \sum_{i \in [n]} \tilde \xvec_i ( 1( y_i > \xvec^T \hat \bvec^o_0 (\tau) ) - (1 - \tau) ) \notag\\
		& = \frac{1}{\sqrt{n}} \sum_{i \in [n]} \tilde \xvec_i \psi_{\tau}(y_i - \xvec_i^T \hat \bvec^o_0 (\tau) )
		= -\frac{1}{\sqrt{n}} \sum_{i \in [n]} \psi_i(\hat \bvec^o_0 (\tau), \tau )  \label{eqn:S1tau expr}.
	\end{align}
	Expanding the right hand side of the last expression, we have
	\begin{align}
		\frac{1}{\sqrt{n}} \sum_{i \in [n]} \psi_i(\bvecOR_0 (\tau), \tau )
		& = \empproc( \psi_i(\bvec^*(\tau), \tau) ) + \underbrace{\empproc( \psi_i(\bvecOR_0(\tau), \tau) ) - \empproc( \psi_i(\bvec^*(\tau), \tau) )}_{\defeq I_1(\tau)} \notag \\			
		&  + \underbrace{ \sqrt{n} \E[ \psi_i(\bvecOR_0(\tau), \tau) ] - \sqrt{n} \E\left[ \psi_i(\bvec^*(\tau), \tau) \right] - \sqrt{n} \bm J(\tau) (\bvecOR_{0,\MSset}(\tau) - \bvec^*_{\MSset}(\tau))}_{\defeq I_2(\tau)} \notag \\
		& + \sqrt{n} \E\left[ \psi_i(\bvec^*(\tau), \tau) \right] + \sqrt{n} \bm J(\tau) (\bvecOR_{0,\MSset}(\tau) - \bvec^*_{\MSset}(\tau)) \notag\\
		& = \frac{1}{\sqrt{n}} \sum_{i \in [n]} \psi_i(\bvec^* (\tau), \tau ) + I_1(\tau) + I_2(\tau)  + \sqrt{n} \bm J(\tau) (\bvecOR_{0,\MSset}(\tau) - \bvec^*_{\MSset}(\tau)). \label{eqn:reduced oracle score 1}
	\end{align}
	
	For a sequence of scalars $d_n \to 0$, define the sets
	\begin{equation}
		\mc R (d_n) \defeq \{(\tau,\bvec_{\MSset}) \in \mc T \times \reals^{m+s} : \norm{\bvec_{\MSset}(\tau) - \bvec_{\MSset}^*(\tau)}_2 \leq d_n \}.\label{eqn:Rn set}
	\end{equation}
	Theorem \ref{thm:oracle rate} establishes that for every $\epsilon > 0$ there exists $B(\epsilon) > 0$ such that $(\tau, \bvecOR_0) \in \mc R (d_n)$, where $d_n = B(\epsilon)\sqrt{\frac{m+s}{n}}$, with probability at least $1 - \epsilon$.
	Lemma 32 of \citet{belloni_conditional_2019} provides that
	\begin{align}
		\sup_{\tau \in \mc T} \norm{ I_1(\tau)}_2 
		& \leq \sup_{(\tau, \bvec) \in \mc R(d_n) }\norm{\mathbb G_n(\psi_i(\bvec, \tau)) - \mathbb G_n(\psi_i(\bvec^*(\tau), \tau))}_2 \notag\\
		& \lesssim_p \left( \frac{(m+s)^{3/2}}{n^{1/2}}\xi_n \log n \right)^{1/2} + \frac{(m+s)}{n^{1/2}} \xi_n \log n \label{eqn:I1 bound}
	\end{align}
	and 
	\begin{align}
		\sup_{\tau \in \mc T} \norm{I_2(\tau)}_2 
		& \leq \sup_{(\tau, \bvec) \in \mc R(d_n) } n^{1/2} \norm{\E[\psi_i(\bvec, \tau)] - \E[\psi_i(\bvec^*(\tau), \tau)] - \bm J(\tau) (\bvec_{\MSset} - \bvec_{\MSset}^*(\tau)) }_2 \notag\\
		& \lesssim_p \frac{(m+s)}{n^{1/2}} \xi_n.\label{eqn:I2 bound}
	\end{align}
	Applying these bounds to \eqref{eqn:reduced oracle score 1}, we have
	\begin{align}
		\frac{1}{\sqrt{n}} \sum_{i \in [n]} \psi_i(\bvecOR_0 (\tau), \tau )
		& = \frac{1}{\sqrt{n}} \sum_{i \in [n]} \psi_i(\bvec^* (\tau), \tau ) + \sqrt{n} \bm J(\tau) (\bvecOR_{0,\MSset}(\tau) - \bvec^*_{\MSset}(\tau)) + R_1(\tau) \label{eqn:reduced oracle score 2}
	\end{align}
	where $\sup_{\tau \in \mc T}\norm {R_1(\tau)}_2 \lesssim_p \frac{(m+s)^{3/4} \xi_n^{1/2} \log^{1/2} n }{n^{1/4}}$ since $(m+s)^3 \xi_n^2\log n = o(n)$ by assumption.
	
	Consider the reduced oracle optimization problem \eqref{eqn:RE.oracle}.
	The dual feasibility conditions guarantee that there exists some $\bm v \in \reals^m$ such that
	\begin{align}
		\bm 0 
		& = \frac{1}{n} \sum_{i \in [n]} \tilde \xvec_i ( \hat a_i^0(\tau) - (1-\tau) ) + 
		\widetilde{ \bm C }^T \bm v \notag \\
		& = -\frac{1}{n} \sum_{i \in [n]} \psi_i(\hat \bvec^o_0 (\tau), \tau )  + 
		\widetilde{ \bm C }^T \bm v + \frac{1}{\sqrt n} S_2(\tau) \label{eqn:reduced oracle KKT}
	\end{align}
	where the final equality follows from \eqref{eqn:rs1} and \eqref{eqn:S1tau expr}.
	Combining \eqref{eqn:reduced oracle score 2} and \eqref{eqn:reduced oracle KKT} and rearranging terms, we have
	\begin{equation}
		\sqrt n \widetilde{\bm C}^T \bm v
		= \frac{1}{\sqrt{n}} \sum_{i \in [n]} \psi_i(\bvec^* (\tau), \tau ) + \sqrt{n} \bm J(\tau) (\bvecOR_{0,\MSset}(\tau) - \bvec^*_{\MSset}(\tau)) + R_2(\tau)\label{eqn:v expr}
	\end{equation}
	where $R_2(\tau) = R_1(\tau) - S_2(\tau)$. 
	We see that
	\begin{multline*}
	\sup_{\tau \in \mc T} \norm{R_2(\tau)}_2 
	\leq  \sup_{\tau \in \mc T} \norm{R_1(\tau)}_2 + \sup_{\tau \in \mc T} \norm{S_2(\tau)}_2\\
	\lesssim_p \left( \frac{(m+s)^{3/2}}{n^{1/2}}\xi_n \log n \right)^{1/2} + \frac{(m+s)}{n^{1/2}} \xi_n \log n + \frac{(m+s)}{n^{1/2}} \xi_n + \frac{s}{n^{1/2}} \xi_n \\
	\lesssim \frac{(m+s)^{3/4}}{n^{1/4}}\xi_n^{1/2} \log^{1/2} n + \frac{(m+s)^{1/4}}{n^{1/4}} \xi_n^{1/2} \log^{1/2} n
	\lesssim \frac{(m+s)^{3/4}\xi_n^{1/2} \log^{1/2} n}{n^{1/4}}
	\end{multline*}
	since $(m+s)^2 \xi_n^2 \log n = o(n)$ under our assumptions.
	
	Rearranging \eqref{eqn:v expr} and multiplying by $[\bm J(\tau)]^{-1}$ yields
	\begin{align}
		\sqrt n (\bvecOR_{0,\MSset}(\tau) - \bvec^*_{\MSset}(\tau))
		& = - [\bm J(\tau)]^{-1} \frac{1}{\sqrt n} \sum_{i \in [n]} \psi_i(\bvec^* (\tau), \tau ) \notag \\
		& \text{\;\;} + \sqrt n [\bm J(\tau)]^{-1} \widetilde{\bm C}^T \bm v 
		- [\bm J(\tau)]^{-1} R_2(\tau) . \label{eqn:reduced oracle baha 1}
	\end{align}
	Multiplying this previous expression by $n^{-1/2}\widetilde{\bm C}$, we find
	\begin{align}
		\widetilde{\bm C}(\bvecOR_{0,\MSset}(\tau) - \bvec^*_{\MSset}(\tau))
		& = - n^{-1}\widetilde{\bm C} [\bm J(\tau)]^{-1} \sum_{i \in [n]} \psi_i(\bvec^* (\tau), \tau ) \notag \\
		& \text{\;\;} +  \bm H(\tau) \bm v 
		- n^{-1/2} \widetilde{\bm C} [\bm J(\tau)]^{-1} R_2(\tau).
	\end{align}
	At the same time, under \ref{as:local alts} we see that 
	$$
	\widetilde{\bm C}(\bvecOR_{0,\MSset}(\tau) - \bvec^*_{\MSset}(\tau))
	= \bm C\bvecOR_{0,\Mset}(\tau) - {\bm C}\bvec^*_{\Mset}(\tau) = \bm \Psi(\tau) - {\bm C} \bvec^*_{\Mset}(\tau) = - \bm \eta(\tau)/\sqrt n.
	$$
	Substituting this into the previous expression, rearranging terms, and multiplying by $[\bm H(\tau)]^{-1}$, we have
	\begin{align*}
		\bm v
		& = - n^{-1/2}[\bm H(\tau)]^{-1} \bm \eta(\tau)
		+ n^{-1}[\bm H(\tau)]^{-1}\widetilde{\bm C} [\bm J(\tau)]^{-1} \sum_{i \in [n]} \psi_i(\bvec^* (\tau), \tau ) \\
		& \text{\; \;} + n^{-1/2} [\bm H(\tau)]^{-1} \widetilde{\bm C} [\bm J(\tau)]^{-1} R_2(\tau) .
	\end{align*}
	Plugging this expression for $\bm v$ into \eqref{eqn:reduced oracle baha 1} and simplifying yields
	\begin{align}
		\sqrt n (\bvecOR_{0,\MSset}(\tau) - \bvec^*_{\MSset}(\tau))
		& = - [\bm J(\tau)]^{-1/2}(\bm I - \bm P(\tau)) [\bm J(\tau)]^{-1/2} \frac{1}{\sqrt n} \sum_{i \in [n]} \psi_i(\bvec^* (\tau), \tau ) \notag \\
		& \text{\; \; } - [\bm J(\tau)]^{-1} \widetilde{\bm C}^T [\bm H(\tau)]^{-1} \bm \eta(\tau) \notag\\
		& \text{\; \; } - \underbrace{\left\{[\bm J(\tau)]^{-1} - [\bm J(\tau)]^{-1} \widetilde{\bm C}^T [\bm H(\tau)]^{-1} \widetilde{\bm C} [\bm J(\tau)]^{-1}  \right\} R_2(\tau)}_{\defeq R_3(\tau)} . \label{eqn:reduced oracle baha 2}
	\end{align}
	We have that $\sup_{\tau \in \mc T} \norm{[\bm J(\tau)]^{-1}} \lesssim 1 $ from Lemma \ref{lemma:bounding lemma}.
	Moreover, briefly dropping  $\tau$ from our notation for ease of presentation, we find
	\begin{align*}
		\norm{\bm J^{-1} \widetilde{\bm C}^T \bm H^{-1}\widetilde{\bm C}\bm J^{-1}}_2
		& \leq 
		\norm{\bm J^{-1/2}}_2^2
		\norm{\bm J^{-1/2} \widetilde{\bm C}^T \bm H^{-1} \widetilde{\bm C} \bm J^{-1/2}}_2\\
		& \lesssim 
		\norm{\bm J^{-1/2} \widetilde{\bm C}^T \bm H^{-1}\widetilde{\bm C} \bm J^{-1/2}}_2\\
		& = \norm{ \widetilde{\bm C} \bm J^{-1} \widetilde{\bm C}^T \bm H^{-1} }_2\\
		& = \norm{ \bm H \bm H^{-1} }_2 = 1
	\end{align*}
	uniformly for all $n$ and $\tau \in \mc T$.
	Thus $\sup_{\tau \in \mc T} \norm{R_3(\tau)}_2 \lesssim \sup_{\tau \in \mc T} \norm{R_2(\tau)}_2$.
	Combining this with \eqref{eqn:reduced oracle baha 2} and the fact that $- \sum_{i \in [n]} \psi_i(\bvec^* (\tau), \tau ) = \sum_{i \in [n]}  \tilde \xvec_i \psi_{\tau}(y_i - \xvec_i^T \bvec^*(\tau))$  yields \eqref{eqn:RE.oracle bahadur}.
	
	We now shift to proving \eqref{eqn:oracle rs rep}.
	Plugging \eqref{eqn:reduced oracle baha 2} into \eqref{eqn:reduced oracle score 2} we find
	\begin{align*}
		\frac{1}{\sqrt{n}} \sum_{i \in [n]} \psi_i(\bvecOR_0 (\tau), \tau )
		& = \frac{1}{\sqrt{n}} \sum_{i \in [n]} \psi_i(\bvec^* (\tau), \tau ) 
		- [\bm J(\tau)]^{1/2}(\bm I - \bm P(\tau)) [\bm J(\tau)]^{-1/2} \frac{1}{\sqrt n} \sum_{i \in [n]} \psi_i(\bvec^* (\tau), \tau ) \\
		& \text{\; \; } - \widetilde{\bm C}^T [\bm H(\tau)]^{-1} \bm \eta(\tau) - \bm J(\tau)R_3(\tau) + R_1(\tau)\\
		& =  [\bm J(\tau)]^{1/2}\bm P(\tau)[\bm J(\tau)]^{-1/2} \frac{1}{\sqrt n} \sum_{i \in [n]} \psi_i(\bvec^* (\tau), \tau ) \\
		& \text{\; \; } - \widetilde{\bm C}^T [\bm H(\tau)]^{-1} \bm \eta(\tau) - \bm J(\tau)R_3(\tau) + R_1(\tau).
	\end{align*}
	Combining \eqref{eqn:rs1}, \eqref{eqn:S1tau expr}, and $- \sum_{i \in [n]} \psi_i(\bvec^* (\tau), \tau ) = \sum_{i \in [n]}  \tilde \xvec_i \psi_{\tau}(y_i - \xvec_i^T \bvec^*(\tau))$ with the previous expression, we have
	\begin{align*}
		\frac{1}{\sqrt{n}} \sum_{i \in [n]} \tilde \xvec_i ( \hat a_i^o(\tau) - (1 - \tau) )
		& =  [\bm J(\tau)]^{1/2}\bm P(\tau)[\bm J(\tau)]^{-1/2} \frac{1}{\sqrt n} \sum_{i \in [n]} \tilde \xvec_i \psi_{\tau}(y_i - \xvec_i^T \bvec^*(\tau)) \\
		& \text{\; \; } + \widetilde{\bm C}^T [\bm H(\tau)]^{-1} \bm \eta(\tau)
		+ \underbrace{\bm J(\tau)R_3(\tau) - R_1(\tau) + S_2(\tau)}_{\defeq R_4(\tau)}.
	\end{align*}
	Since $\sup_{\tau \in \mc T} \norm{\bm J(\tau)} \lesssim 1 $ by Lemma \ref{lemma:bounding lemma}, the bounds we have derived for $R_3(\tau)$, $R_1(\tau)$, and $S_2(\tau)$ imply 
	$\sup_{\tau \in \mc T} \norm{R_4(\tau)}_2 \lesssim \frac{(m+s)^{3/4}\xi_n^{1/2} \log^{1/2} n}{n^{1/4}}$, proving \eqref{eqn:oracle rs rep}.
	
	Lastly, we will prove \eqref{eqn:FE.oracle bahadur}.
	As in \eqref{eqn:reduced oracle score 1}, we can decompose the score as
	\begin{align}
		\frac{1}{\sqrt{n}} \sum_{i \in [n]} \psi_i(\bvecOR_a (\tau), \tau )
		& = \empproc( \psi_i(\bvec^*(\tau), \tau) ) + \underbrace{\empproc( \psi_i(\bvecOR_a(\tau), \tau) ) - \empproc( \psi_i(\bvec^*(\tau), \tau) )}_{\defeq H_1(\tau)} \notag \\			
		&  + \underbrace{ \sqrt{n} \E[ \psi_i(\bvecOR_a(\tau), \tau) ] - \sqrt{n} \E\left[ \psi_i(\bvec^*(\tau), \tau) \right] - \sqrt{n} \bm J(\tau) (\bvecOR_{a,\MSset}(\tau) - \bvec^*_{\MSset}(\tau))}_{\defeq H_2(\tau)} \notag \\
		& + \sqrt{n} \E\left[ \psi_i(\bvec^*(\tau), \tau) \right] + \sqrt{n} \bm J(\tau) (\bvecOR_{a,\MSset}(\tau) - \bvec^*_{\MSset}(\tau)) \notag\\
		& = \frac{1}{\sqrt{n}} \sum_{i \in [n]} \psi_i(\bvec^* (\tau), \tau ) + H_1(\tau) + H_2(\tau)  + \sqrt{n} \bm J(\tau) (\bvecOR_{a,\MSset}(\tau) - \bvec^*_{\MSset}(\tau)).
	\end{align}
	As with the reduced oracle estimator, Theorem \ref{thm:oracle rate} establishes that for every $\epsilon > 0$ there exists $B(\epsilon) > 0$ such that $(\tau, \bvecOR_a) \in \mc R (d_n)$, where $d_n = B(\epsilon)\sqrt{\frac{m+s}{n}}$, with probability at least $1 - \epsilon$.
	As such, $\sup_{\tau \in \mc T} \norm{ H_1(\tau)}_2  \lesssim_p \left( \frac{(m+s)^{3/2}}{n^{1/2}}\xi_n \log n \right)^{1/2} + \frac{(m+s)}{n^{1/2}} \xi_n \log n$ and $\sup_{\tau \in \mc T} \norm{H_2(\tau)}_2 \lesssim_p \frac{(m+s)}{n^{1/2}} \xi_n$ as in \eqref{eqn:I1 bound} and \eqref{eqn:I2 bound}.
	Therefore
	\begin{align}
		\frac{1}{\sqrt{n}} \sum_{i \in [n]} \psi_i(\bvecOR_a (\tau), \tau )
		= \frac{1}{\sqrt{n}} \sum_{i \in [n]} \psi_i(\bvec^*(\tau), \tau ) + \sqrt{n} \bm J(\tau) (\bvecOR_{a,\MSset}(\tau) - \bvec^*_{\MSset}(\tau)) + \widetilde{R}_1(\tau) 
		\label{eqn:full oracle score 1}
	\end{align}
	with $\sup_{\tau \in \mc T}\shortnorm {\widetilde R_1(\tau)}_2 \lesssim \frac{(m+s)^{3/4} \xi_n^{1/2} \log^{1/2} n }{n^{1/4}}$.
	
	Let $ \tilde a_1^o(\tau), \ldots, \tilde a_n^o(\tau) $ denote the solution to the dual problem for \eqref{eqn:FE.oracle}, given by
	\begin{equation}
		\tilde {\bm a}^o(\tau) = \argmax_{\bm a} \{ \bm y^T\bm a: \widetilde \xmat^T \bm a = (1 - \tau)\widetilde\xmat^T \bm 1_n ,\, \bm a \in [0,1]^n  \}. 
	\end{equation}
	Following the same argument applied in \eqref{eqn:rs1} and \eqref{eqn:S2 bound}, we have
	$$
	\frac{1}{\sqrt{n}} \sum_{i \in [n]} \tilde \xvec_i ( \tilde a_i^o(\tau) - (1 - \tau) ) 
	= -\frac{1}{\sqrt{n}} \sum_{i \in [n]} \psi_i(\bvecOR_a (\tau), \tau ) + \widetilde R_2(\tau).
	$$
	where $ \sup_{\tau \in \mc T} \shortnorm{ \widetilde R_2(\tau)}_2 \leq \frac{(m+s)\xi_n}{\sqrt n} $.
	In addition, the dual feasibility conditions guarantee that $\frac{1}{\sqrt{n}} \sum_{i \in [n]} \tilde \xvec_i ( \tilde a_i^o(\tau) - (1 - \tau) ) = \bm 0$.		
	
	Combining these results with \eqref{eqn:full oracle score 1} and rearranging terms, we have
	\begin{equation*}
		\sqrt{n} \bm J(\tau) (\bvecOR_{a,\MSset}(\tau) - \bvec^*_{\MSset}(\tau)) 
		= -\frac{1}{\sqrt{n}} \sum_{i \in [n]} \psi_i(\bvec^*(\tau), \tau ) +  \widetilde{R}_3(\tau)
	\end{equation*}
	where  $ \sup_{\tau \in \mc T} \shortnorm{ \widetilde R_3(\tau)}_2 \leq\frac{(m+s)^{3/4}}{n^{1/4}}\xi_n^{1/2} \log^{1/2} n $.
	Left-multiplying the previous expression by $[\bm J(\tau)]^{-1}$ and noting that $- \sum_{i \in [n]} \psi_i(\bvec^* (\tau), \tau ) = \sum_{i \in [n]}  \tilde \xvec_i \psi_{\tau}(y_i - \xvec_i^T \bvec^*(\tau))$ and $\sup_{\tau \in \mc T} \norm{[\bm J(\tau)]^{-1}}_2 \lesssim 1 $ yields \eqref{eqn:FE.oracle bahadur}.
\end{proof}

\subsubsection{Proof of Theorem \ref{thm:oracle test stat convergence}}
In this section, we will study the asymptotic behavior of the oracle Wald and rank score processes through a series of lemmas, which we will then use to prove Theorem \ref{thm:oracle test stat convergence}.

Define $\bm K(\tau) = \widetilde \Cmat [\bm J(\tau)]^{-1} \bm \Sigma [\bm J(\tau)]^{-1} \widetilde \Cmat^T$.
We will study the processes
\begin{equation}
	\tilde{\bm \omega}_n(\tau) \defeq \sqrt{n} (\bm K(\tau))^{-1/2}(\widetilde \Cmat \bvecOR_{a}(\tau) - \bm \Psi(\tau) ) \label{eqn:Wald known J rewritten}
\end{equation}
and
\begin{equation}
	\widetilde{W}_n(\tau) \defeq \norm{\tilde{\bm \omega}_n(\tau)}_2^2/(\tau (1 - \tau))
\end{equation}
as stepping stones to study $W_n^o(\tau)$.
Likewise, we will study the processes
\begin{equation}
	\tilde {\bm s} (\tau)  
	\defeq \frac{1}{\sqrt{n}} [\bm K(\tau)]^{-1/2} 
	\widetilde{\bm C} [\bm J(\tau)]^{-1}\sum_{i \in [n]} \tilde \xvec_i ( \hat a_i^o(\tau) - (1 - \tau)) \label{eqn:rs with known J}
\end{equation}
and
\begin{equation}
	\widetilde{R}_n(\tau) =  \norm{\tilde{\bm s}_n(\tau)}_2^2/(\tau (1 - \tau))
\end{equation}
as stepping stones to study $R_n^o(\tau)$.

Using the Bahadur representation in Theorem \ref{thm:oracle bahadur}, we establish a Gaussian process approximations for $\tilde{\bm \omega}_n(\tau)$ and $\tilde {\bm s} (\tau)$ conditional on $(\xvec_i)_{i \in [n]}$:
\begin{lemma}\label{lemma:GP for Wald with known}
	Suppose that \ref{as:iid array} - \ref{as:local alts} hold and that $(m + s)^3 \xi_n^2 = o(n^{1 - \delta_1})$ and $r^7\xi_n^6 = o(n^{1 - \delta_2})$ for some $\delta_1, \delta_2 >0$.
	Then there exist sequence of processes $(\bm G_n(\cdot))_{n\geq1}$ and $(\widetilde{\bm G}_n(\cdot))_{n\geq1}$ such that
	\begin{eqnarray}
		\tilde{\bm \omega}_n(\tau) & = & \bm G_n(\tau) + \bm \nu(\tau) + \bm r_n(\tau), \label{eqn:Wald known J GP approx}\\
		\tilde{\bm s}_n(\tau) & = & \widetilde{\bm G}_n(\tau) + \bm \nu(\tau) + \tilde{\bm r}_n(\tau) \label{eqn:rs known J GP approx}
	\end{eqnarray}
	for all $\tau \in \mc T$,
	where each $\bm G_n(\tau)$ (and $\widetilde{\bm G}_n(\tau)$ respectively) is, conditionally on $(\xvec_i)_{i \in [n]}$, a zero-mean Gaussian process with a.s. continuous sample paths and covariance function
	$$
	\E[\bm G_n(\tau)\bm G_n(\tilde \tau)^T |(\xvec_i)_{i \in [n]} ] 
	= \mathbb{E}_n(\bm Z_i \bm Z_i^T) (\tau \vee \tilde \tau - \tau \tilde \tau) 
	\text{\; for all \;} \tau, \tilde{\tau} \in \mc T,
	$$
	where $\bm Z_i = [\bm K(\tau)]^{-1/2} \widetilde \Cmat [\bm J(\tau)]^{-1} \tilde{\xvec}_i$ for $i \in [n]$;
	$\bm \nu(\tau) = [\bm K(\tau)]^{-1/2} \bm \eta(\tau)$;
	and\\
	$
	\sup_{\tau \in \mc T} \norm{\bm r_n(\tau)}_2 = o_p(n^{- \epsilon'})
	$
	and
	$
	\sup_{\tau \in \mc T} \norm{\tilde{\bm r}_n(\tau)}_2 = o_p(n^{- \epsilon'})
	$
	for some $\epsilon' > 0$.
\end{lemma}

Building from this lemma, we can characterize the asymptotic behavior of $\widetilde{W}_n(\cdot)$ and $\widetilde{R}_n(\tau)$.
\begin{lemma} \label{lemma:Bessel for Wald with known J}
	Under the conditions of Lemma \ref{lemma:GP for Wald with known}, if $r$ is fixed and $\mc T \subseteq [\delta, 1 - \delta]$, $0 < \delta < 1/2$, then
	\begin{eqnarray}
		&& \sup_{\tau \in \mc T}\widetilde W_n(\tau) \overset{d}{\to} \sup_{\tau \in \mc T} Q^2_{r, \bm \mu(\tau)}(\tau), \label{eqn:Wald J to Bessel}\\
		&& \sup_{\tau \in \mc T}\widetilde R_n(\tau) \overset{d}{\to} \sup_{\tau \in \mc T} Q^2_{r, \bm \mu(\tau)}(\tau) \label{eqn:rs J to Bessel}
	\end{eqnarray}
	as $n \to \infty$, where $\bm \mu(\tau) = \bm \eta (\tau)^T [\bm K(\tau)]^{-1} \bm \eta(\tau)/(\tau(1 - \tau))$.
\end{lemma}

Recall that we use \citet{Powell1984}'s estimator of the Jacobian:
\begin{equation}
	\hatmat J^o_e (\tau) \defeq \frac{1}{2h} \mathbb{E}_n\left[ 1(|y_i - \xvec_i^T \bvecOR_{e}(\tau)| \leq h ) \tilde{\xvec}_i \tilde{\xvec}_i^T  \right], \label{eqn:J estimator}
\end{equation}
where $h$ is a bandwidth parameter for which $h = h_n \to 0$, for $e \in \{0,a\}$.
The following lemma establishes that $W_n^o(\tau)$ and $\widetilde{W}_n(\tau)$ (and $R_n^o(\tau)$ and $\widetilde{R}_n(\tau)$, respectively) are asymptotically equivalent if $h \to 0$ at an appropriate rate.
\begin{lemma}\label{lemma:Wald J to Jhat}
	Under the assumptions of Lemma \ref{lemma:GP for Wald with known} and the additional conditions that $\mc T \subseteq [\delta, 1 - \delta]$, $0 < \delta < 1/2$,
	$(m+s)^2 \xi_n^2 \log n = o(n^{1 - \epsilon} h)$, $h = o((m+s)^{-1/2}n^{-\epsilon/2})$, and $m+s = o(n^{\epsilon})$ for some $\epsilon > 0$,
	\begin{eqnarray}
		&& \sup_{\tau \in \mc T} | W_n^o(\tau) - \widetilde{W}_n(\tau) | = o_p(1) \label{eqn:Wald to Wald J} \\
		&& \sup_{\tau \in \mc T} | R_n^o(\tau) - \widetilde{R}_n(\tau) | = o_p(1). \label{eqn:rs to rs J}
	\end{eqnarray}
\end{lemma}

\noindent Using these lemmas, we prove Theorem \ref{thm:oracle test stat convergence}.
\begin{proof}[Proof of Theorem \ref{thm:oracle test stat convergence}]
	We will focus on the proof of \eqref{eqn:thm:Wald to Bessel} as the proof of \eqref{eqn:thm:rs to Bessel} follows the same argument.
	We have from Lemma \ref{lemma:Wald J to Jhat} that 
	$ \sup_{\tau \in \mc T} |W_n^o(\tau) - \widetilde{W}_n(\tau) | = o_p(1) $.
	Applying the reverse triangle inequality, this implies 	
	$ |\sup_{\tau \in \mc T} W_n^o(\tau) - \sup_{\tau \in \mc T}\widetilde{W}_n(\tau) | = o_p(1) $.
	We further have from Lemma \ref{lemma:Bessel for Wald with known J} that
	$\sup_{\tau \in \mc T}\widetilde{W}_n(\tau) \overset{d}{\to} \sup_{\tau \in \mc T} Q_{r, \bm \mu(\tau)}^2(\tau)$. 
	Combining these results and applying Slutsky's theorem yields \eqref{eqn:thm:Wald to Bessel}.
\end{proof}

\subsubsection{Proofs of lemmas for Theorem \ref{thm:oracle test stat convergence}}

\begin{proof}[Proof of Lemma \ref{lemma:GP for Wald with known}]
	We begin by proving \eqref{eqn:Wald known J GP approx}.		
	Under \ref{as:local alts}, we see that
	$ \widetilde \Cmat \bvecOR_{a,\MSset}(\tau) - \bm \Psi(\tau) = \widetilde \Cmat \bvecOR_{a, \MSset}(\tau) - (\widetilde \Cmat \bvec_{\MSset}^*(\tau) - \bm \eta(\tau)/\sqrt{n} ) = \widetilde \Cmat ( \bvecOR_{a,\MSset}(\tau) - \bvec_{\MSset}^*(\tau) ) + \bm \eta(\tau)/\sqrt{n}  $.
	Plugging this into \eqref{eqn:Wald known J rewritten} yields
	$$
	\tilde{\bm \omega}_n(\tau) 
	= \sqrt{n} [\bm K(\tau)]^{-1/2} \widetilde \Cmat ( \bvecOR_{a,\MSset}(\tau) - \bvec_{\MSset}^*(\tau) )
	+ [\bm K(\tau)]^{-1/2} \bm \eta(\tau).
	$$
	Theorem \ref{thm:oracle bahadur} provides that
	$$
	\sqrt{n}(\bvecOR_{a, \MSset}(\tau) - \bvec_{\MSset}^*(\tau)) =  \frac{1}{\sqrt n} [\bm J(\tau)]^{-1} \sum_{i \in [n]} \tilde \xvec_i \psi_{\tau}( y_i - \xvec_i^T \bvec^*(\tau) ) + \bm r_{1n}(\tau)
	$$	
	for all $\tau \in \mc T$. Using Lemmas \ref{lemma:Ui uniform} and \ref{lemma:Ui equivalence}, this can be rewritten as
	$$
	\sqrt{n}(\bvecOR_{a, \MSset}(\tau) - \bvec_{\MSset}^*(\tau)) =  \frac{1}{\sqrt n} [\bm J(\tau)]^{-1} \sum_{i \in [n]} \tilde \xvec_i \psi_{\tau}(U_i - \tau) + \bm r_{1n}(\tau)
	$$
	where $U_1, \ldots, U_n$ are i.i.d. Uniform$(0,1)$ conditional on $(\xvec_i)_{i \in [n]}$.
	
	Plugging this back into the previous expression for $\tilde{\bm \omega}_n(\tau)$, we find
	\begin{align}
		\tilde{\bm \omega}_n(\tau) 
		& =  \frac{1}{\sqrt n}  \sum_{i \in [n]} \bm Z_i \psi_{\tau}(U_i - \tau) +  [\bm K(\tau)]^{-1/2} \widetilde \Cmat \bm r_{1n}(\tau)	+ \bm \nu(\tau) \label{eqn:Wald process in terms of Ui}
	\end{align}
	where $\bm Z_i = [\bm K(\tau)]^{-1/2} \widetilde \Cmat \bm J(\tau)^{-1} \tilde{\xvec}_i$ for all $i$ and $\bm \nu(\tau) = [\bm K(\tau)]^{-1/2} \bm \eta(\tau)$.
	The rest of our proof of \eqref{eqn:Wald known J GP approx} will proceed in two stages. First, we will establish that $\frac{1}{\sqrt n}  \sum_{i \in [n]} \bm Z_i \psi_{\tau}(U_i - \tau)$ can be approximated by a Gaussian process. Then, we will derive a rate for $\sup_{\tau \in \mc T}\shortnorm{[\bm K(\tau)]^{-1/2} \widetilde \Cmat \bm r_{1n}(\tau)}_2$.
	
	For $d \in \mathbb Z^+$, let $\mc B_d(1) = \{\bm a \in \reals^d : \norm{\bm a}_2 = 1\}$ denote the unit sphere in $\reals^d$.
	Lemma 13 of \citet{belloni_conditional_2019} provides that if $(\bm Z_i)_{i \in [n]}$ satisfy
	\begin{equation}
		\sup_{\bm \alpha \in S^{m+s} } \mathbb{E}_n[ (\bm \alpha^T\bm Z_i)^2 ] = O(1) \text{, }
		\max_{i \in [n]} \norm{\bm Z_i}_2 = O(\xi_n) \text{, and }
		r^7\xi_n^6 = o(n^{1 - \epsilon}) \label{eqn:Lemma 14 Z conditions}
	\end{equation}
	where $\xi _n$ satisfies $1/\xi_n = O(1)$ and $\epsilon >0$, is a constant, then there exists a sequence of processes $(\bm G_n(\cdot))_{n\geq1}$ such that
	\begin{equation}
		\sup_{\tau \in \mc T} \norm{ \frac{1}{\sqrt n}  \sum_{i \in [n]} \bm Z_i \psi_{\tau}(U_i - \tau) - \bm G_n(\tau) }_2 = o_p(n^{-\epsilon_2'}) \label{eqn:Lemma 14 GP approx}
	\end{equation}
	for some $\epsilon_2' > 0$,
	where each $\bm G_n(\tau)$ is, conditionally on $(\bm Z_i)_{i \in [n]}$, a zero-mean Gaussian process with a.s. continuous sample paths and covariance functions
	\begin{equation}
		\E[\bm G_n(\tau)\bm G_n(\tilde \tau)^T |(\bm Z_i)_{i \in [n]} ] 
		= \mathbb{E}_n(\bm Z_i \bm Z_i^T) (\tau \vee \tilde \tau - \tau \tilde \tau) 
		\text{ for all } \tau, \tilde{\tau} \in \mc T. \label{eqn:Lemma 14 GP variance}
	\end{equation}
	
	We will show that the conditions in \eqref{eqn:Lemma 14 Z conditions} hold with probability converging to $1$ as $n \to \infty$.
	We start by bounding $\max_{i \in [n]} \norm{\bm Z_i}_2$.
	We see that for all $i \in [n]$,
	\begin{align*}
		\norm{\bm Z_i}_2^2
		& = \tr(\bm Z_i^T \bm Z_i  ) = \tr(\bm Z_i \bm Z_i^T )\\
		& = \tr\{   [\bm K(\tau)]^{-1/2} \widetilde \Cmat [\bm J(\tau)]^{-1} \tilde \xvec_i \tilde \xvec_i^T [\bm J(\tau)]^{-1}  \widetilde \Cmat^T [\bm K(\tau)]^{-1/2} \}\\
		& = \tr\{   [\bm J(\tau)]^{-1}  \widetilde \Cmat^T [\bm K(\tau)]^{-1} \widetilde \Cmat [\bm J(\tau)]^{-1} \tilde \xvec_i \tilde \xvec_i^T \}\\
		& \leq \sum_{j = 1}^{m+s}  \lambda_j\{ [\bm J(\tau)]^{-1}  \widetilde \Cmat^T [\bm K(\tau)]^{-1} \widetilde \Cmat [\bm J(\tau)]^{-1} \}\lambda_j\{ \tilde \xvec_i \tilde \xvec_i^T \}\\
		& \leq \lambda_{\max}\{ [\bm J(\tau)]^{-1}  \widetilde \Cmat^T [\bm K(\tau)]^{-1} \widetilde \Cmat [\bm J(\tau)]^{-1} \}\tr \{ \tilde \xvec_i \tilde \xvec_i^T \}
	\end{align*}
	where $\lambda_1(\cdot) \geq \cdots \geq \lambda_{m+s}(\cdot)$ denote the eigenvalues of the matrices and the first inequality follows from von Neumann's trace inequality (which holds since the matrices broken apart in the product are both positive semi-definite).
	The latter term satisfies
	$
	\tr \{ \tilde \xvec_i \tilde \xvec_i^T \} = \norm{\tilde \xvec_i}_2^2 \leq \xi_n^2
	$
	for all $i$.
	Turning to the first term, we find
	\begin{align*}
		\lambda_{\max}\{ [\bm J(\tau)]^{-1} \widetilde \Cmat^T [\bm K(\tau)]^{-1} & \widetilde \Cmat [\bm J(\tau)]^{-1} \}\\
		& = \lambda_{\max}\{ [\bm J(\tau)]^{-1}  \widetilde \Cmat^T [\bm K(\tau)]^{-1} \widetilde \Cmat [\bm J(\tau)]^{-1} \bm \Sigma \bm \Sigma^{-1}\}\\
		& \leq \lambda_{\max} ( \bm \Sigma^{-1} )\lambda_{\max}\{ [\bm J(\tau)]^{-1}  \widetilde \Cmat^T [\bm K(\tau)]^{-1} \widetilde \Cmat [\bm J(\tau)]^{-1} \bm \Sigma \}\\
		& = \lambda_{\max} ( \bm \Sigma^{-1} )\lambda_{\max}\{ [\bm K(\tau)]^{-1} \widetilde \Cmat [\bm J(\tau)]^{-1} \bm \Sigma [\bm J(\tau)]^{-1}  \widetilde \Cmat^T \}\\
		& = \lambda_{\max} ( \bm \Sigma^{-1} )\lambda_{\max}\{\bm I_{m+s}\}\\
		& = \lambda_{\max} ( \bm \Sigma^{-1} ).
	\end{align*}
	By \ref{as:Gram lower bound}, we have that $\bm \Sigma$ is positive definite and
	$\lambda_{\max} ( \bm \Sigma^{-1} ) = [\lambda_{\min}(\bm \Sigma) ]^{-1} = O(1)$.
	Combining this with the previous bounds gives us that 
	$\max_{i \in [n]} \norm{\bm Z_i}_2 = O(\xi_n)$.

	Next we will prove that $\sup_{\bm \alpha \in S^{m+s} } \mathbb{E}_n[ (\bm \alpha^T\bm Z_i)^2 ] = O(1)$ holds with probability converging to $1$.
	We see that 
	\begin{align*}
		\E[ \bm Z_i \bm Z_i^T ]
		& = [\bm K(\tau)]^{-1/2} \widetilde \Cmat [\bm J(\tau)]^{-1} \E[ \tilde \xvec_i \tilde \xvec_i^T ] [\bm J(\tau)]^{-1}  \widetilde \Cmat^T [\bm K(\tau)]^{-1/2}\\
		& =  [\bm K(\tau)]^{-1/2} \bm K(\tau) [\bm K(\tau)]^{-1/2}
		= \bm I_r.
	\end{align*}
	Therefore $\E[\bm \alpha^T \bm Z_i \bm Z_i^T \bm \alpha ] = \bm \alpha^T \bm \alpha = 1$ for any $\bm \alpha \in S^{m+s}$
	and
	$ \sup_{\bm \alpha \in S^{m+s} } \E[ (\bm \alpha^T\bm Z_i)^2 ] = 1 $.
	Define $\delta^2 = \frac{\log n}{n} \cdot \frac{\E[ \max_{i \in [n]}  \norm{\bm Z_i}_2^2  ]}{  \sup_{\bm \alpha \in S^{m+s} } \E[ (\bm \alpha^T\bm Z_i)^2 ]  }$ as in Lemma 28 of \citet{belloni_conditional_2019}.
	By our assumption that $(m+s)^3 \xi_n^2 = o(n^{1 - \delta_1})$ and the fact that $\E[ \max_{i \in [n]}  \norm{\bm Z_i}_2^2] = O(\xi_n^2)$, we see that 
	$ \delta^2 = O( \frac{\log n}{n} \xi_n^2 ) = o(1) $.
	Therefore for sufficiently large $n$, $\delta < 1$ and Lemma 28 of \citet{belloni_conditional_2019} provides that
	$$
	\E\left[ \sup_{\bm \alpha \in S^{m+s}} \left| \mathbb{E}_n[(\bm Z_i^T \bm \alpha)^2] -  \E[(\bm Z_i^T \bm \alpha)^2] \right|  \right] \leq 2 \delta.
	$$
	Be extension, Markov's inequality implies
	$$
	\Pr \left( \sup_{\bm \alpha \in S^{m+s}} \left| \mathbb{E}_n[(\bm Z_i^T \bm \alpha)^2] -  \E[(\bm Z_i^T \bm \alpha)^2] \right| > 1 \right) \leq 2 \delta = o(1).
	$$
	Therefore 
	\begin{align*}
		\sup_{\bm \alpha \in S^{m+s}} \mathbb{E}_n[(\bm Z_i^T \bm \alpha)^2] 
		& \leq \sup_{\bm \alpha \in S^{m+s}} \left| \mathbb{E}_n[(\bm Z_i^T \bm \alpha)^2] -  \E[(\bm Z_i^T \bm \alpha)^2] \right| + \sup_{\bm \alpha \in S^{m+s}}  \E[(\bm Z_i^T \bm \alpha)^2] \\
		& = O_p(1) + 1 = O_p(1).
	\end{align*}
	
	Lastly, we note that $r^7\xi_n^6 = o(n^{1 - \epsilon})$ holds by our assumptions. Therefore the conditions in \eqref{eqn:Lemma 14 Z conditions} hold with probability converging to $1$ as $n \to \infty$ and applying Lemma 13 of \citet{belloni_conditional_2019} yields \eqref{eqn:Lemma 14 GP approx} conditionally on $(\xvec_i)_{i \in [n]}$, where each $\bm G_n(\cdot)$ is a zero-mean Gaussian process with a.s. continuous sample paths and covariance functions given by \eqref{eqn:Lemma 14 GP variance}. 
	
	Next we will bound $\sup_{\tau \in \mc T}\norm{[\bm K(\tau)]^{-1/2} \widetilde \Cmat \bm r_{1n}(\tau)}_2$.
	We see that
	$$	
	\sup_{\tau \in \mc T}\norm{[\bm K(\tau)]^{-1/2} \widetilde \Cmat \bm r_{1n}(\tau)}_2 \leq \sup_{\tau \in \mc T}\norm{[\bm K(\tau)]^{-1/2} \widetilde \Cmat }_2\sup_{\tau \in \mc T}\norm{\bm r_{1n}(\tau)}_2.
	$$
	Theorem \ref{thm:oracle bahadur} provides that $\sup_{\tau \in \mc T}\norm{\bm r_{1n}(\tau)}_2 = o_p(n^{-\epsilon_1'})$ for some $\epsilon_1' > 0$.
	Combining this with \eqref{eqn:lemma:KC bound} in Lemma \ref{lemma:bounding lemma}, we see that $\lVert [\bm K(\tau)]^{-1/2}  \widetilde \Cmat \bm r_{1n}(\tau) \rVert_2 = o_p(n^{-\epsilon_1'})$. 
	Applying this result and \eqref{eqn:Lemma 14 GP approx} to our expression for $\tilde{\bm \omega}_n(\tau)$ in \eqref{eqn:Wald process in terms of Ui} gives us \eqref{eqn:Wald known J GP approx}.\\
	
	We now shift to proving \eqref{eqn:rs known J GP approx}.
	To simplify expressions, we define
	$ \bm A (\tau) =  [\bm K(\tau)]^{-1/2} 
	\widetilde{\Cmat} [\bm J(\tau)]^{-1} $.
	We have from Theorem \ref{thm:oracle bahadur} that
	\begin{align*}
		\tilde{\bm s}(\tau)
		& = \bm A(\tau)[\bm J(\tau)]^{1/2} \bm P(\tau) [\bm J(\tau)]^{-1/2}\frac{1}{\sqrt n} \sum_{i \in [n]} \tilde \xvec_i \psi_{\tau}(y_i - \xvec_i^T \bvec^*(\tau)) \notag \\ 
		& \text{\; \;} + \bm A(\tau) \widetilde{\Cmat}^T \bm H^{-1}(\tau) \bm \eta(\tau) +  \bm A(\tau) \bm r_{3n}(\tau)\\
		& \defeq I_1(\tau) + I_2(\tau) + I_3(\tau).
	\end{align*}
	We will parse each of these terms in turn.
	
	We see that
	\begin{align*}
		\bm A(\tau)[\bm J(\tau)]^{1/2} \bm P(\tau) [\bm J(\tau)]^{-1/2}
		& =  [\bm K(\tau)]^{-1/2} 
		\widetilde{\Cmat} [\bm J(\tau)]^{-1/2} \bm P(\tau) [\bm J(\tau)]^{-1/2}\\
		& = [\bm K(\tau)]^{-1/2} 
		\widetilde{\Cmat} [\bm J(\tau)]^{-1}
		\widetilde{\Cmat}^T
		[\bm H(\tau)]^{-1}
		\widetilde{\Cmat} [\bm J(\tau)]^{-1}\\
		& = [\bm K(\tau)]^{-1/2} 
		\bm H(\tau)
		[\bm H(\tau)]^{-1}
		\widetilde{\Cmat} [\bm J(\tau)]^{-1}\\
		& = [\bm K(\tau)]^{-1/2} 
		\widetilde{\Cmat}[\bm J(\tau)]^{-1}.
	\end{align*}
	Therefore
	\begin{align*}
		I_1(\tau) = [\bm K(\tau)]^{-1/2} 
		\widetilde{\Cmat} [\bm J(\tau)]^{-1} \frac{1}{\sqrt n} \sum_{i \in [n]} \tilde \xvec_i \psi_{\tau}(y_i - \xvec_i^T \bvec^*(\tau))
		= \frac{1}{\sqrt n} \sum_{i \in [n]} \bm Z_i \psi_{\tau}(U_i - \tau)
	\end{align*}
	where $U_i$ is defined as in Lemmas \ref{lemma:Ui uniform} and \ref{lemma:Ui equivalence}.
	As we showed above, the conditions for \eqref{eqn:Lemma 14 GP approx} to hold are satisfied and we have
	$$
	\sup_{\tau \in \mc T} \norm {I_1(\tau) - \widetilde{\bm G}_n(\tau) }_2 = o_p(n^{-\epsilon'_1})
	$$
	for some $\epsilon'_1 > 0$.
	
	Moving on to $I_2(\tau)$, we see that
	\begin{align*}
		I_2(\tau) 
		& = [\bm K(\tau)]^{-1/2} 
		\widetilde{\Cmat} [\bm J(\tau)]^{-1} \widetilde{\Cmat}^T \bm H^{-1}(\tau) \bm \eta(\tau) 
		= [\bm K(\tau)]^{-1/2} \bm H(\tau) \bm H^{-1}(\tau) \bm \eta(\tau)\\
		& = [\bm K(\tau)]^{-1/2}  \bm \eta(\tau)
		= \bm \nu (\tau).
	\end{align*}
	Lastly, we note that
	$$
	\norm{ I_3(\tau) }_2 \leq \norm{ [\bm K(\tau)]^{-1/2} \widetilde{\Cmat} }_2 \norm{ [\bm J(\tau)]^{-1} }_2 \norm{\bm r_{3n} (\tau)}_2.
	$$
	Theorem \ref{thm:oracle bahadur}  provides that $\sup_{\tau \in \mc T} \norm{\bm r_{3n} (\tau)}_2 = o_p(n^{-\epsilon'_2})$ for some $\epsilon'_2 > 0$ while \eqref{eqn:lemma:Jinv bound}  and \eqref{eqn:lemma:KC bound} from Lemma \ref{lemma:bounding lemma} bound the other terms in the above expression uniformly over $\tau \in \mc T$, so that
	$\sup_{\tau \in \mc T}  \norm{I_3(\tau) }_2 = o_p(n^{-\epsilon'_2}) $.
	Combining these results for $I_1(\tau)$, $I_2(\tau)$, and $I_3(\tau)$ gives us \eqref{eqn:rs known J GP approx}.
\end{proof} 

\begin{proof}[Proof of Lemma \ref{lemma:Bessel for Wald with known J}]
	We focus here on proving \eqref{eqn:Wald J to Bessel} as an equivalent argument can be used to prove \eqref{eqn:rs J to Bessel}.

	Let $\bm \nu (\tau) = [\bm K(\tau)]^{-1/2} \bm \eta(\tau)$.		
	Lemma \ref{lemma:GP Zi to EZi} provides that
	\begin{align*}
		\sup_{\tau \in \mc T} \norm{ \frac{(\mathbb{E}_n[\bm Z_i \bm Z_i^T] )^{1/2} \bm B_r(\tau) + \bm \nu(\tau)}{\sqrt{\tau (1 - \tau)} } - Q_{r,\bm \mu(\tau)}(\tau)  }_2
		& \lesssim
		\sup_{\tau \in \mc T} \norm{ (\mathbb{E}_n[\bm Z_i \bm Z_i^T] )^{1/2} \bm B_r(\tau) - \bm B_r(\tau) }_2\\
		& = o_p(1)
	\end{align*}
	under our assumptions on $r$ and $\xi_n$, with the first inequality following from the fact that $\mc T \subseteq [\delta, 1 -\delta]$.
	The reverse triangle inequality then yields
	\begin{equation*}
		\left|\sup_{\tau \in \mc T} \frac{\norm{ (\mathbb{E}_n[\bm Z_i \bm Z_i^T] )^{1/2} \bm B_r(\tau) + \bm \nu(\tau)}_2}{\sqrt{\tau (1 - \tau)}} - \sup_{\tau \in \mc T} Q_{r,\bm \mu(\tau)}(\tau) \right| = o_p(1), \label{eqn:sup diff Zi to EZi}
	\end{equation*}
	which further implies
	\begin{equation}
		\sup_{\tau \in \mc T} \frac{\norm{ (\mathbb{E}_n[\bm Z_i \bm Z_i^T] )^{1/2} \bm B_r(\tau) + \bm \nu(\tau)}_2}{\sqrt{\tau (1 - \tau)}} \overset{d}{\to} \sup_{\tau \in \mc T} Q_{r,\bm \mu(\tau)}(\tau).\label{eqn:EZi conv in dist}
	\end{equation}
	
	Now consider $\bm G_n(\cdot)$ from Lemma \ref{lemma:GP for Wald with known} and $(\mathbb{E}_n[\bm Z_i \bm Z_i^T ])^{1/2}\bm B_r(\cdot)$.
	Since they are both zero-mean Gaussian processes with the same covariance function and their sample paths are a.s. continuous, we have
	$ \bm G_n(\cdot) \overset{d}{=} (\mathbb{E}_n[\bm Z_i \bm Z_i^T ])^{1/2}\bm B_r(\cdot) $, conditionally on $(\xvec_i)_{i \in [n]}$. 
	where for two processes $A(\cdot)$ and $B(\cdot)$, $A(\cdot) \overset{d}{=} B(\cdot)$ means that for any finite set $0 < \tau_1 < \ldots < \tau_n < 1$, 
	$(A(\tau_1) \cdots A(\tau_n))^T \overset{d}{=} (B(\tau_1) \cdots B(\tau_n))^T $.
	This further implies $ \bm G_n(\cdot) + \bm \nu(\cdot) \overset{d}{=} (\mathbb{E}_n[\bm Z_i \bm Z_i^T ])^{1/2}\bm B_r(\cdot) + \bm \nu(\cdot) $, conditionally on $(\xvec_i)_{i \in [n]}$.
	Therefore for $z \in \reals$
	$$
	\Pr\left(\sup_{\tau \in \mc T} \frac{\norm{\bm G_n(\tau) + \bm \nu (\tau)}_2}{\sqrt{\tau(1- \tau)}} \leq z \middle| (\xvec_i)_{i \in [n]}\hspace{-0.5mm} \right) = \Pr\left(\sup_{\tau \in \mc T} \frac{\norm{(\mathbb{E}_n[\bm Z_i \bm Z_i^T ])^{1/2}\bm B_r(\tau) + \bm \nu(\tau)}_2}{\sqrt{\tau(1- \tau)}} \leq z \middle| (\xvec_i)_{i \in [n]} \hspace{-0.5mm} \right)
	$$
	which by the Tower law implies
	$$
	\Pr\left(\sup_{\tau \in \mc T} \frac{\norm{\bm G_n(\tau) + \bm \nu (\tau)}_2}{\sqrt{\tau(1- \tau)}} \leq z \right) = \Pr\left(\sup_{\tau \in \mc T} \frac{\norm{(\mathbb{E}_n[\bm Z_i \bm Z_i^T ])^{1/2}\bm B_r(\tau) + \bm \nu(\tau)}_2}{\sqrt{\tau(1- \tau)}} \leq z \right).
	$$
	Combining this with \eqref{eqn:EZi conv in dist} yields
	\begin{equation}
		\sup_{\tau \in \mc T} \frac{\norm{\bm G_n(\tau) + \bm \nu (\tau)}_2}{\sqrt{\tau(1- \tau)}} \overset{d}{\to} \sup_{\tau \in \mc T} Q_{r, \bm \mu(\tau)} (\tau) . \label{eqn:G conv in dist}
	\end{equation}
	
	Lemma \ref{lemma:GP for Wald with known} implies that
	\begin{align*}
		\sup_{\tau \in \mc T} \norm{ \frac{\tilde{\bm \omega}(\tau) }{\sqrt{\tau(1 - \tau)}} -  \frac{\bm G_n(\tau) + \bm \nu(\tau)}{\sqrt{\tau(1 - \tau)}}  }_2
		& = \sup_{\tau \in \mc T} \frac{1}{\sqrt{\tau(1 - \tau)}} \norm{ \tilde{\bm \omega}(\tau) -  \bm G_n(\tau) - \bm \nu(\tau) }_2\\
		& \lesssim \sup_{\tau \in \mc T} \norm{ \tilde{\bm \omega}(\tau) -  \bm G_n(\tau) - \bm \nu(\tau) }_2\\
		& = o_p(n^{-\epsilon'})
	\end{align*}
	conditionally on $(\xvec_i)_{i \in [n]}$, where the first inequality holds since $\mc T \subseteq [\delta, 1 - \delta]$ and the final line follows from the lemma.
	Applying the reverse triangle inequality, we then have
	\begin{align*}
		\left|\sup_{\tau \in \mc T} \norm{ \frac{\tilde{\bm \omega}(\tau) }{\sqrt{\tau(1 - \tau)}} }_2 - \sup_{\tau \in \mc T} \norm{ \frac{\bm G_n(\tau) + \bm \nu(\tau)}{\sqrt{\tau(1 - \tau)}}  }_2\right|
		& = o_p(n^{-\epsilon'}) = o(1)
	\end{align*}
	conditionally on $(\xvec_i)_{i \in [n]}$---that is, for every $\eps > 0$
	$$
		\Pr\left(  \left|\sup_{\tau \in \mc T} \norm{ \frac{\tilde{\bm \omega}(\tau) }{\sqrt{\tau(1 - \tau)}} }_2 - \sup_{\tau \in \mc T} \norm{ \frac{\bm G_n(\tau) + \bm \nu(\tau)}{\sqrt{\tau(1 - \tau)}}  }_2\right| > \eps \middle| (\xvec_i)_{i \in [n]} \right) \to 0.
	$$
	By the Tower Law and the Bounded Convergence Theorem, this implies that the unconditional probability satisfies
	$$
	\Pr\left(  \left|\sup_{\tau \in \mc T} \norm{ \frac{\tilde{\bm \omega}(\tau) }{\sqrt{\tau(1 - \tau)}} }_2 - \sup_{\tau \in \mc T} \norm{ \frac{\bm G_n(\tau) + \bm \nu(\tau)}{\sqrt{\tau(1 - \tau)}}  }_2\right| > \eps \right) \to 0.
	$$
	Combining this result with \eqref{eqn:G conv in dist} and applying Slutsky's Theorem, we find
	$$  \sup_{\tau \in \mc T} \norm{ \frac{\tilde{\bm \omega}(\tau) }{\sqrt{\tau(1 - \tau)}} }_2 \overset{d}{\to} \sup_{\tau \in \mc T} Q_{r, \bm \mu(\tau)}(\tau) .$$
	From here, the Continuous Mapping Theorem yields \eqref{eqn:Wald J to Bessel}.
\end{proof}

\begin{proof}[Proof of Lemma \ref{lemma:Wald J to Jhat}]
	We begin with the proof of \eqref{eqn:Wald to Wald J}.
	As we established in the proof of Lemma \ref{lemma:GP for Wald with known},
	$ \widetilde \Cmat \bvecOR_{a,\MSset}(\tau) - \bm \Psi(\tau) = \widetilde \Cmat ( \bvecOR_{a,\MSset}(\tau) - \bvec_{\MSset}^*(\tau) ) + \bm \eta(\tau)/\sqrt{n}  $.
	From here, we see that
	\begin{align*}
		\sup_{\tau \in \mc T}| W_n^o&(\tau)  - \widetilde{W}_n(\tau) | \\
		& = \sup_{\tau \in \mc T}\frac{1}{\tau (1 - \tau)}|n(\widetilde \Cmat \bvecOR_{a}(\tau) - \bm \Psi(\tau) )^T\left( [\hatmat{K}(\tau)]^{-1} - [\bm K(\tau)]^{-1}  \right)(\widetilde \Cmat \bvecOR_{a}(\tau) - \bm \Psi(\tau) ) |\\
		& \leq \frac{1}{\delta (1 - \delta)}\sup_{\tau \in \mc T}|n(\widetilde \Cmat \bvecOR_{a}(\tau) - \bm \Psi(\tau) )^T\left( [\hatmat{K}(\tau)]^{-1} - [\bm K(\tau)]^{-1}  \right)(\widetilde \Cmat \bvecOR_{a}(\tau) - \bm \Psi(\tau) ) |\\
		& \lesssim \sup_{\tau \in \mc T}|n (\widetilde \Cmat \bvecOR_{a}(\tau) - \bm \Psi(\tau) )^T [\bm K(\tau)]^{-1/2}\\
		& \hspace{0.4cm} \times \left( [\bm K(\tau)]^{1/2}[\hatmat{K}(\tau)]^{-1} [\bm K(\tau)]^{1/2} - \bm I_r  \right) [\bm K(\tau)]^{-1/2}(\widetilde \Cmat \bvecOR_{a}(\tau) - \bm \Psi(\tau) ) |\\
		& \leq \sup_{\tau \in \mc T}  \norm{ [\bm K(\tau)]^{1/2} [\hatmat K(\tau)]^{-1}  [\bm K(\tau)]^{1/2} - \bm I_r}_2\\
		& \text{\; \; \;} \times \sup_{\tau \in \mc T}\norm{\sqrt{n} [\bm K(\tau)]^{-1/2}\widetilde \Cmat ( \bvecOR_{a,\MSset}(\tau) - \bvec_{\MSset}^*(\tau) ) +  [\bm K(\tau)]^{-1/2} \bm \eta(\tau)  }_2^2,
	\end{align*}
	where the first inequality follows since $\mc T \subseteq [\delta, 1 - \delta] $.
	
	Starting with the second term, we see that
	\begin{align*}
		\sup_{\tau \in \mc T}\lVert \sqrt{n} [\bm K(\tau)]^{-1/2}& \widetilde \Cmat ( \bvecOR_{a,\MSset}(\tau) - \bvec_{\MSset}^*(\tau) ) \rVert_2\\
		& \leq \sup_{\tau \in \mc T} \norm{[\bm K(\tau)]^{-1/2}\widetilde \Cmat}_2 \sup_{\tau \in \mc T} \norm{\sqrt{n} [\bm ( \bvecOR_{a,\MSset}(\tau) - \bvec_{\MSset}^*(\tau) ) }_2\\
		& = O_p(\sqrt{m+s})
	\end{align*}
	by \eqref{eqn:lemma:KC bound} and Theorem \ref{thm:oracle rate}.
	Likewise, we find 
	\begin{align*}
		\sup_{\tau \in \mc T}\norm{ [\bm K(\tau)]^{-1/2} \bm \eta(\tau)  }_2
		& \leq  \sup_{\tau \in \mc T}\norm{ [\bm K(\tau)]^{-1/2} }_2\sup_{\tau \in \mc T}\norm{ \bm \eta(\tau)  }_2\\
		& \leq  \sup_{\tau \in \mc T}( \norm{ [\bm K(\tau)]^{-1} }_2)^{1/2} \sup_{\tau \in \mc T}\norm{ \bm \eta(\tau)  }_2\\
		& 
		= O\left( \sqrt{m+s } \right)
	\end{align*}
	by \eqref{eqn:lemma:Kinv bound} and \ref{as:local alts}.
	
	Under our assumptions that $(m+s)^2 \xi_n^2 \log n = o(n^{1 - \epsilon} h)$ and $h = o((m+s)^{-1/2}n^{-\epsilon/2})$, \eqref{eqn:lemma K - I bound} in Lemma \ref{lemma:bounding lemma} implies
	\begin{align*}
		\sup_{\tau \in \mc T} \norm{ [\bm K(\tau)]^{1/2} [\hatmat K(\tau)]^{-1}  [\bm K(\tau)]^{1/2} - \bm I_r}_2 
		& \lesssim_p \sqrt{\frac{(m+s)\xi_n^2 \log n}{n h}} + h  + \sqrt{\frac{\xi_n^2 \log n}{n}}\\
		& = o((m+s)^{-1/2}n^{-\epsilon/2}).
	\end{align*}
	
	Combining these bounds, we have
	$$
	\sup_{\tau \in \mc T}| W_n^o(\tau)  - \widetilde{W}_n(\tau) | 
	= o_p((m+s)^{1/2}n^{-\epsilon/2}) = o(1),
	$$
	where the final equality follows from our assumption that $m+s = o(n^\epsilon)$, giving us \eqref{eqn:Wald to Wald J}.
	
	We now shift to proving \eqref{eqn:rs to rs J}.
	For ease of notation, we define $\hatvec b^o(\tau) = (\hat a^0_1(\tau), \ldots, \hat a^0_n(\tau))^T - \mathbbm{1}_n (1 - \tau) $ so that 
	$
	\sum_{i \in [n]} \tilde \xvec_i (\hat a^o_i(\tau) - (1 - \tau)) = \widetilde{\xmat}^T \hatvec b^o(\tau)
	$. We drop $\tau$ from our notation for the $\bm J$ and $\bm K$ matrices and their estimators simplify expressions.
	
	Using this notation, we have that
	\begin{align}
		\sup_{\tau \in \mc T}| R_n^o&(\tau) - \widetilde{R}_n(\tau) | \notag \\
		& =  \sup_{\tau \in \mc T} \frac{1} {\tau(1 - \tau)} \left| \frac{1}{n} \hatvec b^o(\tau)^T\widetilde \xmat \left( \hatmat J^{-1} \widetilde{\Cmat}^T \hatmat K^{-1} \widetilde{\Cmat} \hatmat J^{-1} - \bm J^{-1} \widetilde{\Cmat}^T \bm K^{-1} \widetilde{\Cmat} \bm J^{-1} \right) \widetilde{\xmat}^T \hatvec b^o(\tau)  \right | \notag \\
		& \lesssim  \sup_{\tau \in \mc T} \norm{ \hatmat J^{-1} \widetilde{\Cmat}^T \hatmat K^{-1} \widetilde{\Cmat} \hatmat J^{-1} - \bm J^{-1} \widetilde{\Cmat}^T \bm K^{-1} \widetilde{\Cmat} \bm J^{-1} }_2 \norm{\frac{1}{\sqrt{n}}\widetilde{\xmat}^T \hatvec b^o(\tau) }_2^2, \label{eqn:rs J vs Jhat diff}
	\end{align}
	again using the fact that $\mc T \subseteq [\delta, 1 - \delta]$.
	Focusing on the first term in \eqref{eqn:rs J vs Jhat diff}, we see that
	\begin{multline*}
		\norm{ \hatmat J^{-1} \widetilde{\Cmat}^T \hatmat K^{-1} \widetilde{\Cmat} \hatmat J^{-1} - \bm J^{-1} \widetilde{\Cmat}^T \bm K^{-1} \widetilde{\Cmat} \bm J^{-1} }_2 \\
		\leq \underbrace{\norm{ \hatmat J^{-1} \widetilde{\Cmat}^T \hatmat K^{-1} \widetilde{\Cmat} \hatmat J^{-1} - \hatmat J^{-1} \widetilde{\Cmat}^T \bm K^{-1} \widetilde{\Cmat} \hatmat J^{-1} }_2}_{\defeq I_1(\tau)}
		+ \underbrace{\norm{ \hatmat J^{-1} \widetilde{\Cmat}^T \bm K^{-1} \widetilde{\Cmat} \hatmat J^{-1} - \bm J^{-1} \widetilde{\Cmat}^T \bm K^{-1} \widetilde{\Cmat} \bm J^{-1} }_2}_{\defeq I_2(\tau)} .
	\end{multline*}
	Applying \eqref{eqn:lemma:hatJinv bound}, \eqref{eqn:lemma:KC bound}, and \eqref{eqn:lemma K - I bound} from Lemma \ref{lemma:bounding lemma}, we then find 
	\begin{align*}
		I_1(\tau)
		& \leq \norm{ \widetilde \Cmat^T \hatmat K^{-1} \widetilde{\Cmat} - \widetilde \Cmat^T \bm K^{-1} \widetilde{\Cmat} }_2 \norm{\hatmat J^{-1}}_2^2\\
		& \lesssim_p \norm{ \widetilde \Cmat^T \bm K^{-1/2} \left( \bm K^{1/2} \hatmat K^{-1} \bm K^{1/2} - \bm I  \right) \bm K^{-1/2}\widetilde{\Cmat} }_2\\
		& \leq \norm{  \bm K^{1/2} \hatmat K^{-1} \bm K^{1/2} - \bm I }_2 \norm{\bm K^{-1/2}\widetilde{\Cmat}}_2^2\\
		& \lesssim_p \sqrt{\frac{(m+s)\xi_n^2 \log n}{n h}} + h  + \sqrt{\frac{\xi_n^2 \log n}{n}}\\
		& \ll (m+s)^{-1/2} n^{-\epsilon/2}
	\end{align*}
	uniformly over $\tau \in \mc T$ where the final inequality follows from our assumptions that $(m+s)^2 \xi_n^2 \log n = o(n^{1 - \epsilon} h)$ and $h = o((m+s)^{-1/2}n^{-\epsilon/2})$.
	Likewise, we apply \eqref{eqn:lemma:Jinv bound}, \eqref{eqn:lemma:J bound}, \eqref{eqn:lemma:KC bound}, and \eqref{eqn:lemma:hatJinv - Jinv bound} from Lemma \ref{lemma:bounding lemma} to find
	\begin{align*}
		I_2(\tau)
		& \leq \norm{ \hatmat J^{-1} \widetilde{\Cmat}^T \bm K^{-1} \widetilde{\Cmat} \hatmat J^{-1} - \bm J^{-1} \widetilde{\Cmat}^T \bm K^{-1} \widetilde{\Cmat} \hatmat J^{-1} }_2 
		+ \norm{ \bm J^{-1} \widetilde{\Cmat}^T \bm K^{-1} \widetilde{\Cmat} \hatmat J^{-1} - \bm J^{-1} \widetilde{\Cmat}^T \bm K^{-1} \widetilde{\Cmat} \bm J^{-1} }_2\\
		& \leq \norm{ \hatmat J^{-1} - \bm J^{-1} }_2 \norm{\widetilde{\Cmat}^T \bm K^{-1}\widetilde{\Cmat}}_2 \norm{ \hatmat J^{-1} }_2 +  \norm{ \bm J^{-1} }_2\norm{\widetilde{\Cmat}^T \bm K^{-1}\widetilde{\Cmat}}_2\norm{ \hatmat J^{-1} - \bm J^{-1} }_2 \\
		& \lesssim_p \norm{\widetilde{\Cmat}^T \bm K^{-1}\widetilde{\Cmat}}_2\norm{ \hatmat J^{-1} - \bm J^{-1}}_2
		= \norm{ \bm K^{-1/2}\widetilde{\Cmat}}_2^2 \norm{ \hatmat J^{-1} - \bm J^{-1}}_2\\
		& \lesssim_p \sqrt{\frac{(m+s)\xi_n^2 \log n}{n h}} + h 
		\ll (m+s)^{-1/2} n^{-\epsilon/2}
	\end{align*}
	uniformly for all $\tau \in \mc T$.
	
	Moving on to the second term in \eqref{eqn:rs J vs Jhat diff}, we have from \eqref{eqn:rs1}-\eqref{eqn:reduced oracle score 2} that
	\begin{align*}
		\frac{1}{\sqrt n} \widetilde \xmat^T \hatvec b^o(\tau)
		= \frac{1}{\sqrt n} \sum_{i \in [n]} \tilde \xvec_i \psi_{\tau}(y_i - \xvec_i^T \bvec^*(\tau))
		- \sqrt{n} \bm J(\tau) (\bvecOR_{0,\MSset}(\tau) - \bvec^*_{\MSset}(\tau)) + \tilde{\bm r}(\tau)
	\end{align*}
	where $\sup_{\tau \in \mc T} \norm{\tilde{\bm r}(\tau)}_2 \lesssim_p \frac{(m+s)^{3/4} \xi_n^{1/2} \log^{1/2}n }{n^{1/4}} \lesssim \sqrt{m+s}$.
	
	Lemmas  \ref{lemma:Ui uniform}  and \ref{lemma:Ui equivalence} imply that $\frac{1}{\sqrt n} \sum_{i \in [n]} \tilde \xvec_i \psi_{\tau}(y_i - \xvec_i^T \bvec^*(\tau)) = \frac{1}{\sqrt n} \sum_{i \in [n]} \tilde \xvec_i \psi_{\tau}(U_i - \tau)$,
	where $U_1, \ldots, U_n$ are i.i.d. Uniform$(0,1)$ conditional on $(\xvec_i)_{i \in [n]}$.
	Lemma 31 of \citet{belloni_conditional_2019} provides
	$$\sup_{\tau \in \mc T} \norm{\frac{1}{\sqrt n} \sum_{i \in [n]} \tilde \xvec_i \psi_{\tau}(U_i - \tau)}_2 \lesssim_p \sqrt{m+s}.$$
	In addition, we find that
	\begin{align*}
		\sup_{\tau \in \mc T} \norm{\sqrt{n} \bm J(\tau) (\bvecOR_{0,\MSset}(\tau) - \bvec^*_{\MSset}(\tau))}_2
		& \leq \sup_{\tau \in \mc T} \norm{ \bm J(\tau) } \norm{\sqrt{n} (\bvecOR_{0,\MSset}(\tau) - \bvec^*_{\MSset}(\tau))}_2 \lesssim_p \sqrt{m+s}
	\end{align*}
	by \eqref{eqn:lemma:J bound} and Theorem \ref{thm:oracle rate}.
	
	Applying these bounds to \eqref{eqn:rs J vs Jhat diff}, we have
	$$
	\sup_{\tau \in \mc T}  | R_n^o(\tau) - \widetilde{R}_n(\tau)  |
	= o_p ((m+s)^{1/2}n^{-\epsilon/2}) = o(1)
	$$
	where the final equality follows from our assumption that $m+s \ll n^{\epsilon}$, finishing the proof		
\end{proof}

\subsubsection{Supporting lemmas}

\begin{lemma} \label{lemma:bounding lemma}
	Under \ref{as:iid array}-\ref{as:local alts}, the following hold:
	\begin{eqnarray}
		\exists C_1 \ni \lambda_{\max} \{ [\bm J(\tau) ]^{-1} \} & \leq & C_1  \text{\; \;}\forall_{n,\tau \in \overline{\mc T}}  \label{eqn:lemma:Jinv bound}\\
		\exists C_2 \ni \lambda_{\max} \{ \bm J(\tau) \} & \leq & C_2  \text{\; \;}\forall_{n,\tau \in \overline{\mc T}} \label{eqn:lemma:J bound}\\
		\sup_{\tau \in \mc T} \norm{ [\bm K(\tau)]^{-1/2} \widetilde{\Cmat} }_2 & = & O(1) \label{eqn:lemma:KC bound}\\
		\sup_{\tau \in \mc T} \norm{ [\bm K(\tau)]^{-1} }_2 & = & O(1). \label{eqn:lemma:Kinv bound}
	\end{eqnarray}
	If in addition $h = o(1)$ and $ (m + s) \xi_n^2 \log^2 n = o(nh) $, then	
	\begin{eqnarray}	
		\sup_{\tau \in \mc T} \norm{ [\hatmat J(\tau)]^{-1} }_2 & = & O_p(1)\label{eqn:lemma:hatJinv bound}\\
		\sup_{\tau \in \mc T} \norm{ [\hatmat J(\tau)]^{-1} - [\bm J(\tau)]^{-1}  }_2 & = & O_p\left( \sqrt{\frac{(m+s)\xi_n^2 \log n}{n h}} + h \right) \text{\; \;}\label{eqn:lemma:hatJinv - Jinv bound} \\
		\sup_{\tau \in \mc T} \norm{ \bm K^{1/2} \hatmat K^{-1}  \bm K^{1/2} - \bm I_r}_2 & = & O_p\left( \sqrt{\frac{(m+s)\xi_n^2 \log n}{n h}} + h  + \sqrt{\frac{\xi_n^2 \log n}{n}} \right) \text{ \; \;} \label{eqn:lemma K - I bound} 
	\end{eqnarray}
	where $\bm K = \bm K(\tau)$ and $\hatmat K = \hatmat K(\tau)$.
\end{lemma}

\begin{proof}[Proof of Lemma \ref{lemma:bounding lemma}]
	Throughout this proof, we will drop the $\tau$ argument from $\bm J(\tau), \bm K(\tau)$, etc. where the dependence on $\tau$ is clear from context.
	
	We see that \eqref{eqn:lemma:Jinv bound} and \eqref{eqn:lemma:J bound} both follow immediately from \ref{as:density conditions} and \ref{as:Gram lower bound}.
	Next we will prove \eqref{eqn:lemma:KC bound}.
	Conditions \ref{as:density conditions} and \ref{as:Gram lower bound} imply that both $\bm J$ and $\bm \Sigma$ have their eigenvalues bounded from above and away from zero uniformly over $n$. Therefore both of these matrices and, by extension, $\bm J^{-1}\bm \Sigma \bm J^{-1}$ are positive definite.
	Using this fact, we find that for all $\tau$
	\begin{align}
		\big \lVert &\bm K^{-1/2}  \widetilde \Cmat \big \rVert_2 \notag \\
		& = \norm{\bm K^{-1/2} \widetilde \Cmat (\bm J^{-1}\bm \Sigma \bm J^{-1})^{1/2} (\bm J^{-1}\bm \Sigma \bm J^{-1})^{-1/2}}_2 \notag \\
		& \leq \norm{\bm K^{-1/2} \widetilde \Cmat (\bm J^{-1}\bm \Sigma \bm J^{-1})^{1/2}}_2 \norm{(\bm J^{-1}\bm \Sigma \bm J^{-1})^{-1/2}}_2. \label{eqn:KC prod bound}
	\end{align}
	The first term in this bound simplifies to
	\begin{align*}
		\lVert \bm K^{-1/2} \widetilde \Cmat (\bm J^{-1}\bm \Sigma & \bm J^{-1})^{1/2} \rVert_2\\
		& = \lambda_{\max}^{1/2} \{\bm K^{-1/2} \widetilde \Cmat \bm J^{-1}\bm \Sigma \bm J^{-1} \widetilde \Cmat^T  \bm K^{-1/2} \}\\
		& = \lambda_{\max}^{1/2} \{\bm K^{-1/2} \bm K  \bm K^{-1/2} \}\\
		& = \lambda_{\max}^{1/2}\{ \bm I_r \} = 1.
	\end{align*}
	Turning to the second term, we find
	\begin{align*}
		\norm{(\bm J^{-1}\bm \Sigma \bm J^{-1})^{-1/2}}_2
		& = \lambda_{\max}^{1/2}\{ (\bm J^{-1}\bm \Sigma \bm J^{-1})^{-1}  \} \\
		& = \lambda_{\max}^{1/2}\{ \bm J\bm \Sigma^{-1} \bm J \} \\
		& \leq \lambda_{\max} \{\bm J\} \lambda_{\max}^{1/2} \{ \bm \Sigma^{-1} \}\\
		& = O(1)
	\end{align*}
	for all $\tau \in \overline{\mc T}$ by \eqref{eqn:lemma:J bound} and \ref{as:Gram lower bound}.
	Applying these to \eqref{eqn:KC prod bound} yields \eqref{eqn:lemma:KC bound}.

	Moving on to \eqref{eqn:lemma:Kinv bound}, \ref{as:Gram lower bound} implies that there exists a constant $C_1$ such that $\lambda_{\min}\{\bm \Sigma\} > C_1$ for all $n$.
	Moreover, \eqref{eqn:lemma:J bound} provides that there exists a constant $C_2$ such that
	$\lambda_{\min} \{ \bm J^{-2} \} = (\lambda_{\max}\{ \bm J  \} )^{-2} > C_2$ for all $n, \tau \in \overline{\mc T}$.
	Since 
	$$
	\lambda_{\min} \{ \bm J^{-2} \}  = \min_{\bm a \neq \bm 0} \frac{\bm a^T \bm J^{-2} \bm a}{\norm{\bm a}_2^2} = \min_{\bm a \neq \bm 0} \frac{ \norm{\bm J^{-1} \bm a}_2^2 }{\norm{\bm a}_2^2},
	$$
	this implies $\norm{\bm J^{-1} \bm a}_2^2 > C_2 \norm{\bm a}_2^2$ for all $\bm a, n, \bm \tau \in \overline{\mc T}$.
	From here, we find
	\begin{align*}
		\lambda_{\min}\{ \bm J^{-1} \bm \Sigma \bm J^{-1} \}
		& = \min_{\bm a \neq \bm 0} \frac{\bm a^T\bm J^{-1} \bm \Sigma \bm J^{-1} \bm a}{\norm{\bm a}_2^2}\\
		& \geq C_2\min_{\bm a \neq \bm 0} \frac{\bm a^T\bm J^{-1} \bm \Sigma \bm J^{-1} \bm a}{\norm{\bm J^{-1} \bm a}_2^2}\\
		& \geq C_2\min_{\bm a \neq \bm 0} \frac{\bm a^T \bm \Sigma \bm a}{\norm{\bm a}_2^2}\\
		& = C_2 \lambda_{\min}\{ \bm \Sigma \} > C_1 C_2
	\end{align*}
	for all $n, \tau \in \overline{\mc T}$.
	We know from \ref{as:local alts} that there exists a constant $C_3$ such that $\lambda_{\min}\{\widetilde{\Cmat} \widetilde{\Cmat}^T  \} = \lambda_{\min}\{ \Cmat \Cmat^T \} = (\lambda_{\max} \{ (\Cmat \Cmat^T)^{-1} \} )^{-1} > C_3$ for all $n$.
	By a similar argument to the one above, $\norm{\widetilde{\Cmat}^T \bm a}_2 \geq C_3 \norm{\bm a}_2^2 $ for all $\bm a$ and therefore
	\begin{align*}
		\lambda_{\min}\{ \bm K\}
		& = \min_{\bm a \neq \bm 0}\frac{ \bm a^T \widetilde{\Cmat}  \bm J^{-1} \bm \Sigma \bm J^{-1} \widetilde{\Cmat}^T \bm a }{\norm{\bm a}_2^2} \\
		& = C_3 \min_{\bm a \neq \bm 0}\frac{ \bm a^T  \bm J^{-1} \bm \Sigma \bm J^{-1} \bm a }{\norm{\bm a}_2^2}
		> C_1C_2C_3
	\end{align*}
	for all $n, \tau \in \overline{\mc T}$.
	Thus 
	$$
	\sup_{\tau \in \mc T}\norm{\bm K^{-1}}_2 = \sup_{\tau \in \mc T} \lambda_{\max}\{  \bm K^{-1} \}
	=  \sup_{\tau \in \mc T} (\lambda_{\min}\{ \bm K \})^{-1} < (C_1C_2C_3)^{-1}
	$$
	for all $n$, establishing \eqref{eqn:lemma:Kinv bound}.
	
	Moving on to \eqref{eqn:lemma:hatJinv bound}, we know that
	$ \sup_{\tau \in \mc T} \norm{\hatmat J^{-1} }  = \sup_{\tau \in \mc T} \lambda_{\max}\{\hatmat J^{-1} \}$.
	We find that
	\begin{align}
		\lambda_{\min} \{ \hatmat J \} 
		& = \min_{ \norm{\bm v}_2 = 1 } \bm v^T \hatmat{J} \bm v \notag\\
		& \geq \min_{ \norm{\bm v}_2 = 1 } \bm v^T \bm{J} \bm v - \sup_{ \norm{\bm v}_2 = 1} |\bm v^T(\bm J - \hatmat J)\bm v| \notag \\
		& = \lambda_{\min}\{ \bm J \} - \norm{\bm J - \hatmat J}_2 \label{eqn:J eig lower bound}
	\end{align}
	where the inequality follows from Lemma S.5 of \citet{shi2019linear}.
	We know from \eqref{eqn:lemma:Jinv bound} that there exists $C$ such that $\lambda_{\min}\{ \bm J \} > C$ for all $n, \tau \in \overline{\mc T}$.
	Moreover, Lemma 29 of \citet{belloni_conditional_2019} provides that
	\begin{equation}
		\sup_{\tau \in \mc T}  \norm{\hatmat J - \bm J}_2 =  O_p\left( \sqrt{\frac{(m+s)\xi_n^2 \log n}{n h}} + h \right) = o(1). \label{eqn:Bet lemma 30 J}
	\end{equation}
	when $h = o(1)$ and $ (m + s) \xi_n^2 \log^2 n = o(nh)$.
	Therefore if $n$ is sufficiently large, then $\norm{\hatmat J - \bm J}_2 < C/2$ and $\lambda_{\min} \{ \hatmat J \} > C/2 $ for all $\tau \in \mc T$ with probability converging to $1$.
	By extension, we have
	$ \sup_{\tau \in \mc T} \lambda_{\max}\{\hatmat J^{-1} \} 
	= \sup_{\tau \in \mc T}( \lambda_{\min}\{\hatmat J \} )^{-1} < 2/C$
	with probability converging to $1$, giving us \eqref{eqn:lemma:hatJinv bound}.
	
	To establish \eqref{eqn:lemma:hatJinv - Jinv bound}, we note that
	\begin{align*}
		\sup_{\tau \in \mc T} \norm{ \hatmat J^{-1} - \bm J^{-1} }_2
		& = \sup_{\tau \in \mc T} \norm{ \hatmat J^{-1}(\bm J - \hatmat J ) \bm J^{-1} }_2\\
		& \leq \sup_{\tau \in \mc T} \norm{ \hatmat J^{-1} }_2 \sup_{\tau \in \mc T} \norm{ \bm J - \hatmat J }_2 \sup_{\tau \in \mc T} \norm{  \bm J^{-1} }_2\\
		& = O_p\left( \sqrt{\frac{(m+s)\xi_n^2 \log n}{n h}} + h \right) 
	\end{align*}
	by \eqref{eqn:lemma:Jinv bound}, \eqref{eqn:lemma:hatJinv bound}, and \eqref{eqn:Bet lemma 30 J}.
	
	Finally we will prove \eqref{eqn:lemma K - I bound}.
	We see that
	\begin{align}
		& \sup_{\tau \in \mc T} \lVert \bm K^{1/2}  \hatmat K^{-1} \bm K^{1/2} - \bm I_r \rVert_2 \notag \\
		& \hspace{0.2cm} = \sup_{\tau \in \mc T} \norm{ \left(\bm K^{-1/2} \hatmat K  \bm K^{-1/2}\right)^{-1} - \bm I_r}_2 \notag \\
		& \hspace{0.2cm} \leq \sup_{\tau \in \mc T} \norm{ \left(\bm K^{-1/2} \hatmat K  \bm K^{-1/2}\right)^{-1} }_2 \sup_{\tau \in \mc T} \norm{ \bm K^{-1/2} \hatmat K  \bm K^{-1/2}- \bm I_r}_2 \notag \\
		& \defeq R_1 \times R_2 \label{eqn:K-I = R1 R2}
	\end{align}
	We will bound each term in turn, starting with $R_2$.
	We find
	\begin{align*}
		R_2
		& = \sup_{\tau \in \mc T} \norm{ \bm K^{-1/2} \left( \hatmat K  - \bm K \right) \bm K^{-1/2} }_2\\
		& = \sup_{\tau \in \mc T} \norm{ \bm K^{-1/2} \widetilde{\Cmat}\left( \hatmat J^{-1} \hatmat \Sigma \hatmat J^{-1} -  \bm J^{-1} \bm \Sigma  \bm J^{-1}  \right) \widetilde{\Cmat}^T \bm K^{-1/2} }_2\\
		& \leq \sup_{\tau \in \mc T} \norm{\bm K^{-1/2} \widetilde{\Cmat}}_2^2 \sup_{\tau \in \mc T} \norm{\hatmat J^{-1} \hatmat \Sigma  \hatmat J^{-1} -  \bm J^{-1} \bm \Sigma  \bm J^{-1}}_2. 
	\end{align*}
	The first term is $O(1)$ by \eqref{eqn:lemma:KC bound}.
	To bound the second term, we note that
	\begin{align*}
		\sup_{\tau \in \mc T} \lVert \hatmat J^{-1} & \hatmat \Sigma  \hatmat J^{-1}  -  \bm J^{-1} \bm \Sigma  \bm J^{-1} \rVert_2 \\
		& = \sup_{\tau \in \mc T} \norm{\hatmat J^{-1} \hatmat \Sigma  \hatmat J^{-1} - \hatmat J^{-1} \bm \Sigma  \hatmat J^{-1} + \hatmat J^{-1} \bm \Sigma  \hatmat J^{-1} -  \bm J^{-1} \bm \Sigma  \bm J^{-1}}_2 \\
		& \leq \sup_{\tau \in \mc T} \norm{\hatmat J^{-1} \hatmat \Sigma  \hatmat J^{-1} - \hatmat J^{-1} \bm \Sigma  \hatmat J^{-1}}_2 + \sup_{\tau \in \mc T} \norm{\hatmat J^{-1} \bm \Sigma  \hatmat J^{-1} -  \bm J^{-1} \bm \Sigma  \bm J^{-1}}_2.
	\end{align*}
	We see that the first term in this sum satisfies
	\begin{align*}
		\sup_{\tau \in \mc T} \norm{\hatmat J^{-1} \hatmat \Sigma  \hatmat J^{-1} - \hatmat J^{-1} \bm \Sigma  \hatmat J^{-1}}_2 
		& = \sup_{\tau \in \mc T} \norm{\hatmat J^{-1} \left( \hatmat \Sigma  - \bm \Sigma \right) \hatmat J^{-1}}_2 \\
		& \leq \sup_{\tau \in \mc T} \norm{\hatmat J^{-1} }_2^2  \norm{\hatmat \Sigma  - \bm \Sigma}_2\\
		& = O_p\left( \sqrt{\frac{\xi_n^2 \log n}{n}} \right) 
	\end{align*}
	by \eqref{eqn:lemma:hatJinv bound} and Lemma 29 of \citet{belloni_conditional_2019}.
	Moving on to the second term in the sum, we find
	\begin{align*}
		\sup_{\tau \in \mc T} \lVert & \hatmat J^{-1}  \bm \Sigma  \hatmat J^{-1}  -  \bm J^{-1} \bm \Sigma  \bm J^{-1} \rVert_2\\
		& = \sup_{\tau \in \mc T} \norm{\hatmat J^{-1} \bm \Sigma  \hatmat J^{-1} - \bm J^{-1} \bm \Sigma  \hatmat J^{-1} + \bm J^{-1} \bm \Sigma  \hatmat J^{-1} -  \bm J^{-1} \bm \Sigma  \bm J^{-1}}_2 \\
		& \leq \sup_{\tau \in \mc T} \norm{( \hatmat J^{-1} - \bm J^{-1}) \bm \Sigma  \hatmat J^{-1}}_2  + \norm{\bm J^{-1} \bm \Sigma (\hatmat J^{-1} -  \bm J^{-1})}_2\\
		& \leq \sup_{\tau \in \mc T} \norm{\hatmat J^{-1} - \bm J^{-1}}_2 \norm{\bm \Sigma }_2 \sup_{\tau \in \mc T} \norm{\hatmat J^{-1}}_2 
		+ \sup_{\tau \in \mc T} \norm{\bm J^{-1}}_2 \norm{\bm \Sigma}_2 \sup_{\tau \in \mc T} \norm{\hatmat J^{-1} -  \bm J^{-1}}_2 \\
		& =  O_p\left( \sqrt{\frac{(m+s)\xi_n^2 \log n}{n h}} + h \right) 
	\end{align*}
	by \eqref{eqn:lemma:Jinv bound}, \eqref{eqn:lemma:hatJinv bound}, \eqref{eqn:lemma:hatJinv - Jinv bound}, and \ref{as:Gram lower bound}.
	Combining these results, we have
	$$
	R_2 =  O_p\left( \sqrt{\frac{(m+s)\xi_n^2 \log n}{n h}} + h + \sqrt{\frac{\xi_n^2 \log n}{n}}\right)
	$$
	
	Moving on to $R_1$ in \eqref{eqn:K-I = R1 R2}, we note that
	$$
	R_1 = \lambda_{\max} \left \{\left(\bm K^{-1/2} \hatmat K  \bm K^{-1/2}\right)^{-1} \right \}.
	$$
	By the same argument used in \eqref{eqn:J eig lower bound}, we see that for all $\tau$
	$$
	\lambda_{\min} \left \{ \bm K^{-1/2} \hatmat K  \bm K^{-1/2} \right \}
	\geq \lambda_{\min}\{ \bm I_r \} - \norm{ \bm I_r - \bm K^{-1/2} \hatmat K  \bm K^{-1/2} }_2.
	$$
	We recognize the latter term as $R_2$ and note that under the rate conditions  $h = o(1)$ and $ (m + s) \xi_n^2 \log^2 n = o(nh) $, $R_2 = o_p(1)$. Therefore for sufficiently large $n$, $R_2 < 1/2$ and, by extension, $\lambda_{\min} \left \{ \bm K^{-1/2} \hatmat K  \bm K^{-1/2} \right \} > 1/2$ for all $\tau \in \overline{ \mc T }$ with probability converging to $1$.
	This further implies
	$$
	\lambda_{\max} \left \{\left(\bm K^{-1/2} \hatmat K  \bm K^{-1/2}\right)^{-1} \right \} = \left( \lambda_{\min} \left \{ \bm K^{-1/2} \hatmat K  \bm K^{-1/2} \right \} \right)^{-1} < 2
	$$
	for all $\tau \in \overline{ \mc T }$ with probability converging to $1$, meaning that $R_1 = O_p(1)$.
	Applying our bounds for $R_1$ and $R_2$ to \eqref{eqn:K-I = R1 R2} yields \eqref{eqn:lemma K - I bound}.
\end{proof}

\begin{lemma} \label{lemma:GP Zi to EZi}
	Define $\bm Z_i$ as in the proof of Lemma \ref{lemma:GP for Wald with known} and let $\bm B_r(\cdot)$ denote an $r$-dimensional standard Brownian bridge.
	If $\xi_n^2 \log n = o(n)$, then
	\begin{equation*}
		\sup_{\tau \in \mc T} \norm{ (\mathbb{E}_n[\bm Z_i \bm Z_i^T] )^{1/2}  \bm B_r(\tau) - \bm B_r(\tau) }_2
		=  O_p\left( \sqrt{\frac{ r^3 \xi_n^2 \log n }{n}} \right).
	\end{equation*}
\end{lemma}

\begin{proof}[Proof of Lemma \ref{lemma:GP Zi to EZi}]
	Define $\bm r_n(\tau) = (\mathbb{E}_n[\bm Z_i \bm Z_i^T] )^{1/2} \bm B_r(\tau) - (\E[\bm Z_i \bm Z_i^T] )^{1/2} \bm B_r(\tau)$.
	We see that
	\begin{align}
		\sup_{\tau \in \mc T} \norm{\bm r_n(\tau)}_2
		& \leq \sup_{\tau \in \mc T} \norm{(\mathbb{E}_n[\bm Z_i \bm Z_i^T] )^{1/2}  - (\E[\bm Z_i \bm Z_i^T] )^{1/2} }_2 \sup_{\tau \in \mc T} \norm{\bm B_r(\tau)}_2.
	 \label{eqn:GP diff bound}
	\end{align}
	Focusing on the second term in the expression above, we see that
	\begin{align}
		\E[ \sup_{\tau \in \mc T} \norm{ \bm B_r(\tau) }_2 ]
		& = \E \bigg[ \sup_{\tau \in \mc T} \bigg( \sum_{j = 1}^r \bm B_{r,j}^2(\tau) \bigg)^{1/2} \bigg]
		\leq  \E \bigg[ \bigg( \sum_{j = 1}^r  \sup_{\tau \in \mc T} \bm B_{r,j}^2(\tau) \bigg)^{1/2} \bigg] \notag\\
		& \leq  \bigg( \sum_{j = 1}^r  \E \bigg[  \sup_{\tau \in \mc T} \bm B_{r,j}^2(\tau) \bigg] \bigg)^{1/2}
		= O(\sqrt{r}) \label{eqn:expectation of Bessel sup bound}
	\end{align}
	where the second inequality follows from Jensen's inequality.
	Markov's inequality then implies that $\sup_{\tau \in \mc T} \norm{\bm B_r(\tau)}_2 = O_p(\sqrt{r}) $.
	
	We now turn our focus to the first term	in \eqref{eqn:GP diff bound}.
	We define $\bm A = \bm A(\tau) = [\bm K(\tau)]^{-1/2} \widetilde \Cmat [\bm J(\tau)]^{-1}$, temporarily dropping $\tau$ from our notation for ease of presentation.
	By definition, we see that
	$ \mathbb{E}[ \bm Z_i \bm Z_i^T] = \bm A \hatmat \Sigma \bm A^T $ and 
	$ \E[ \bm Z_i \bm Z_i^T] = \bm A \bm \Sigma \bm A^T $.
	Applying the Powers-Stormer inequality \citep{powers_free_1970} and relations between matrix norms, we have
	\begin{align*}
		\norm{ (\bm A \hatmat \Sigma \bm A^T)^{1/2} - (\bm A \bm \Sigma \bm A^T)^{1/2} }_2
		& \leq \norm{ (\bm A \hatmat \Sigma \bm A^T)^{1/2} - (\bm A \bm \Sigma \bm A^T)^{1/2} }_F\\
		& \leq  \norm{ \bm A \hatmat \Sigma \bm A^T - \bm A \bm \Sigma \bm A^T }_*\\
		& \leq  \sqrt{r} \norm{ \bm A \hatmat \Sigma \bm A^T - \bm A \bm \Sigma \bm A^T }_F \\
		& \leq  r \norm{ \bm A \hatmat \Sigma \bm A^T - \bm A \bm \Sigma \bm A^T }_2
	\end{align*}
	for all $\tau$,	where $\norm{\cdot}_F$ and $\norm{\cdot}_*$ denote the Frobenius and nuclear norm, respectively.
	From here, we find
	\begin{align*}
		\sup_{\tau \in \mc T} \norm{ \bm A \hatmat \Sigma \bm A^T - \bm A \bm \Sigma \bm A^T }_2
		& = \sup_{\tau \in \mc T} \norm{ \bm A (\hatmat \Sigma - \bm \Sigma )\bm A^T }_2 
		\leq \sup_{\tau \in \mc T} \norm{\bm A}_2^2 \norm{ \hatmat \Sigma - \bm \Sigma }_2.
	\end{align*}
	We see that 
	\begin{align*}
		\norm{\bm A}_2^2
		& = \norm{\bm A^T \bm A}_2
		= \norm{ \bm J^{-1} \widetilde{\Cmat}^T \bm K^{-1} \widetilde{\Cmat} \bm J^{-1} }_2\\
		& = \norm{ \bm J^{-1} \widetilde{\Cmat}^T \bm K^{-1} \widetilde{\Cmat} \bm J^{-1} \bm \Sigma \bm \Sigma^{-1}}_2\\
		& \leq \norm{ \bm J^{-1} \widetilde{\Cmat}^T \bm K^{-1} \widetilde{\Cmat} \bm J^{-1} \bm \Sigma}_2 \norm{\bm \Sigma^{-1}}_2\\
		& = \lambda_{\max} \{ \bm J^{-1} \widetilde{\Cmat}^T \bm K^{-1} \widetilde{\Cmat} \bm J^{-1} \bm \Sigma\} \norm{\bm \Sigma^{-1}}_2\\
		& = \lambda_{\max} \{ \bm K^{-1} \bm K\} \norm{\bm \Sigma^{-1}}_2
		= \lambda_{\max}\{ \bm I_r \} \norm{\bm \Sigma^{-1}}_2\\
		& = \norm{\bm \Sigma^{-1}}_2
	\end{align*}
	for all $\tau$.
	Since $\xi_n^2 \log n = o(n)$, Lemma 29 of \citet{belloni_conditional_2019} establishes that $\norm{ \hatmat \Sigma - \bm \Sigma }_2 = O_p\left( \sqrt{\frac{\xi_n^2 \log n}{n}}\right)$ and \ref{as:Gram lower bound} provides that $\norm{\bm \Sigma^{-1}}_2 = O(1)$.
	Applying these bounds to \eqref{eqn:GP diff bound} yields
	\begin{align}
		\sup_{\tau \in \mc T} \norm{\bm r_n(\tau)}_2
		& = O_p\left( \sqrt{\frac{ r^3 \xi_n^2 \log n }{n}} \right).
	\end{align}
	The fact that $\E[\bm Z_i \bm Z_i^T] = \bm I_r$, as we showed in the proof of Lemma \ref{lemma:GP for Wald with known}, completes the proof.
\end{proof}

\subsection{Proofs of partial penalized estimator and test results}

\begin{proof}[Proof of Theorem \ref{thm:part pen rates}]
	The same argument suffices to prove \eqref{eqn:FE.part pen rate} and \eqref{eqn:RE.part pen rate} so we will present only the former.
	Let $\epsilon > 0$ and define $q_a = m+s$.
	Theorems \ref{thm:oracle rate} and \ref{thm:strong oracle} provide that there exist $N_1, N_2, M \geq 0$ such that
	$$\Pr \left(  \sup_{\tau \in \mc T} \norm{ \bvecOR_a (\tau) - \bvec^* (\tau) }_2 > M \sqrt{\frac{q_a}{n}}\right) < \epsilon/2$$ for all $n > N_1$ and $$\Pr\left(\sup_{\tau \in \mc T} \norm{\hat \bvec^{(2)}_a (\tau) - \bvecOR_a (\tau)}_2 \neq 0 \right) < \epsilon/2 $$ for all $n > N_2$.
	Therefore
	\begin{align*}
		\Pr \Bigg(  \sup_{\tau \in \mc T} & \norm{ \hat \bvec^{(2)}_a (\tau) - \bvec^* (\tau) }_2  \leq M \sqrt{\frac{q_a}{n}}\Bigg)\\
		& \geq \Pr\left(\left\{\sup_{\tau \in \mc T} \norm{\hat \bvec^{(2)}_a (\tau) - \bvecOR_a (\tau)}_2 = 0 \right\} \bigcap \left\{ \sup_{\tau \in \mc T} \norm{ \bvecOR_a (\tau) - \bvec^* (\tau) }_2 \leq M \sqrt{\frac{q_a}{n}} \right\}\right)\\
		& \geq 1 - \Pr\left(\sup_{\tau \in \mc T} \norm{\hat \bvec^{(2)}_a (\tau) - \bvecOR_a (\tau)}_2 \neq 0 \right) - \Pr \left(  \sup_{\tau \in \mc T} \norm{ \bvecOR_a (\tau) - \bvec^* (\tau) }_2 > M \sqrt{\frac{q_a}{n}}\right)\\
		& \geq 1 - \epsilon
	\end{align*}
	for all $n > \max\{N_1, N_2\}$ and \eqref{eqn:FE.part pen rate} holds by definition.
	The proof of \eqref{eqn:RE.part pen rate} is the same, except that $q_0 = m-r+s$ takes the place of $q_a$.
\end{proof}

\begin{proof}[Proof of Theorem \ref{thm:part pen bahadur}]
	The same argument suffices to prove \eqref{eqn:FE.part pen bahadur} and \eqref{eqn:RE.part pen bahadur}, so we will focus on the former.
	We have from Theorem \ref{thm:oracle bahadur} that
	$$
	\bm r_n(\tau) = \sqrt{n}(\bvecOR_{a, \MSset}(\tau) - \bvec_{\MSset}^*(\tau)) - \frac{1}{\sqrt n} [\bm J(\tau)]^{-1} \sum_{i \in [n]} \tilde \xvec_i \psi_{\tau}(y_i - \xvec_i^T \bvec^*(\tau)) 
	$$
	satisfies 
	$
	\sup_{\tau \in \mc T} \norm{\bm r_n(\tau)}_2 = O_p\left(\frac{(m + s)^{3/4}\xi_n^{1/2} \log^{1/2}n}{n^{1/4}}\right)
	$.
	Define
	$$
	\bar {\bm r}_n(\tau) = \sqrt{n}(\hat \bvec_{a, \MSset}^{(2)}(\tau) - \bvec_{\MSset}^*(\tau)) - \frac{1}{\sqrt n} [\bm J(\tau)]^{-1} \sum_{i \in [n]} \tilde \xvec_i \psi_{\tau}(y_i - \xvec_i^T \bvec^*(\tau)). 
	$$
	Then 
	$
	\bar {\bm r}_n(\tau) = \bm r_n(\tau) + \sqrt{n}(\hat \bvec_{a, \MSset}^{(2)}(\tau) - \bvecOR_{a, \MSset}(\tau))
	$.
	Following that same argument as in the proof of Theorem \ref{thm:part pen rates} and leveraging Theorem \ref{thm:strong oracle}, we can show that for every $\epsilon > 0$ there exist $M,N \geq 0$, such that
	$$
	\Pr\left( \sup_{\tau \in \mc T} \norm{\bar {\bm r}_n(\tau)}_2 \leq M\frac{(m + s)^{3/4}\xi_n^{1/2} \log^{1/2}n}{n^{1/4}} \right) \geq 1 - \epsilon
	$$
	for all $n > N$. 
	This completes the proof.
\end{proof}

\begin{proof}[Proof of Theorem \ref{thm:part pen bessel}]
	We start by proving \eqref{eqn:thm:pp Wald to Bessel}.
	Clearly if $\sup_{\tau \in \mc T} \norm {\hat \bvec^{(2)}_a(\tau) - \bvecOR_a(\tau)  }_2 = 0$, then 
	$\sup_{\tau \in \mc T} |W_n^{(2)}(\tau) - W_n^o(\tau)  | = 0 .$	
	Let $\epsilon > 0$. Theorem \ref{thm:strong oracle} provides that there exists $N \geq 0$ such that $\Pr\left( \sup_{\tau \in \mc T} \shortnorm {\hat \bvec^{(2)}_a(\tau) - \bvecOR_a(\tau)  }_2  = 0 \right) > 1 - \epsilon$ for all $n > N$.
	As such, we find that for all $n > N$
	\begin{align*}
		\Pr\left(\sup_{\tau \in \mc T} |W_n^{(2)}(\tau) - W_n^o(\tau)  | = 0\right)
		& \geq  \Pr\left( \sup_{\tau \in \mc T} \norm {\hat \bvec^{(2)}_a(\tau) - \bvecOR_a(\tau)  }_2  = 0 \right)
		> 1 - \epsilon,
	\end{align*}
	which implies $\sup_{\tau \in \mc T} | W_n^{(2)}(\tau) - W_n^o(\tau) | = o_p(1)$.
	Combining this result with Lemma \ref{lemma:Wald J to Jhat} and applying the triangle inequality yields
	$\sup_{\tau \in \mc T} |W_n^{(2)}(\tau) - \widetilde {W}_n(\tau) | = o_p(1)$.
	Taking this result together with Lemma \ref{lemma:Bessel for Wald with known J} and applying Slutsky's Theorem yields \eqref{eqn:thm:pp Wald to Bessel}.
	
	Moving to the rank statistic, we aim to show that the partial penalized dual problem \eqref{eqn:rank scores} and oracle dual problem \eqref{eqn:oracle rank scores} are equivalent with probability converging to $1$.
	Since the objective function is the same in both cases, we just need to show that the dual feasibility conditions are equivalent. Indexing by $j \in \{0,\ldots,p\}$, 
	the partial penalized feasibility conditions can be expressed as
	$$
	\xmat_{j}^T \bm a = (1 - \tau)\xmat_{j}^T \bm 1_n +  \Cmat_j^T\bm \lambda \text{ for } j \in \Mset, \; \;|\xmat_{j}^T (\bm a - (1-\tau) \bm 1_n ) | \leq \hat w_j^{(B-1)} \text{ for } j \in \Mset^c
	$$
	and the oracle dual feasibility conditions can be expressed as
	$$
	\xmat_{j}^T \bm a = (1 - \tau)\xmat_{j}^T \bm 1_n +  \Cmat_j^T\bm \lambda \text{ for } j \in \Mset, \; \; \xmat_{j}^T (\bm a - (1-\tau) \bm 1_n ) = 0 \text{ for } j \in \Sset.
	$$
	We have three cases to check:
	\begin{itemize}
		\item \textit{Case 1:} $j \in \Mset$ --- The condition is the same in both dual problems.
		\item \textit{Case 2:} $j \in \Sset$ --- Lemma \ref{lemma:global33modAssumptions2} guarantees that $\inf_{j \in \Sset}  \sup_{\tau \in \mc T}|\hat{\beta}^o_{0,j}(\tau)| > d\lambda$ with probability converging to $1$. 
		When $\sup_{\tau \in \mc T}|\hat{\beta}^o_{0,j}(\tau)| > d\lambda$, we have $\sup_{\tau \in \mc T} \hat w_j^{(B-1)}(\tau) = 0$. By extension, the condition in the partial penalized dual problem simplifies to $\xmat_{j}^T \bm a = (1-\tau) \xmat_{j}^T\bm 1_n $ uniformly over $\tau$, matching the conditions in the oracle dual problem.
		\item \textit{Case 3:} $j \in \MSset^c$ --- Lemma \ref{lemma:global33modAssumptions2} provides that $	\sup_{\tau \in \mc T, \bm w(\tau) \in \partial_{{\MSset^c}} \mc Q_{\tau,n}(\hat{\bvec}_e^{o}(\tau))} ||\bm w(\tau)||_{\infty} < d_1\lambda$ with probability converging to $1$. At the same time, $\hat \beta_{0,j}^o = 0$ by definition and, consequently, $\hat w_j^{(B-1)}(\tau) = p_{\lambda}'(| \hat \beta_{0,j}^o |) \geq d_1 \lambda$.
		Therefore we have $|\xmat_{j}^T (\bm a - (1-\tau) \bm 1_n ) | \leq \hat w_j^{(B-1)}$ for all $\bm a \in [0,1]^n$, meaning that the constraint in \eqref{eqn:rank scores} holds trivially, with probability converging to $1$.
	\end{itemize}

	Thus the  the partial penalized dual problem \eqref{eqn:rank scores} and oracle dual problem \eqref{eqn:oracle rank scores} are equivalent with probability converging to $1$.
	As an immediate consequence, $\sup_{\tau} \norm{ \hat \avec^o(\tau) - \hat \avec(\tau) }_2 = 0$ with probability converging to $1$.
	When $\sup_{\tau \in \mc T} \shortnorm {\hat \bvec^{(2)}_0(\tau) - \bvecOR_0(\tau)  }_2 = 0$ and $\sup_{\tau} \shortnorm{ \hat \avec^o(\tau) - \hat \avec(\tau) }_2 = 0$, clearly $\sup_{\tau \in \mc T}| R_n^{(2)}(\tau) - R_n^{o}(\tau)| = 0 $. From here
	an equivalent argument to the one used above to prove \eqref{eqn:thm:pp Wald to Bessel} proves \eqref{eqn:thm:pp RS to Bessel}.
\end{proof}

\subsection{GIC model-selection consistency}\label{app:gic_selection}

\subsubsection{Conditions}\label{app:gic_notation}

We use $\mc Q_{\tau,n}$ from \eqref{eqn:qndefine}, the empirical mean
$\empmean$, and the sets $\Mset$ and $\Sset$ defined in the main text.
Throughout this appendix, Conditions \ref{as:iid array} and
\ref{as:density conditions} are imposed, and $\mc T$ is the interval
specified in the subsection containing Theorem \ref{thm:GIC.consistency}.
For $e=0$, either \eqref{hyp:linear hypothesis} or Condition
\ref{as:local alts} is imposed. 
Note that, under \eqref{hyp:linear hypothesis} $\bm\eta(\tau)=\bm0$.
Let $f_-$, $f_+$, and $L_f$ denote, respectively, the positive lower
bound, the upper bound, and the derivative bound in Condition
\ref{as:density conditions}.

For $A\subset\Mset^c\setminus\{0\}$, let
\[
D_A=\{0\}\cup\Mset\cup A.
\]
The intercept is unpenalized in Algorithms
\ref{algo:LLA for full estimator}--\ref{algo:LLA for reduced estimator};
accordingly, the support restrictions in \eqref{eqn:FE.oracle} and
\eqref{eqn:RE.oracle} include coordinate $0$.
Define
\begin{align}
\widehat L_{a,A}(\tau)
&=\inf_{\bm b_{D_A^c}=\bm0}\mc Q_{\tau,n}(\bm b),
\label{eqn:gic_aux_loss_full}\\
\widehat L_{0,A}(\tau)
&=\inf_{\substack{\bm b_{D_A^c}=\bm0\\
		\Cmat\bm b_{\Mset}=\bm\Psi(\tau)}}\mc Q_{\tau,n}(\bm b),
\label{eqn:gic_aux_loss_reduced}\\
\operatorname{GIC}^{\circ}_e(A)
&=\int_{\mc T}\log\widehat L_{e,A}(\tau)\,d\tau+|A|\phi_n.
\label{eqn:gic_auxiliary}
\end{align}
For each fixed $e$, replacing $|A|$ in \eqref{eqn:gic_auxiliary} by
$|D_A|-r1(e=0)$ adds the constant $(m+1-r1(e=0))\phi_n$ for every $A$.

The sets of candidates used in the proof are
\begin{align*}
\mathcal U_n&=\{A\subset\Mset^c\setminus\{0\}:|A|\le K_{\mathrm G,n},\ \Sset\nsubseteq A\},\\
\mathcal O_n&=\{A\subset\Mset^c\setminus\{0\}:|A|\le K_{\mathrm G,n},\ \Sset\subsetneq A\},\\
\mathcal A_n&=\{A\subset\Mset^c\setminus\{0\}:|A\setminus\Sset|\le K_{\mathrm G,n}\},\\
\mathcal B_n&=\{\Sset\}\cup\{A\cup(\Sset\setminus\{j\}):
A\in\mathcal U_n,\ j\in\Sset\setminus A\}.
\end{align*}
For use in the inclusion comparisons, write
$A^{+j}=A\cup(\Sset\setminus\{j\})$.
Define the excess loss by
\begin{equation}\label{eqn:gic_excess_loss}
\ell_{\tau,\bm b}
=\rho_\tau(Y-\xvec^T\bm b)
-\rho_\tau\{Y-\xvec^T\bvec^*(\tau)\}.
\end{equation}
Define
\begin{equation}\label{eqn:gic_dimensions}
d_{\mathrm G,n}=m+1+s+K_{\mathrm G,n}.
\end{equation}
\begin{equation}\label{eqn:gic_sparse_envelope}
\xi_{\mathrm G,n}
=\sup_{\substack{\|\bm v\|_2=1\\
		|\operatorname{supp}(\bm v)|\le d_{\mathrm G,n}}}
\operatorname*{ess\,sup}|\xvec^T\bm v|.
\end{equation}
\begin{equation}\label{eqn:gic_log_factor}
L_{\mathrm G,n}=\log\bigl(p\vee n\vee d_{\mathrm G,n}
\vee\lceil\xi_{\mathrm G,n}\rceil\vee3\bigr).
\end{equation}
By the definition of $\xi_n$,
\eqref{eqn:gic_sparse_envelope} gives $\xi_n\le\xi_{\mathrm G,n}$. For a sufficiently large fixed constant $A_0$, let
\begin{equation}\label{eqn:gic_an}
a_{\mathrm G,n}=A_0\left\{
\sqrt{\frac{d_{\mathrm G,n}L_{\mathrm G,n}}n}
+\frac{\xi_{\mathrm G,n}d_{\mathrm G,n}L_{\mathrm G,n}}n\right\},
\end{equation}
\begin{equation}\label{eqn:gic_mu}
\mu_{\mathrm G,n}=
\begin{cases}
	\displaystyle\min_{j\in\Sset}\int_{\mc T}|\beta_j^*(\tau)|\,d\tau,&s\ge1,\\
	\infty,&s=0.
\end{cases}
\end{equation}

We assume the following additional conditions in Theorem \ref{thm:GIC.consistency}:
\begin{enumerate}[label=(G\arabic*),ref=(G\arabic*),leftmargin=2.5\parindent]
\item\label{as:gic_design}
There are constants $K_x,\kappa_-,\kappa_+,q_0>0$ such that
\[
\|\xvec\|_\infty\le K_x
\quad\text{almost surely}.
\]
For every $\bm v$ with $\|\bm v\|_2=1$ and $|\operatorname{supp}(\bm v)|\le d_{\mathrm G,n}$,
\[
\kappa_-
\le
\E(\xvec^T\bm v)^2
\le
\kappa_+,
\qquad
\frac{\{\E(\xvec^T\bm v)^2\}^{3/2}}
{\E|\xvec^T\bm v|^3}
\ge q_0.
\]

\item\label{as:gic_moments}
There are constants $B_\beta<\infty$, $B_Y<\infty$, and $\tau_0\in\mc T$ such that
\[
\sup_{\tau\in\mc T}\max_{j\in D_{\Sset}}
|\beta_j^*(\tau)|\le B_\beta,
\qquad
\sup_{n\ge1}\E|Y-\xvec^T\bvec^*(\tau_0)|\le B_Y.
\]

\item\label{as:gic_underfit_rates}
The sequence in \eqref{eqn:gic_an} satisfies $a_{\mathrm G,n}\to0$.
If $s\ge1$, then
\[
\phi_n\to0,
\qquad
a_{\mathrm G,n}=o(\mu_{\mathrm G,n}),
\qquad
s\phi_n=o(\mu_{\mathrm G,n}^2).
\]

\item\label{as:gic_overfit_rates}
The sequences in \eqref{eqn:gic_dimensions} and \eqref{eqn:gic_sparse_envelope} satisfy
\begin{equation}\label{eqn:gic_sharp_growth}
	\frac{(\xi_{\mathrm G,n})^2d_{\mathrm G,n}^3L_{\mathrm G,n}^2}{n}\to0,
\end{equation}
and the penalty satisfies
\begin{equation}\label{eqn:gic_penalty_growth}
	\frac{L_{\mathrm G,n}}{n}=o(\phi_n).
\end{equation}

\item\label{as:gic_oracle_path}
For each $e\in\{0,a\}$, the finite tuning set $\mathcal P_{e,n}$ in \eqref{eqn:gic_lambda_selector} contains a deterministic sequence $\lambda_{e,n}^{\circ}$ such that
\begin{equation}\label{eqn:gic_oracle_path_assumption}
	P\left(
	\left\{
	\sup_{\tau\in\mc T}
	\left\|\widehat{\bvec}_e^{(2)}(\tau;\lambda_{e,n}^{\circ})-\bvecOR_e(\tau)\right\|_2>0
	\right\}
	\cup
	\left\{\widehat{\Sset}_e(\lambda_{e,n}^{\circ})\ne\Sset\right\}
	\right)
	\to0.
\end{equation}
\end{enumerate}

Condition \ref{as:gic_oracle_path} holds for the global LLA estimators by Theorem \ref{thm:strong oracle}.
The growth requirement \eqref{eqn:gic_sharp_growth} in Condition \ref{as:gic_overfit_rates} implies
\begin{equation}\label{eqn:gic_sharp_consequences}
\xi_{\mathrm G,n}\sqrt{\frac{d_{\mathrm G,n}L_{\mathrm G,n}}{n}}\to0,
\qquad
\xi_{\mathrm G,n}\frac{d_{\mathrm G,n}L_{\mathrm G,n}}{\sqrt n}=o(L_{\mathrm G,n}^{1/2}),
\qquad
\xi_{\mathrm G,n}\frac{d_{\mathrm G,n}^{3/2}L_{\mathrm G,n}}{\sqrt n}=o(1).
\end{equation}
Condition \ref{as:gic_design} gives
\[
\xi_{\mathrm G,n}^2\le K_x^2d_{\mathrm G,n},
\qquad
\frac{\xi_{\mathrm G,n}^2d_{\mathrm G,n}^3L_{\mathrm G,n}^2}{n}
\le K_x^2\frac{d_{\mathrm G,n}^4L_{\mathrm G,n}^2}{n}.
\]
Thus $d_{\mathrm G,n}^4L_{\mathrm G,n}^2=o(n)$ suffices for
\eqref{eqn:gic_sharp_growth}.

For the reduced estimator, use $\widetilde{\bvec}^*$ from the proof of
Theorem \ref{thm:oracle rate}, with the bound \eqref{eqn:beta tilde rate}.
To express coefficients after multiplication by $\sqrt n$, write
\begin{equation}\label{eqn:gic_feasible_shift}
\bm v_a(\tau)=\bm0,
\qquad
\bm v_0(\tau)=\sqrt n\{\widetilde{\bvec}^*(\tau)-\bvec^*(\tau)\}.
\end{equation}
The definition of $\widetilde{\bvec}^*$ gives
\[
(\bm v_0)_{\Mset}(\tau)=-\Cmat^T(\Cmat\Cmat^T)^{-1}\bm\eta(\tau),
\qquad
(\bm v_0)_{\Mset^c}(\tau)=\bm0,
\]
and Condition \ref{as:local alts} gives
\begin{equation}\label{eqn:gic_feasible_shift_bound}
\begin{aligned}
	\|\bm v_0(\tau)\|_2^2
	&=\bm\eta(\tau)^T(\Cmat\Cmat^T)^{-1}\bm\eta(\tau)
	\le C\|\bm\eta(\tau)\|_2^2,\\
	\sup_{\tau\in\mc T}\|\bm v_0(\tau)\|_2^2
	&=O(m-r+s).
\end{aligned}
\end{equation}

\subsubsection{Underfitted candidates}\label{app:gic_underfit}

\begin{lemma}\label{lem:gic_quantile_path}
For all $\tau,t\in\mc T$,
\[
|\xvec^T\{\bvec^*(\tau)-\bvec^*(t)\}|
\le
f_-^{-1}|\tau-t|
\quad\text{almost surely}.
\]
\end{lemma}

\begin{proof}
The argument used in the proof of Lemma 19 of \citet{belloni_conditional_2019} gives the required bound. For $\tau<t$, \eqref{eqn:conditional quantile linear expr} and Condition \ref{as:density conditions}(ii) imply, almost surely,
\begin{align}
	t-\tau
	=\int_{\xvec^T\bvec^*(\tau)}^{\xvec^T\bvec^*(t)}
	f_{y\mid\xvec}(y\mid\xvec)\,dy\geq f_-\xvec^T\{\bvec^*(t)-\bvec^*(\tau)\}\ge0.
	\label{eqn:gic_quantile_path_bound}
\end{align}
For $t<\tau$, apply \eqref{eqn:gic_quantile_path_bound} to $(t,\tau)$.
For $t=\tau$, the difference is zero.
\end{proof}

\begin{lemma}\label{lem:gic_loss_scale}
There is a constant $c_\rho>0$ such that
\[
0<c_\rho
\le
\inf_{\tau\in\mc T}\E \rho_\tau\{Y-\xvec^T\bvec^*(\tau)\}
\le
\sup_{\tau\in\mc T}\E \rho_\tau\{Y-\xvec^T\bvec^*(\tau)\}
\le
B_Y+f_-^{-1},
\]
Moreover,
\[
\sup_{\tau\in\mc T}\mc Q_{\tau,n}\{\bvec^*(\tau)\}=O_p(1),
\qquad
P\left\{\inf_{\tau\in\mc T}\mc Q_{\tau,n}\{\bvec^*(\tau)\}<c_\rho\right\}\to0.
\]
\end{lemma}

\begin{proof}
By \eqref{eqn:gic_quantile_path_bound}, Condition \ref{as:gic_moments}, and $\rho_\tau(u)\le |u|$,
\begin{equation}\label{eqn:gic_loss_upper}
	\sup_{\tau\in\mc T}\rho_\tau\{Y-\xvec^T\bvec^*(\tau)\}
	\le |Y-\xvec^T\bvec^*(\tau_0)|+f_-^{-1},
	\qquad
	\E \sup_{\tau\in\mc T}\rho_\tau\{Y-\xvec^T\bvec^*(\tau)\}\le B_Y+f_-^{-1}.
\end{equation}
Consequently, for every $M>0$, Markov's inequality gives
\[
P\left\{\sup_{\tau\in\mc T}\mc Q_{\tau,n}\{\bvec^*(\tau)\}>M\right\}
\le \frac{B_Y+f_-^{-1}}{M}.
\]
For the lower bound, let
\[
a_0=\frac{1-\tau_U}{2f_+},\qquad
p_0=\frac{1-\tau_U}{2},\qquad
A_i=\{Y_i\ge\xvec_i^T\bvec^*(\tau_U)+a_0\}.
\]
Condition \ref{as:density conditions}(i) yields
\begin{align*}
	\Pr(A_i\mid\xvec_i)
	&=1-\tau_U-
	\int_{\xvec_i^T\bvec^*(\tau_U)}^{\xvec_i^T\bvec^*(\tau_U)+a_0}
	f_{y\mid\xvec}(y\mid\xvec_i)\,dy\\
	&\ge 1-\tau_U-f_+a_0=p_0.
\end{align*}
For every $\tau\in\mc T$, the monotonicity of conditional quantiles implies
\[
\rho_\tau\{Y_i-\xvec_i^T\bvec^*(\tau)\}\ge\tau_La_0\,1(A_i),\qquad
\inf_{\tau\in\mc T}\E \rho_\tau\{Y-\xvec^T\bvec^*(\tau)\}\ge\tau_La_0p_0.
\]
Set $c_\rho=\tau_La_0p_0/2$. The variables $1(A_i)$ are independent and lie in $[0,1]$. Hoeffding's inequality therefore gives
\begin{equation}\label{eqn:gic_loss_lower_probability}
	\begin{aligned}
		P\left\{\inf_{\tau\in\mc T}\mc Q_{\tau,n}\{\bvec^*(\tau)\}<c_\rho\right\}
		&\le \Pr\left\{\frac1n\sum_{i=1}^n1(A_i)<\frac{p_0}{2}\right\}\\
		&\le \exp(-np_0^2/2)\longrightarrow0.
	\end{aligned}
\end{equation}
\end{proof}

\begin{lemma}\label{lem:gic_population_smallball}
Let
\[
\pi_0=\frac{\kappa_-^3q_0^2}{8\kappa_+^3}.
\]
Every $\bm v$ with $\|\bm v\|_2=1$ and $|\operatorname{supp}(\bm v)|\le d_{\mathrm G,n}$ satisfies
\[
|\xvec^T\bm v|\le\xi_{\mathrm G,n}
\quad\text{almost surely},
\qquad
\Pr\{|\xvec^T\bm v|\ge(\kappa_-/2)^{1/2}\}\ge\pi_0.
\]
\end{lemma}

\begin{proof}
The bound $|\xvec^T\bm v|\le\xi_{\mathrm G,n}$ follows from \eqref{eqn:gic_sparse_envelope}. For the probability bound, set $Z=\xvec^T\bm v$. Condition \ref{as:gic_design} gives
\[
\E Z^2\ge\kappa_-,\qquad
\E|Z|^3\le\kappa_+^{3/2}/q_0.
\]
H\"older's inequality implies
\begin{align*}
	\kappa_-
	&\le\E\left[Z^21\{|Z|<(\kappa_-/2)^{1/2}\}\right]
	+\E\left[Z^21\{|Z|\ge(\kappa_-/2)^{1/2}\}\right]\\
	&\le\frac{\kappa_-}{2}
	+(\E|Z|^3)^{2/3}
	\Pr\{|Z|\ge(\kappa_-/2)^{1/2}\}^{1/3}\\
	&\le\frac{\kappa_-}{2}
	+\kappa_+q_0^{-2/3}
	\Pr\{|Z|\ge(\kappa_-/2)^{1/2}\}^{1/3}.
\end{align*}
Thus
\[
\Pr\{|Z|\ge(\kappa_-/2)^{1/2}\}
\ge\left(\frac{\kappa_-q_0^{2/3}}{2\kappa_+}\right)^3
=\pi_0.
\]
\end{proof}

\begin{lemma}\label{lem:gic_excess_identity}
For all $\tau,t\in\mc T$ and $\bm b,\bm c\in\reals^{p+1}$,
\[
\ell_{\tau,\bm b}
=
\int_{\xvec^T\bvec^*(\tau)}^{\xvec^T\bm b}
\{1(Y\le u)-\tau\}\,du,
\]
\[
|\ell_{\tau,\bm b}|
\le
|\xvec^T\{\bm b-\bvec^*(\tau)\}|,
\]
and
\begin{align*}
	|\ell_{\tau,\bm b}-\ell_{t,\bm c}|
	&\le
	|\xvec^T\{\bm b-\bvec^*(\tau)-\bm c+\bvec^*(t)\}|\\
	&\quad+2f_-^{-1}|\tau-t|
	+|\tau-t|\,|\xvec^T\{\bm c-\bvec^*(t)\}|.
\end{align*}
\end{lemma}

\begin{proof}
For fixed $Y$, the function $z\mapsto\rho_\tau(Y-z)$ is absolutely continuous and its derivative equals $1(Y\le z)-\tau$ for Lebesgue-almost every $z$. Hence
\begin{equation}\label{eqn:gic_excess_integral}
	\ell_{\tau,\bm b}
	=\int_{\xvec^T\bvec^*(\tau)}^{\xvec^T\bm b}
	\{1(Y\le u)-\tau\}\,du,
	\qquad
	|\ell_{\tau,\bm b}|
	\le |\xvec^T\{\bm b-\bvec^*(\tau)\}|.
\end{equation}
Subtracting the two integral representations gives
\begin{align*}
	\ell_{\tau,\bm b}-\ell_{t,\bm c}
	=\int_{\xvec^T\bm c}^{\xvec^T\bm b}\{1(Y\le u)-\tau\}\,du +\int_{\xvec^T\bvec^*(\tau)}^{\xvec^T\bvec^*(t)}
	\{1(Y\le u)-\tau\}\,du
	+(t-\tau)\xvec^T\{\bm c-\bvec^*(t)\}.
\end{align*}
By \eqref{eqn:gic_quantile_path_bound}, we get
\begin{align*}
	|\xvec^T(\bm b-\bm c)|
	&\le|\xvec^T\{\bm b-\bvec^*(\tau)-\bm c+\bvec^*(t)\}|
	+f_-^{-1}|\tau-t|,\\
	|\ell_{\tau,\bm b}-\ell_{t,\bm c}|
	&\le|\xvec^T\{\bm b-\bvec^*(\tau)-\bm c+\bvec^*(t)\}|+2f_-^{-1}|\tau-t|
	+|\tau-t|\,|\xvec^T\{\bm c-\bvec^*(t)\}|.
\end{align*}
\end{proof}

\begin{lemma}\label{lem:gic_population_curvature}
There are constants $c_1,c_2>0$ such that, for every $A\in\mathcal A_n$, $\tau\in\mc T$, and $\bm b$ satisfying $\bm b_{D_A^c}=\bm0$,
\[
\E \ell_{\tau,\bm b}
\ge
c_1\{\|\bm b-\bvec^*(\tau)\|_2^2\wedge \|\bm b-\bvec^*(\tau)\|_2\},
\]
and
\[
\E \ell_{\tau,\bm b}
\ge
c_2\|\bm b-\bvec^*(\tau)\|_2^2
\quad\text{if }\|\bm b-\bvec^*(\tau)\|_2\le c_2.
\]
\end{lemma}

\begin{proof}
We use the population loss expansion in Lemma 7.1 of \citet{zheng2015globally}. Write $\bm\delta=\bm b-\bvec^*(\tau)$ and $r=\|\bm\delta\|_2$. The support of $\bm\delta$ is contained in $D_A\cup\Sset$, whose cardinality is at most $d_{\mathrm G,n}$. Conditional expectation in \eqref{eqn:gic_excess_integral} and the equality
$\Pr\{Y\le\xvec^T\bvec^*(\tau)\mid\xvec\}=\tau$ give
\begin{equation}\label{eqn:gic_population_expansion}
	\E \ell_{\tau,\bm b}
	=\E\left[
	(\xvec^T\bm\delta)^2\int_0^1(1-u)
	f_{y\mid\xvec}\{\xvec^T\bvec^*(\tau)+u\xvec^T\bm\delta\mid\xvec\}\,du
	\right].
\end{equation}
Conditions \ref{as:density conditions}(ii)--(iii) and \ref{as:gic_design} imply
\begin{align}
	\E \ell_{\tau,\bm b}
	&\ge \frac{f_-}{2}\E(\xvec^T\bm\delta)^2
	-\frac{L_f}{6}\E|\xvec^T\bm\delta|^3\notag\\
	&\ge\frac{f_-\kappa_-}{2}r^2
	-\frac{L_f\kappa_+^{3/2}}{6q_0}r^3.
	\label{eqn:gic_population_cubic}
\end{align}
Choose
\[
c_2=\min\left\{1,\frac{f_-\kappa_-}{4},
\frac{3f_-\kappa_-q_0}{2(1+L_f)\kappa_+^{3/2}}\right\},
\qquad c_1=\frac{f_-\kappa_-c_2}{4}.
\]
For $r\le c_2$, \eqref{eqn:gic_population_cubic} yields
\begin{equation}\label{eqn:gic_population_local_lower}
	\E \ell_{\tau,\bm b}\ge\frac{f_-\kappa_-}{4}r^2\ge c_2r^2.
\end{equation}
For $r>c_2$, convexity of the check loss and $\ell_{\tau,\bvec^*(\tau)}=0$ imply
\[
\E \ell_{\tau,\bvec^*(\tau)+(c_2/r)\bm\delta}
\le\frac{c_2}{r}\E \ell_{\tau,\bm b}.
\]
The vector $(c_2/r)\bm\delta$ has the same support as $\bm\delta$. Applying \eqref{eqn:gic_population_local_lower} to this vector gives
\begin{equation}\label{eqn:gic_population_global_lower}
	\E \ell_{\tau,\bm b}
	\ge\frac{r}{c_2}\frac{f_-\kappa_-}{4}c_2^2
	=c_1r.
\end{equation}
Since $c_2\le1$ and $c_1\le f_-\kappa_-/4$, equations
\eqref{eqn:gic_population_local_lower}--\eqref{eqn:gic_population_global_lower} give
$\E \ell_{\tau,\bm b}\ge c_1(r^2\wedge r)$ for every $r\ge0$.
\end{proof}

\begin{lemma}\label{lem:gic_empirical_smallball}
Let
\[
c_{\mathrm{sb}}
=
\frac{(\kappa_-/2)^{1/2}\pi_0}{4}.
\]
\[
P\left\{
\inf_{\substack{\|\bm v\|_2=1\\
		|\operatorname{supp}(\bm v)|\le d_{\mathrm G,n}}}
\empmean |\xvec^T\bm v|
<
c_{\mathrm{sb}}
\right\}\to0.
\]
\end{lemma}

\begin{proof}
Let $\varepsilon=(\kappa_-/2)^{1/2}/(2\xi_{\mathrm G,n})$. For every $D\subset\{0,\ldots,p\}$ with $1\le|D|\le d_{\mathrm G,n}$, choose a finite set $\mathcal N_D$ such that
\begin{equation}\label{eqn:gic_smallball_cover}
	\begin{gathered}
		\mathcal N_D\subset\{\bm u:\operatorname{supp}(\bm u)\subset D,\ \|\bm u\|_2=1\},
		\qquad |\mathcal N_D|\le(1+2/\varepsilon)^{|D|},\\
		\sup_{\substack{\operatorname{supp}(\bm v)\subset D\\\|\bm v\|_2=1}}
		\min_{\bm u\in\mathcal N_D}\|\bm v-\bm u\|_2\le\varepsilon.
	\end{gathered}
\end{equation}
The definition of $L_{\mathrm G,n}$ in \eqref{eqn:gic_log_factor} implies
\begin{align}
	\sum_{\substack{D\subset\{0,\ldots,p\}\\1\le|D|\le d_{\mathrm G,n}}}|\mathcal N_D|
	&\le\sum_{k=1}^{d_{\mathrm G,n}}\binom{p+1}{k}
	\{1+4\xi_{\mathrm G,n}(2/\kappa_-)^{1/2}\}^k\notag\\
	&\le d_{\mathrm G,n}\{C(p+1)(1+\xi_{\mathrm G,n})\}^{d_{\mathrm G,n}}
	\le\exp(Cd_{\mathrm G,n}L_{\mathrm G,n}).
	\label{eqn:gic_smallball_cardinality}
\end{align}
Lemma \ref{lem:gic_population_smallball} and Hoeffding's inequality give, for every $\bm u\in\mathcal N_D$,
\[
\Pr\left\{\empmean 1\{|\xvec^T\bm u|\ge(\kappa_-/2)^{1/2}\}<\frac{\pi_0}{2}\right\}
\le e^{-n\pi_0^2/2}.
\]
By \eqref{eqn:gic_an} and Condition \ref{as:gic_underfit_rates}, we obtain
\[
\frac{d_{\mathrm G,n}L_{\mathrm G,n}}{n}
\le\frac{a_{\mathrm G,n}^2}{A_0^2}\longrightarrow0.
\]
A union bound using \eqref{eqn:gic_smallball_cardinality} therefore yields
\begin{equation}\label{eqn:gic_smallball_probability}
	\Pr\left\{
	\min_D\min_{\bm u\in\mathcal N_D}
	\empmean 1\{|\xvec^T\bm u|\ge(\kappa_-/2)^{1/2}\}<\frac{\pi_0}{2}
	\right\}
	\le\exp\{Cd_{\mathrm G,n}L_{\mathrm G,n}-n\pi_0^2/2\}\longrightarrow0.
\end{equation}
For every admissible $\bm v$, choose $\bm u\in\mathcal N_{\operatorname{supp}(\bm v)}$ as in \eqref{eqn:gic_smallball_cover}. By \eqref{eqn:gic_sparse_envelope},
\[
|\xvec^T(\bm u-\bm v)|\le\xi_{\mathrm G,n}\varepsilon
=\frac{(\kappa_-/2)^{1/2}}{2}.
\]
Thus, simultaneously for all such $\bm v$, we have
\begin{align*}
	\empmean |\xvec^T\bm v|
	&\ge\frac{(\kappa_-/2)^{1/2}}{2}
	\empmean 1\{|\xvec^T\bm v|\ge(\kappa_-/2)^{1/2}/2\}\\
	&\ge\frac{(\kappa_-/2)^{1/2}}{2}
	\empmean 1\{|\xvec^T\bm u|\ge(\kappa_-/2)^{1/2}\}
	\ge\frac{(\kappa_-/2)^{1/2}\pi_0}{4}
\end{align*}
on the complement of the event in \eqref{eqn:gic_smallball_probability}.
\end{proof}

\begin{lemma}\label{lem:gic_localization}
Except on a set of outer probability tending to zero, for every $B\in\mathcal B_n$ and every $\tau\in\mc T$,
\begin{align*}
	\widehat L_{a,B}(\tau)
	&=
	\inf_{\substack{\bm b_{D_B^c}=\bm0\\
			\|\bm b-\bvec^*(\tau)\|_2\le n}}
	\mc Q_{\tau,n}(\bm b),\\
	\widehat L_{0,B}(\tau)
	&=
	\inf_{\substack{\bm b_{D_B^c}=\bm0,\ 
			\Cmat\bm b_{\Mset}=\bm\Psi(\tau)\\
			\|\bm b-\bvec^*(\tau)\|_2\le n}}
	\mc Q_{\tau,n}(\bm b).
\end{align*}
\end{lemma}

\begin{proof}
Let $c_{\mc T}=\tau_L\wedge(1-\tau_U)$ and
$F=|Y-\xvec^T\bvec^*(\tau_0)|+f_-^{-1}$. By \eqref{eqn:gic_loss_upper},
$\empmean F=O_p(1)$. Lemma \ref{lem:gic_empirical_smallball} gives an event whose probability tends to one on which, for every $B\in\mathcal B_n$, $\tau\in\mc T$, and $\bm b$ with $\bm b_{D_B^c}=\bm0$,
\begin{align}
	\mc Q_{\tau,n}(\bm b)
	&\ge c_{\mc T}\empmean |Y-\xvec^T\bm b|\notag\\
	&\ge c_{\mc T}\empmean |\xvec^T\{\bm b-\bvec^*(\tau)\}|-c_{\mc T}\empmean F\notag\\
	&\ge c_{\mc T}c_{\mathrm{sb}}\|\bm b-\bvec^*(\tau)\|_2-c_{\mc T}\empmean F.
	\label{eqn:gic_localization_lower}
\end{align}
For either $e\in\{0,a\}$, consider
\[
\bm b^0_{e,B}(\tau)
=\bvec^*(\tau)+\frac{\bm v_e(\tau)}{\sqrt n}
-\sum_{j\in\Sset\setminus B}\beta_j^*(\tau)\bm e_j.
\]
By \eqref{eqn:gic_feasible_shift}, one has
\begin{equation}\label{eqn:gic_underfit_competitor_feasible}
	(\bm b^0_{e,B})_{D_B^c}=\bm0,\qquad
	\Cmat(\bm b^0_{0,B})_{\Mset}
	=\Cmat\bvec^*_{\Mset}-\frac{\bm\eta}{\sqrt n}=\bm\Psi.
\end{equation}
Here $\Sset\setminus B$ is empty or a singleton. Equations \eqref{eqn:gic_feasible_shift_bound} and \eqref{eqn:gic_sparse_envelope}, together with Condition \ref{as:gic_moments}, imply
\begin{align}
	\sup_{B,\tau}\|\bm b^0_{e,B}(\tau)-\bvec^*(\tau)\|_2
	&\le B_\beta+C\sqrt{\frac{d_{\mathrm G,n}}{n}}<n,\notag\\
	\sup_{B,\tau}\mc Q_{\tau,n}\{\bm b^0_{e,B}(\tau)\}
	&\le \empmean F+K_xB_\beta
	+C\xi_{\mathrm G,n}\sqrt{\frac{d_{\mathrm G,n}}{n}}.
	\label{eqn:gic_underfit_competitor_upper}
\end{align}
Condition \ref{as:gic_design} gives $\xi_{\mathrm G,n}\le K_x\sqrt{d_{\mathrm G,n}}$. Hence \eqref{eqn:gic_an} and $a_{\mathrm G,n}\to0$ yield
\[
\frac{\xi_{\mathrm G,n}\sqrt{d_{\mathrm G,n}/n}}{n}
\le\frac{K_xd_{\mathrm G,n}}{n^{3/2}}
\le\frac{K_xa_{\mathrm G,n}^2}{A_0^2\sqrt nL_{\mathrm G,n}}
\longrightarrow0.
\]
Therefore
\begin{align*}
	\Pr\bigg\{c_{\mc T}c_{\mathrm{sb}}n
	&\le(1+c_{\mc T})\empmean F+K_xB_\beta
	+C\xi_{\mathrm G,n}\sqrt{d_{\mathrm G,n}/n}\bigg\}
	\longrightarrow0.
\end{align*}
On the complement of this event, equations \eqref{eqn:gic_localization_lower} and \eqref{eqn:gic_underfit_competitor_upper} imply
\begin{align*}
	\inf_{\substack{\bm b_{D_B^c}=\bm0\\\|\bm b-\bvec^*(\tau)\|_2>n}}
	\mc Q_{\tau,n}(\bm b)
	&>\empmean \rho_\tau\{Y-\xvec^T\bm b^0_{a,B}(\tau)\},\\
	\inf_{\substack{\bm b_{D_B^c}=\bm0,\ \Cmat\bm b_{\Mset}=\bm\Psi(\tau)\\\|\bm b-\bvec^*(\tau)\|_2>n}}
	\mc Q_{\tau,n}(\bm b)
	&>\empmean \rho_\tau\{Y-\xvec^T\bm b^0_{0,B}(\tau)\}.
\end{align*}
The vector in \eqref{eqn:gic_underfit_competitor_feasible} has distance less than $n$ by \eqref{eqn:gic_underfit_competitor_upper}; the infimum can consequently be restricted to $\|\bm b-\bvec^*(\tau)\|_2\le n$.
\end{proof}

\begin{lemma}\label{lem:gic_empirical_process}
There is a constant $C_2<\infty$ such that
\[
P\left\{
\sup_{A\in\mathcal A_n}
\sup_{\tau\in\mc T}
\sup_{\substack{\bm b_{D_A^c}=\bm0\\
		\|\bm b-\bvec^*(\tau)\|_2\le n}}
\frac{|(\empmean -\E )\ell_{\tau,\bm b}|}
{\|\bm b-\bvec^*(\tau)\|_2+a_{\mathrm G,n}}
>
C_2a_{\mathrm G,n}
\right\}\to0.
\]
\end{lemma}

\begin{proof}
We use the finite-dimensional covering calculation in Theorem 2.7.11 of \citet{van_der_vaart_weak_1996}, followed by Bernstein's inequality. The quantile-regression function class in Lemma 24 of \citet{belloni_conditional_2019} is treated by a similar covering calculation. For $r>0$ and $A\in\mathcal A_n$, let
\[
\mathcal F(r,A)
=\{\ell_{\tau,\bm b}:\tau\in\mc T,\ \bm b_{D_A^c}=\bm0,
\ \|\bm b-\bvec^*(\tau)\|_2\le r\}.
\]
For $f=\ell_{\tau,\bm b}\in\mathcal F(r,A)$, equations \eqref{eqn:gic_excess_integral} and \eqref{eqn:gic_sparse_envelope}, together with Condition \ref{as:gic_design}, give
\begin{equation}\label{eqn:gic_ep_envelope}
	\|f\|_\infty\le\xi_{\mathrm G,n}r,
	\qquad \E f^2\le\kappa_+r^2.
\end{equation}
For $\bm\delta,\bm\gamma$ supported on $D_A\cup\Sset$ with norms at most $r$, Lemma \ref{lem:gic_excess_identity} gives
\begin{equation}\label{eqn:gic_ep_lipschitz}
	\begin{aligned}
		&\|\ell_{\tau,\bvec^*(\tau)+\bm\delta}
		-\ell_{t,\bvec^*(t)+\bm\gamma}\|_\infty\\
		&\qquad\le\xi_{\mathrm G,n}\|\bm\delta-\bm\gamma\|_2
		+(2f_-^{-1}+\xi_{\mathrm G,n}r)|\tau-t|.
	\end{aligned}
\end{equation}
Fix $r\in[a_{\mathrm G,n},2n]$. Choose a finite subset $\mathcal V_{r,A}$ of
$\{\bm\delta:\operatorname{supp}(\bm\delta)\subset D_A\cup\Sset,\ \|\bm\delta\|_2\le r\}$
whose Euclidean covering radius is $a_{\mathrm G,n}r/(4\xi_{\mathrm G,n})$. Since $|D_A\cup\Sset|\le d_{\mathrm G,n}$,
\begin{equation}\label{eqn:gic_ep_vector_cardinality}
	|\mathcal V_{r,A}|
	\le\left(1+\frac{8\xi_{\mathrm G,n}}{a_{\mathrm G,n}}\right)^{d_{\mathrm G,n}}.
\end{equation}
Choose a finite grid $\mathcal T_r\subset\mc T$ containing both endpoints
and with maximal spacing at most
\[
\frac{a_{\mathrm G,n}r}{4(2f_-^{-1}+\xi_{\mathrm G,n}r)}.
\]
Then
\begin{equation}\label{eqn:gic_ep_quantile_cardinality}
	|\mathcal T_r|
	\le2+\frac{4(\tau_U-\tau_L)(2f_-^{-1}+\xi_{\mathrm G,n}r)}{a_{\mathrm G,n}r}
	\le C\{1+a_{\mathrm G,n}^{-2}+\xi_{\mathrm G,n}/a_{\mathrm G,n}\}.
\end{equation}
Define
\[
\mathcal N_{r,A}
=\{\ell_{t,\bvec^*(t)+\bm\gamma}:t\in\mathcal T_r,
\ \bm\gamma\in\mathcal V_{r,A}\}.
\]
Equation \eqref{eqn:gic_ep_lipschitz} gives
\begin{equation}\label{eqn:gic_ep_approximation}
	\sup_{f\in\mathcal F(r,A)}\min_{g\in\mathcal N_{r,A}}\|f-g\|_\infty
	\le\frac{a_{\mathrm G,n}r}{2}.
\end{equation}
By \eqref{eqn:gic_an}, $a_{\mathrm G,n}^{-1}\le\sqrt n/A_0$ eventually. The definition of $L_{\mathrm G,n}$ in \eqref{eqn:gic_log_factor} consequently implies
\[
\log(1+a_{\mathrm G,n}^{-1})+\log(1+\xi_{\mathrm G,n})
+\log(p+1)+\log(d_{\mathrm G,n}+1)\le CL_{\mathrm G,n}.
\]
Moreover, we get
\[
|\mathcal A_n|
\le2^s\sum_{j=0}^{K_{\mathrm G,n}}\binom{p+1}{j}
\le2^s(K_{\mathrm G,n}+1)(p+1)^{K_{\mathrm G,n}}.
\]
Combining these bounds with \eqref{eqn:gic_ep_vector_cardinality}--\eqref{eqn:gic_ep_quantile_cardinality} yields
\begin{equation}\label{eqn:gic_ep_cardinality}
	\log\left(\sum_{A\in\mathcal A_n}|\mathcal N_{r,A}|\right)
	\le Cd_{\mathrm G,n}L_{\mathrm G,n}.
\end{equation}
For every $g\in\mathcal N_{r,A}$, Lemma \ref{lem:gic_excess_identity}
and Condition \ref{as:gic_design} give
\[
|g-\E g|\le2\xi_{\mathrm G,n}r,\qquad
\operatorname{Var}(g)\le\kappa_+r^2.
\]
Bernstein's inequality and \eqref{eqn:gic_ep_cardinality} imply, for every $u\ge1$,
\begin{align}
	&\Pr\left\{\max_{A\in\mathcal A_n}\max_{g\in\mathcal N_{r,A}}
	|(\empmean -\E )g|>ur a_{\mathrm G,n}\right\}\notag\\
	&\qquad\le2\exp\left\{
	Cd_{\mathrm G,n}L_{\mathrm G,n}
	-\frac{nu^2a_{\mathrm G,n}^2}
	{2\kappa_++4u\xi_{\mathrm G,n}a_{\mathrm G,n}/3}
	\right\}.
	\label{eqn:gic_ep_bernstein}
\end{align}
The two terms in \eqref{eqn:gic_an} yield, respectively,
\[
nu^2a_{\mathrm G,n}^2\ge u^2A_0^2d_{\mathrm G,n}L_{\mathrm G,n},
\qquad
\frac{nu a_{\mathrm G,n}}{\xi_{\mathrm G,n}}
\ge uA_0d_{\mathrm G,n}L_{\mathrm G,n}.
\]
Since $1/(x+y)\ge\frac12\min(1/x,1/y)$ for $x,y>0$, we have
\begin{align}
	\frac{nu^2a_{\mathrm G,n}^2}{2\kappa_++4u\xi_{\mathrm G,n}a_{\mathrm G,n}/3}
	&\ge c\min\left\{nu^2a_{\mathrm G,n}^2,
	\frac{nu a_{\mathrm G,n}}{\xi_{\mathrm G,n}}\right\}\notag\\
	&\ge c\min\{u^2A_0^2,uA_0\}d_{\mathrm G,n}L_{\mathrm G,n}.
	\label{eqn:gic_ep_exponent}
\end{align}
Choose a fixed $u\ge1$ so that $c\min(u^2A_0^2,uA_0)>C+4$. For
\[
0\le k\le\left\lceil\log_2(n/a_{\mathrm G,n})\right\rceil,
\qquad r=2^ka_{\mathrm G,n},
\]
equations \eqref{eqn:gic_ep_bernstein}--\eqref{eqn:gic_ep_exponent} bound each failure probability by $2e^{-4d_{\mathrm G,n}L_{\mathrm G,n}}$. Furthermore,
\[
1+\left\lceil\log_2(n/a_{\mathrm G,n})\right\rceil
\le CL_{\mathrm G,n}.
\]
It follows that, except on a measurable event with probability at most
$CL_{\mathrm G,n}e^{-4d_{\mathrm G,n}L_{\mathrm G,n}}=o(1)$,
\begin{equation}\label{eqn:gic_ep_grid_event}
	|(\empmean -\E )g|\le u2^ka_{\mathrm G,n}^2
	\quad\text{for every }A,k,\ g\in\mathcal N_{2^ka_{\mathrm G,n},A}.
\end{equation}
For $f\in\mathcal F(2^ka_{\mathrm G,n},A)$, equations \eqref{eqn:gic_ep_approximation} and \eqref{eqn:gic_ep_grid_event} give
\begin{align*}
	|(\empmean -\E )f|
	&\le\min_{g\in\mathcal N_{2^ka_{\mathrm G,n},A}}
	\{|(\empmean -\E )g|+\empmean |f-g|+\E |f-g|\}\\
	&\le (u+1)2^ka_{\mathrm G,n}^2.
\end{align*}
When $\|\bm b-\bvec^*(\tau)\|_2\le a_{\mathrm G,n}$, choose $k=0$.
Otherwise choose $k$ with $2^{k-1}a_{\mathrm G,n}<\|\bm b-\bvec^*(\tau)\|_2\le2^ka_{\mathrm G,n}$.
In both cases,
\[
|(\empmean -\E )\ell_{\tau,\bm b}|
\le2(u+1)a_{\mathrm G,n}(\|\bm b-\bvec^*(\tau)\|_2+a_{\mathrm G,n}).
\]
This inequality holds simultaneously for all parameters in the lemma outside the measurable event bounded after \eqref{eqn:gic_ep_grid_event}. Thus the stated probability conclusion holds with $C_2=2(u+1)$.
\end{proof}

\begin{lemma}\label{lem:gic_empirical_curvature}
There are constants $c_3,C_3>0$ such that
\[
P\left\{
\inf_{A\in\mathcal A_n}
\inf_{\tau\in\mc T}
\inf_{\substack{\bm b_{D_A^c}=\bm0\\
		\|\bm b-\bvec^*(\tau)\|_2\le n}}
\left\{
\empmean \ell_{\tau,\bm b}
-c_3(\|\bm b-\bvec^*(\tau)\|_2^2\wedge1)
\right\}
<
-C_3a_{\mathrm G,n}^2
\right\}\to0.
\]
\end{lemma}

\begin{proof}
Lemmas \ref{lem:gic_population_curvature} and \ref{lem:gic_empirical_process} imply, simultaneously for all parameters in the statement, with outer probability tending to one,
\begin{equation}\label{eqn:gic_empirical_curvature_initial}
	\empmean \ell_{\tau,\bm b}
	\ge c_1(r^2\wedge r)-C_2a_{\mathrm G,n}r-C_2a_{\mathrm G,n}^2,
	\qquad r=\|\bm b-\bvec^*(\tau)\|_2.
\end{equation}
Set $c_3=c_1/4$ and $C_3=C_2+C_2^2/(2c_1)$. For $r\le1$, Young's inequality gives
\[
C_2a_{\mathrm G,n}r
\le\frac{c_1}{2}r^2+\frac{C_2^2}{2c_1}a_{\mathrm G,n}^2.
\]
Consequently, \eqref{eqn:gic_empirical_curvature_initial} yields
\[
\empmean \ell_{\tau,\bm b}\ge\frac{c_1}{2}r^2-C_3a_{\mathrm G,n}^2
\ge c_3(r^2\wedge1)-C_3a_{\mathrm G,n}^2.
\]
For $1<r\le n$, $C_2a_{\mathrm G,n}\le c_1/2$ eventually, and
\begin{align*}
	\empmean \ell_{\tau,\bm b}
	&\ge(c_1-C_2a_{\mathrm G,n})r-C_2a_{\mathrm G,n}^2\\
	&\ge c_1/2-C_2a_{\mathrm G,n}^2
	\ge c_3(r^2\wedge1)-C_3a_{\mathrm G,n}^2.
\end{align*}
\end{proof}

\begin{lemma}\label{lem:gic_one_missing_gap}
If $s\ge1$, then, for every $\varepsilon\in(0,1)$, there is $c_\varepsilon>0$ such that
\[
\limsup_{n\to\infty}
P\left\{
\inf_{\substack{B\in\mathcal B_n\\
		\Sset\setminus B\ne\varnothing}}
\int_{\mc T}
\log\frac{\widehat L_{e,B}(\tau)}{\widehat L_{e,\Sset}(\tau)}\,d\tau
<
c_\varepsilon\mu_{\mathrm G,n}^2
\right\}
\le\varepsilon.
\]
\end{lemma}

\begin{proof}
Fix $\varepsilon\in(0,1)$. Lemma \ref{lem:gic_loss_scale} gives a constant $M_\varepsilon$ such that
\begin{equation}\label{eqn:gic_underfit_loss_event}
	\limsup_{n\to\infty}P\left\{
	\sup_{\tau\in\mc T}\mc Q_{\tau,n}\{\bvec^*(\tau)\}>M_\varepsilon
	\ \text{or}\ \inf_{\tau\in\mc T}\mc Q_{\tau,n}\{\bvec^*(\tau)\}<c_\rho
	\right\}\le\varepsilon/2.
\end{equation}
We will work on the complement of this event and of the events in Lemmas
\ref{lem:gic_localization}, \ref{lem:gic_empirical_process}, and \ref{lem:gic_empirical_curvature}.
The union of the excluded events has probability with limit superior at most $\varepsilon/2$.

First, \eqref{eqn:gic_population_expansion} and Conditions \ref{as:density conditions}(i) and \ref{as:gic_design} imply, for every vector with $|\operatorname{supp}\{\bm b-\bvec^*(\tau)\}|\le d_{\mathrm G,n}$,
\begin{equation}\label{eqn:gic_population_upper}
	0\le \E \ell_{\tau,\bm b}
	\le\frac{f_+}{2}\E\left[(\xvec^T\{\bm b-\bvec^*(\tau)\})^2\right]
	\le\frac{f_+\kappa_+}{2}\|\bm b-\bvec^*(\tau)\|_2^2.
\end{equation}
For $B=\Sset$, the feasible vector in \eqref{eqn:gic_underfit_competitor_feasible} is
$\bm b^0_{e,\Sset}(\tau)=\bvec^*(\tau)+\bm v_e(\tau)/\sqrt n$. Equations \eqref{eqn:gic_feasible_shift_bound} and \eqref{eqn:gic_an} give
\begin{equation}\label{eqn:gic_underfit_true_radius}
	\sup_\tau \|\bm b^0_{e,\Sset}(\tau)-\bvec^*(\tau)\|_2
	\le C\sqrt{\frac{m-r+s}{n}}
	\le C\sqrt{\frac{d_{\mathrm G,n}}{n}}
	\le\frac{C a_{\mathrm G,n}}{A_0\sqrt{L_{\mathrm G,n}}}.
\end{equation}
By \eqref{eqn:gic_population_upper}, \eqref{eqn:gic_underfit_true_radius}, and Lemma \ref{lem:gic_empirical_process}, we get
\begin{align}
	\sup_\tau \empmean \ell_{\tau,\bm b^0_{e,\Sset}(\tau)}
	&\le C\frac{d_{\mathrm G,n}}{n}
	+C_2a_{\mathrm G,n}\left(C\sqrt{\frac{d_{\mathrm G,n}}{n}}+a_{\mathrm G,n}\right)\notag\\
	&\le C a_{\mathrm G,n}^2.
	\label{eqn:gic_underfit_true_loss}
\end{align}
The feasible comparison and Lemmas \ref{lem:gic_localization} and \ref{lem:gic_empirical_curvature} now yield
\begin{equation}\label{eqn:gic_underfit_true_loss_bounds}
	\mc Q_{\tau,n}\{\bvec^*(\tau)\}-C_3a_{\mathrm G,n}^2
	\le\widehat L_{e,\Sset}(\tau)
	\le \mc Q_{\tau,n}\{\bvec^*(\tau)\}+C a_{\mathrm G,n}^2.
\end{equation}
For $\Sset\setminus B=\{j\}$, the two terms in
$\bm b^0_{e,B}(\tau)-\bvec^*(\tau)=\bm v_e(\tau)/\sqrt n-\beta_j^*(\tau)\bm e_j$
have disjoint supports. Therefore
\begin{equation}\label{eqn:gic_underfit_missing_radius}
	\|\bm b^0_{e,B}(\tau)-\bvec^*(\tau)\|_2^2
	=\frac{\|\bm v_e(\tau)\|_2^2}{n}+\beta_j^*(\tau)^2
	\le C\frac{d_{\mathrm G,n}}{n}+B_\beta^2.
\end{equation}
Equations \eqref{eqn:gic_population_upper}, \eqref{eqn:gic_underfit_missing_radius}, and Lemma \ref{lem:gic_empirical_process} give
\begin{align}
	\widehat L_{e,B}(\tau)
	&\le \mc Q_{\tau,n}\{\bvec^*(\tau)\}+\empmean \ell_{\tau,\bm b^0_{e,B}(\tau)}\notag\\
	&\le M_\varepsilon
	+\frac{f_+\kappa_+}{2}\left(B_\beta^2+C\frac{d_{\mathrm G,n}}{n}\right)
	+C_2a_{\mathrm G,n}\left(B_\beta+C\sqrt{\frac{d_{\mathrm G,n}}{n}}+a_{\mathrm G,n}\right).
	\label{eqn:gic_underfit_missing_upper}
\end{align}
Since $a_{\mathrm G,n}\to0$ and $d_{\mathrm G,n}/n\to0$, equations
\eqref{eqn:gic_underfit_loss_event}, \eqref{eqn:gic_underfit_true_loss_bounds},
and \eqref{eqn:gic_underfit_missing_upper}, together with Lemma \ref{lem:gic_empirical_curvature}, imply that a constant $M'_{\varepsilon}<\infty$ can be chosen so that, for all sufficiently large $n$,
\begin{equation}\label{eqn:gic_underfit_loss_interval}
	\frac{c_\rho}{2}\le\widehat L_{e,B}(\tau),\widehat L_{e,\Sset}(\tau)
	\le M'_{\varepsilon}
	\quad\text{for every }B\in\mathcal B_n,\ \tau\in\mc T.
\end{equation}
Every vector feasible for $B$ has $b_j=0$, and hence
\[
\|\bm b-\bvec^*(\tau)\|_2\ge|\beta_j^*(\tau)|,\qquad
\|\bm b-\bvec^*(\tau)\|_2^2\wedge1
\ge\frac{\beta_j^*(\tau)^2}{1\vee B_\beta^2}.
\]
Let $c_5=c_3/(1\vee B_\beta^2)$. Lemmas \ref{lem:gic_localization} and \ref{lem:gic_empirical_curvature}, and \eqref{eqn:gic_underfit_true_loss_bounds}, give
\begin{equation}\label{eqn:gic_one_missing_pointwise}
	\widehat L_{e,B}(\tau)-\widehat L_{e,\Sset}(\tau)
	\ge c_5\beta_j^*(\tau)^2-Ca_{\mathrm G,n}^2.
\end{equation}
In particular, we get
\begin{equation}\label{eqn:gic_one_missing_negative}
	\{\widehat L_{e,\Sset}(\tau)-\widehat L_{e,B}(\tau)\}_+
	\le Ca_{\mathrm G,n}^2.
\end{equation}
Integrating \eqref{eqn:gic_one_missing_pointwise}, applying Cauchy--Schwarz, and using \eqref{eqn:gic_mu}, we obtain
\begin{align}
	\int_{\mc T}\{\widehat L_{e,B}(\tau)-\widehat L_{e,\Sset}(\tau)\}\,d\tau
	&\ge c_5\int_{\mc T}\beta_j^*(\tau)^2\,d\tau
	-C(\tau_U-\tau_L)a_{\mathrm G,n}^2\notag\\
	&\ge\frac{c_5}{\tau_U-\tau_L}
	\left(\int_{\mc T}|\beta_j^*(\tau)|\,d\tau\right)^2
	-C(\tau_U-\tau_L)a_{\mathrm G,n}^2\notag\\
	&\ge\frac{c_5\mu_{\mathrm G,n}^2}{\tau_U-\tau_L}
	-C(\tau_U-\tau_L)a_{\mathrm G,n}^2.
	\label{eqn:gic_underfit_integrated_difference}
\end{align}
For $u,v\in[c_\rho/2,M'_{\varepsilon}]$, direct integration of $1/t$ gives
\begin{equation}\label{eqn:gic_underfit_log_difference}
	\log u-\log v
	\ge\frac{u-v}{M'_{\varepsilon}}
	-\left(\frac{2}{c_\rho}-\frac{1}{M'_{\varepsilon}}\right)(v-u)_+.
\end{equation}
Using \eqref{eqn:gic_underfit_loss_interval} in \eqref{eqn:gic_underfit_log_difference}, then integrating and applying
\eqref{eqn:gic_one_missing_negative} and \eqref{eqn:gic_underfit_integrated_difference}, yields
\begin{align*}
	\int_{\mc T}\log\frac{\widehat L_{e,B}(\tau)}{\widehat L_{e,\Sset}(\tau)}\,d\tau
	&\ge\frac{c_5\mu_{\mathrm G,n}^2}{(\tau_U-\tau_L)M'_{\varepsilon}}
	-C_{\varepsilon}a_{\mathrm G,n}^2\\
	&\ge\frac{c_5\mu_{\mathrm G,n}^2}{2(\tau_U-\tau_L)M'_{\varepsilon}}
\end{align*}
for all sufficiently large $n$, since $a_{\mathrm G,n}=o(\mu_{\mathrm G,n})$ by Condition \ref{as:gic_underfit_rates}.
The last inequality holds simultaneously for every $B$ in the statement on an event whose complementary probability has limit superior at most $\varepsilon/2$.
Thus the assertion holds with $c_\varepsilon=c_5/\{2(\tau_U-\tau_L)M'_{\varepsilon}\}$.
\end{proof}

\begin{theorem}\label{thm:gic_underfit}
Under Conditions \ref{as:iid array}, \ref{as:density conditions}, and \ref{as:gic_design}--\ref{as:gic_underfit_rates}, and, when $e=0$, either \eqref{hyp:linear hypothesis} or Condition \ref{as:local alts},
\[
P\left\{
\inf_{A\in\mathcal U_n}
\operatorname{GIC}^{\circ}_e(A)
\le
\operatorname{GIC}^{\circ}_e(\Sset)
\right\}
\to0.
\]
\end{theorem}

\begin{proof}
As in the first comparison in the proof of Lemma 7.9 of \citet{zheng2015globally}, we use inclusion of the candidate parameter sets. If $s=0$, then $\mathcal U_n=\varnothing$. For $s\ge1$, fix $\varepsilon\in(0,1)$ and choose, for each $A\in\mathcal U_n$, $j_A=\min(\Sset\setminus A)$. By the definition of $A^{+j_A}$, one has
\[
A\subset A^{+j_A},\qquad
\Sset\setminus A^{+j_A}=\{j_A\},\qquad
A^{+j_A}\in\mathcal B_n.
\]
The equality constraint for $e=0$ is the same for both candidates. Thus \eqref{eqn:gic_aux_loss_full}--\eqref{eqn:gic_aux_loss_reduced} imply
\begin{equation}\label{eqn:gic_underfit_inclusion}
	\widehat L_{e,A}(\tau)\ge\widehat L_{e,A^{+j_A}}(\tau)
	\quad\text{for all }\tau\in\mc T.
\end{equation}
Lemma \ref{lem:gic_one_missing_gap} and \eqref{eqn:gic_underfit_inclusion} give
\begin{align*}
	&\limsup_{n\to\infty}P\left\{
	\inf_{A\in\mathcal U_n}
	\int_{\mc T}\log\frac{\widehat L_{e,A}(\tau)}{\widehat L_{e,\Sset}(\tau)}\,d\tau
	<c_\varepsilon\mu_{\mathrm G,n}^2
	\right\}\le\varepsilon.
\end{align*}
On the complement of this event, \eqref{eqn:gic_auxiliary} and $|A|-s\ge-s$ imply
\begin{align*}
	\inf_{A\in\mathcal U_n}
	\{\operatorname{GIC}^{\circ}_e(A)-\operatorname{GIC}^{\circ}_e(\Sset)\}
	&\ge c_\varepsilon\mu_{\mathrm G,n}^2-s\phi_n\\
	&=\mu_{\mathrm G,n}^2\left(c_\varepsilon-\frac{s\phi_n}{\mu_{\mathrm G,n}^2}\right)>0
\end{align*}
for all sufficiently large $n$, by Condition \ref{as:gic_underfit_rates}. Therefore the limit superior of the probability in the theorem is at most $\varepsilon$ for every $\varepsilon>0$, and equals zero.
\end{proof}

\subsubsection{Overfitted candidates}\label{app:gic_overfit}

Fix $e\in\{0,a\}$. For $A\in\mathcal O_n\cup\{\Sset\}$, define
\[
\mathcal H_A=
\begin{cases}
\{\bm h\in\reals^{p+1}:\bm h_{D_A^c}=\bm0\},&e=a,\\
\{\bm h\in\reals^{p+1}:\bm h_{D_A^c}=\bm0,\ \Cmat\bm h_{\Mset}=\bm0\},&e=0.
\end{cases}
\]
Since $\operatorname{rank}(\Cmat)=r$ and $\Mset\subset D_A$,
\[
\dim(\mathcal H_A)=|D_A|-r1(e=0),
\qquad
\dim(\mathcal H_A)-\dim(\mathcal H_{\Sset})=|A|-s.
\]
For vectors supported in $D_A$, let
\begin{equation}\label{eqn:gic_jacobian_form}
J_\tau(\bm h,\bm g)
=\E\left[
f_{y\mid\xvec}\{\xvec^T\bvec^*(\tau)\mid\xvec\}
(\xvec^T\bm h)(\xvec^T\bm g)
\right].
\end{equation}
Conditions \ref{as:density conditions} and \ref{as:gic_design} imply
\begin{equation}\label{eqn:gic_jacobian_equivalence}
f_-\kappa_-\|\bm h\|_2^2
\le J_\tau(\bm h,\bm h)
\le f_+\kappa_+\|\bm h\|_2^2
\qquad (|\operatorname{supp}(\bm h)|\le d_{\mathrm G,n}).
\end{equation}
For $A\in\mathcal O_n$, define
\[
\mathcal K_A(\tau)
=\{\bm g\in\mathcal H_A:
J_\tau(\bm u,\bm g)=0\text{ for every }\bm u\in\mathcal H_{\Sset}\}.
\]
Then
\begin{equation}\label{eqn:gic_orthogonal_sum}
\mathcal H_A=\mathcal H_{\Sset}\oplus_{J_\tau}\mathcal K_A(\tau),
\qquad
\dim\{\mathcal K_A(\tau)\}=|A|-s.
\end{equation}

For all $\bm h,\bm g$ supported in $D_A$, define
\begin{align}
M_A(\bm h,\tau)
&=n\empmean \left[
\rho_\tau\{Y-\xvec^T\bvec^*(\tau)-n^{-1/2}\xvec^T\bm h\}
-\rho_\tau\{Y-\xvec^T\bvec^*(\tau)\}
\right],\label{eqn:gic_local_criterion}\\
G_n(\bm g,\tau)
&=\frac1{\sqrt n}\sum_{i=1}^n
\xvec_i^T\bm g\,[\tau-1\{Y_i\le\xvec_i^T\bvec^*(\tau)\}],
\label{eqn:gic_score_process}\\
Z_A(\bm h,\tau;\bm g)
&=\frac1{\sqrt n}\sum_{i=1}^n\xvec_i^T\bm g
\left[
1\{Y_i\le\xvec_i^T\bvec^*(\tau)+n^{-1/2}\xvec_i^T\bm h\}
-1\{Y_i\le\xvec_i^T\bvec^*(\tau)\}
\right],\label{eqn:gic_score_perturbation}\\
\overline Z_A(\bm h,\tau;\bm g)
&=\sqrt n\,\E\left[\xvec^T\bm g
\left[
1\{Y\le\xvec^T\bvec^*(\tau)+n^{-1/2}\xvec^T\bm h\}
-1\{Y\le\xvec^T\bvec^*(\tau)\}
\right]\right].\label{eqn:gic_score_perturbation_mean}
\end{align}
The vector $\bm v_e(\tau)$ in \eqref{eqn:gic_feasible_shift} satisfies
\[
\{\bm h:\bm h_{D_A^c}=\bm0,
\ \Cmat[\bvec^*_{\Mset}(\tau)+n^{-1/2}\bm h_{\Mset}]=\bm\Psi(\tau)\}
=\bm v_0(\tau)+\mathcal H_A
\qquad(e=0).
\]
Let $\widehat{\bm h}_A(\tau)$ be any minimizer of
$M_A(\cdot,\tau)$ over $\bm v_e(\tau)+\mathcal H_A$.
The feasible set is nonempty and affine. The criterion is a finite-valued
polyhedral convex function bounded below by $-n\mc Q_{\tau,n}\{\bvec^*(\tau)\}$; therefore the
minimum is attained. By \eqref{eqn:gic_aux_loss_full} and
\eqref{eqn:gic_aux_loss_reduced},
\begin{equation}\label{eqn:gic_local_loss_identity}
\widehat L_{e,A}(\tau)
=\mc Q_{\tau,n}\{\bvec^*(\tau)\}+n^{-1}M_A\{\widehat{\bm h}_A(\tau),\tau\}.
\end{equation}

For fixed $M\ge1$, let
\[
\mathcal L_A(M,\tau)
=\{\bm h:\bm h_{D_A^c}=\bm0,\
J_\tau(\bm h,\bm h)\le M|D_A|L_{\mathrm G,n}\}.
\]
This set imposes no restriction on $\Cmat\bm h_{\Mset}$.
Taylor's expansion of the conditional distribution function gives
\begin{align}
&\left|\overline Z_A(\bm h,\tau;\bm g)-J_\tau(\bm h,\bm g)\right|
\notag\\
&\quad\le\frac{L_f}{2\sqrt n}\E\{|\xvec^T\bm g|(\xvec^T\bm h)^2\}
\le\frac{L_f\xi_{\mathrm G,n}\|\bm h\|_2}{2\sqrt n}
\{\E(\xvec^T\bm g)^2\}^{1/2}\{\E(\xvec^T\bm h)^2\}^{1/2}
\notag\\
&\quad\le\frac{L_f\kappa_+\xi_{\mathrm G,n}}{2\sqrt n}
\|\bm h\|_2^2\|\bm g\|_2.
\label{eqn:gic_score_mean_expansion}
\end{align}
Here \eqref{eqn:gic_sparse_envelope} bounds one factor
$|\xvec^T\bm h|$, and Condition \ref{as:gic_design} bounds the two second moments.

\begin{lemma}\label{lem:gic_overfit_process_bounds}
Under Conditions \ref{as:iid array}, \ref{as:density conditions}, and
\ref{as:gic_design}--\ref{as:gic_overfit_rates}, for every fixed $M\ge1$
there is a constant $C<\infty$ such that
\begin{align}
	P\left\{
	\exists A\in\mathcal O_n\cup\{\Sset\},\ \tau\in\mc T:
	\sup_{\substack{\bm g_{D_A^c}=\bm0\\J_\tau(\bm g,\bm g)\le1}}
	|G_n(\bm g,\tau)|>C\sqrt{|D_A|L_{\mathrm G,n}}
	\right\}&\to0,\label{eqn:gic_score_size}\\
	P\left\{
	\sup_{A\in\mathcal O_n\cup\{\Sset\}}
	\sup_{\tau\in\mc T}
	\sup_{\bm h\in\mathcal L_A(M,\tau)}
	\sup_{\substack{\bm g_{D_A^c}=\bm0\\J_\tau(\bm g,\bm g)\le1}}
	|Z_A(\bm h,\tau;\bm g)-\overline Z_A(\bm h,\tau;\bm g)|
	>\epsilon\sqrt{L_{\mathrm G,n}}
	\right\}&\to0\label{eqn:gic_score_modulus}
\end{align}
for every $\epsilon>0$. In addition,
\begin{equation}\label{eqn:gic_extra_score}
	P\left\{
	\exists A\in\mathcal O_n,\ \tau\in\mc T:
	\sup_{\substack{\bm g\in\mathcal K_A(\tau)\\J_\tau(\bm g,\bm g)\le1}}
	|G_n(\bm g,\tau)|>C\sqrt{(|A|-s)L_{\mathrm G,n}}
	\right\}\to0.
\end{equation}
The constants in \eqref{eqn:gic_score_size} and
\eqref{eqn:gic_extra_score} do not depend on $M$.
\end{lemma}

\begin{proof}
We use the products and indicator differences in Lemmas 23--24 and 26
of \citet{belloni_conditional_2019}. We construct pointwise brackets
and apply Bernstein's inequality.
We note that all constants
in this proof are independent of $n,A,\tau$, and the integer used to
index candidate sizes. Only constants used to prove
\eqref{eqn:gic_score_modulus} may depend on the fixed $M$.

Fix $1\le k\le d_{\mathrm G,n}$ and consider candidates with $|D_A|=k$.
For $\bm g_{D_A^c}=\bm0$ and $J_\tau(\bm g,\bm g)\le1$, write
\[
f_{\bm g,\tau}(Y,\xvec)
=(\xvec^T\bm g)[\tau-1\{Y\le\xvec^T\bvec^*(\tau)\}].
\]
By \eqref{eqn:conditional quantile linear expr},
$\E f_{\bm g,\tau}=0$. Equations \eqref{eqn:gic_sparse_envelope} and
\eqref{eqn:gic_jacobian_equivalence} imply
\begin{equation}\label{eqn:gic_score_envelope_variance}
	\|\bm g\|_2\le(f_-\kappa_-)^{-1/2},
	\qquad
	|f_{\bm g,\tau}|\le C\xi_{\mathrm G,n},
	\qquad
	\E f_{\bm g,\tau}^2\le C.
\end{equation}
For $0<\delta\le1$, select finite sets of coefficient vectors and quantiles
so that each admissible pair has an approximation $(\bm g_0,t)$ satisfying
\[
\|\bm g-\bm g_0\|_2\le\delta,
\qquad |\tau-t|\le\delta^2,
\qquad \|\bm g_0\|_2\le(f_-\kappa_-)^{-1/2}+1.
\]
Lemma \ref{lem:gic_quantile_path} gives the pointwise bound
\begin{align}
	|f_{\bm g,\tau}-f_{\bm g_0,t}|
	&\le\xi_{\mathrm G,n}\delta+C\xi_{\mathrm G,n}\delta^2+C\xi_{\mathrm G,n}
	1\{|Y-\xvec^T\bvec^*(t)|\le f_-^{-1}\delta^2\}
	=:b_{\bm g_0,t,\delta}.
	\label{eqn:gic_score_bracket_width}
\end{align}
The coefficient $C$ is chosen to bound $\|\bm g_0\|_2$ in
\eqref{eqn:gic_score_envelope_variance}. Condition
\ref{as:density conditions} implies
\[
\Pr\{|Y-\xvec^T\bvec^*(t)|\le f_-^{-1}\delta^2\}
\le2f_+f_-^{-1}\delta^2,
\qquad
\|b_{\bm g_0,t,\delta}\|_{L_2}\le C\xi_{\mathrm G,n}\delta.
\]
Thus $[f_{\bm g_0,t}-b_{\bm g_0,t,\delta},
f_{\bm g_0,t}+b_{\bm g_0,t,\delta}]$ is a pointwise bracket.
The covering bound used in Theorem 2.7.11 of
\citet{van_der_vaart_weak_1996} yields at most
$(1+C/\delta)^k$ coefficient vectors. The number of quantiles is at most
$1+\delta^{-2}$, and the number of candidate supports with $|D_A|=k$ is at most
$(p+1)^k$. Therefore, for the union $\mathcal F$ of these score functions, we have
\begin{align}
	\log N_{[]}\{C\xi_{\mathrm G,n}\delta,\mathcal F,L_2\}
	&\le k\log(p+1)+k\log(1+C/\delta)+\log(1+\delta^{-2})
	\notag\\
	&\le Ck\{L_{\mathrm G,n}+\log(C/\delta)\}.
	\label{eqn:gic_score_brackets}
\end{align}

We compute the finite-bracket probability bound. Suppose a
class $\mathcal F$ has envelope $C\xi_{\mathrm G,n}$,
$\sup_{f\in\mathcal F}\E f^2\le\sigma^2$, and satisfies
\eqref{eqn:gic_score_brackets}. For
$c/\sqrt n\le t\le C\xi_{\mathrm G,n}\sqrt n$, take
$\delta=t/(C_1\xi_{\mathrm G,n}\sqrt n)$, where $C_1$ is sufficiently
large that every bracket has $L_2$ width at most $t/(8\sqrt n)$ and
$\delta\le1$. Since
\[
\delta^{-1}\le C_2\xi_{\mathrm G,n}n,
\qquad
\log(C/\delta)\le C_3L_{\mathrm G,n},
\]
the total number of endpoints is at most $\exp(C_4kL_{\mathrm G,n})$.
For $l\le f\le u$, one has
\begin{align*}
	\sqrt n(\empmean -\E )f
	&\le\sqrt n(\empmean -\E )u+\sqrt nP(u-l),\\
	\sqrt n(\empmean -\E )f
	&\ge\sqrt n(\empmean -\E )l-\sqrt nP(u-l),\\
	\sqrt nP(u-l)&\le\sqrt n\|u-l\|_{L_2}\le t/8.
\end{align*}
For every nonempty bracket, choose any function $f$ inside it. Its
endpoints $v\in\{l,u\}$ satisfy
\[
\E v^2\le2\E f^2+2\E (u-l)^2\le2\sigma^2+t^2/(32n),
\qquad |v-\E v|\le C\xi_{\mathrm G,n}.
\]
Bernstein's inequality applied to each endpoint gives
\begin{equation}\label{eqn:gic_bracket_bernstein}
	P^*\left\{\sup_{f\in\mathcal F}|\sqrt n(\empmean -\E )f|>t\right\}
	\le2\exp\left\{
	C_4kL_{\mathrm G,n}
	-\frac{ct^2}{\sigma^2+t^2/n+\xi_{\mathrm G,n}t/\sqrt n}
	\right\}.
\end{equation}
All brackets constructed here have pointwise endpoints bounded by
$C\xi_{\mathrm G,n}$, as required in this calculation.

Apply \eqref{eqn:gic_bracket_bernstein} to
\eqref{eqn:gic_score_envelope_variance} with
$t=C_0\sqrt{kL_{\mathrm G,n}}$. By \eqref{eqn:gic_sharp_consequences},
\[
\sup_{k\le d_{\mathrm G,n}}
\left\{\frac{t^2}{n}+\frac{\xi_{\mathrm G,n}t}{\sqrt n}\right\}
\le C_0^2\frac{d_{\mathrm G,n}L_{\mathrm G,n}}n
+C_0\xi_{\mathrm G,n}\sqrt{\frac{d_{\mathrm G,n}L_{\mathrm G,n}}n}
\longrightarrow0.
\]
Consequently the exponent in \eqref{eqn:gic_bracket_bernstein} is at most
$(C_4-c_1C_0^2)kL_{\mathrm G,n}$ for all sufficiently large $n$.
Choose $C_0$ so that $c_1C_0^2-C_4\ge4$. Then
\[
\sum_{k=1}^{d_{\mathrm G,n}}2e^{-4kL_{\mathrm G,n}}
\le\frac{2e^{-4L_{\mathrm G,n}}}{1-e^{-4L_{\mathrm G,n}}}\to0,
\]
which proves \eqref{eqn:gic_score_size}.

For \eqref{eqn:gic_score_modulus}, we use a similar argument used in Lemma 32 of \citet{belloni_conditional_2019}.
For $|D_A|=k$, $\bm h\in\mathcal L_A(M,\tau)$, and admissible $\bm g$, let
\[
f_{\bm h,\bm g,\tau}
=(\xvec^T\bm g)
\left[1\{Y\le\xvec^T\bvec^*(\tau)+n^{-1/2}\xvec^T\bm h\}
-1\{Y\le\xvec^T\bvec^*(\tau)\}\right].
\]
Its absolute value is at most $C\xi_{\mathrm G,n}$. Conditioning on
$\xvec$, then using \eqref{eqn:gic_sparse_envelope},
\eqref{eqn:gic_jacobian_equivalence}, and Condition \ref{as:gic_design}, gives
\begin{align}
	\E f_{\bm h,\bm g,\tau}^2
	&\le\frac{f_+}{\sqrt n}\E\{(\xvec^T\bm g)^2|\xvec^T\bm h|\}
	\notag\\
	&\le\frac{f_+\xi_{\mathrm G,n}\|\bm h\|_2}{\sqrt n}
	\E(\xvec^T\bm g)^2
	\le C\xi_{\mathrm G,n}\sqrt{\frac{kL_{\mathrm G,n}}n}.
	\label{eqn:gic_perturbation_variance}
\end{align}
Choose approximations supported in the same $D_A$ with
\[
\|\bm g-\bm g_0\|_2\le\delta,
\qquad
\|\bm h-\bm h_0\|_2\le\frac{\delta^2\sqrt n}{\xi_{\mathrm G,n}},
\qquad
|\tau-t|\le\delta^2.
\]
Both $\bm h$ and $\bm h_0$ may be chosen in the Euclidean ball of radius
$C\sqrt{kL_{\mathrm G,n}}$ given by \eqref{eqn:gic_jacobian_equivalence}.
With $c_2=1+f_-^{-1}$, the pointwise difference is bounded by
\begin{align}
	|f_{\bm h,\bm g,\tau}-f_{\bm h_0,\bm g_0,t}|
	&\le\xi_{\mathrm G,n}\delta
	+C\xi_{\mathrm G,n}
	1\{|Y-\xvec^T\bvec^*(t)-n^{-1/2}\xvec^T\bm h_0|\le c_2\delta^2\}
	\notag\\
	&\quad+C\xi_{\mathrm G,n}
	1\{|Y-\xvec^T\bvec^*(t)|\le c_2\delta^2\}.
	\label{eqn:gic_perturbation_bracket_width}
\end{align}
Each indicator on the right has expectation at most $2f_+c_2\delta^2$.
The $L_2$ norm of the right side is therefore at most
$C\xi_{\mathrm G,n}\delta$. Adding and subtracting this bound from
$f_{\bm h_0,\bm g_0,t}$ produces pointwise brackets. Their number is at most
\[
(p+1)^k(1+C/\delta)^k
\left(1+\frac{C\xi_{\mathrm G,n}\sqrt{kL_{\mathrm G,n}}}{\delta^2\sqrt n}\right)^k
(1+\delta^{-2}).
\]
Since $k\le d_{\mathrm G,n}\le e^{L_{\mathrm G,n}}$,
$\xi_{\mathrm G,n}\le e^{L_{\mathrm G,n}}$, and
$L_{\mathrm G,n}\le e^{L_{\mathrm G,n}}$, its logarithm is at most
$Ck\{L_{\mathrm G,n}+\log(C/\delta)\}$. Hence
\eqref{eqn:gic_bracket_bernstein} applies with
$\sigma^2=C\xi_{\mathrm G,n}\sqrt{kL_{\mathrm G,n}/n}$.

Set
\begin{equation}\label{eqn:gic_perturbation_threshold}
	t=C_0\left[
	\left\{\frac{\xi_{\mathrm G,n}(kL_{\mathrm G,n})^{3/2}}{\sqrt n}\right\}^{1/2}
	+\frac{\xi_{\mathrm G,n}kL_{\mathrm G,n}}{\sqrt n}
	\right].
\end{equation}
Condition \ref{as:gic_design} and \eqref{eqn:gic_sparse_envelope} imply
$\xi_{\mathrm G,n}\ge\sqrt{\kappa_-}$. Thus, we get
\[
t\ge C_0\sqrt{\kappa_-}\,\frac{kL_{\mathrm G,n}}{\sqrt n}
\ge\frac{C_0\sqrt{\kappa_-}\log3}{\sqrt n}.
\]
Equation \eqref{eqn:gic_sharp_growth} also gives
$d_{\mathrm G,n}L_{\mathrm G,n}/n\to0$. Consequently, we obtain
\begin{align*}
	\frac{t}{\xi_{\mathrm G,n}\sqrt n}
	&=C_0\left[
	\xi_{\mathrm G,n}^{-1/2}
	\left(\frac{kL_{\mathrm G,n}}n\right)^{3/4}
	+\frac{kL_{\mathrm G,n}}n\right]\\
	&\le C_0\left[
	\kappa_-^{-1/4}
	\left(\frac{d_{\mathrm G,n}L_{\mathrm G,n}}n\right)^{3/4}
	+\frac{d_{\mathrm G,n}L_{\mathrm G,n}}n\right]\to0.
\end{align*}
To evaluate the exponent in \eqref{eqn:gic_bracket_bernstein},
\eqref{eqn:gic_perturbation_threshold} implies separately
\begin{align*}
	C\xi_{\mathrm G,n}\sqrt{\frac{kL_{\mathrm G,n}}n}
	&\le\frac{C}{C_0^2}\frac{t^2}{kL_{\mathrm G,n}},\\
	\frac{\xi_{\mathrm G,n}t}{\sqrt n}
	&\le\frac1{C_0}\frac{t^2}{kL_{\mathrm G,n}},\\
	\frac{t^2}{n}
	&\le\frac{d_{\mathrm G,n}L_{\mathrm G,n}}n
	\frac{t^2}{kL_{\mathrm G,n}}.
\end{align*}
It follows that
\[
\frac{ct^2}{
	C\xi_{\mathrm G,n}\sqrt{kL_{\mathrm G,n}/n}+t^2/n+\xi_{\mathrm G,n}t/\sqrt n}
\ge
\frac{ckL_{\mathrm G,n}}{
	C/C_0^2+1/C_0+d_{\mathrm G,n}L_{\mathrm G,n}/n}.
\]
Choose $C_0$ sufficiently large, then $n$ sufficiently large, so that
the right side is at least $(C_4+4)kL_{\mathrm G,n}$. The sum of the
resulting probabilities is bounded by
$2e^{-4L_{\mathrm G,n}}/(1-e^{-4L_{\mathrm G,n}})$.
After division of \eqref{eqn:gic_perturbation_threshold} by
$\sqrt{L_{\mathrm G,n}}$, its maximum over $k\le d_{\mathrm G,n}$ is at most
\[
C_0\left[
\left\{\frac{\xi_{\mathrm G,n}^2d_{\mathrm G,n}^3L_{\mathrm G,n}}n\right\}^{1/4}
+\left\{\frac{\xi_{\mathrm G,n}^2d_{\mathrm G,n}^2L_{\mathrm G,n}}n\right\}^{1/2}
\right]\to0
\]
by \eqref{eqn:gic_sharp_growth}. This proves \eqref{eqn:gic_score_modulus}.

For \eqref{eqn:gic_extra_score}, we use an analogous argument used in Lemmas 18--19 of \citet{belloni_conditional_2019}.
Lemma \ref{lem:gic_quantile_path} and Condition \ref{as:density conditions}
give, for vectors supported in a common $D_A$,
\begin{align}
	|J_\tau(\bm h,\bm g)-J_t(\bm h,\bm g)|
	&\le L_ff_-^{-1}|\tau-t|\E|\xvec^T\bm h\,\xvec^T\bm g|
	\notag\\
	&\le L_ff_-^{-1}\kappa_+|\tau-t|\|\bm h\|_2\|\bm g\|_2.
	\label{eqn:gic_jacobian_lipschitz}
\end{align}
Choose a matrix $\bm U$ whose columns form a Euclidean orthonormal basis
of $\mathcal H_{\Sset}$. For each $A\in\mathcal O_n$, let $\bm V_A$ have
columns $\bm e_j$, $j\in A\setminus\Sset$, in increasing order. Thus
\[
\bm U^T\bm U=I_{\dim(\mathcal H_{\Sset})},\qquad \bm V_A^T\bm V_A=I_{|A|-s},
\qquad \bm U^T\bm V_A=0,
\qquad
\mathcal H_A=\operatorname{range}[\bm U,\bm V_A].
\]
Write $\bm U_i$ and $(\bm V_A)_j$ for the respective columns.
For this proof, define
\[
[\bm B(\tau)]_{ij}=J_\tau(\bm U_i,\bm U_j),
\qquad
[\bm R_A(\tau)]_{ij}=J_\tau(\bm U_i,(\bm V_A)_j),
\qquad
\bm T_A(\tau)=\bm V_A-\bm U\bm B(\tau)^{-1}\bm R_A(\tau).
\]
Equations \eqref{eqn:gic_jacobian_equivalence} and
\eqref{eqn:gic_jacobian_lipschitz} give
\begin{align}
	\lambda_{\min}\{\bm B(\tau)\}&\ge f_-\kappa_-,
	&\|\bm R_A(\tau)\|_{\mathrm{op}}&\le f_+\kappa_+,
	\notag\\
	\|\bm B(\tau)-\bm B(t)\|_{\mathrm{op}}
	+\|\bm R_A(\tau)-\bm R_A(t)\|_{\mathrm{op}}
	&\le C|\tau-t|.
	\label{eqn:gic_orthogonal_matrices}
\end{align}
For every $\bm z\in\reals^{|A|-s}$,
\[
J_\tau\{\bm U\bm a,\bm T_A(\tau)\bm z\}
=\bm a^T[\bm R_A(\tau)-\bm B(\tau)\bm B(\tau)^{-1}\bm R_A(\tau)]\bm z=0,
\qquad
\bm V_A^T\bm T_A(\tau)\bm z=\bm z.
\]
Conversely, if $\bm U\bm a+\bm V_A\bm z\in\mathcal K_A(\tau)$, then
$\bm B(\tau)\bm a+\bm R_A(\tau)\bm z=0$, so
$\bm a=-\bm B(\tau)^{-1}\bm R_A(\tau)\bm z$. Therefore
\begin{equation}\label{eqn:gic_orthogonal_parameterization}
	\mathcal K_A(\tau)=\{\bm T_A(\tau)\bm z:\bm z\in\reals^{|A|-s}\},
	\qquad
	\|\bm z\|_2\le\|\bm T_A(\tau)\bm z\|_2\le C\|\bm z\|_2.
\end{equation}
The upper bound follows from \eqref{eqn:gic_orthogonal_matrices} and
$\|\bm T_A(\tau)\|_{\mathrm{op}}
\le1+(f_-\kappa_-)^{-1}f_+\kappa_+$.
The exact inverse identity yields
\begin{align*}
	\bm B(\tau)^{-1}-\bm B(t)^{-1}
	&=\bm B(\tau)^{-1}[\bm B(t)-\bm B(\tau)]\bm B(t)^{-1},\\
	\|\bm T_A(\tau)-\bm T_A(t)\|_{\mathrm{op}}
	&\le\|\bm B(\tau)^{-1}-\bm B(t)^{-1}\|_{\mathrm{op}}
	\|\bm R_A(\tau)\|_{\mathrm{op}}\\
	&\quad+\|\bm B(t)^{-1}\|_{\mathrm{op}}
	\|\bm R_A(\tau)-\bm R_A(t)\|_{\mathrm{op}}\\
	&\le C|\tau-t|.
\end{align*}
Consequently, we get
\begin{equation}\label{eqn:gic_orthogonal_parameter_lipschitz}
	\|\bm T_A(\tau)\bm z-\bm T_A(t)\bm z_0\|_2
	\le C\|\bm z-\bm z_0\|_2+C|\tau-t|\|\bm z_0\|_2.
\end{equation}
If $J_\tau\{\bm T_A(\tau)\bm z,\bm T_A(\tau)\bm z\}\le1$, then
\eqref{eqn:gic_jacobian_equivalence} and
\eqref{eqn:gic_orthogonal_parameterization} imply
$\|\bm z\|_2\le(f_-\kappa_-)^{-1/2}$.
Fix $j=|A|-s\ge1$. Choose coefficient and quantile approximations with
$\|\bm z-\bm z_0\|_2\le\delta$ and $|\tau-t|\le\delta^2$.
By \eqref{eqn:gic_orthogonal_parameter_lipschitz},
\eqref{eqn:gic_score_bracket_width} holds with $C\xi_{\mathrm G,n}\delta$
in place of its first term. There are at most $(p+1)^j$ choices of
$A\setminus\Sset$, $(1+C/\delta)^j$ choices of $\bm z_0$, and
$1+\delta^{-2}$ choices of $t$. Thus the logarithm of the number of brackets
is at most
\[
j\log(p+1)+j\log(1+C/\delta)+\log(1+\delta^{-2})
\le Cj\{L_{\mathrm G,n}+\log(C/\delta)\}.
\]
Their envelope and variance obey \eqref{eqn:gic_score_envelope_variance}.
Equation \eqref{eqn:gic_bracket_bernstein}, with $k$ replaced by $j$ and
$t=C_0\sqrt{jL_{\mathrm G,n}}$, therefore gives a probability at most
$2e^{-4jL_{\mathrm G,n}}$ for sufficiently large $C_0$ and $n$.
Summing over $j$ proves \eqref{eqn:gic_extra_score}.
\end{proof}

\begin{lemma}\label{lem:gic_overfit_localization}
Under the conditions of Lemma \ref{lem:gic_overfit_process_bounds}, and
either \eqref{hyp:linear hypothesis} or Condition \ref{as:local alts}
when $e=0$, there is a fixed $M<\infty$ such that
\[
P\{\exists A\in\mathcal O_n\cup\{\Sset\},\ \tau\in\mc T:
\widehat{\bm h}_A(\tau)\notin\mathcal L_A(M,\tau)\}\to0.
\]
The assertion holds for every minimizer.
\end{lemma}

\begin{proof}
Knight's identity, in the form used in Lemma 7.1 of
\citet{zheng2015globally}, gives
\begin{align}
	M_A(\bm h,\tau)
	&=-G_n(\bm h,\tau)
	+\sum_{i=1}^n\int_0^{n^{-1/2}\xvec_i^T\bm h}
	\left[1\{Y_i-\xvec_i^T\bvec^*(\tau)\le u\}
	-1\{Y_i-\xvec_i^T\bvec^*(\tau)\le0\}\right]du
	\notag\\
	&=-G_n(\bm h,\tau)+\int_0^1Z_A(t\bm h,\tau;\bm h)\,dt.
	\label{eqn:gic_knight_identity}
\end{align}
For $\bm h\in\mathcal L_A(M,\tau)$,
\eqref{eqn:gic_score_mean_expansion} and
\eqref{eqn:gic_jacobian_equivalence} imply
\begin{align*}
	\left|\int_0^1\overline Z_A(t\bm h,\tau;\bm h)\,dt
	-\frac12J_\tau(\bm h,\bm h)\right|
	&\le C\frac{\xi_{\mathrm G,n}}{\sqrt n}\|\bm h\|_2^3\\
	&\le\frac{CM\xi_{\mathrm G,n}d_{\mathrm G,n}L_{\mathrm G,n}}{\sqrt n}
	J_\tau(\bm h,\bm h)^{1/2}.
\end{align*}
Together with \eqref{eqn:gic_score_modulus} for $\epsilon=1$, this proves
that, outside an event of outer probability tending to zero,
\begin{align}
	\left|M_A(\bm h,\tau)+G_n(\bm h,\tau)-\frac12J_\tau(\bm h,\bm h)\right|
	&\le\left[\sqrt{L_{\mathrm G,n}}
	+\frac{CM\xi_{\mathrm G,n}d_{\mathrm G,n}L_{\mathrm G,n}}{\sqrt n}\right]
	J_\tau(\bm h,\bm h)^{1/2}
	\notag\\
	&\le2\sqrt{L_{\mathrm G,n}}\,J_\tau(\bm h,\bm h)^{1/2}
	\label{eqn:gic_local_quadratic_error}
\end{align}
uniformly over $A,\tau$, and $\bm h\in\mathcal L_A(M,\tau)$, for all
sufficiently large $n$. The last inequality is
\eqref{eqn:gic_sharp_consequences} with fixed $M$.

Write $k=|D_A|$. Equations \eqref{eqn:gic_feasible_shift_bound} and
\eqref{eqn:gic_jacobian_equivalence} give a constant $C_v$ such that
\begin{equation}\label{eqn:gic_shift_quadratic_bound}
	J_\tau\{\bm v_e(\tau),\bm v_e(\tau)\}
	\le C_v\sup_{t\in\mc T}\|\bm v_e(t)\|_2^2,
	\qquad
	J_\tau\{\bm v_e(\tau),\bm v_e(\tau)\}\le C_vk.
\end{equation}
Increase $C_v$ if necessary. For any fixed $M>C_v$, the vector
$\bm v_e(\tau)$ belongs to $\mathcal L_A(M,\tau)$ for all sufficiently
large $n$. By \eqref{eqn:gic_score_size},
\eqref{eqn:gic_local_quadratic_error}, and
\eqref{eqn:gic_shift_quadratic_bound}, we get
\begin{align}
	M_A\{\bm v_e(\tau),\tau\}
	&\le C\sqrt{kL_{\mathrm G,n}}\sqrt{C_vk}
	+\frac{C_vk}{2}+2\sqrt{L_{\mathrm G,n}}\sqrt{C_vk}
	\notag\\
	&\le C_1kL_{\mathrm G,n}.
	\label{eqn:gic_feasible_criterion_bound}
\end{align}
The constant $C_1$ is independent of $M$. On the boundary
$J_\tau(\bm h,\bm h)=MkL_{\mathrm G,n}$,
\eqref{eqn:gic_score_size} and \eqref{eqn:gic_local_quadratic_error} give
\[
M_A(\bm h,\tau)
\ge\left\{\frac M2-C\sqrt M-2\sqrt M\right\}kL_{\mathrm G,n}.
\]
Choose fixed $M>C_v$ so that $M/2-(C+2)\sqrt M>C_1+1$.
This boundary value then exceeds
$M_A\{\bm v_e(\tau),\tau\}$ by \eqref{eqn:gic_feasible_criterion_bound}.
If a minimizer $\widehat{\bm h}$ were outside the ball, continuity would
give $t\in(0,1)$ such that
\[
\bm h_t=(1-t)\bm v_e(\tau)+t\widehat{\bm h},
\qquad J_\tau(\bm h_t,\bm h_t)=MkL_{\mathrm G,n}.
\]
The affine feasible set contains $\bm h_t$, and convexity gives
\[
M_A(\bm h_t,\tau)
\le(1-t)M_A\{\bm v_e(\tau),\tau\}+tM_A(\widehat{\bm h},\tau)
\le M_A\{\bm v_e(\tau),\tau\},
\]
contradicting the boundary inequality. All comparisons are simultaneous
over $A$ and $\tau$, which proves the lemma.
\end{proof}

\begin{theorem}\label{thm:gic_overfit}
Under the conditions of Lemma \ref{lem:gic_overfit_localization}, there is
$C<\infty$ such that
\begin{equation}\label{eqn:gic_overfit_loss_gain}
	P\left\{
	\sup_{A\in\mathcal O_n}\sup_{\tau\in\mc T}
	\frac{n\{\widehat L_{e,\Sset}(\tau)-\widehat L_{e,A}(\tau)\}}
	{(|A|-s)L_{\mathrm G,n}+\sup_{t\in\mc T}\|\bm v_e(t)\|_2^2}
	>C\right\}\to0.
\end{equation}
The numerator is nonnegative. If $e=a$, or if $e=0$ and
\eqref{eqn:gic_local_alt_rate} holds, then
\[
P\left\{\inf_{A\in\mathcal O_n}
[\operatorname{GIC}^{\circ}_e(A)-\operatorname{GIC}^{\circ}_e(\Sset)]
\le0\right\}\to0.
\]
\end{theorem}

\begin{proof}
By \eqref{eqn:gic_aux_loss_full} and \eqref{eqn:gic_aux_loss_reduced},
$\widehat L_{e,A}(\tau)\le\widehat L_{e,\Sset}(\tau)$ for every
$A\in\mathcal O_n$. Apply Lemmas \ref{lem:gic_overfit_process_bounds} and
\ref{lem:gic_overfit_localization}. For each $A,\tau$, equation
\eqref{eqn:gic_orthogonal_sum} yields unique vectors $\bm u,\bm w$ such that
\[
\widehat{\bm h}_A(\tau)=\bm v_e(\tau)+\bm u+\bm w,
\quad\bm u\in\mathcal H_{\Sset},\quad\bm w\in\mathcal K_A(\tau),
\quad J_\tau(\bm u,\bm w)=0.
\]
Let $k=|D_A|$. The orthogonality and
\eqref{eqn:gic_shift_quadratic_bound} imply, for every $0\le t\le1$,
\begin{align*}
	J_\tau(\bm u+t\bm w,\bm u+t\bm w)^{1/2}
	&\le J_\tau(\bm u+\bm w,\bm u+\bm w)^{1/2}\\
	&\le J_\tau\{\widehat{\bm h}_A(\tau),\widehat{\bm h}_A(\tau)\}^{1/2}
	+J_\tau\{\bm v_e(\tau),\bm v_e(\tau)\}^{1/2},\\
	J_\tau\{\bm v_e(\tau)+\bm u+t\bm w,
	\bm v_e(\tau)+\bm u+t\bm w\}^{1/2}
	&\le(\sqrt M+2\sqrt{C_v})\sqrt{kL_{\mathrm G,n}}.
\end{align*}
Thus every $\bm v_e(\tau)+\bm u+t\bm w$ belongs to
$\mathcal L_A((\sqrt M+2\sqrt{C_v})^2,\tau)$; equation
\eqref{eqn:gic_score_modulus} applies on this whole segment.
Since $\bm v_e(\tau)+\bm u$ is feasible for $\Sset$, the exact
integral identity in \eqref{eqn:gic_knight_identity}, applied along this
segment, gives
\begin{align}
	0\le n\{\widehat L_{e,\Sset}(\tau)-\widehat L_{e,A}(\tau)\}
	&\le M_A\{\bm v_e(\tau)+\bm u,\tau\}
	-M_A\{\bm v_e(\tau)+\bm u+\bm w,\tau\}
	\notag\\
	&=G_n(\bm w,\tau)
	-\int_0^1Z_A\{\bm v_e(\tau)+\bm u+t\bm w,\tau;\bm w\}\,dt.
	\label{eqn:gic_affine_loss_difference}
\end{align}
Write $z=J_\tau(\bm w,\bm w)^{1/2}$. If $z=0$, then $\bm w=0$
by \eqref{eqn:gic_jacobian_equivalence}, and the loss difference is zero.
For $z>0$, equation \eqref{eqn:gic_extra_score} gives
\[
|G_n(\bm w,\tau)|\le C\sqrt{(|A|-s)L_{\mathrm G,n}}\,z.
\]
Moreover, we have
\begin{align}
	\int_0^1J_\tau\{\bm v_e(\tau)+\bm u+t\bm w,\bm w\}\,dt
	&=J_\tau\{\bm v_e(\tau),\bm w\}+\frac12z^2,
	\notag\\
	|J_\tau\{\bm v_e(\tau),\bm w\}|
	&\le\sqrt{C_v}\,\sup_{t\in\mc T}\|\bm v_e(t)\|_2\,z.
	\label{eqn:gic_affine_quadratic_difference}
\end{align}
Equations \eqref{eqn:gic_score_mean_expansion},
\eqref{eqn:gic_score_modulus}, and
\eqref{eqn:gic_jacobian_equivalence}, applied with $\bm g=\bm w/z$, imply
\begin{align}
	&\left|\int_0^1
	\left[Z_A\{\bm v_e(\tau)+\bm u+t\bm w,\tau;\bm w\}
	-J_\tau\{\bm v_e(\tau)+\bm u+t\bm w,\bm w\}\right]dt\right|
	\notag\\
	&\quad\le
	\left[\sqrt{L_{\mathrm G,n}}
	+C\frac{\xi_{\mathrm G,n}d_{\mathrm G,n}L_{\mathrm G,n}}{\sqrt n}\right]z
	\le2\sqrt{L_{\mathrm G,n}}\,z.
	\label{eqn:gic_affine_remainder}
\end{align}
The last inequality follows from \eqref{eqn:gic_sharp_consequences}.
Equations \eqref{eqn:gic_affine_loss_difference}--\eqref{eqn:gic_affine_remainder}
and $|A|-s\ge1$ give
\begin{align*}
	n\{\widehat L_{e,\Sset}(\tau)-\widehat L_{e,A}(\tau)\}
	&\le\left[C\sqrt{(|A|-s)L_{\mathrm G,n}}
	+C\sup_{t\in\mc T}\|\bm v_e(t)\|_2\right]z-\frac12z^2\\
	&\le\frac12\left[C\sqrt{(|A|-s)L_{\mathrm G,n}}
	+C\sup_{t\in\mc T}\|\bm v_e(t)\|_2\right]^2\\
	&\le C\left[(|A|-s)L_{\mathrm G,n}
	+\sup_{t\in\mc T}\|\bm v_e(t)\|_2^2\right].
\end{align*}
This proves \eqref{eqn:gic_overfit_loss_gain}.

Every summand integral in \eqref{eqn:gic_knight_identity} is nonnegative.
Equations \eqref{eqn:gic_score_size} and
\eqref{eqn:gic_local_loss_identity}, together with Lemma
\ref{lem:gic_overfit_localization}, therefore imply
\[
\widehat L_{e,A}(\tau)
\ge \mc Q_{\tau,n}\{\bvec^*(\tau)\}-\frac1n|G_n\{\widehat{\bm h}_A(\tau),\tau\}|
\ge \mc Q_{\tau,n}\{\bvec^*(\tau)\}-C\frac{d_{\mathrm G,n}L_{\mathrm G,n}}n.
\]
Since $\xi_{\mathrm G,n}^2\ge\kappa_-$,
\eqref{eqn:gic_sharp_growth} implies
$d_{\mathrm G,n}L_{\mathrm G,n}/n\to0$. Lemma \ref{lem:gic_loss_scale} gives
\begin{equation}\label{eqn:gic_overfit_positive_loss}
	P^*\left\{\inf_{A\in\mathcal O_n\cup\{\Sset\}}\inf_{\tau\in\mc T}
	\widehat L_{e,A}(\tau)<c_\rho/2\right\}\to0.
\end{equation}
On its complementary event, we get
\begin{align*}
	\int_{\mc T}\log\frac{\widehat L_{e,A}(\tau)}{\widehat L_{e,\Sset}(\tau)}\,d\tau
	&=-\int_{\mc T}\int_{\widehat L_{e,A}(\tau)}^{\widehat L_{e,\Sset}(\tau)}
	\frac{du}{u}\,d\tau\\
	&\ge-\frac2{c_\rho}\int_{\mc T}
	[\widehat L_{e,\Sset}(\tau)-\widehat L_{e,A}(\tau)]\,d\tau\\
	&\ge-\frac Cn\left[(|A|-s)L_{\mathrm G,n}
	+\sup_{t\in\mc T}\|\bm v_e(t)\|_2^2\right].
\end{align*}
It follows that
\begin{align*}
	\operatorname{GIC}^{\circ}_e(A)-\operatorname{GIC}^{\circ}_e(\Sset)
	&\ge(|A|-s)\phi_n
	-\frac Cn\left[(|A|-s)L_{\mathrm G,n}
	+\sup_{t\in\mc T}\|\bm v_e(t)\|_2^2\right]\\
	&\ge(|A|-s)\phi_n\left[
	1-C\frac{L_{\mathrm G,n}}{n\phi_n}
	-C\frac{\sup_{t\in\mc T}\|\bm v_e(t)\|_2^2}{n\phi_n}\right]>0
\end{align*}
for all sufficiently large $n$. Indeed, \eqref{eqn:gic_penalty_growth}
gives $L_{\mathrm G,n}/(n\phi_n)\to0$, while
\[
\frac{\sup_{t\in\mc T}\|\bm v_a(t)\|_2^2}{n\phi_n}=0,
\qquad
\frac{\sup_{t\in\mc T}\|\bm v_0(t)\|_2^2}{n\phi_n}
\le C\frac{\sup_{t\in\mc T}\|\bm\eta(t)\|_2^2}{n\phi_n}\to0
\]
by \eqref{eqn:gic_feasible_shift_bound} and
\eqref{eqn:gic_local_alt_rate}; under \eqref{hyp:linear hypothesis},
$\bm v_0$ is identically zero.
\end{proof}

\begin{theorem}\label{thm:gic_auxiliary_separation}
Under the conditions of Theorem \ref{thm:GIC.consistency},
\begin{equation}\label{eqn:gic_auxiliary_separation}
	P\left\{
	\inf_{\substack{A\subset\Mset^c\setminus\{0\},\ |A|\le K_{\mathrm G,n}\\A\ne\Sset}}
	[\operatorname{GIC}^{\circ}_e(A)-\operatorname{GIC}^{\circ}_e(\Sset)]\le0
	\right\}\to0.
\end{equation}
In addition,
\begin{equation}\label{eqn:gic_all_positive_loss}
	P\left\{
	\inf_{\substack{A\subset\Mset^c\setminus\{0\}\\|A|\le K_{\mathrm G,n}}}
	\inf_{\tau\in\mc T}\widehat L_{e,A}(\tau)<c_\rho/2\right\}\to0.
\end{equation}
\end{theorem}

\begin{proof}
The candidates in \eqref{eqn:gic_auxiliary_separation} form the disjoint
union $\mathcal U_n\cup\mathcal O_n$. The probability in that equation is
therefore bounded by
\begin{align*}
	P\left\{\inf_{A\in\mathcal U_n}
	[\operatorname{GIC}^{\circ}_e(A)-\operatorname{GIC}^{\circ}_e(\Sset)]\le0\right\}+P\left\{\inf_{A\in\mathcal O_n}
	[\operatorname{GIC}^{\circ}_e(A)-\operatorname{GIC}^{\circ}_e(\Sset)]\le0\right\}\to0
\end{align*}
by Theorems \ref{thm:gic_underfit} and \ref{thm:gic_overfit}.
For $A\in\mathcal U_n$, choose $j_A=\min(\Sset\setminus A)$.
Then $A\subset A^{+j_A}$ and $A^{+j_A}\in\mathcal B_n$. Lemmas
\ref{lem:gic_localization} and \ref{lem:gic_empirical_curvature} give
\[
\widehat L_{e,A}(\tau)\ge\widehat L_{e,A^{+j_A}}(\tau)
\ge \mc Q_{\tau,n}\{\bvec^*(\tau)\}-C_3a_{\mathrm G,n}^2\ge c_\rho/2
\]
uniformly over $A\in\mathcal U_n$ and $\tau\in\mc T$, outside events of
probability tending to zero, by Lemma \ref{lem:gic_loss_scale} and
Condition \ref{as:gic_underfit_rates}. For
$A\in\mathcal O_n\cup\{\Sset\}$, use
\eqref{eqn:gic_overfit_positive_loss}. This proves
\eqref{eqn:gic_all_positive_loss}.
\end{proof}

\subsubsection{Selection of the tuning parameter}\label{app:gic_lla_bridge}

\begin{proof}[Proof of Theorem \ref{thm:GIC.consistency}]
We use the comparison of an estimated loss with the minimum on its
selected coordinates from the proof of Corollary 3.2 of
\citet{zheng2015globally}.
Fix $e\in\{0,a\}$ and $\lambda\in\mathcal P_{e,n}$ satisfying
$|\widehat{\Sset}_e(\lambda)|\le K_{\mathrm G,n}$.
Equation \eqref{eqn:gic_selected_support} and Algorithms
\ref{algo:LLA for full estimator}--\ref{algo:LLA for reduced estimator} give
\[
\widehat{\bvec}^{(2)}_{e,D_{\widehat{\Sset}_e(\lambda)}^c}(\tau;\lambda)=\bm0,
\qquad
\Cmat\widehat{\bvec}^{(2)}_{0,\Mset}(\tau;\lambda)=\bm\Psi(\tau).
\]
Thus \eqref{eqn:gic_aux_loss_full}--\eqref{eqn:gic_aux_loss_reduced} imply
\begin{equation}\label{eqn:gic_path_lower_bound}
	\mc Q_{\tau,n}\{\widehat{\bvec}_e^{(2)}(\tau;\lambda)\}
	\ge\widehat L_{e,\widehat{\Sset}_e(\lambda)}(\tau).
\end{equation}

On the event that all losses in \eqref{eqn:gic_all_positive_loss} are
at least $c_\rho/2$, equations \eqref{eqn:gic_lla},
\eqref{eqn:gic_auxiliary}, and \eqref{eqn:gic_path_lower_bound} give
\begin{equation}\label{eqn:gic_path_aux_lower}
	\operatorname{GIC}_{e;\mc T}(\lambda)
	\ge\operatorname{GIC}^{\circ}_e\{\widehat{\Sset}_e(\lambda)\}.
\end{equation}
Condition \ref{as:gic_oracle_path} gives, outside an event of outer
probability tending to zero,
\[
\widehat{\bvec}^{(2)}_e(\tau;\lambda_{e,n}^{\circ})
=\bvecOR_e(\tau)\quad(\tau\in\mc T),
\qquad
\widehat{\Sset}_e(\lambda_{e,n}^{\circ})=\Sset.
\]
Since $s\le K_{\mathrm G,n}$, this tuning parameter is admissible in
\eqref{eqn:gic_lambda_selector}. By \eqref{eqn:FE.oracle},
\eqref{eqn:RE.oracle}, \eqref{eqn:gic_aux_loss_full}, and
\eqref{eqn:gic_aux_loss_reduced}, we get
\[
\mc Q_{\tau,n}\{\bvecOR_e(\tau)\}=\widehat L_{e,\Sset}(\tau).
\]
Together with \eqref{eqn:gic_lla} and \eqref{eqn:gic_auxiliary}, this yields
\begin{equation}\label{eqn:gic_oracle_path_equality}
	\operatorname{GIC}_{e;\mc T}(\lambda_{e,n}^{\circ})
	=\operatorname{GIC}^{\circ}_e(\Sset).
\end{equation}

Outside the union of the events in
\eqref{eqn:gic_oracle_path_assumption},
\eqref{eqn:gic_auxiliary_separation}, and
\eqref{eqn:gic_all_positive_loss}, every admissible $\lambda$ with
$\widehat{\Sset}_e(\lambda)\ne\Sset$ satisfies
\[
\begin{aligned}
	\operatorname{GIC}_{e;\mc T}(\lambda)
	\ge\operatorname{GIC}^{\circ}_e\{\widehat{\Sset}_e(\lambda)\}>\operatorname{GIC}^{\circ}_e(\Sset)
	=\operatorname{GIC}_{e;\mc T}(\lambda_{e,n}^{\circ}),
\end{aligned}
\]
by \eqref{eqn:gic_path_aux_lower} and
\eqref{eqn:gic_oracle_path_equality}. The finite minimization set is
nonempty because it contains $\lambda_{e,n}^{\circ}$. Hence
\[
\argmin_{\lambda\in\mathcal P_{e,n}:\,
	|\widehat{\Sset}_e(\lambda)|\le K_{\mathrm G,n}}
\operatorname{GIC}_{e;\mc T}(\lambda)
\subset
\{\lambda\in\mathcal P_{e,n}:\widehat{\Sset}_e(\lambda)=\Sset\}.
\]
Outer subadditivity and the probability bounds in
\eqref{eqn:gic_oracle_path_assumption},
\eqref{eqn:gic_auxiliary_separation}, and
\eqref{eqn:gic_all_positive_loss} prove
\eqref{eqn:gic_selection_consistency}.
\end{proof}

\subsection{Power comparison}\label{sec:power_comparison_details}

We make explicit the homoskedasticity required for the comparison in
Remark \ref{rem:Power_comparison}. Under this error structure, recall that the model has a common
regression error independent of $\xvec_i$, with distribution function
$F$ and density $f$.  Note that
\[
f_{y\mid\xvec}
\{Q_y(\tau\mid\xvec)\mid\xvec\}
=
f\{F^{-1}(\tau)\},
\qquad
\bm J(\tau)
=
f\{F^{-1}(\tau)\}\bm\Sigma.
\]	
Let \(\Sigma_X=\E(\xvec_i\xvec_i^T)\) denote the full Gram matrix, and let $\bm e_j$ denote the standard basis
vector corresponding to the tested coordinate in $\Sigma_X$ and in its
principal submatrix $\bm\Sigma$. 
Lemma S.1 of \citet{shi2019linear} gives
\begin{equation}\label{eqn:information.comparison}
	\bm e_j^T\bm\Sigma^{-1}\bm e_j
	\leq
	\bm e_j^T\Sigma_X^{-1}\bm e_j.
\end{equation}

\subsubsection{Pointwise local alternatives}

Fix $\tau\in\mc T$ and consider
\[
\beta_j^*(\tau)=n^{-1/2}h(\tau),
\qquad
h(\tau)\ne0.
\]
By the results above, the Bahadur representation
\eqref{eqn:FE.part pen bahadur}, together with Theorem 4.2 and
Lemma 7 of \citet{bradic2017uniform}, gives the respective pointwise
asymptotic variances
\[
\frac{\tau(1-\tau)}{[f\{F^{-1}(\tau)\}]^2}
\bm e_j^T\bm\Sigma^{-1}\bm e_j
\quad\text{and}\quad
\frac{\tau(1-\tau)}{[f\{F^{-1}(\tau)\}]^2}
\bm e_j^T\Sigma_X^{-1}\bm e_j.
\]
Consequently, the noncentrality parameters of the two pointwise Wald
limits are
\begin{equation}\label{eqn:gamma-ours-bk}
	\begin{split}
	\gamma_{\mathrm{ours}}(\tau)
	&=
	\frac{
		h^2(\tau)[f\{F^{-1}(\tau)\}]^2
	}{
		\tau(1-\tau)
		\bm e_j^T\bm\Sigma^{-1}\bm e_j
	},
	\notag\\
	\gamma_{\mathrm{BK}}(\tau)
	&=
	\frac{
		h^2(\tau)[f\{F^{-1}(\tau)\}]^2
	}{
		\tau(1-\tau)
		\bm e_j^T\Sigma_X^{-1}\bm e_j
	}.
\end{split}
\end{equation}
The display \eqref{eqn:information.comparison} implies
$\gamma_{\mathrm{ours}}(\tau)\geq\gamma_{\mathrm{BK}}(\tau)$.
By Theorem \ref{thm:part pen bessel},
$\gamma_{\mathrm{ours}}(\tau)$ is the common pointwise noncentrality
parameter of our Wald and rank score statistics.  Since the upper-tail
probability of a noncentral $\chi_1^2$ variable is increasing in its
noncentrality parameter, our pointwise asymptotic local power is no
smaller than that of \citet{bradic2017uniform}.

For the pointwise procedure of \citet{yan2023confidence}, let
$\varepsilon_i(\tau)=y_i-\xvec_i^T\bvec^*(\tau)$ and write their
population Hessian as
\[
\widetilde{\bm J}_{h_n}(\tau)
=
\E\left[
K_{h_n}\{-\varepsilon_i(\tau)\}
\xvec_i\xvec_i^T
\right].
\]
Under the same homoskedasticity and the conditions of Theorems 3--5
of \citet{yan2023confidence},
\[
\widetilde{\bm J}_{h_n}(\tau)
=
\E[K_{h_n}\{-\varepsilon_i(\tau)\}]\Sigma_X,
\qquad
\E[K_{h_n}\{-\varepsilon_i(\tau)\}]
\longrightarrow
f\{F^{-1}(\tau)\}.
\]
Consequently, the population counterpart of the studentizing variance
in equation (19) of \citet{yan2023confidence} satisfies
\begin{align*}
	\tau(1-\tau)\bm e_j^T
	\widetilde{\bm J}_{h_n}^{-1}(\tau)
	\Sigma_X
	\widetilde{\bm J}_{h_n}^{-1}(\tau)\bm e_j\longrightarrow
	\frac{\tau(1-\tau)}{[f\{F^{-1}(\tau)\}]^2}
	\bm e_j^T\Sigma_X^{-1}\bm e_j.
\end{align*}
Theorems 4--5 of \citet{yan2023confidence} therefore give the same
pointwise limiting variance, and hence the same limiting noncentrality
parameter, as the procedure of \citet{bradic2017uniform}.  Now, equation
\eqref{eqn:information.comparison} applies again.

\subsubsection{Global local alternatives}

We first remark that \citet{bradic2017uniform} formulate and prove their results under the additive location-shift model (see the display (1.1) therein), for which only the intercept component of \(\beta^*(\tau)\) varies with \(\tau\).
Now, consider a global local alternative
\[
\beta_j^*(\tau)=n^{-1/2}h,
\qquad
\tau\in\mc T.
\]
Theorem 4.2 of
\citet{bradic2017uniform} gives their uniform Bahadur representation,
and Theorem 5.1 gives their null Brownian-bridge limit.  Together with
Theorem \ref{thm:part pen bessel}, these results show that the two
asymptotic target laws may be represented using the same standard
Brownian bridge $\bm B_1$, with respective deterministic drift terms
\begin{align*}
	\nu_{\mathrm{ours}}(\tau)
	=
	\frac{
		hf\{F^{-1}(\tau)\}
	}{
		(\bm e_j^T\bm\Sigma^{-1}\bm e_j)^{1/2}
	}, \quad
	\nu_{\mathrm{BK}}(\tau)
	=
	\frac{
		hf\{F^{-1}(\tau)\}
	}{
		(\bm e_j^T\Sigma_X^{-1}\bm e_j)^{1/2}
	}.
\end{align*}
Again, we get
\begin{equation}\label{eqn:mu-scaling}
	\nu_{\mathrm{ours}}(\tau)
	=
	\left\{
	\frac{
		\bm e_j^T\Sigma_X^{-1}\bm e_j
	}{
		\bm e_j^T\bm\Sigma^{-1}\bm e_j
	}
	\right\}^{1/2}
	\nu_{\mathrm{BK}}(\tau),
\end{equation}
where the multiplier is independent of $\tau$ and is at least one by
\eqref{eqn:information.comparison}.

Let $c_\alpha$ be the upper-$\alpha$ critical value of
\[
\sup_{\tau\in\mc T}
\frac{|\bm B_1(\tau)|^2}{\tau(1-\tau)},
\]
and define
\[
\mathcal A_\alpha
=
\left\{
g\in C(\mc T):
\sup_{\tau\in\mc T}
\frac{|g(\tau)|^2}{\tau(1-\tau)}
\leq c_\alpha
\right\}.
\]
The set $\mathcal A_\alpha$ is closed, convex, and symmetric.
Anderson's inequality of \citet{anderson1955integral} for centered
Gaussian measures, also discussed by \citet{lifshits2012lectures},
implies that, for every fixed $\nu\in C(\mc T)$, the map
\[
t\longmapsto
\Pr\{\bm B_1+t\nu\in\mathcal A_\alpha\}
\]
is non-increasing on $[0,\infty)$.  Applying this result with
$\nu=\nu_{\mathrm{BK}}$ and the multiplier in
\eqref{eqn:mu-scaling} gives
\begin{align*}
	\Pr\left\{
	\bm B_1+\nu_{\mathrm{ours}}
	\in\mathcal A_\alpha
	\right\}\leq
	\Pr\left\{
	\bm B_1+\nu_{\mathrm{BK}}
	\in\mathcal A_\alpha
	\right\}.
\end{align*}
Taking complements proves that our global asymptotic local power is
no smaller than that of \citet{bradic2017uniform}.

	\subsection{Grid interpolation for the reduced two-step LLA process}\label{app:grid_interpolation}

We first derive the required path bound from
\eqref{eqn:conditional quantile linear expr} and Conditions
\ref{as:density conditions}(ii)--(iii) and
\ref{as:Gram lower bound}.

\begin{lemma}\label{lem:grid_truth_modulus}
	Under \eqref{eqn:conditional quantile linear expr} and Conditions
	\ref{as:density conditions}(ii)--(iii) and
	\ref{as:Gram lower bound}, there is a finite constant $C>0$,
	independent of $n$, such that
	\[
	\left\|
	\bvec^*(\tau)-\bvec^*(t)
	\right\|_2
	\leq
	C|\tau-t|,
	\qquad
	\tau,t\in\mc T.
	\]
\end{lemma}

\begin{proof}
	The inverse-conditional-CDF argument in the proof of Lemma 19 of
	\citet{belloni_conditional_2019} gives
	\[
	\left|
	Q_y(\tau\mid\xvec)-Q_y(t\mid\xvec)
	\right|
	\lesssim
	|\tau-t|
	\]
	almost surely, uniformly over $\tau,t\in\mc T$ and $n$.  
	By the definition of $\Sset$,
	$\bvec^*_{(\{0\}\cup\Mset\cup\Sset)^c}(u)=\bm0$ for every
	$u\in\mc T$.  Hence
	\begin{align*}
		&
		\left[
		\bvec^*_{\{0\}\cup\Mset\cup\Sset}(\tau)
		-
		\bvec^*_{\{0\}\cup\Mset\cup\Sset}(t)
		\right]^T
		\bm\Sigma
		\left[
		\bvec^*_{\{0\}\cup\Mset\cup\Sset}(\tau)
		-
		\bvec^*_{\{0\}\cup\Mset\cup\Sset}(t)
		\right]\\
		&\qquad=
		\E\left[
		\left\{
		Q_y(\tau\mid\xvec)-Q_y(t\mid\xvec)
		\right\}^2
		\right]
		\lesssim
		|\tau-t|^2.
	\end{align*}
	Condition \ref{as:Gram lower bound} and the zero coefficients
	outside $\{0\}\cup\Mset\cup\Sset$ prove the assertion.
\end{proof}

\begin{lemma}\label{lem:grid_point_rate}
	Suppose Conditions
	\ref{as:iid array}--\ref{as:new_re} and
	\ref{as:l1re}--\ref{as:s_size} hold, together with
	\[
	(m+s)\xi_n^2\log^2n=o(n),
	\qquad
	m-r+s\geq1,
	\qquad
	\sqrt n\,\Delta_n=O(1).
	\]
	Then
	\[
	\Pr\left\{
	\max_{t\in\mathcal G_n}
	\left\|
	\widehat{\bvec}_{0,\mathcal G_n}^{(2)}(t)
	-
	\widehat{\bvec}_0^o(t)
	\right\|_2
	=0
	\right\}
	\longrightarrow1,
	\]
	and
	\[
	\max_{t\in\mathcal G_n}
	\left\|
	\widehat{\bvec}_{0,\mathcal G_n}^{(2)}(t)
	-
	\bvec^*(t)
	\right\|_2
	=
	O_p\left(
	\sqrt{\frac{m-r+s}{n}}
	\right).
	\]
\end{lemma}

\begin{proof}
Recall that we already showed
	\[
	\sup_{\tau\in\mc T}
	\left\|
	\widehat{\bvec}_0^{(0)}(\tau)-\bvec^*(\tau)
	\right\|_\infty
	=
	o_p(\lambda).
	\]
	The mesh condition and Condition \ref{as:lambda} imply
	\[
	\Delta_n=O(n^{-1/2})=o(\lambda),
	\qquad
	\sqrt{\frac{m-r+s}{n}}=o(\lambda).
	\]
	
	If $\Sset\neq\varnothing$, Lemma
	\ref{lem:grid_truth_modulus} and Condition
	\ref{as:strength of signal} yield
	\begin{align*}
		\inf_{j\in\Sset}
		\max_{t\in\mathcal G_n}
		\left|
		\widehat\beta_{0,\mathcal G_n,j}^{(0)}(t)
		\right|
		&\geq
		\inf_{j\in\Sset}
		\sup_{\tau\in\mc T}
		|\beta_j^*(\tau)|
		-
		C\Delta_n
		-
		\sup_{\tau\in\mc T}
		\left\|
		\widehat{\bvec}_0^{(0)}(\tau)-\bvec^*(\tau)
		\right\|_\infty\\
		&>
		d\lambda
	\end{align*}
	with probability tending to one.  For every
	$j\in(\{0\}\cup\Mset\cup\Sset)^c$,
	\[
	\max_{t\in\mathcal G_n}
	\left|
	\widehat\beta_{0,\mathcal G_n,j}^{(0)}(t)
	\right|
	\leq
	\sup_{\tau\in\mc T}
	\left\|
	\widehat{\bvec}_0^{(0)}(\tau)-\bvec^*(\tau)
	\right\|_\infty
	=
	o_p(\lambda)
	<
	d_2\lambda
	\]
	with probability tending to one.  Thus Conditions
	\ref{as:FCP ii}--\ref{as:FCP iv} imply that the first grid LLA
	weights are zero on $\Sset$ and at least $d_1\lambda$ on
	$(\{0\}\cup\Mset\cup\Sset)^c$. 
	
	The event in \eqref{eqn:thm33as2} holds uniformly over $\mc T$
	and hence over $\mathcal G_n\subset\mc T$.  The deterministic
	subgradient comparison in the proof of Theorem
	\ref{thm:global33mod}, applied at each $t\in\mathcal G_n$ with
	the preceding grid weights, therefore gives
	\[
	\max_{t\in\mathcal G_n}
	\left\|
	\widehat{\bvec}_{0,\mathcal G_n}^{(1)}(t)
	-
	\widehat{\bvec}_0^o(t)
	\right\|_2
	=0
	\]
	with probability tending to one.
	
	By Theorem \ref{thm:oracle rate},
	we obtain	\[
	\sup_{\tau\in\mc T}
	\left\|
	\widehat{\bvec}_0^o(\tau)-\bvec^*(\tau)
	\right\|_2
	=
	O_p\left(
	\sqrt{\frac{m-r+s}{n}}
	\right)
	=
	o_p(\lambda).
	\]
	Now, Lemma
	\ref{lem:grid_truth_modulus} and Condition
	\ref{as:strength of signal} consequently give
	\begin{align*}
		\inf_{j\in\Sset}
		\max_{t\in\mathcal G_n}
		\left|
		\widehat\beta_{0,j}^o(t)
		\right|
		&\geq
		\inf_{j\in\Sset}
		\sup_{\tau\in\mc T}
		|\beta_j^*(\tau)|
		-
		C\Delta_n
		-
		\sup_{\tau\in\mc T}
		\left\|
		\widehat{\bvec}_0^o(\tau)-\bvec^*(\tau)
		\right\|_2\\
		&>
		d\lambda
	\end{align*}
	with probability tending to one.  Moreover,
	$\widehat\beta_{0,j}^o(t)=0$ for
	$j\in(\{0\}\cup\Mset\cup\Sset)^c$.  Conditions
	\ref{as:FCP ii}--\ref{as:FCP iv} therefore classify the second
	grid LLA weights in the same manner.  Applying the deterministic
	subgradient comparison in the proof of Theorem
	\ref{thm:global33modStep2} at every grid point proves
	\[
	\Pr\left\{
	\max_{t\in\mathcal G_n}
	\left\|
	\widehat{\bvec}_{0,\mathcal G_n}^{(2)}(t)
	-
	\widehat{\bvec}_0^o(t)
	\right\|_2
	=0
	\right\}
	\longrightarrow1.
	\]
	On this event,
	\[
	\max_{t\in\mathcal G_n}
	\left\|
	\widehat{\bvec}_{0,\mathcal G_n}^{(2)}(t)
	-
	\bvec^*(t)
	\right\|_2
	\leq
	\sup_{\tau\in\mc T}
	\left\|
	\widehat{\bvec}_0^o(\tau)-\bvec^*(\tau)
	\right\|_2,
	\]
	and the claimed convergence rate follows from Theorem
	\ref{thm:oracle rate}.
\end{proof}

\begin{proof}[Proof of Theorem \ref{thm:grid_interpolation}]
	First, it is obvious that $\widehat{\bvec}_{0,\mathcal G_n}^{(2)}(\cdot)$ is c\`adl\`ag.
	Next, we prove the other assertions.
	Fix $\tau\in\mc T$ and let $t=\pi_n(\tau)$.  Since the knot
	estimate is feasible, we have $\Cmat
	\widehat{\bvec}_{0,\mathcal G_n,\Mset}^{(2)}(t)
	=
	\bm\Psi(t)$.
	Equation \eqref{eqn:grid_interpolated_estimator} therefore gives
	\begin{align*}
		\Cmat
		\widehat{\bvec}_{0,\mathcal G_n,\Mset}^{(2)}(\tau)
		&=
		\bm\Psi(t)
		+
		\Cmat\Cmat^T(\Cmat\Cmat^T)^{-1}
		\{\bm\Psi(\tau)-\bm\Psi(t)\}\\
		&=
		\bm\Psi(\tau).
	\end{align*}
	Condition \ref{as:local alts} gives
	\[
	\bm\Psi(\tau)-\bm\Psi(t)
	=
	\Cmat
	\{\bvec_{\Mset}^*(\tau)-\bvec_{\Mset}^*(t)\}
	-
	\frac{\bm\eta(\tau)-\bm\eta(t)}{\sqrt n}.
	\]
	Note that the matrix $\Cmat^T(\Cmat\Cmat^T)^{-1}\Cmat$
	is the orthogonal projection operator onto the row space of
	$\Cmat$.  Hence, we get
	\begin{align*}
		\left\|
		\left\{
		\bm I_m-
		\Cmat^T(\Cmat\Cmat^T)^{-1}\Cmat
		\right\}
		\{\bvec_{\Mset}^*(t)-\bvec_{\Mset}^*(\tau)\}
		\right\|_2^2+
		\left\|
		\bvec_{\Mset^c}^*(t)-\bvec_{\Mset^c}^*(\tau)
		\right\|_2^2
		\leq
		\left\|
		\bvec^*(t)-\bvec^*(\tau)
		\right\|_2^2.
	\end{align*}
	Subtracting $\bvec^*(\tau)$ from
	\eqref{eqn:grid_interpolated_estimator}, adding and subtracting
	$\bvec^*(t)$, and applying the preceding projection inequality
	yield
	\begin{align*}
		\left\|
		\widehat{\bvec}_{0,\mathcal G_n}^{(2)}(\tau)
		-
		\bvec^*(\tau)
		\right\|_2
		&\leq
		\left\|
		\widehat{\bvec}_{0,\mathcal G_n}^{(2)}(t)
		-
		\bvec^*(t)
		\right\|_2
		+
		\left\|
		\bvec^*(t)-\bvec^*(\tau)
		\right\|_2\\
		&\qquad+
		\left\|
		\Cmat^T(\Cmat\Cmat^T)^{-1}
		\right\|_{\mathrm{op}}
		\frac{
			\|\bm\eta(\tau)-\bm\eta(t)\|_2
		}{\sqrt n}.
	\end{align*}
	Condition \ref{as:local alts} also provides
	\[
	\left\|
	\Cmat^T(\Cmat\Cmat^T)^{-1}
	\right\|_{\mathrm{op}}^2
	=
	\lambda_{\max}
	\{(\Cmat\Cmat^T)^{-1}\}
	=
	O(1)
	\]
	and
	\[
	\frac{
		\|\bm\eta(\tau)-\bm\eta(t)\|_2
	}{\sqrt n}
	\lesssim
	\sqrt{m-r+s}\,|\tau-t|.
	\]
	Together with Lemma \ref{lem:grid_truth_modulus} and
	$m-r+s\geq1$, this implies
	\[
	\left\|
	\widehat{\bvec}_{0,\mathcal G_n}^{(2)}(\tau)
	-
	\bvec^*(\tau)
	\right\|_2
	\leq
	\left\|
	\widehat{\bvec}_{0,\mathcal G_n}^{(2)}(t)
	-
	\bvec^*(t)
	\right\|_2
	+
	C_\star\sqrt{m-r+s}\,|\tau-t|.
	\]
	Since $|\tau-t|\leq\Delta_n$, taking the supremum over
	$\tau\in\mc T$ proves \eqref{eqn:grid_interpolation_bound}.
	
	Lemma \ref{lem:grid_point_rate} and
	$\sqrt n\,\Delta_n=O(1)$ yield
	\[
	\sup_{\tau\in\mc T}
	\left\|
	\widehat{\bvec}_{0,\mathcal G_n}^{(2)}(\tau)
	-
	\bvec^*(\tau)
	\right\|_2
	=
	O_p\left(
	\sqrt{\frac{m-r+s}{n}}
	\right).
	\]
\end{proof}

\begin{proof}[Proof of Corollary \ref{cor:grid_bahadur}]
	For $u\in\mc T$, write, within this proof,
	\begin{align*}
		\bm U_n(u)
		&=
		\frac1{\sqrt n}
		\sum_{i\in[n]}
		\tilde\xvec_i
		\psi_u\{y_i-\xvec_i^T\bvec^*(u)\},\\
		\bm A(u)
		&=
		[\bm J(u)]^{-1/2}
		\{\bm I-\bm P(u)\}
		[\bm J(u)]^{-1/2},\\
		\bm D(u)
		&=
		[\bm J(u)]^{-1}
		\widetilde{\bm C}^{T}
		\bm H^{-1}(u).
	\end{align*}
	By \eqref{eqn:RE.part pen bahadur}, we obtain
	\[
	\sqrt n
	\left\{
	\widehat{\bvec}_{0,\MSset}^{(2)}(u)
	-
	\bvec_{\MSset}^*(u)
	\right\}
	=
	\bm A(u)\bm U_n(u)
	-
	\bm D(u)\bm\eta(u)
	+
	\tilde{\bm r}_n(u),
	\]
	where Theorem \ref{thm:part pen bahadur} provides
	\[
	\sup_{u\in\mc T}
	\|\tilde{\bm r}_n(u)\|_2=o_p(1).
	\]
	
	We first compute the required operator bounds.  Recall that, uniformly over $u,v\in\mc T$,
	\[
	\left|
	Q_y(u\mid\xvec)-Q_y(v\mid\xvec)
	\right|
	\lesssim |u-v|
	\]
	almost surely, and hence
	\[
	\left|
	f_{y\mid\xvec}\{Q_y(u\mid\xvec)\mid\xvec\}
	-
	f_{y\mid\xvec}\{Q_y(v\mid\xvec)\mid\xvec\}
	\right|
	\lesssim |u-v|.
	\]
	It follows from Condition \ref{as:Gram lower bound} that
	\[
	\|\bm J(u)-\bm J(v)\|_{\mathrm{op}}
	\lesssim |u-v|.
	\]
	Conditions \ref{as:density conditions}(i)--(ii) and
	\ref{as:Gram lower bound} also imply that the eigenvalues of
	$\bm J(u)$ are bounded from above and away from zero uniformly
	over $u\in\mc T$ and $n$.
	
	Let $\widetilde{\bm Z}$ be the orthonormal null-space matrix used
	in the proof of Theorem \ref{thm:oracle rate}.  One has
	\[
	\bm A(u)
	=
	\widetilde{\bm Z}
	\left\{
	\widetilde{\bm Z}^{T}\bm J(u)\widetilde{\bm Z}
	\right\}^{-1}
	\widetilde{\bm Z}^{T}.
	\]
	The preceding bounds and the inverse identity therefore yield
	\[
	\sup_{u\in\mc T}\|\bm A(u)\|_{\mathrm{op}}\lesssim1,
	\qquad
	\|\bm A(u)-\bm A(v)\|_{\mathrm{op}}
	\lesssim |u-v|.
	\]
	
	For every $\bm a\in\reals^r$, $\bm D(u)\bm a$ is the unique
	minimizer of $\bm b^T\bm J(u)\bm b$ subject to
	$\widetilde{\bm C}\bm b=\bm a$.  The vector
	\[
	\widetilde{\bm C}^{T}
	(\Cmat\Cmat^T)^{-1}\bm a
	\]
	is feasible because
	$\widetilde{\bm C}\widetilde{\bm C}^{T}=\Cmat\Cmat^T$.
	Condition \ref{as:local alts} and the uniform eigenvalue bounds
	for $\bm J(u)$ consequently imply
	\[
	\sup_{u\in\mc T}\|\bm D(u)\|_{\mathrm{op}}\lesssim1.
	\]
	Moreover,
	\[
	\widetilde{\bm C}\bm D(u)=\bm I_r,
	\qquad
	\bm A(u)\widetilde{\bm C}^{T}=\bm 0,
	\qquad
	\bm A(u)\bm J(u)
	=
	\bm I-\bm D(u)\widetilde{\bm C}.
	\]
	These identities imply
	\[
	\bm D(u)-\bm D(v)
	=
	-\bm A(u)
	\{\bm J(u)-\bm J(v)\}
	\bm D(v),
	\]
	and hence
	\[
	\|\bm D(u)-\bm D(v)\|_{\mathrm{op}}
	\lesssim |u-v|.
	\]
	Condition \ref{as:local alts} now implies
	\begin{equation}\label{eqn:grid_bias_modulus}
		\left\|
		\bm D(u)\bm\eta(u)
		-
		\bm D(v)\bm\eta(v)
		\right\|_2
		\lesssim
		\left\{
		\sqrt{(m-r+s)n}
		+
		\sqrt{m-r+s}
		\right\}|u-v|.
	\end{equation}
	
	We next control the score increment.  Lemmas
	\ref{lemma:Ui uniform} and \ref{lemma:Ui equivalence}, together
	with monotonicity in $u$, give, outside one null event,
	simultaneously over $u\in\mc T$,
	\[
	\bm U_n(u)
	=
	\frac1{\sqrt n}
	\sum_{i\in[n]}
	\tilde\xvec_i\{u-1(U_i<u)\}.
	\]
	Consider the scalar functions
	\[
	(\tilde\xvec,U)
	\longmapsto
	(\bm\alpha^T\tilde\xvec)
	\left[
	(u-v)-\{1(U<u)-1(U<v)\}
	\right],
	\]
	where $u,v\in\mc T$, $|u-v|\leq\Delta_n$, and
	$\|\bm\alpha\|_2\leq1$.  Conditional uniformity of $U$ implies
	that every such function is centered and has second moment
	bounded by $C\Delta_n$, while its envelope is bounded by
	$C\xi_n$.  The covering argument in the proof of Lemma 24,
	together with Lemmas 22 and 23, of
	\citet{belloni_conditional_2019} gives uniform entropy bounded by
	\[
	C(m+s)\log(C/\varepsilon),
	\qquad
	0<\varepsilon\leq1.
	\]
	Lemma 21(3) of \citet{belloni_conditional_2019} and duality
	therefore yield
	\begin{equation}\label{eqn:grid_score_modulus}
		\sup_{\substack{u,v\in\mc T\\|u-v|\leq\Delta_n}}
		\|\bm U_n(u)-\bm U_n(v)\|_2
		=
		O_p\!\left\{
		\sqrt{(m+s)\Delta_n\log n}
		+
		\frac{(m+s)\xi_n\log n}{\sqrt n}
		\right\}.
	\end{equation}
	Condition \ref{as:Gram lower bound} gives
	\[
	\xi_n^2
	\geq
	\E\|\tilde\xvec\|_2^2
	=
	\operatorname{tr}(\bm\Sigma)
	\gtrsim m+s.
	\]
	Thus $(m+s)^3\xi_n^2=o(n^{1-\epsilon})$ implies
	$(m+s)^4=o(n^{1-\epsilon})$.  The mesh and growth conditions give
	\begin{align*}
		(m+s)\Delta_n\log n
		&=
		\{\sqrt{n(m+s)}\,\Delta_n\}
		\sqrt{\frac{m+s}{n}}\log n
		=o(1),\\
		\frac{(m+s)^2\xi_n^2\log^2n}{n}
		&=
		\frac{(m+s)^3\xi_n^2}{n^{1-\epsilon}}
		\frac{\log^2n}{(m+s)n^\epsilon}
		=o(1).
	\end{align*}
	Consequently, the right-hand side of
	\eqref{eqn:grid_score_modulus} is $o_p(1)$.  Lemma 31 of
	\citet{belloni_conditional_2019} also gives
	\[
	\sup_{u\in\mc T}\|\bm U_n(u)\|_2
	=
	O_p\{\sqrt{m+s}\}.
	\]
	
	Set $t=\pi_n(\tau)$.  Since $|t-\tau|\leq\Delta_n$,
	the preceding operator and score bounds imply
	\begin{align*}
	\sup_{\tau\in\mc T}
		\|\bm A(t)\bm U_n(t)-\bm A(\tau)\bm U_n(\tau)\|_2
	&	\leq
		\sup_{\tau\in\mc T}
		\|\bm A(t)\|_{\mathrm{op}}
		\|\bm U_n(t)-\bm U_n(\tau)\|_2+
		\sup_{\tau\in\mc T}
		\|\bm A(t)-\bm A(\tau)\|_{\mathrm{op}}
		\|\bm U_n(\tau)\|_2\\
		&=o_p(1),
	\end{align*}
	since $\sqrt{m+s}\,\Delta_n=o(1)$ holds under the mesh condition.
	Similarly, \eqref{eqn:grid_bias_modulus},
	$m-r+s\leq m+s$, and the mesh condition give
	\[
	\sup_{\tau\in\mc T}
	\|\bm D(t)\bm\eta(t)-\bm D(\tau)\bm\eta(\tau)\|_2
	=o(1).
	\]
	
	The deterministic interpolation term is controlled directly by
	Condition \ref{as:local alts} and Lemma
	\ref{lem:grid_truth_modulus}.  In fact, Condition
	\ref{as:local alts} gives
	\[
	\bm\Psi(\tau)-\bm\Psi(t)
	=
	\widetilde{\bm C}
	\{\bvec_{\MSset}^*(\tau)-\bvec_{\MSset}^*(t)\}
	-
\frac{\bm\eta(\tau)-\bm\eta(t)}{\sqrt n}.
\]
Therefore, we get
\begin{align*}
	&\bvec_{\MSset}^*(t)-\bvec_{\MSset}^*(\tau)
	+
	\widetilde{\bm C}^{T}
	(\Cmat\Cmat^T)^{-1}
	\{\bm\Psi(\tau)-\bm\Psi(t)\}\\
	&\qquad=
	\left\{
	\bm I-
	\widetilde{\bm C}^{T}
	(\Cmat\Cmat^T)^{-1}
	\widetilde{\bm C}
	\right\}
	\{\bvec_{\MSset}^*(t)-\bvec_{\MSset}^*(\tau)\}-
	\widetilde{\bm C}^{T}
	(\Cmat\Cmat^T)^{-1}
	\frac{\bm\eta(\tau)-\bm\eta(t)}{\sqrt n}.
\end{align*}
Moreover, recall that
Condition \ref{as:local alts} gives
\[
\left\|
\widetilde{\bm C}^{T}
(\Cmat\Cmat^T)^{-1}
\right\|_{\mathrm{op}}^2
=
\lambda_{\max}\{(\Cmat\Cmat^T)^{-1}\}
\lesssim1.
\]
It follows from Lemma \ref{lem:grid_truth_modulus}, Condition
\ref{as:local alts}, and $m-r+s\geq1$ that
\begin{align*}
	\sup_{\tau\in\mc T}
	\sqrt n
	\left\|
	\bvec_{\MSset}^*(t)-\bvec_{\MSset}^*(\tau)
	+
	\widetilde{\bm C}^{T}
	(\Cmat\Cmat^T)^{-1}
	\{\bm\Psi(\tau)-\bm\Psi(t)\}
	\right\|_2\lesssim
	\sqrt{n(m-r+s)}\,\Delta_n
	=o(1).
\end{align*}

The displayed growth condition implies
\[
\frac{(m+s)\xi_n^2\log^2n}{n}
=
\frac{(m+s)^3\xi_n^2}{n^{1-\epsilon}}
\frac{\log^2n}{(m+s)^2n^\epsilon}
=o(1).
\]
The mesh condition implies
$\sqrt n\,\Delta_n=O(1)$.  Lemma
\ref{lem:grid_point_rate} therefore implies that, with
probability tending to one, simultaneously for every
$t\in\mathcal G_n$,
\[
\widehat{\bvec}_{0,\mathcal G_n}^{(2)}(t)
=
\widehat{\bvec}_0^o(t).
\]
On this event, \eqref{eqn:grid_interpolated_estimator} and
\eqref{eqn:RE.oracle bahadur}, evaluated at $t=\pi_n(\tau)$,
give
\begin{align*}
	\sqrt n
	\left\{
	\widehat{\bvec}_{0,\mathcal G_n,\MSset}^{(2)}(\tau)
	-
	\bvec_{\MSset}^*(\tau)
	\right\}
	&=
	\bm A(t)\bm U_n(t)
	-
	\bm D(t)\bm\eta(t)
	+
	\bm r_{2n}(t)\\
	&\quad+
	\sqrt n
	\left[
	\bvec_{\MSset}^*(t)-\bvec_{\MSset}^*(\tau)
	+
	\widetilde{\bm C}^{T}
	(\Cmat\Cmat^T)^{-1}
	\{\bm\Psi(\tau)-\bm\Psi(t)\}
	\right].
\end{align*}
Define
\begin{align*}
	\widetilde{\bm r}_{\mathcal G,n}(\tau)
	=
	\sqrt n
	\left\{
	\widehat{\bvec}_{0,\mathcal G_n,\MSset}^{(2)}(\tau)
	-
	\bvec_{\MSset}^*(\tau)
	\right\}-
	\bm A(\tau)\bm U_n(\tau)
	+
	\bm D(\tau)\bm\eta(\tau).
\end{align*}
On the preceding event, the displayed expansion shows that
\begin{align*}
	\widetilde{\bm r}_{\mathcal G,n}(\tau)
	&=
	\bm A(t)\bm U_n(t)-\bm A(\tau)\bm U_n(\tau)\\
	&\quad-
	\{\bm D(t)\bm\eta(t)-\bm D(\tau)\bm\eta(\tau)\}
	+
	\bm r_{2n}(t)\\
	&\quad+
	\sqrt n
	\left[
	\bvec_{\MSset}^*(t)-\bvec_{\MSset}^*(\tau)
	+
	\widetilde{\bm C}^{T}
	(\Cmat\Cmat^T)^{-1}
	\{\bm\Psi(\tau)-\bm\Psi(t)\}
	\right].
\end{align*}
The preceding bounds and
$\sup_{u\in\mc T}\|\bm r_{2n}(u)\|_2=o_p(1)$ from Theorem
\ref{thm:oracle bahadur} indicate
\[
\sup_{\tau\in\mc T}
\|\widetilde{\bm r}_{\mathcal G,n}(\tau)\|_2=o_p(1).
\]
The preceding display and the definitions of $\bm A(\tau)$,
$\bm U_n(\tau)$, and $\bm D(\tau)$ prove
\eqref{eqn:RE.grid bahadur}.

Finally, subtracting \eqref{eqn:RE.part pen bahadur} from
\eqref{eqn:RE.grid bahadur} yields
\[
\sqrt n
\left\{
\widehat{\bvec}_{0,\mathcal G_n,\MSset}^{(2)}(\tau)
-
\widehat{\bvec}_{0,\MSset}^{(2)}(\tau)
\right\}
=
\widetilde{\bm r}_{\mathcal G,n}(\tau)
-
\tilde{\bm r}_n(\tau).
\]
Taking the supremum over $\tau\in\mc T$ proves
\eqref{eqn:grid_exact_equivalence}.
\end{proof}

\clearpage
\section{Additional Simulation Results} 

\subsection{QQ plots of $p$-values under a heteroskedastic model} \label{sec:sim QQ plots}
In this section, we present QQ-plots of the $p$-values from the global and pointwise inference simulations in Section \ref{sec:additional simulation}.
Figures \ref{fig:add_hetero_p150_band_qq} - \ref{fig:add_hetero_p450_point_qq} plot the quantiles of the $p$-values from each testing procedure against the quantiles of a Uniform$(0,1)$ distribution in simulations where the null hypotheses are true (i.e. $h = 0$).
We see that the distribution of $p$-values for the QLLA rank score test very nearly matches both the theoretical distribution of $p$-values and the distribution of oracle rank score test $p$-values. Both the QLLA and oracle Wald tests appear to be conservative in the global inference simulations, though the QLLA test is more so. 
In the pointwise simulations the $p$-values for both Wald tests are closer to the theoretical Uniform$(0,1)$ distribution, particularly for testing H$_{(1)}$ and H$_{(3)}$.

\begin{figure}[H]
	\centering
	\includegraphics[width=\textwidth]{"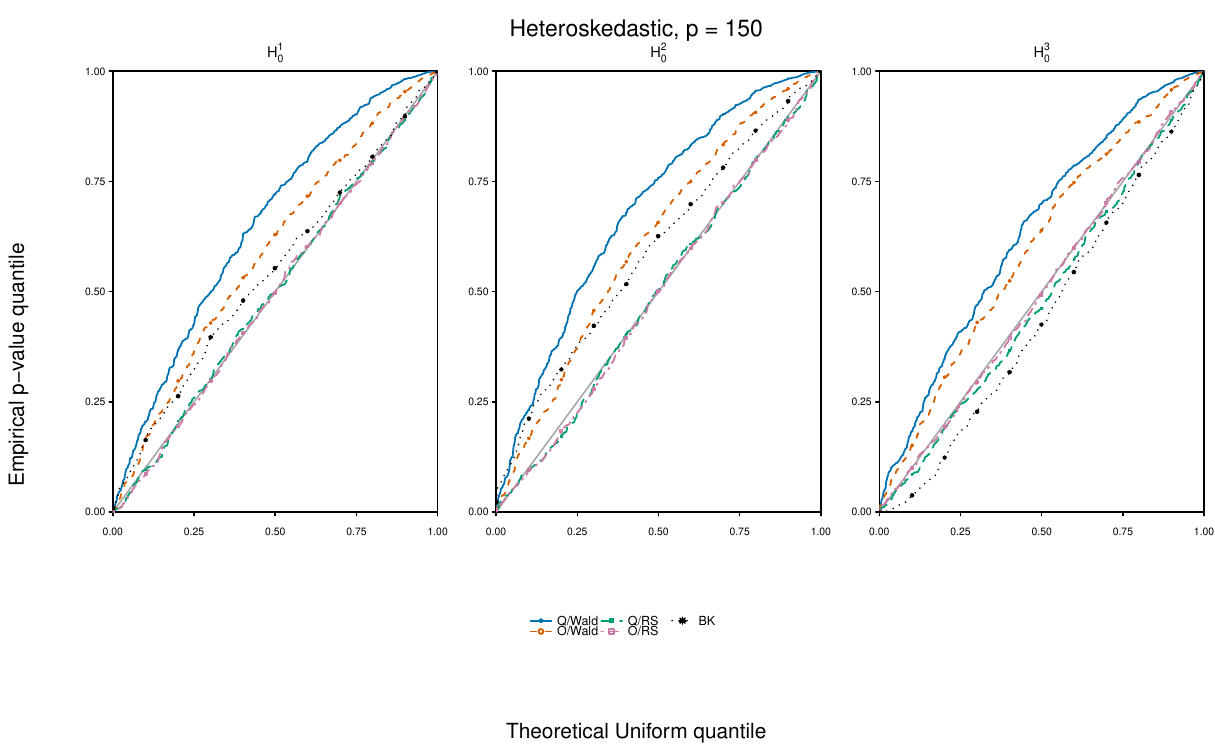"}
	\caption{Null $p$-value Q-Q plots for global inference in the
		heteroskedastic $p=150$ design.  }
	\label{fig:add_hetero_p150_band_qq}
\end{figure}

\begin{figure}[H]
	\centering
	\includegraphics[width=\textwidth]{"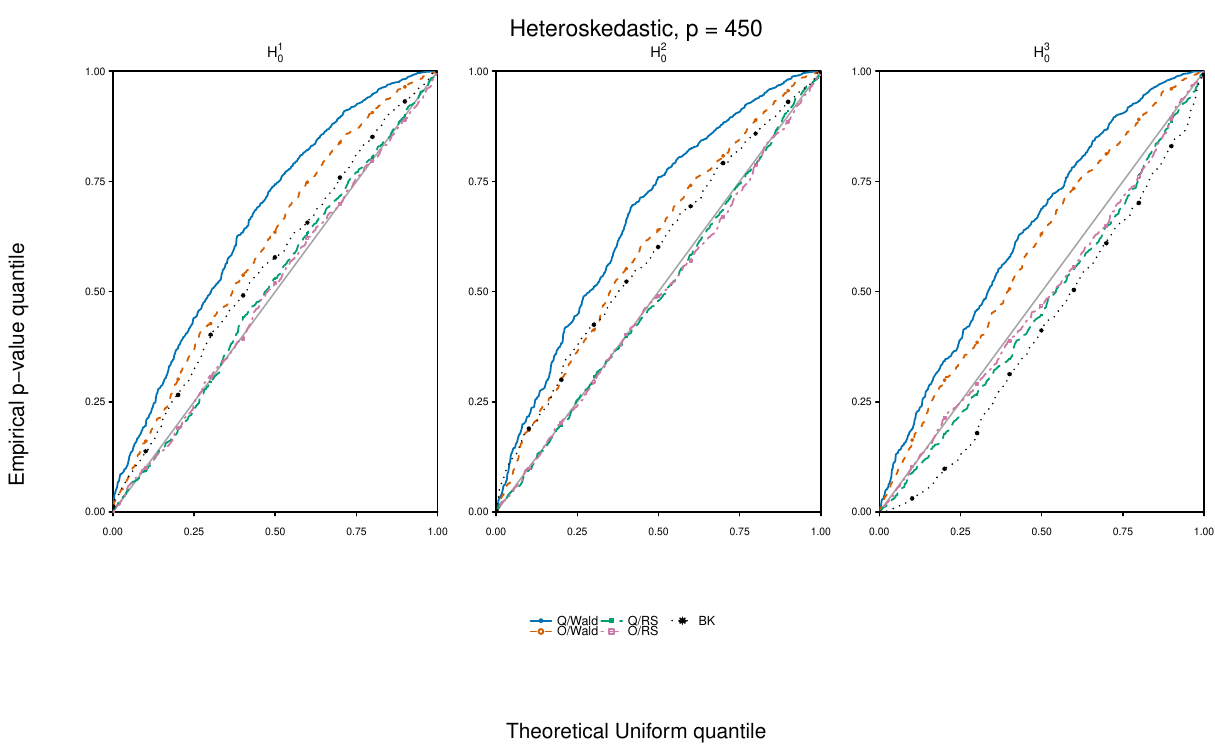"}
	\caption{Null $p$-value Q-Q plots for global inference in the
		heteroskedastic $p=450$ design.}
	\label{fig:add_hetero_p450_band_qq}
\end{figure}

\begin{figure}[H]
	\centering
	\includegraphics[width=\textwidth]{"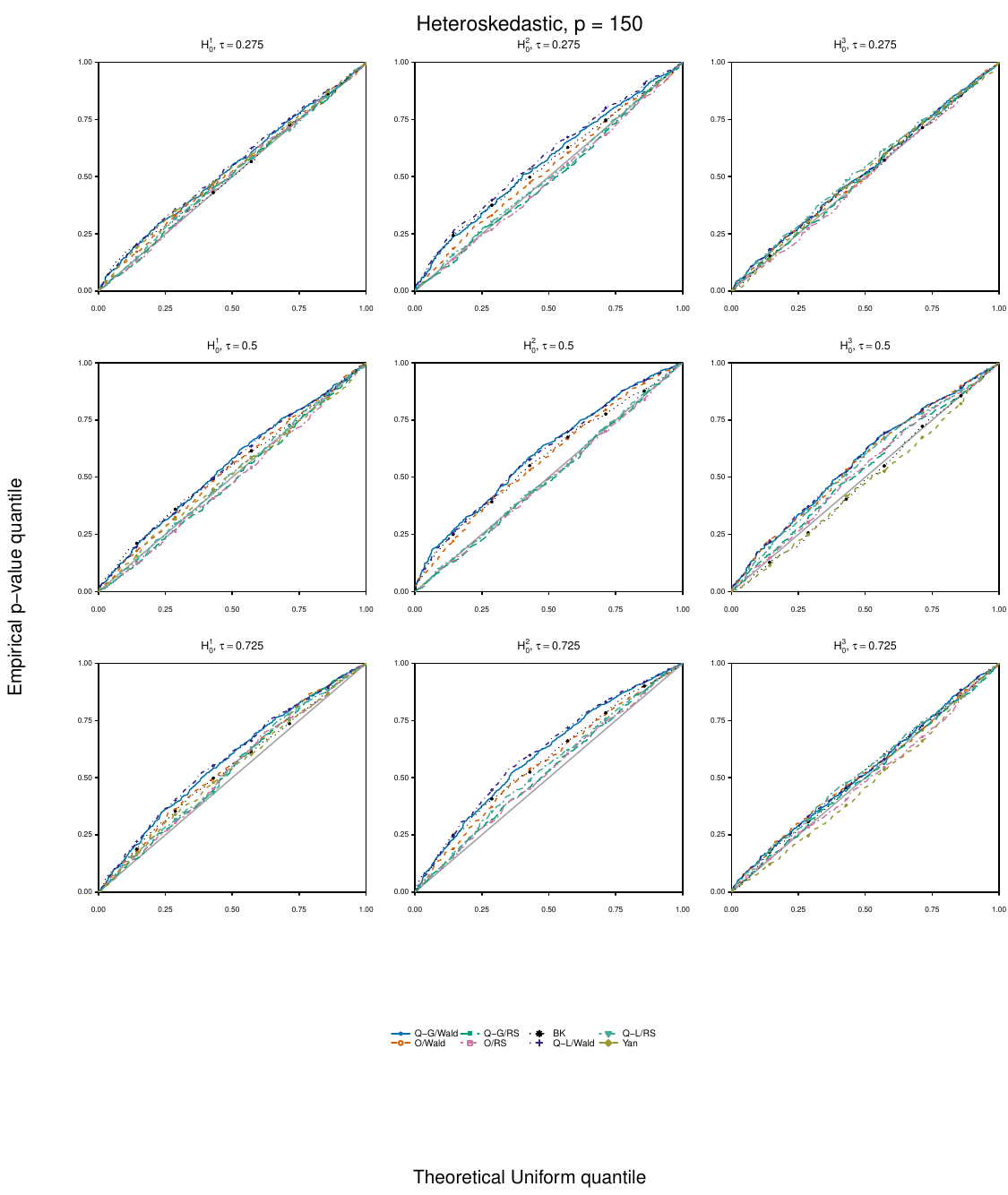"}
	\caption{Null $p$-value Q-Q plots for pointwise inference in the
		heteroskedastic $p=150$ design.}
	\label{fig:add_hetero_p150_point_qq}
\end{figure}

\begin{figure}[H]
	\centering
	\includegraphics[width=\textwidth]{"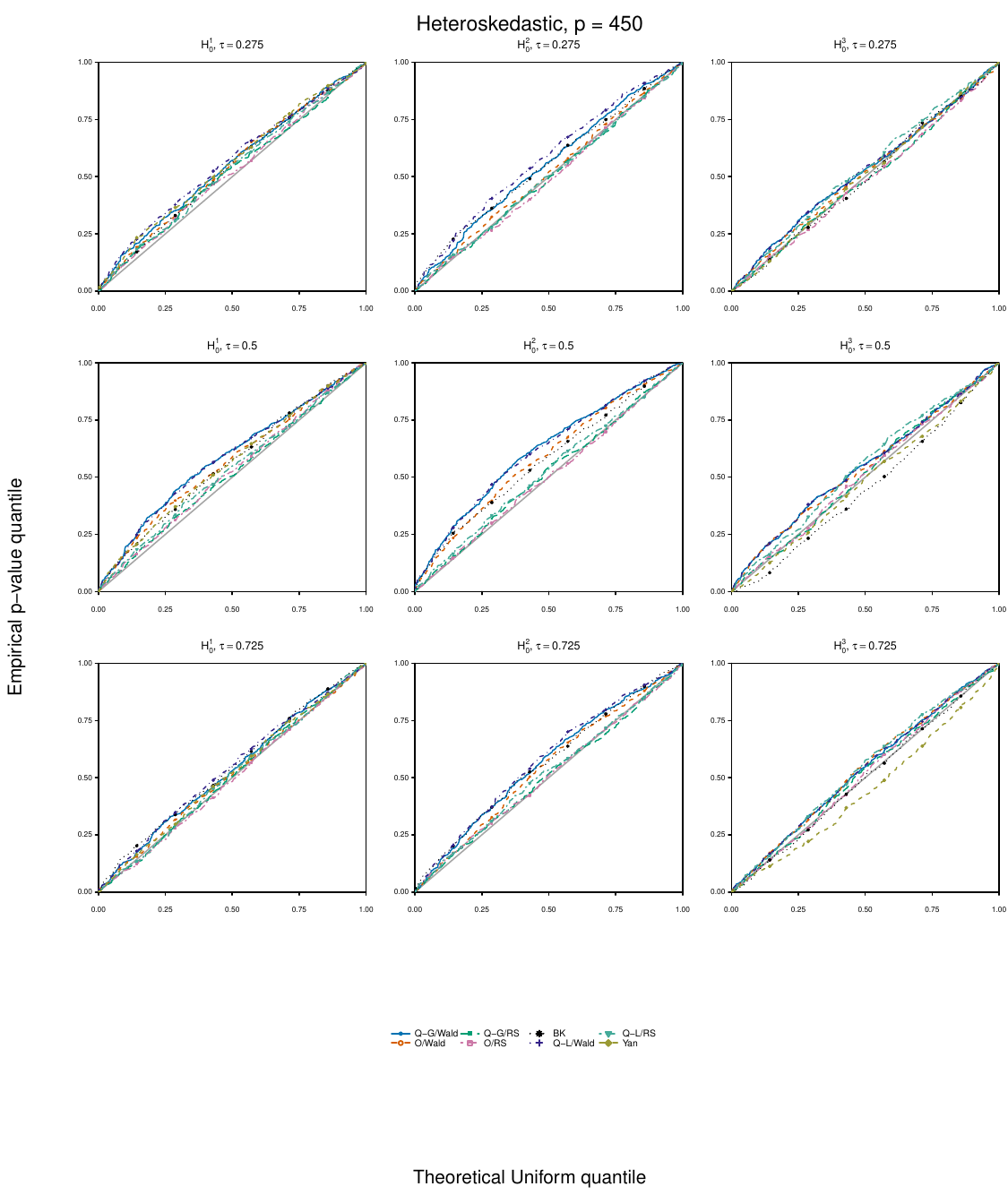"}
	\caption{Null $p$-value Q-Q plots for pointwise inference in the
		heteroskedastic $p=450$ design.}
	\label{fig:add_hetero_p450_point_qq}
\end{figure}

\subsection{Estimation comparison with homoskedastic design}\label{app:homo estimation comparison}
In this section, we conduct a estimation comparison under a location-shift model  
$$y_i = \beta_0^* + \xvec_i^T \bvec^* + \eps_i \text{\; for \;} i= 1, \ldots, n $$ 
with $\beta_0^* = 1 $, $\bvec^* = (2, 0, -2, 0, 1, 0, -1, 0, \ldots)^T$, and $\eps_i \overset{iid}{\sim} t_2$.
As in Section \ref{sec: estimation comparison}, we generate $\xvec_i \overset{iid}{\sim} N(\bm 0, \bm \Sigma_{\xvec})$ with $ [\bm \Sigma_{\xvec}]_{ij} = 0.5^{|i-j|}$ for $i,j \in [p]$ and fit the models at $80$ equally-spaced quantiles in $\mc T = [0.25, 0.75]$.
In each replication we generate $n = 200$ observations with either $p = 100$ or $p = 400$ predictors, and we perform $1000$ replications per setting.

Table \ref{tbl:homo est comp} reports the results from the homoskedastic model simulations.
The same patterns hold as in the heteroskedastic simulations, with the globally adaptive LLA estimators outperforming the alternative estimators in terms of the mean $\ell_1$ and $\ell_2$ estimation errors and the mean number of false negative variable selections, and the globally adaptive estimators outperforming their locally adaptive counterparts.

\spacingset{1}
\begin{table}[t]
	\centering
	\caption{Homoskedastic model estimation comparison} \label{tbl:homo est comp}
	\begin{tabular}[t]{l|lllll}			
		& $p$ & $\ell_1$ error & $\ell_2$ error & FP & FN\\
		\hline
		G-LLA(SCAD) & 100 & \textbf{0.565 (0.011)} & \textbf{0.137 (0.009)} & 0.026 (0.004) & \textbf{0.035 (0.008)}\\
		G-LLA(MCP) & 100 & \textbf{0.517 (0.007)} & \textbf{0.1 (0.005)} & \textbf{0.025 (0.005)} & \textbf{0.012 (0.004)}\\
		L-LLA(SCAD) & 100 & 1.032 (0.02) & 0.556 (0.02) & 0.062 (0.005) & 0.248 (0.018)\\
		L-LLA(MCP) & 100 & 0.83 (0.016) & 0.374 (0.015) & 0.048 (0.005) & 0.151 (0.013)\\
		G-Zheng & 100 & 0.978 (0.022) & 0.465 (0.019) & 0.092 (0.009) & 0.284 (0.016)\\
		L-Zheng & 100 & 1.716 (0.024) & 1.181 (0.023) & \textbf{0.018 (0.003)} & 0.877 (0.018)\\
		L1 & 100 & 4.195 (0.02) & 4.123 (0.04) & 0.066 (0.005) & 0.96 (0.02)\\
		Oracle & 100 & 0.47 (0.004) & 0.071 (0.001) & 0 (0) & 0 (0)\\
		\hline
		G-LLA(SCAD) & 400 & \textbf{1.131 (0.027)} & \textbf{0.645 (0.028)} & 0.029 (0.003) & \textbf{0.432 (0.026)}\\
		G-LLA(MCP) & 400 & \textbf{0.828 (0.022)} & \textbf{0.362 (0.02)} & 0.038 (0.005) & \textbf{0.203 (0.018)}\\
		L-LLA(SCAD) & 400 & 2.094 (0.028) & 1.722 (0.037) & 0.031 (0.003) & 1.052 (0.03)\\
		L-LLA(MCP) & 400 & 1.707 (0.027) & 1.258 (0.032) & 0.029 (0.003) & 0.78 (0.028)\\
		G-Zheng & 400 & 1.932 (0.023) & 1.372 (0.024) & 0.034 (0.006) & 1.075 (0.021)\\
		L-Zheng & 400 & 2.598 (0.021) & 2.168 (0.029) & \textbf{0.004 (0.001)} & 1.653 (0.015)\\
		L1 & 400 & 5.123 (0.019) & 6.116 (0.05) & \textbf{0.017 (0.002)} & 1.719 (0.019)\\
		Oracle & 400 & 0.472 (0.004) & 0.071 (0.001) & 0 (0) & 0 (0)\\
		\hline
	\end{tabular}
\end{table}
\spacingset{\myspacing}

\subsection{Test comparison with homoskedastic design} \label{sec:additional simulation results homoskedasticity}

In this section, we present the homoskedastic counterpart to the heteroskedastic model test comparison in Section \ref{sec:additional simulation}. We use the same dimensions, effect values, hypotheses, 81-knot grid, procedures as in those simulations, and generate the data from
\[
\xvec_i\overset{\mathrm{iid}}{\sim}N(\bm 0, \bm \Sigma_x),
\qquad [\bm \Sigma_x]_{jk}=0.5^{|j-k|},
\qquad
y_i=\beta_0^*+\xvec_i^T\bvec_h^*+\eps_i,
\qquad \eps_i\overset{\mathrm{iid}}{\sim}t_5.
\]
The conditional-quantile slopes are constant in $\tau$, and the intercept is $1+F_{t_5}^{-1}(\tau)$.

We present power comparisons for global and pointwise inference with $p \in \{150, 450\}$ as well as QQ-plots of the $p$-values under the null hypotheses (as in Section \ref{sec:sim QQ plots}).
Figures \ref{fig:add_homo_p150_band_rates} and \ref{fig:add_homo_p450_band_rates} plot the rejection rates for the global comparison and Figures \ref{fig:add_homo_p150_point_rates} and \ref{fig:add_homo_p450_point_rates} plot the rejection rates for the pointwise comparison.
Figures \ref{fig:add_homo_p150_band_qq} and \ref{fig:add_homo_p450_band_qq} present QQ-plots of the $p$-values for the global comparison and Figures \ref{fig:add_homo_p150_point_qq} and \ref{fig:add_homo_p450_point_qq} present the QQ-plots for the pointwise comparison.
The same general patterns hold in these results as in the heteroskedastic model simulations. 

\begin{figure}[H]
\centering
\includegraphics[width=\textwidth]{"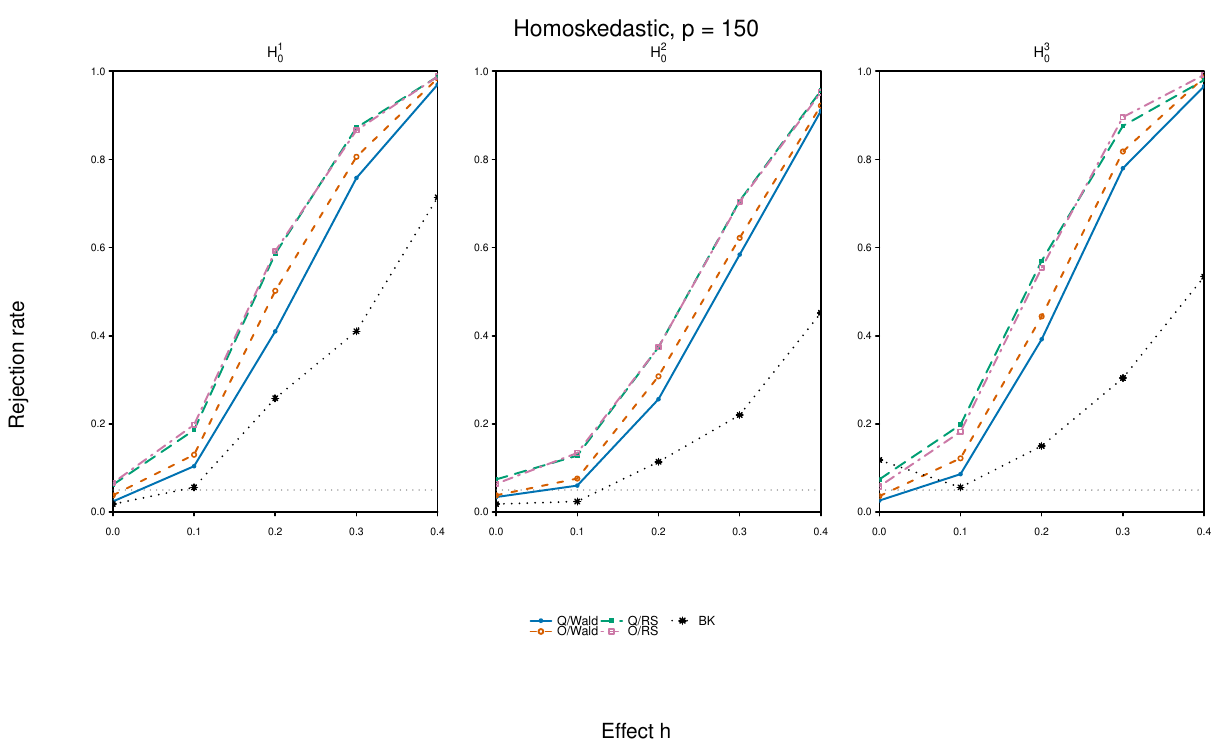"}
\caption{Global rejection rates at level $0.05$ in the
homoskedastic $p=150$ design.}
\label{fig:add_homo_p150_band_rates}
\end{figure}

\begin{figure}[H]
	\centering
	\includegraphics[width=\textwidth]{"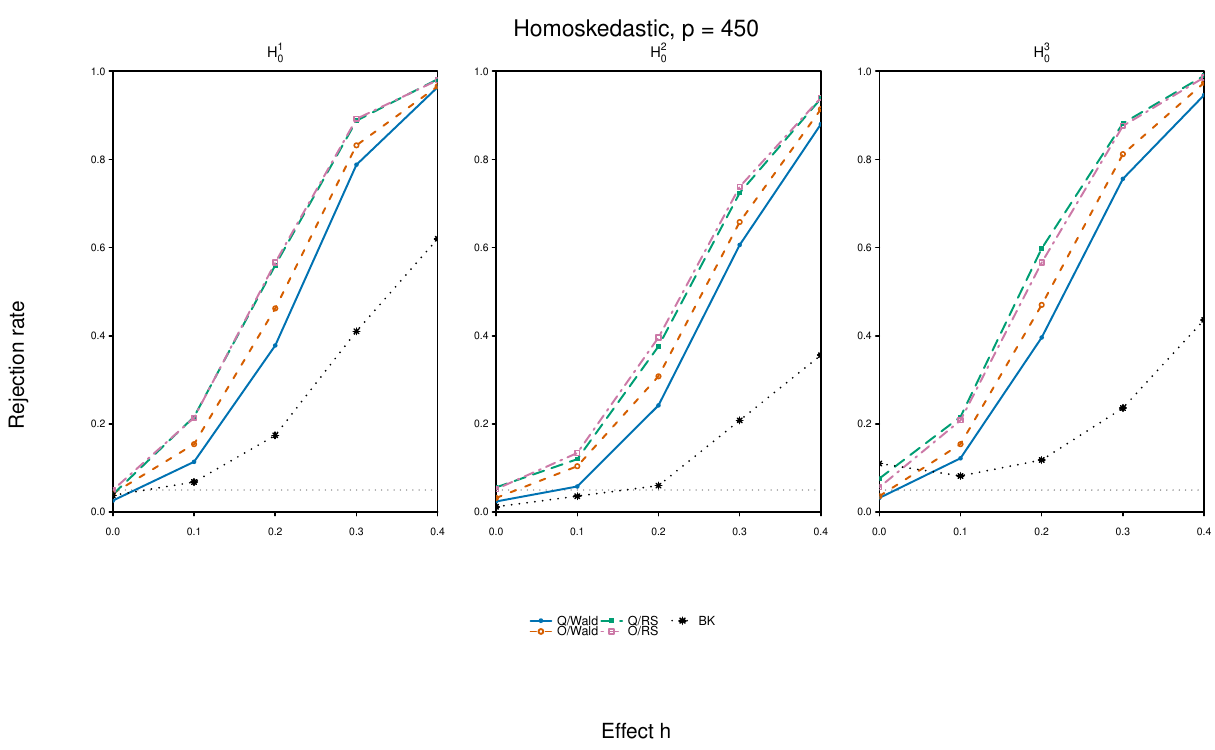"}
	\caption{Global rejection rates at level $0.05$ in the
		homoskedastic $p=450$ design.}
	\label{fig:add_homo_p450_band_rates}
\end{figure}

\begin{figure}[H]
\centering
\includegraphics[width=\textwidth]{"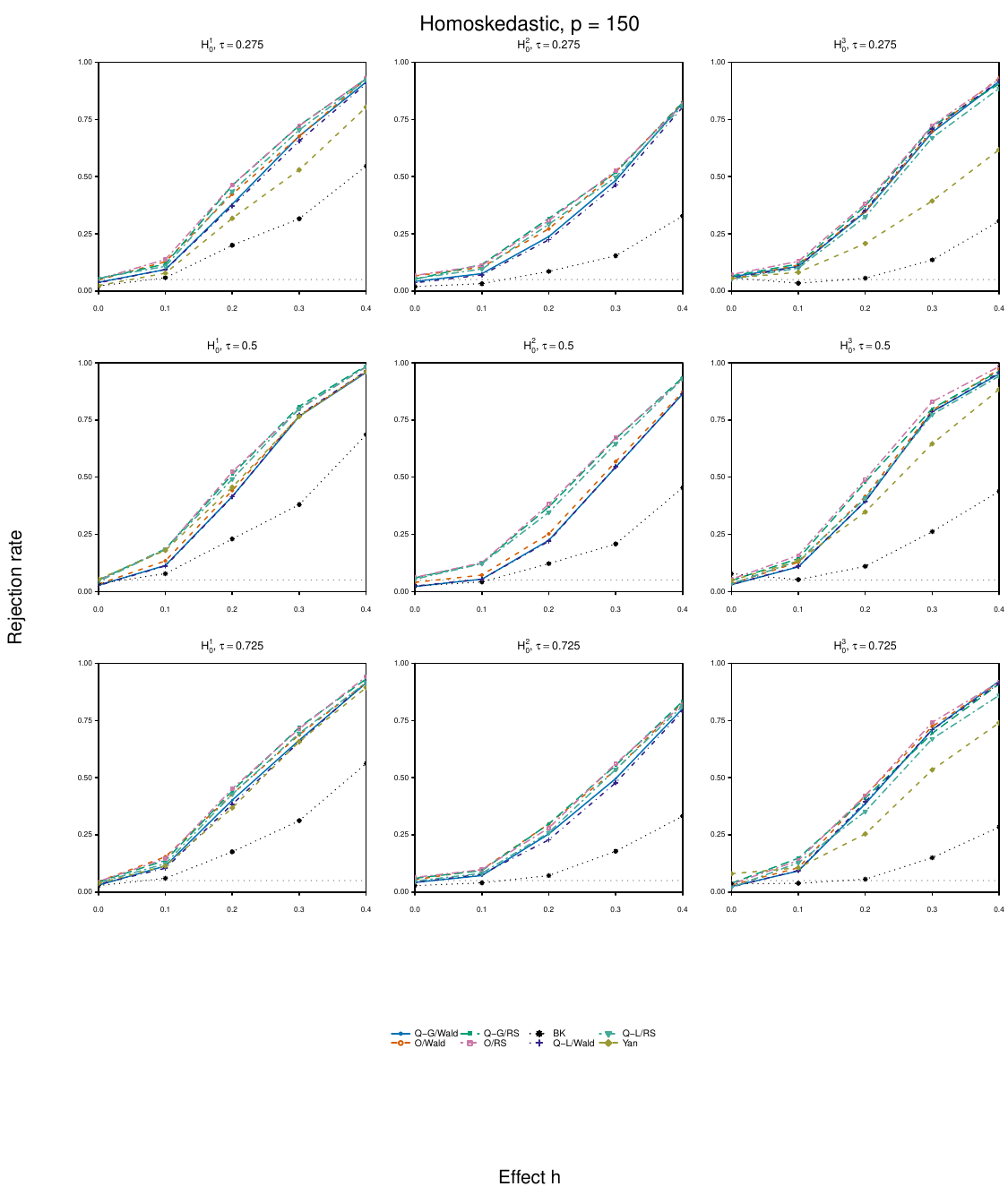"}
\caption{Pointwise rejection rates at level $0.05$ in the
homoskedastic $p=150$ design.}
\label{fig:add_homo_p150_point_rates}
\end{figure}

\begin{figure}[H]
\centering
\includegraphics[width=\textwidth]{"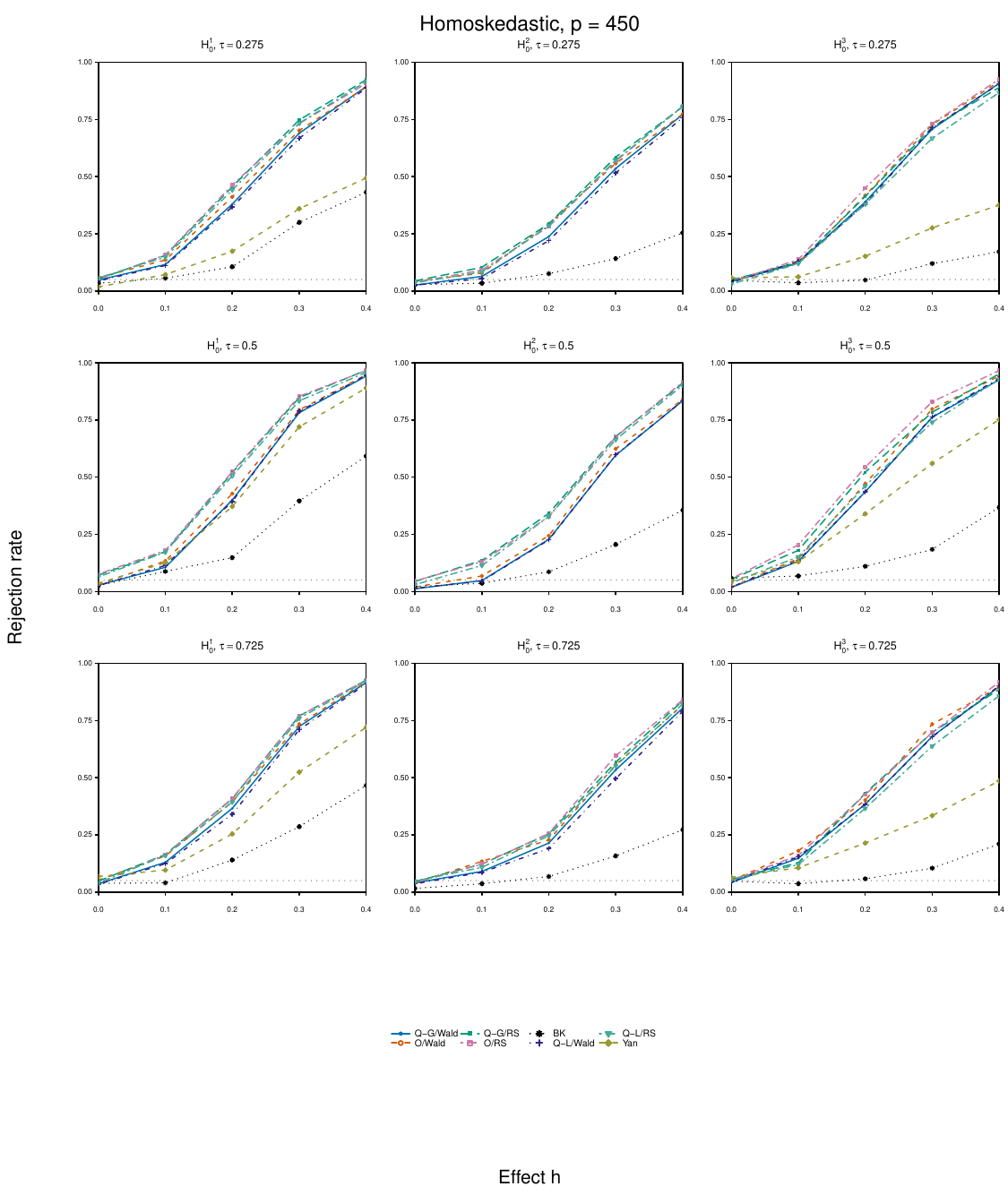"}
\caption{Pointwise rejection rates at level $0.05$ in the
homoskedastic $p=450$ design.}
\label{fig:add_homo_p450_point_rates}
\end{figure}

\begin{figure}[H]
	\centering
	\includegraphics[width=\textwidth]{"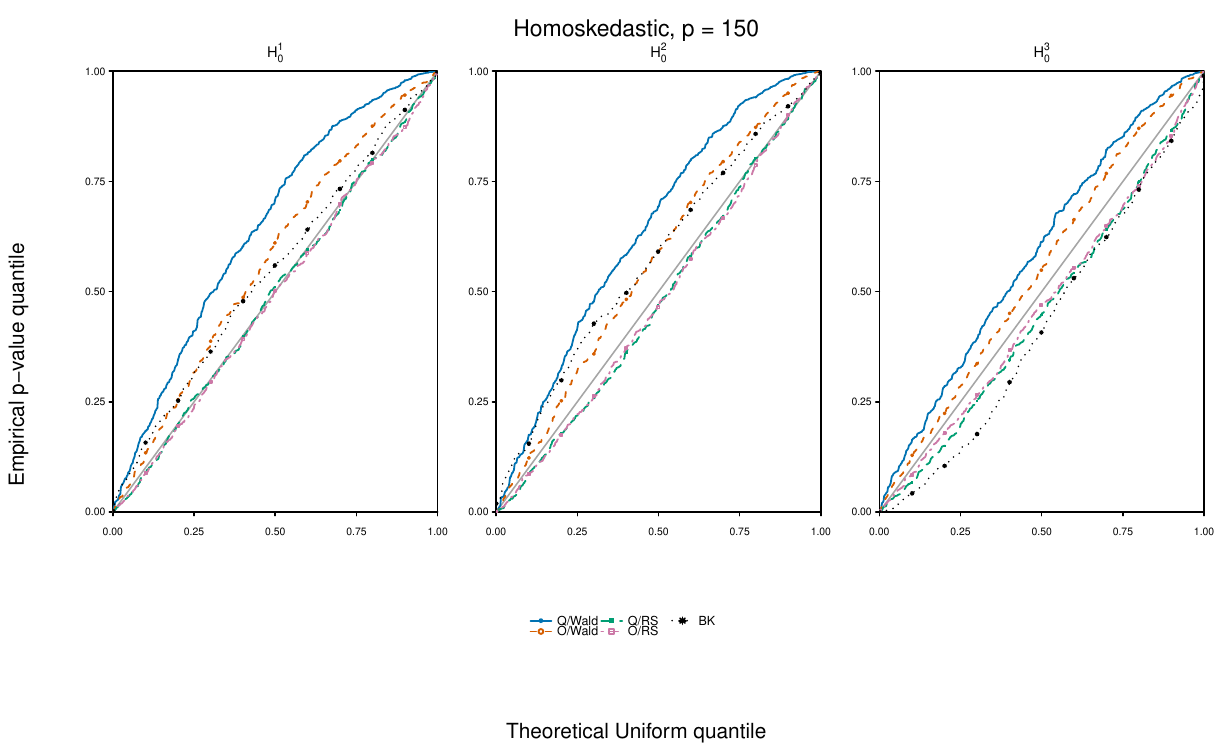"}
	\caption{Null $p$-value Q-Q plots for global inference in the
		homoskedastic $p=150$ design.}
	\label{fig:add_homo_p150_band_qq}
\end{figure}

\begin{figure}[H]
	\centering
	\includegraphics[width=\textwidth]{"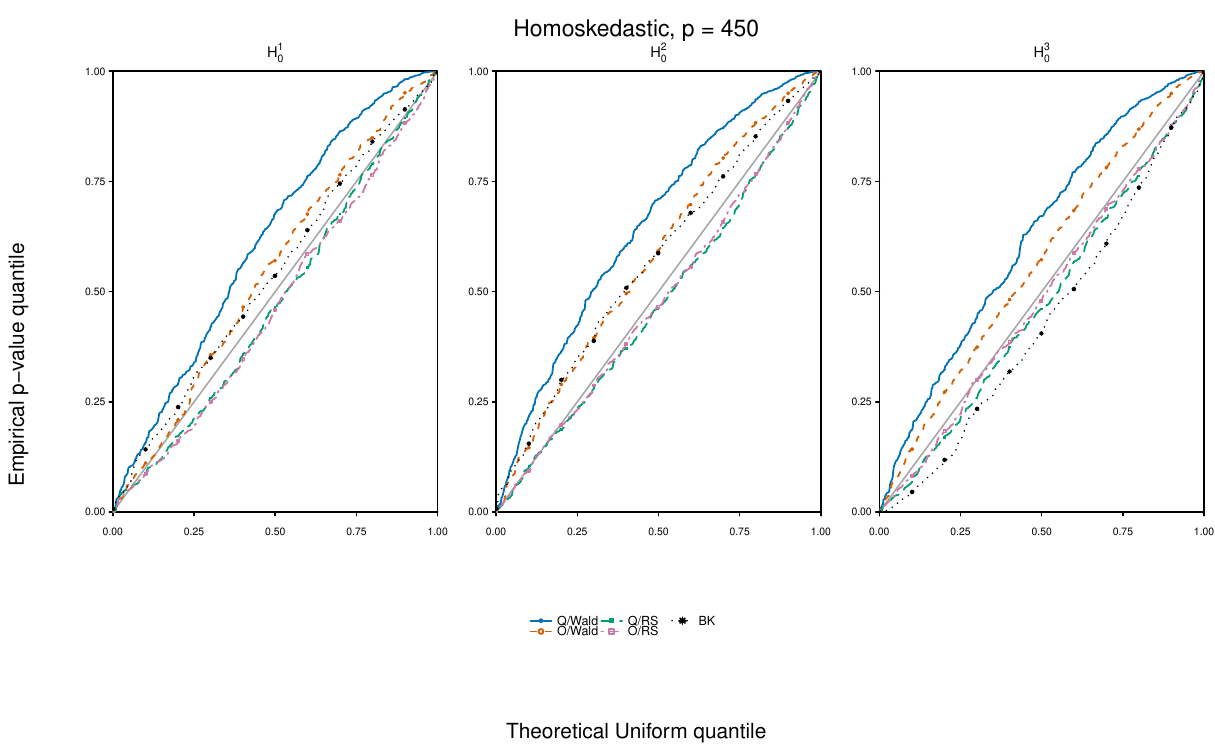"}
	\caption{Null $p$-value Q-Q plots for global inference in the
		homoskedastic $p=450$ design.}
	\label{fig:add_homo_p450_band_qq}
\end{figure}

\begin{figure}[H]
	\centering
	\includegraphics[width=\textwidth]{"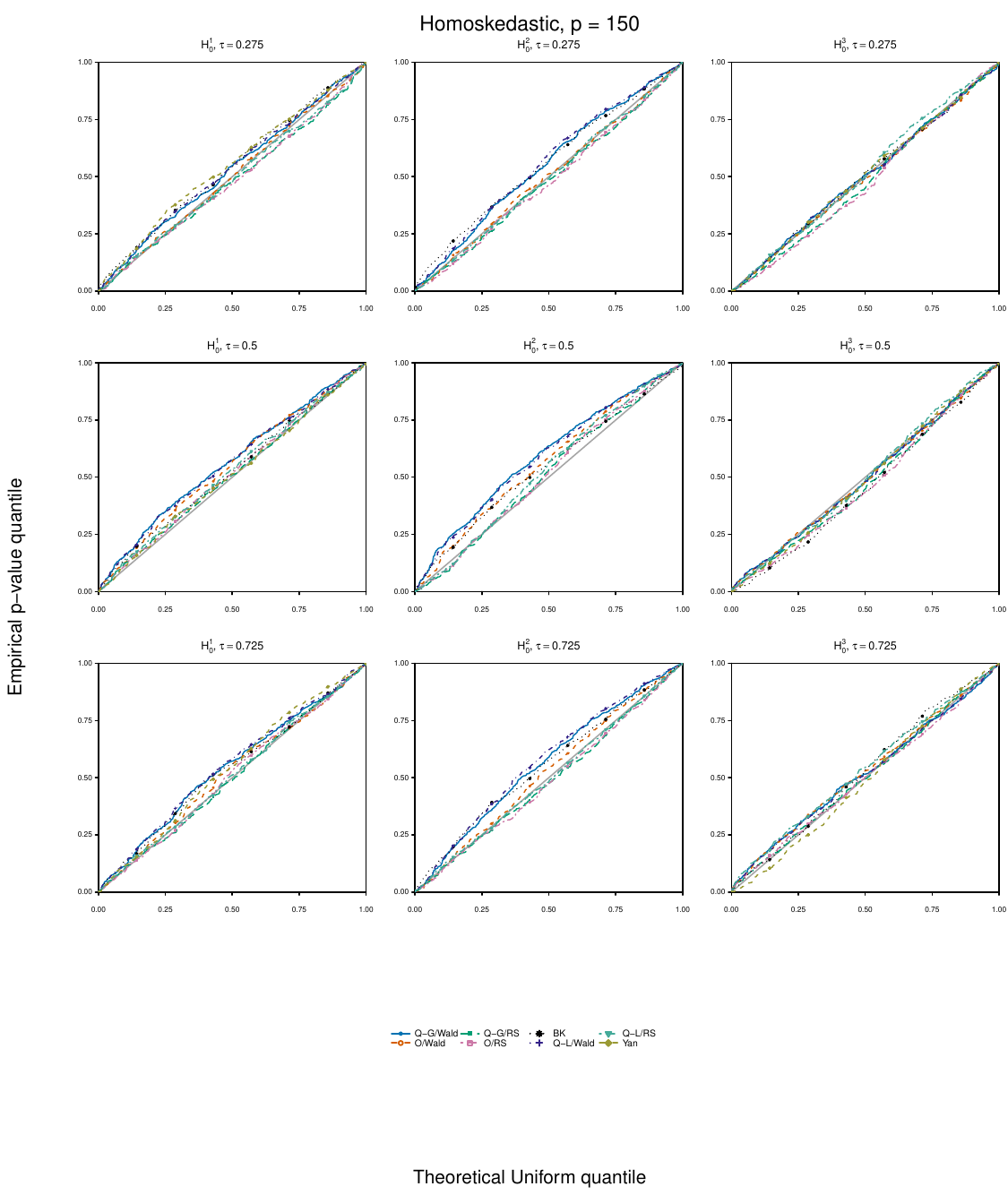"}
	\caption{Null $p$-value Q-Q plots for pointwise inference in the
		homoskedastic $p=150$ design.}
	\label{fig:add_homo_p150_point_qq}
\end{figure}

\begin{figure}[H]
	\centering
	\includegraphics[width=\textwidth]{"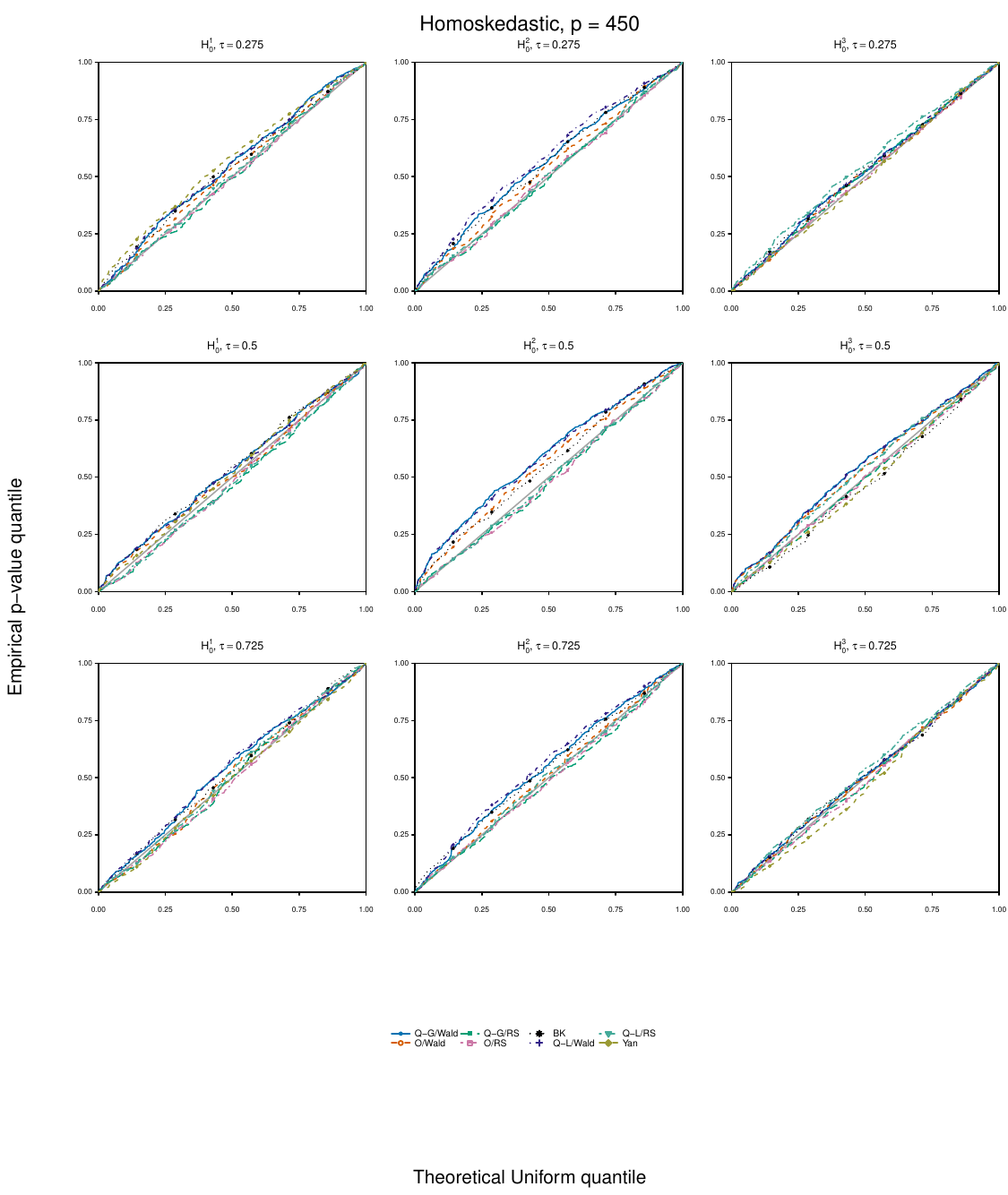"}
	\caption{Null $p$-value Q-Q plots for pointwise inference in the
		homoskedastic $p=450$ design.}
	\label{fig:add_homo_p450_point_qq}
\end{figure}

\end{document}